\PassOptionsToPackage{unicode}{hyperref}
\PassOptionsToPackage{hyphens}{url}
\PassOptionsToPackage{dvipsnames,svgnames,x11names}{xcolor}

\documentclass[12pt]{article}

\usepackage{amsmath,amssymb}
\usepackage{iftex}
\ifPDFTeX
  \usepackage[T1]{fontenc}
  \usepackage[utf8]{inputenc}
  \usepackage{textcomp}
\else
  \usepackage{unicode-math}
  \defaultfontfeatures{Scale=MatchLowercase}
  \defaultfontfeatures[\rmfamily]{Ligatures=TeX,Scale=1}
\fi
\usepackage{lmodern}
\IfFileExists{upquote.sty}{\usepackage{upquote}}{}
\IfFileExists{microtype.sty}{%
  \usepackage[]{microtype}
  \UseMicrotypeSet[protrusion]{basicmath}
}{}
\makeatletter
\@ifundefined{KOMAClassName}{%
  \IfFileExists{parskip.sty}{%
    \usepackage{parskip}
  }{%
    \setlength{\parindent}{0pt}
    \setlength{\parskip}{6pt plus 2pt minus 1pt}}
}{%
  \KOMAoptions{parskip=half}}
\makeatother
\usepackage{xcolor}
\makeatletter
\ifx\paragraph\undefined\else
  \let\oldparagraph\paragraph
  \renewcommand{\paragraph}{
    \@ifstar
      \xxxParagraphStar
      \xxxParagraphNoStar
  }
  \newcommand{\xxxParagraphStar}[1]{\oldparagraph*{#1}\mbox{}}
  \newcommand{\xxxParagraphNoStar}[1]{\oldparagraph{#1}\mbox{}}
\fi
\ifx\subparagraph\undefined\else
  \let\oldsubparagraph\subparagraph
  \renewcommand{\subparagraph}{
    \@ifstar
      \xxxSubParagraphStar
      \xxxSubParagraphNoStar
  }
  \newcommand{\xxxSubParagraphStar}[1]{\oldsubparagraph*{#1}\mbox{}}
  \newcommand{\xxxSubParagraphNoStar}[1]{\oldsubparagraph{#1}\mbox{}}
\fi
\makeatother

\usepackage{longtable,booktabs,array}
\usepackage{calc}
\usepackage{etoolbox}
\makeatletter
\patchcmd\longtable{\par}{\if@noskipsec\mbox{}\fi\par}{}{}
\makeatother
\IfFileExists{footnotehyper.sty}{\usepackage{footnotehyper}}{\usepackage{footnote}}
\makesavenoteenv{longtable}
\usepackage{graphicx,psfrag,epsf}
\makeatletter
\def\maxwidth{\ifdim\Gin@nat@width>\linewidth\linewidth\else\Gin@nat@width\fi}
\def\maxheight{\ifdim\Gin@nat@height>\textheight\textheight\else\Gin@nat@height\fi}
\makeatother
\setkeys{Gin}{width=\maxwidth,height=\maxheight,keepaspectratio}
\makeatletter
\def\fps@figure{htbp}
\makeatother

\makeatletter
\@ifpackageloaded{caption}{}{\usepackage{caption}}
\AtBeginDocument{%
\ifdefined\contentsname
  \renewcommand*\contentsname{Table of contents}
\else
  \newcommand\contentsname{Table of contents}
\fi
\ifdefined\listfigurename
  \renewcommand*\listfigurename{List of Figures}
\else
  \newcommand\listfigurename{List of Figures}
\fi
\ifdefined\listtablename
  \renewcommand*\listtablename{List of Tables}
\else
  \newcommand\listtablename{List of Tables}
\fi
\ifdefined\figurename
  \renewcommand*\figurename{Figure}
\else
  \newcommand\figurename{Figure}
\fi
\ifdefined\tablename
  \renewcommand*\tablename{Table}
\else
  \newcommand\tablename{Table}
\fi
}
\@ifpackageloaded{float}{}{\usepackage{float}}
\floatstyle{ruled}
\@ifundefined{c@chapter}{\newfloat{codelisting}{h}{lop}}{\newfloat{codelisting}{h}{lop}[chapter]}
\floatname{codelisting}{Listing}

\makeatother
\makeatletter
\@ifpackageloaded{caption}{}{\usepackage{caption}}
\makeatother

\ifLuaTeX
  \usepackage{selnolig}
\fi

\usepackage{natbib}
\usepackage{enumerate}
\usepackage{setspace}
\usepackage{url}
\usepackage{float}
\usepackage{amsthm}
\usepackage{amsfonts}
\usepackage{amsbsy}
\usepackage{mathtools}
\usepackage{mathrsfs}
\usepackage{bbm}
\usepackage{array}
\usepackage{tabularx}
\usepackage{multirow}
\usepackage{makecell}
\usepackage{threeparttable}
\usepackage{rotating}
\usepackage{placeins}
\usepackage{subfig}
\usepackage{comment}
\usepackage{wrapfig}
\usepackage{appendix}
\usepackage{changepage}
\usepackage{centernot}

\usepackage{bookmark}
\IfFileExists{xurl.sty}{\usepackage{xurl}}{}
\usepackage{hyperref}
\hypersetup{
  pdftitle={Title},
  pdfauthor={Author 1; Author 2},
  pdfkeywords={3 to 6 keywords, that do not appear in the title},
  colorlinks=true,
  linkcolor={blue},
  filecolor={Maroon},
  citecolor={Blue},
  urlcolor={Blue},
  pdfcreator={LaTeX via pandoc}}

\theoremstyle{plain}
\newtheorem{theorem}{Theorem}

\theoremstyle{definition}

\newtheorem{assumption}{Assumption}

\theoremstyle{remark}
\newtheorem{remark}{Remark}

\newcommand{\E}{\mathbbm{E}}

\newcommand{\Var}{\text{Var}}

\newcommand{\B}{\boldsymbol{\beta}}

\newcommand{\A}{\mathbf{A}}
\newcommand{\Aone}{\mathbf{A}^{(1)}_{jt}}
\newcommand{\Atwo}{\mathbf{A}^{(2,n^{(1)})}_{jt}}
\newcommand{\Atwolago}{\mathbf{A}^{(2, n^{(1)})}_{jt}}
\newcommand{\Atwodifflago}{\mathbf{A}^{(2, n^{(1)})}_{\Tilde{j}\Tilde{t}}}

\newcommand{\Yone}{Y^{(1)}_{ijt}}

\newcommand{\Ytwo}{Y^{(2,n^{(1)})}_{ijt}}
\newcommand{\Ytwolago}{Y^{(2,n^{(1)})}_{ijt}}
\newcommand{\Ytwodifflago}{Y^{(2,n^{(1)})}_{\Tilde{i}\Tilde{j}\Tilde{t}}}

\newcommand{\X}{\mathbf{X}}
\newcommand{\BoldY}{\mathbf{Y}}
\newcommand{\BarA}{\overline{\A}}
\newcommand{\Bolda}{\mathbf{a}}
\newcommand{\Bara}{\overline{\A}}
\newcommand*\BoldL{\ensuremath{\boldsymbol\ell}}
\newcommand{\BoldLone}{\boldsymbol{\ell}^{(1)}_{j}}

\newcommand{\BoldLtwo}{\boldsymbol{\ell}^{(2)}_{j}}
\newcommand{\BoldLtwolago}{\boldsymbol{\ell}^{(2)}_{j}}

\newcommand{\BoldLtwodifflago}{\boldsymbol{\ell}^{(2)}_{\Tilde{j}}}
\newcommand{\BoldZ}{\boldsymbol{Z}}

\newcommand{\PastN}{\overline{n}_{k-}}
\newcommand{\CurB}{\mathscr{B}}
\newcommand{\ConvP}{\overset{P}{\to}}
\newcommand{\ConvD}{\overset{\mathcal{D}}{\to}}
\newcommand\numberthis{\addtocounter{equation}{1}\tag{\theequation}}
\newcommand{\independent}{\perp\mkern-9.5mu\perp}

\newcommand{\anon}{1}

\begin{document}

\def\spacingset#1{\renewcommand{\baselinestretch}%
{#1}\small\normalsize} \spacingset{1}

\if1\anon
{
  \title{\bf Addressing Confounding by Indication 
  Through (Un)Measured Center Characteristics in 
  Learn-As-you-GO (LAGO)~Trials}
  \author{Minh Thu Bui\thanks{This research was supported by NIH Grant HL167936 (Drs.~Lok~and~Spiegelman). The EXTRA-CVD study was supported by NIH grant HL142099 (Drs.~Longenecker, Webel,~and~Bosworth).},
    Department of Mathematics and Statistics, 
    Boston University\\
    Chris T. Longenecker,
    Department of Global Health, University of 
    Washington\\
    Ante Bing,
    Department of Mathematics and Statistics, 
    Boston University\\
    Donna Spiegelman,
    Department of Biostatistics, Yale University\\
    Allison R. Webel,
    School of Nursing, University of Washington\\
    Hayden B. Bosworth,
    Department of Medicine, Duke University \\
    Judith J. Lok,
    Department of Mathematics and Statistics, 
    Boston University}
  \maketitle
} \fi
\if0\anon
{
  \bigskip
  \bigskip
  \bigskip
  \begin{center}
    {\large\bf Addressing Confounding by Indication 
    Through (Un)Measured Center~Characteristics in 
    Learn-As-you-GO (LAGO) Trials}
  \end{center}
  \medskip
} \fi

\bigskip

\begin{abstract}
The Learn-As-you-Go (LAGO) design is an adaptive clinical trial design allowing modifications to multicomponent intervention packages across stages. Centers participate in more than one stage, as is common in large-scale implementation trials. In LAGO trials, center characteristics may act as confounders, predicting both the intervention package and the outcomes. We extend LAGO theory by introducing fixed center effects to control for confounding by indication through measured and unmeasured center characteristics. Conditioning on center characteristics by including fixed center effects ensures asymptotic results hold without requiring explicit characterization of unmeasured confounders. Our methods apply even with small numbers of centers. LAGO theory is established for continuous outcomes following a generalized linear model and binary outcomes following a logistic regression model, unifying theory across outcome types. Point- and interval estimators are derived, and consistency and asymptotic normality are established. Valid hypothesis tests for the overall intervention effect are provided, and the optimal intervention package minimizing cost subject to a target outcome mean is obtained via constrained optimization.\\\\
\noindent {\it Keywords:}
unmeasured confounding, implementation science, adaptive designs.
\par
\end{abstract}

\def\thefigure{\arabic{figure}}
\def\thetable{\arabic{table}}

\renewcommand{\theequation}{\thesection.\arabic{equation}}
\spacingset{1.5}
\section{Introduction}
Adaptive clinical trials offer greater flexibility than traditional designs. The FDA defines adaptive design as pre-planned modifications based on accumulating data while maintaining study integrity \citep{FDA}. \citet{pallmann2018adaptive} highlight key advantages: more efficient use of resources, fewer required participants, and flexibility across trial stages. Though existing adaptive designs address modifications such as sample size re-estimation, participant population, treatment arm selection, and others, none allow changes to the intervention packages. This gap motivated the development of Learn-As-you-GO (LAGO), an adaptive trial design that allows modifications to multicomponent intervention packages over different stages.

LAGO trials have $K > 1$ stages. The stage 1 recommended intervention package is determined prior to the trial. After stage 1, the data collected so far are used to determine the recommended intervention package for stage~2. This process repeats, and the recommended intervention package for the final stage is determined using the data from all previous stages. One of the goals of LAGO is to estimate the optimal intervention package, for example, minimizing the cost of the package while attaining a pre-specified mean outcome. Another goal of LAGO is to estimate the effect of the final estimated optimal intervention package and its components.

\citet{nevo2021analysis} established the first theoretical LAGO framework, modeling binary outcomes with logistic regression. \citet{bing2023learn} extended this theory to continuous outcomes using a generalized linear model (GLM), and \citet{bing2025learn} proposed unconditional and conditional power approaches for LAGO optimizations. In LAGO, recommended interventions converge in probability to a fixed package as stage participation increases. Although these previous works established consistency and asymptotic normality, they did not account for confounding by indication or allow centers to participate in multiple stages.

This motivates extending LAGO to accommodate confounding by indication due to measured and unmeasured center characteristics. Specifically, we introduce three additional extensions to the LAGO framework. First, time-invariant measured and unmeasured center characteristics may confound the relationship between intervention package implemented and outcomes in the centers. Second, centers may participate in multiple stages, reflecting the design of many large-scale implementation trials. Third, the participant-level error terms in the model described in Section~\ref{notations} are assumed to be independent, though not necessarily identically distributed. These differences pose statistical challenges. 

We establish a LAGO framework using fixed center effects to account for confounding by measured and unmeasured time-invariant center characteristics. Fixed effects modeling is well-established for controlling unmeasured time-invariant confounding \citep{hausman1981panel, allison2009fixed}, and has been shown to perform comparably to or better than modern causal inference methods \citep{spiegelman2018evaluating}. Consistency and asymptotic normality of the intervention effect estimators hold without requiring explicit characterization of unmeasured confounders, as the limiting distribution depends only on the fixed center effect parameters. Unlike existing work on LAGO, which assumes no unmeasured confounding, and in contrast to random effects approaches that require the number of centers $J \rightarrow \infty$, our fixed effects approach addresses unmeasured confounding by indication and applies also with small numbers of centers.

Two trial designs often compared to LAGO are Sequential Multiple Assignment Randomized Trials (SMART) \citep{almirall2014introduction} and Multiphase Optimization Strategy Trials (MOST) \citep{collins2007multiphase}. SMART aims to optimize treatment sequences within participants through sequential randomization \citep{almirall2014introduction}, whereas LAGO optimizes intervention packages across stages with new participants at each stage. LAGO shares MOST's goal of identifying the optimal intervention package, but combines data from all stages in the final analysis rather than relying solely on a separate evaluation phase.  \citet{beidas2022promises} describes LAGO as a practical alternative in implementation science, allowing optimization of implementation strategies based on demonstrated needs in ongoing trials.

This article addresses confounding by indication arising from measured and unmeasured center characteristics in LAGO trials. Section~\ref{notations} introduces the notation, assumptions, proofs for continuous outcomes, and methods for hypothesis testing, confidence sets, and confidence bands, with extensions to binary outcomes and general GLMs provided in Appendices~\ref{proofs_binary} and~\ref{proofs_glm_general_link}. Section~\ref{simulations} presents simulation results. Section~\ref{real_data} applies the proposed methods to the A Nurse-Led Intervention to Extend the HIV Treatment Cascade for Cardiovascular Disease Intervention (EXTRA-CVD) study \citep{longenecker2024nurse}. Section~\ref{conclusions} concludes with a discussion and future research directions.

\section{LAGO Notation and Assumptions}\label{notations}

Let $K > 1$ be the number of stages, $J$ the number of centers, and $J_k$ the number of centers at stage $k$. An intervention package $\Bolda \in \mathbb{R}^P$ has $P$ components, with known cost $C(\Bolda)$, say in dollars. The optimal intervention minimizes $C(\Bolda)$ subject to the mean outcome reaching a prespecified goal $\theta$. Let $\BoldZ_j$ denote measured and unmeasured characteristics of center $j$ (e.g.\ yearly participant volume, staff-to-participant ratio), with $\overline{\BoldZ} = (\BoldZ_1^\top, \ldots, \BoldZ_J^\top)^\top$. These center characteristics may be predictive of both the actual intervention and the outcome. Let $T^{(1)}$ and $T^{(2)}$ denote the number of visit times in stages 1 and 2, respectively, and $t = 1,\ldots, T^{(1)}$ visits in stage 1 and $t = T^{(1)}+1, \ldots, T^{(2)}$ visits in stage 2.  Let $Y_{ijt}^{(k)}(\Bolda)$ be the counterfactual outcome for participant $i$ in center $j$ during visit $t$ at stage $k$ under intervention $\Bolda$, $\boldsymbol{Y}_j^{(k)}(\Bolda)$ the outcomes of all participants and all visits in center $j$ in stage $k$, and $\overline{\boldsymbol{Y}}^{(k)}(\Bolda)$ the outcomes of in stage $k$ under intervention $\Bolda$. Let $\mathbbm{1}[\cdot]$ denote the indicator function. The main text focuses on $K = 2$; Appendix~\ref{extensionK2} covers $K > 2$.

Let $n_{jt}^{(k)}$ the number of participants in center $j$ during visit $t$ in stage $k$, $n_{j}^{(k)}$ be the~number of participants in center~$j$ in stage~$k$, and $n^{(k)} = \sum_{j=1}^{J^{(k)}} n_{j}^{(k)}$ the total number of participants in stage~$k$. 

\begin{assumption}\label{ratioConvergence}
The ratio of the number of participants during visit $t$ in each center $j$ in stage $k$, $n_{jt}^{(k)}$, to the total sample size $n$, $n_{jt}^{(k)} / n$, converges to a constant $\alpha_{jt}^{(k)} = \lim_{n\rightarrow \infty} n_{jt}^{(k)} / n$ as $n \rightarrow \infty$.
\end{assumption}

Consider the following model for the counterfactual outcomes of participant $i$ in center $j$ during visit $t$ with center characteristics $\BoldZ_j$ under intervention package $\Bolda$:
\begin{align}\label{mainmodelAssp1}
    Y_{ijt}^{(k)}\left(\Bolda \right) = \beta_0 + \boldsymbol{\beta}_A^{T}\Bolda +  \boldsymbol{\beta}_z^{\top} \BoldZ_j + \eta \cdot \mathbbm{1}[k = 2] + b_j + \epsilon_{ijt}^{(k)}(\Bolda),
\end{align}
with $i = 1,..., n_{jt}^{(k)}$, $j = 1,..., J^{(k)}$, and $t = 1,.\ldots, T^{(1)}$ for a visit in stage 1 and $t = T^{(1)}+1, \ldots, T^{(2)}$ for a visit in stage 2. Here, $b_j$ is a random center effect of center $j$, independent and identically distributed (i.i.d.) with $\E(b_j) = 0$ and $\Var(b_j) = \sigma^2_{b}$; $\epsilon_{ijt}^{(k)}(\Bolda)$ is the participant-level error term, which has mean 0 for each intervention $\Bolda$ given the center characteristics~$\BoldZ_j$. 

A more general version of the model from equation (\ref{mainmodelAssp1}) is
\begin{equation}\label{eqn:main}
    Y_{ijt}^{(k)}\left(\Bolda \right) = \boldsymbol{\beta}_{A}^{\top} \Bolda + \gamma_j + \eta \cdot \mathbbm{1}[k = 2]+ \epsilon_{ijt}^{(k)}(\Bolda) ,
\end{equation} 
with $i = 1,..., n_{jt}^{(k)}$, $j = 1,..., J^{(k)}$, and $t = 1,.\ldots, T^{(1)}$ for a visit in stage 1 and $t = T^{(1)}+1, \ldots, T^{(2)}$ for a visit in stage 2. Here, $\gamma_j =  \beta_0 + b_j + \boldsymbol{\beta}_z^{T} \cdot \BoldZ_j$, which is henceforth considered as a fixed center effect. Including $\gamma_j$ in the outcome model effectively conditions on the random center effects $b_j$. Therefore, throughout the remainder of this work, we treat the $\gamma_j$ as fixed and, when model (\ref{mainmodelAssp1}) holds, condition on the realized values of $b_j$. Our main interest lies in the vector $\boldsymbol{\beta}_A$ and the null hypothesis of no intervention effect, $H_0: \boldsymbol{\beta}_A = \mathbf{0}$. 

The model of equation \eqref{eqn:main} leads to the following assumptions:
\begin{assumption}[Main Model]\label{mainmodel}
Define  
\begin{align*}
    \BoldL_{j}^{(k)} = \begin{pmatrix} \mathbbm{1}[j=1] \\\mathbbm{1}[j=2] \\ \vdots \\ \mathbbm{1}[j=J] \\ \mathbbm{1}[k=2] \end{pmatrix} \; \; \text{and} \; \; \boldsymbol{\beta}^{T} =  \left(\boldsymbol{\beta}^{T}_A, \boldsymbol{\gamma}^{T}, \eta  \right) \; \; \text{with} \; \; \boldsymbol{\gamma}^{T} = \left(\gamma_1, \ldots, \gamma_J\right).
\end{align*}

Assume that for all $\Bolda$,

\begin{equation}\label{eqn:main_matrixform}
    Y_{ijt}^{(k)}\left(\Bolda \right) = \B^{T} \begin{pmatrix} \Bolda \\  \BoldL_{j}^{(k)} \end{pmatrix} + \epsilon_{ijt}^{(k)}(\Bolda),
\end{equation}
with $\E [ \epsilon_{ijt}^{(k)}(\Bolda) \mid \overline{\BoldZ} ] = 0$, is a correctly specified model with $\B^{*T} = \left(\B^{*T}_A, \boldsymbol{\gamma}^{*T},\eta^{*} \right)$ the true parameters.
\end{assumption}

\begin{assumption}\label{asspContributeMoreThanOnce}
Every center $j$ contributes to more than one stage $k$.
\end{assumption}

\noindent Assumption \ref{asspContributeMoreThanOnce} is needed for the fixed effects $\gamma_j$ to be identified because $\B_A^{\top}\Bolda$ can be absorbed into $\gamma_j$ unless the center implements different intervention packages across stages, analogously to how a fixed individual effect in panel regression cannot be identified without repeated observations \citep{hausman1981panel}. 

 \begin{assumption}[Consistency] \label{asspConsistency} If center $j$ during visit $t$ in stage $k$ does not receive any intervention, $\A_{jt}^{(k)} = \boldsymbol{0}$, $Y_{ijt}^{(k)} = Y_{ijt}^{(k)}(\boldsymbol{0})$. If center $j$ during visit $t$ in stage $k$ receives intervention $\A_{jt}^{(k)} = \Bolda$, $Y_{ijt}^{(k)} = Y_{ijt}^{(k)}(\Bolda)$.
\end{assumption}

One goal of LAGO is to estimate the optimal intervention package. Assume $\theta$ is the prespecified outcome goal. For example, the optimal intervention package solves the optimization problem so that, on average, the mean outcome under a constant intervention $\boldsymbol{x}$ is at least $\theta$, while minimizing cost

\begin{align}\label{optimizationCost}
\underset{\boldsymbol{x}}{\text{min}} \; C(\boldsymbol{x}) \quad \text{subject to} \quad \frac{1}{n}\sum_{j=1}^{J} \sum_{t=1}^{T} n_{jt} \, \mathbb{E}[Y_{ijt}(\boldsymbol{x})| \boldsymbol{\beta}] \geq \theta,
\end{align}

with each component $p = 1,\ldots, P$ of $\boldsymbol{x}$ in the interval $[L_p, U_p]$.

This article considers a linear and a cubic cost function, $C(\boldsymbol{x})$. We assume that there is a unique solution to the optimization problem of equation (\ref{optimizationCost}), depending continuously on $\B$ in an open neighborhood of $\B^{*}$. Details on the formulation and computational implementation of this optimization problem under linear and cubic cost functions are provided in Nevo et al. (\citeyear{nevo2021analysis}) and Bing et al. (\citeyear{bing2023learn}).

At stage 1, investigators predetermine a recommended intervention $\boldsymbol{x}^{(1)}$ 
(or $\boldsymbol{x}_j^{(1)}$ per center). The actual intervention at center 
$j$ during visit $t$ is given by $\A_{jt}^{(1)} = h_{jt}^{(1)}(\boldsymbol{x}^{(1)})$, 
where $h_{jt}^{(1)}$ is a deterministic continuous function capturing center- and 
time-specific deviations from the recommended intervention. Let $Y_{ijt}^{(1)}$ denote 
the outcome of participant $i$ in center $j$ during visit $t$, and let 
$\boldsymbol{Y}_{jt}^{(1)} = (Y_{1jt}^{(1)}, \ldots, Y_{n_{jt}^{(1)}jt}^{(1)})$ 
collect the outcomes across all participants in center $j$ during visit $t$. 
The stage-1 interventions, outcomes, and errors aggregated across centers 
at each visit $t$ are $\overline{\A}_t^{(1)} = (\A_{1t}^{(1)}, \ldots, \A_{J^{(1)}t}^{(1)})$, $\overline{\boldsymbol{Y}}_t^{(1)} = (\boldsymbol{Y}_{1t}^{(1)}, \ldots, \boldsymbol{Y}_{J^{(1)}t}^{(1)})$, and $\overline{\boldsymbol{\epsilon}}_t^{(1)} =(\boldsymbol{\epsilon}_{1t}^{(1)}, \ldots, \boldsymbol{\epsilon}_{J^{(1)}t}^{(1)})$. We write
$\overline{\A}^{(1)} = \bigl(\overline{\A}_1^{(1)}, \ldots, \overline{\A}_{T^{(1)}}^{(1)}\bigr)$, $\overline{\boldsymbol{Y}}^{(1)} = \bigl(\overline{\boldsymbol{Y}}_1^{(1)}, \ldots, \overline{\boldsymbol{Y}}_{T^{(1)}}^{(1)}\bigr)$, and $\overline{\boldsymbol{\epsilon}}^{(1)} = \bigl(\overline{\boldsymbol{\epsilon}}_1^{(1)}, \ldots, \overline{\boldsymbol{\epsilon}}_{T^{(1)}}^{(1)}\bigr)$. Center characteristics $\overline{\BoldZ} = (\BoldZ_1, \ldots, \BoldZ_J)$ are assumed time-invariant and thus require no further temporal indexing. We use subscript $(-j)$ to denote all centers except $j$; for example, $\overline{\A}_{(-j)}^{(1)}$ denotes the interventions implemented across all stage~1 centers except center $j$. Although $\boldsymbol{x}^{(1)}$ may depend on $j$, for ease of exposition, we henceforth assume it does not.

After stage 1, the stage 1 data are analyzed to determine the recommended intervention for stage 2. The recommended intervention for stage 2, denoted $\widehat{\boldsymbol{x}}^{\text{opt},(2,n^{(1)})}$, is determined by a function $g$ that takes stage 1 outcomes and intervention packages as inputs: $\widehat{\boldsymbol{x}}^{\text{opt},(2,n^{(1)})} = g \bigl(\BarA^{(1)}, \overline{\boldsymbol{Y}}^{(1)}\bigr)$. For example, the recommended intervention $\widehat{\boldsymbol{x}}^{\text{opt},(2,n^{(1)})}$ solves the optimization problem from equation (\ref{optimizationCost}) using the estimated $\boldsymbol{\widehat{\beta}}^{(1)}$ based on the stage 1 data instead of the true $\B$, that is, it solves
\begin{align}
\underset{\boldsymbol{x}}{\text{min}} \; C(\boldsymbol{x}) \quad \text{subject to} \quad \frac{1}{n}\sum_{j=1}^{J} \sum_{t}^{T} n_{jt} \, \mathbb{E}[Y_{ijt}(\boldsymbol{x})|\boldsymbol{\widehat{\beta}}^{(1)}] \geq \theta,
\end{align}
where $\widehat{\B}^{(1)T} = (\widehat{\B}^{(1)T}_A, \widehat{\boldsymbol{\gamma}}^{(1)T}, 0)$ is the stage 1 estimate of $\B^{*T}$, with $\widehat{\eta} = 0$ since no stage 2 data are available yet.

As in \citet{nevo2021analysis} and \citet{bing2023learn}, $\X^{(2,n^{(1)})}$ may correspond to $\widehat{\boldsymbol{x}}^{\text{opt},(2,n^{(1)})}$, but our framework holds for any recommendation function $\X^{(2,n^{(1)})} = g \bigl(\BarA^{(1)}, \overline{\boldsymbol{Y}}^{(1)}\bigr)$ for which $\X^{(2,n^{(1)})}$ converges in probability to a constant.

At stage 2, the actual intervention at center $j$ during visit $t$ is given 
by $\A_{jt}^{(2,n^{(1)})} = h_{jt}^{(2)}(\widehat{\boldsymbol{x}}^{\text{opt},(2,n^{(1)})})$, 
where $h_{jt}^{(2)}$ is a deterministic continuous function capturing center- and 
time-specific deviations from the estimated optimal intervention. Let $Y_{ijt}^{(2,n^{(1)})}$ 
denote the outcome of participant $i$ in center $j$ during visit $t$, and let 
$\boldsymbol{Y}_{jt}^{(2,n^{(1)})} = (Y_{1jt}^{(2,n^{(1)})}, \ldots, Y_{n_{jt}^{(2)}jt}^{(2,n^{(1)})})$ 
collect the outcomes across all participants in center $j$ during visit $t$. 
The stage~2 interventions, outcomes, and errors aggregated across centers 
at each visit $t$ are $\overline{\A}_t^{(2,n^{(1)})} = (\A_{1t}^{(2,n^{(1)})}, \ldots, \A_{J^{(2)}t}^{(2,n^{(1)})})$, 
$\overline{\boldsymbol{Y}}_t^{(2,n^{(1)})} = (\boldsymbol{Y}_{1t}^{(2,n^{(1)})}, \ldots, \boldsymbol{Y}_{J^{(2)}t}^{(2,n^{(1)})})$, 
and $\overline{\boldsymbol{\epsilon}}_t^{(2,n^{(1)})} = (\boldsymbol{\epsilon}_{1t}^{(2,n^{(1)})}, \ldots, 
\boldsymbol{\epsilon}_{J^{(2)}t}^{(2,n^{(1)})})$. We write
$\overline{\A}^{(2,n^{(1)})} = \bigl(\overline{\A}_1^{(2,n^{(1)})}, \ldots, 
\overline{\A}_{T^{(1)}+1,\ldots, T^{(2)}}^{(2,n^{(1)})}\bigr)$, 
$\overline{\boldsymbol{Y}}^{(2,n^{(1)})} = \bigl(\overline{\boldsymbol{Y}}_1^{(2,n^{(1)})}, \ldots, 
\overline{\boldsymbol{Y}}_{T^{(1)}+1,\ldots, T^{(2)}}^{(2,n^{(1)})}\bigr)$, and 
$\overline{\boldsymbol{\epsilon}}^{(2,n^{(1)})} = \bigl(\overline{\boldsymbol{\epsilon}}_1^{(2,n^{(1)})}, 
\ldots, \overline{\boldsymbol{\epsilon}}_{T^{(1)}+1,\ldots, T^{(2)}}^{(2,n^{(1)})}\bigr)$.

Our estimand is $\B_A$, the causal effect of the actual intervention $\A$ on the outcome, rather than the effect of the recommended intervention $\boldsymbol{x}$, because centers may deviate from the recommendations through $h_{jt}^{(k)}$. We focus on the effect of the actual intervention because it directly determines participant exposure and thus the outcome. Within each stage $k$, identification of $\B_A$ requires variation in the actual interventions $\A_{jt}^{(k)}$ across visits $t$ within center $j$,  as captured by the function $h_{jt}^{(k)}$.

\begin{assumption}[Conditional Independence of Outcome Errors]\label{asspMinh}
Conditional on $\overline{\A}^{(k)}$ and $\overline{\BoldZ}$, the outcome errors $\epsilon_{ijt}^{(k)}$ are independent across centers and, within each stage $k$ in each center $j$, independent across participants. Also, conditional on $\overline{\A}^{(k)}$ and $\overline{\BoldZ}$, for participants present in multiple stages, their outcome errors are independent across stages.
\end{assumption}
\noindent That is, dependencies within the center are captured by the $\gamma_j$.
\begin{assumption}[LAGO Assumption]\label{asspAnte}
Conditional on $\widehat{\boldsymbol{x}}^{\text{opt},(2,n^{(1)})}$ and $\overline{\BoldZ}$, the stage 2 $\bigl(\overline{\A}^{(2, n^{(1)})}, \overline{\boldsymbol{\epsilon}}^{(2, n^{(1)})}\bigr)$ are independent of the stage 1 $\bigl(\overline{\A}^{(1)}, \overline{\boldsymbol{\epsilon}}^{(1)} \bigr)$. 
\end{assumption}

\noindent Assumption \ref{asspAnte} states that learning takes place only through the determination of the recommended intervention, given the center characteristics.

\begin{remark}
Under Assumption~\ref{mainmodel}, Assumption~\ref{asspMinh} implies that conditional on $\overline{\A}^{(k)}$ and $\overline{\BoldZ}$, the outcomes $Y_{ijt}^{(k)}$ are independent across centers and individuals, with within-center similarity captured by the center effects $\gamma_j$. Assumption~\ref{asspAnte} implies that conditional on $\widehat{\boldsymbol{x}}^{\text{opt},(2,n^{(1)})}$ and $\overline{\BoldZ}$, $\bigl(\overline{\A}^{(2,n^{(1)})}, \overline{\boldsymbol{Y}}^{(2,n^{(1)})}\bigr)$ is independent of $\bigl(\overline{\A}^{(1)}, \overline{\boldsymbol{Y}}^{(1)}\bigr)$.
\end{remark}

\begin{assumption}[Limiting Recommendations]\label{asspRecommendedIntervention} The stage 2 recommended intervention $\widehat{\boldsymbol{x}}^{\text{opt}, (2,n^{(1)})}$ converges in probability to a limit $\boldsymbol{x}^{(2)}$ as $n^{(1)} \rightarrow \infty$.
\end{assumption}

\noindent The LAGO framework incorporates a shrinkage algorithm. Shrinkage is needed when the desired outcome goal cannot be achieved within the feasible intervention space, preventing the recommended intervention from bouncing between the upper and lower limits under the null intervention effect. The shrinkage ensures that the stage 2 recommended intervention converges to a well-defined limit, so Assumption \ref{asspRecommendedIntervention} holds also under the null hypothesis where there are no intervention effects. See Section 5 of the Supplementary Materials from \citet{nevo2021analysis} or Appendix B in \citet{bing2025learn} for a discussion of the shrinkage algorithm. 

\begin{remark}\label{convergenceofA}
Under Assumption \ref{asspRecommendedIntervention} and the assumed continuity of the function $h_{jt}$ that maps recommended interventions onto actual intervention in center $j$, the Continuous Mapping Theorem implies that $\A_{jt}^{(2,n^{(1)})} = h_{jt}^{(2)} \bigl( \widehat{\boldsymbol{x}}^{\text{opt}, (2,n^{(1)})} \bigr)$ converges in probability to $\Bolda_{jt}^{(2)} = h_{jt}^{(2)}\bigl( \boldsymbol{x}^{(2)}\bigr)$.
\end{remark}

We illustrate the 2-stage LAGO design as follows:

\begin{align*}
\boldsymbol{x}^{(1)} \xrightarrow[]{h_{jt}^{(1)}} \underset{\substack{j=1,\dots,J^{(1)} \\ t= 1,\ldots, T^{(1)}}}{\A_{jt}^{(1)}} \xrightarrow[]{\BoldZ_j}\underset{\substack{ i=1,\dots,n_{jt}^{(1)} \\ j=1,\dots,J^{(1)} \\ t = 1,\ldots, T^{(1)}}}{Y_{ijt}^{(1)}} \xrightarrow[\text{analyze data}]{\overline{\BoldY}^{(1)}, \BarA^{(1)}} \boldsymbol{\hat{x}}^{\text{opt}, (2, n^{(1)})} \xrightarrow[]{h_{jt}^{(2)}} \underset{\substack{j=1,\dots,J^{(2)} \\ t=T^{(1)}+1,\ldots, T^{(2)}}}{\A_{jt}^{(2,n^{(1)})}} \xrightarrow[]{\BoldZ_j} \underset{\substack{i=1,\dots,n_{jt}^{(2)} \\ j=1,\dots,J^{(2)} \\ t=T^{(1)}+1,\ldots, T^{(2)}}}{Y_{ijt}^{(2, n^{(1)})}}.    
\end{align*}

\begin{assumption}[Outcome Errors]\label{epsilonassumption}
For a fixed intervention package $\Bolda$, the distribution of $\epsilon_{ijt}^{(k)}(\Bolda)$ depends only on $\BoldZ_j$ and is independent of $\Bolda$, with mean $\E \bigl( \epsilon_{ijt}^{(k)}(\Bolda) \mid \overline{\BoldZ} \bigr) = 0$. Assume that for some $\delta > 0$, $\max_{j} \E \bigl( |\epsilon_{ijt}^{(k)}(\Bolda)|^{2+\delta} \mid \BoldZ_j \bigr) \leq \mathcal{C}_\epsilon$. With $\sigma^2(\BoldZ_j) = \Var \bigl( \epsilon_{ijt}^{(k)}(\Bolda) \mid \BoldZ_j \bigr)$, define $\sigma^2 = \max_{j} \sigma^2(\BoldZ_j)$.
\end{assumption}

\begin{assumption}[Conditional Exchangeability]\label{asspCondExcha}
For fixed $\Bolda$ and all $i = 1,\ldots,n_{jt}^{(k)}$, $j = 1,\ldots,J^{(k)}$, $k = 1,\ldots,K$, and $t$ a visit in stage $k$, the counterfactual outcome is conditionally independent of the actual and recommended interventions given center characteristics: $Y_{ijt}^{(k)}(\Bolda) \independent (\BarA^{(k)}, \X^{(k)}) \mid~\overline{\BoldZ}$.
\end{assumption}

\begin{assumption}[Conditional Independence Between Centers]\label{asspcenterIndep}
For fixed $\Bolda$ and all $i = 1,\ldots,n_{jt}^{(k)}$, $j = 1,\ldots,J^{(k)}$, $k = 1,\ldots,K$, and $t$ a visit in stage $k$, $Y_{ijt}^{(k)}(\Bolda) \independent \overline{\BoldZ}_{(-j)} \mid  \BoldZ_j$.
That is, the counterfactual outcomes do not depend on center characteristics of other centers.
\end{assumption}

\noindent Assumption \ref{asspCondExcha} states that the counterfactual outcome $Y_{ijt}^{(k)}(\Bolda)$ of participant $i$ in center $j$ in stage $k$ is not predicted by actual and recommended interventions of any center in stage $k$ conditional on $\BoldZ$. Combining Assumptions \ref{asspCondExcha} and \ref{asspcenterIndep}, $\E \bigl(Y_{ijt}^{(k)}(\Bolda) \mid  \BarA^{(k)}, \X^{(k)}, \overline{\BoldZ} \bigl) = \E \bigl(Y_{ijt}^{(k)}(\Bolda) \mid \overline{\BoldZ} \bigr) = \E \bigl(Y_{ijt}^{(k)}(\Bolda) \mid \BoldZ_j \bigr)$.

\begin{assumption}[Compact Space]\label{asspCompact}
The intervention $\Bolda$ and vector $\boldsymbol{\beta}$ take values in compact spaces, with components bounded by constants $\mathcal{C}_A$ and $\mathcal{C}_\beta$, respectively.
\end{assumption}

\subsection[Unbiased estimating equations for B]{Unbiased estimating equations for $\B$:}
The estimating equations used to estimate $\B$ are

\begin{align}\label{eqn:OLS_ignore_ee}
    &0 = U(\boldsymbol{\beta}) \nonumber\\
    &= \frac{1}{n} \left\{ 
        \sum_{t=1}^{T^{(1)}} \sum_{j=1}^{J^{(1)}} \sum_{i=1}^{n_{jt}^{(1)}} 
        \begin{pmatrix} \Aone\\ \BoldLone \end{pmatrix} 
        \left[Y_{ijt}^{(1)} - 
        \begin{pmatrix} \boldsymbol{\beta}_A \\\B_{\ell} 
        \end{pmatrix}^{\top} 
        \begin{pmatrix} \Aone\\ \BoldLone \end{pmatrix} 
        \right] \right.\nonumber \\
    & \left.
        +\sum_{t=T^{(1)}+1}^{T^{(2)}} \sum_{j=1}^{J^{(2)}} \sum_{i=1}^{n_{jt}^{(2)}} 
        \begin{pmatrix} \Atwolago\\ \BoldLtwolago \end{pmatrix} 
        \left[Y_{ijt}^{(2,n^{(1)})} - 
        \begin{pmatrix} \boldsymbol{\beta}_A \\ \B_{\ell} 
        \end{pmatrix}^{\top} 
        \begin{pmatrix} \Atwolago\\ \BoldLtwolago \end{pmatrix} 
        \right]
    \right\}.
\end{align}

The asymptotic analysis of $\widehat{\B}$ has challenges due to the dependence between $\overline{\boldsymbol{Y}}^{(1)}$ and $\bigl(\BarA^{(2,n^{(1))}}, \overline{\boldsymbol{Y}}^{(2,n^{(1)})}\bigr)$, which means that, in contrast to most of the causal inference literature, the estimating function $U(\B)$ from equation (\ref{eqn:OLS_ignore_ee}) cannot be expressed as a sum of independent terms.

\begin{theorem}\label{unbiasedEE} (Unbiased estimating equations). Under Assumptions \ref{mainmodel}, \ref{epsilonassumption}, \ref{asspCondExcha}, \ref{asspcenterIndep}, and \ref{asspConsistency}, conditional on the fixed center characteristics $\overline{\BoldZ}$, $U(\B)$ are unbiased estimating equations~for~$\B^{*}$.
\end{theorem}

To prove Theorem \ref{unbiasedEE}, we show that conditional on the fixed center characteristics $\overline{\BoldZ}$, the expectation of each term in equation (\ref{eqn:OLS_ignore_ee}) equals 0. Although $\BoldZ_j$ may be unmeasured, conditioning on $\BoldZ_j$ is appropriate as it captures the center-level confounding absorbed by $\gamma_j$.

Considering $\overline{\BoldZ}$ as fixed, we show that the solution $\widehat{\B}$ to the unbiased estimating equations~(\ref{eqn:OLS_ignore_ee}) converges in probability to $\B^*$ and is asymptotically normal.

\subsection[Consistency of B-hat]{Consistency of $\widehat{\boldsymbol{\B}}$:}

\begin{theorem}\label{theoremConsistencyWithLearning} (Consistency). Under Assumptions \ref{ratioConvergence}-\ref{asspCompact} and conditional on $\overline{\BoldZ}$, in the absence of multicollinearity, $\widehat{\B} \overset{P}{\rightarrow} \B^{*}$. 
\end{theorem}

Define 
\begin{align}\label{smallu}
u(\boldsymbol{\beta}) &= \sum_{t=1}^{T^{(1)}} \sum_{j=1}^{J^{(1)}} \alpha_{jt}^{(1)}  \begin{pmatrix} \A_{jt}^{(1)}\\ \BoldLone \end{pmatrix} \left [ \begin{pmatrix} \boldsymbol{\beta}^{*}_A \\   \B^{*}_{\ell} \end{pmatrix}^{\top} - \begin{pmatrix} \boldsymbol{\beta}_A \\   \B_{\ell} \end{pmatrix}^{\top} \right ] \begin{pmatrix} \A_{jt}^{(1)} \\  \BoldLone \end{pmatrix} \nonumber \\
&+  \sum_{t=T^{(1)}+1}^{T^{(2)}} \sum_{j=1}^{J^{(2)}} \alpha_{jt}^{(2)}  \begin{pmatrix} \Bolda_{jt}^{(2)} \\  \BoldLtwo \end{pmatrix} \left [ \begin{pmatrix} \boldsymbol{\beta}^{*}_A \\   \B^{*}_{\ell} \end{pmatrix}^{\top} - \begin{pmatrix} \boldsymbol{\beta}_A \\   \B_{\ell} \end{pmatrix}^{\top} \right ] \begin{pmatrix} \Bolda_{jt}^{(2)} \\  \BoldLtwo \end{pmatrix}.
\end{align}
To prove Theorem \ref{theoremConsistencyWithLearning}, we use Theorem 5.9 of \citet{van2000asymptotic} and show that the two conditions of this theorem are satisfied:  
\begin{align}\label{VaartFirst}
\underset{\boldsymbol{\beta} \in \CurB}{\sup} ||U(\boldsymbol{\beta}) - u(\boldsymbol{\beta})||\overset{p}{\to} 0,  
\end{align} 
and, for every $\epsilon > 0$,  
\begin{align}\label{VaartSecond}
\underset{\B \in \CurB:||\B - \B^{*}||>\epsilon}{\inf} ||u(\boldsymbol{\beta})|| > 0 = ||u(\boldsymbol{\beta}^{*})||. 
\end{align} 
To prove equation~(\ref{VaartFirst}), Appendix~\ref{ConsistencyProofAppendix} decomposes $U(\B) - u(\B)$ into five components, each converging in probability to 0.

The first component is $U(\B^{*})$. We use Chebyshev's inequality on each entry of $U(\B^{*})$ to prove that it converges to 0 in probability. We first establish that the two summands in $U(\B^{*})$ are uncorrelated conditional on the center characteristics $\overline{\BoldZ}$. Then, we prove that the variance of each entry of each stage converges to 0. Thus, the expectation equals 0 (Theorem \ref{unbiasedEE}), the variance converges to 0, and with Chebyshev's inequality, $U(\B^{*}) \ConvP 0$. 

The supremum over $\B$ of the other components is proven to converge to 0 in probability because of Compact Space Assumption \ref{asspCompact}, Remark \ref{convergenceofA}, and the Continuous Mapping Theorem. Equation (\ref{VaartFirst}) follows.

The condition of equation (\ref{VaartSecond}) holds because  $||u(\boldsymbol{\beta}^{*})|| = 0$ from equation (\ref{smallu}), while the strict concavity of $u(\B)$ is ensured by the negative definiteness of its derivative at $\B$ in the absence of multicollinearity in the limit. Hence, with the two conditions in Theorem 5.9 of \citet{van2000asymptotic} satisfied, Theorem \ref{theoremConsistencyWithLearning} holds.

\subsection[Asymptotic Normality]{Asymptotic Normality of $\widehat{\boldsymbol{\beta}}$:}

\begin{theorem}\label{AsymptoticTheoremWithLearning} (Asymptotic Normality). Under Assumptions \ref{ratioConvergence}-\ref{asspCompact}, conditional on $\overline{\BoldZ}$, 
    $\sqrt{n}(\widehat{\B} - \B^{*}) \overset{\mathcal{D}}{\rightarrow} \mathcal{N}\left(0, J^{*-1} \cdot V \cdot J^{*-T}\right)$,
    with
    \begin{align*}
        &J^{*} = \sum_{t=1}^{T^{(1)}} \sum_{j=1}^{J^{(1)}} \alpha_{jt}^{(1)}  \begin{pmatrix} \A_{jt}^{(1)}\\   \BoldLone \end{pmatrix} \begin{pmatrix} \A_{jt}^{(1)}\\   \BoldLone \end{pmatrix}^{\top} + \sum_{t=T^{(1)}+1}^{T^{(2)}} \sum_{j=1}^{J^{(2)}} \alpha_{jt}^{(2)}  \begin{pmatrix} \Bolda_{jt}^{(2)}\\   \BoldLtwo \end{pmatrix} \begin{pmatrix} \Bolda_{jt}^{(2)}\\   \BoldLtwo \end{pmatrix}^{\top}, \\
        &V = \sum_{t=1}^{T^{(1)}} \sum_{j=1}^{J^{(1)}} \alpha_{jt}^{(1)} \begin{pmatrix} \A_{jt}^{(1)}\\  \BoldLone \end{pmatrix}  \begin{pmatrix} \A_{jt}^{(1)}\\   \BoldLone \end{pmatrix}^{\top} \sigma^2 \left(\BoldZ_j \right) + \sum_{t=T^{(1)}+1}^{T^{(2)}} \sum_{j=1}^{J^{(2)}} \alpha_{jt}^{(2)} \begin{pmatrix} \Bolda_{jt}^{(2)}\\   \BoldLtwo \end{pmatrix}  \begin{pmatrix} \Bolda_{jt}^{(2)}\\  \BoldLtwo \end{pmatrix}^{\top} \sigma^2 \left(\BoldZ_j \right).
    \end{align*}
\end{theorem}

To prove Theorem \ref{AsymptoticTheoremWithLearning}, first, we use the Mean Value Theorem to show that the limiting distribution of $ \sqrt{n}(\widehat{\B} - \B^{*})$ is the same as the limiting distribution of $-\sqrt{n}\left[\left( \left .\frac{\partial}{\partial \boldsymbol{\beta}} \right |_{\Tilde{\boldsymbol{\beta}}}U(\boldsymbol{\beta}) \right)\right]^{-1} U(\boldsymbol{\beta}^{*})$, where for each row of $\left .\frac{\partial}{\partial \boldsymbol{\beta}} \right |_{\Tilde{\boldsymbol{\beta}}}U(\boldsymbol{\beta})$, $\widetilde{\B}$ takes a possibly row-dependent value between $\widehat{\B}$ and $\B^{*}$. Second, we show that $\left(-\left .\frac{\partial}{\partial \boldsymbol{\beta}} \right |_{\Tilde{\boldsymbol{\beta}}}U(\boldsymbol{\beta}) \right)^{-1}$ converges in probability to $J^{*-1}$. 

To show that $\sqrt{n}\,U(\B^{*})$ converges in distribution to a normal random variable, we decompose it into two terms,
\begin{align*}
    U_{1\_2,n} &= \frac{1}{\sqrt{n}} \left [ \sum_{t=1}^{T^{(1)}} \sum_{j=1}^{J^{(1)}} \sum_{i=1}^{n_{jt}^{(1)}} \begin{pmatrix} \Aone\\   \BoldLone \end{pmatrix} \epsilon_{ijt}^{(1)} +\sum_{t=T^{(1)}+1}^{T^{(2)}} \sum_{j=1}^{J^{(2)}} \sum_{i=1}^{n_{jt}^{(2)}} \begin{pmatrix} \Bolda_{jt}^{(2)}\\   \BoldLtwo \end{pmatrix} \epsilon_{ijt}^{(2,n^{(1)})} \right ],\\
    U_{2\_2,n} &=\frac{1}{\sqrt{n}}\sum_{t=T^{(1)}+1}^{T^{(2)}} \sum_{j=1}^{J^{(2)}} \sum_{i=1}^{n_{jt}^{(2)}} \left [\begin{pmatrix} \Atwolago\\   \BoldLtwolago \end{pmatrix}  -  \begin{pmatrix} \Bolda_{jt}^{(2)}\\   \BoldLtwo \end{pmatrix} \right ]\epsilon_{ijt}^{(2,n^{(1)})}. 
\end{align*}
Appendix \ref{AsymptoticProofAppendix} establishes that $U_{1\_2,n}$ converges to a normal distribution while $U_{2\_2,n}$ converges to 0 in probability. Regarding $U_{1\_2,n}$, by Outcome Errors Assumption \ref{epsilonassumption}, the $\epsilon_{ijt}^{(2,n^{(1)})}$ can be replaced by the errors $\epsilon_{ijt}^{(2)}$ under the limit $\Bolda_{jt}^{(2)}$ without changing the distribution of  $U_{1\_2,n}$. To prove that $U_{1\_2,n}$ converges to a normal distribution, since $\epsilon_{ijt}^{(1)}$ and $\epsilon_{ijt}^{(2)}$ are independent but not i.i.d.\ with distribution possibly different in each center $j$, the Crámer-Wold device and the Lyapunov Central Limit Theorem (CLT) are used \citep{billingsley1995}. Next, to establish that $U_{2\_2,n} \ConvP 0$, Appendix \ref{AsymptoticProofAppendix} shows that $\E(U_{2\_2,n}) = 0$ and $\E(U_{2\_2,n}^2) \rightarrow 0$, which imply convergence in probability to 0 through Chebyshev's inequality. Theorem \ref{AsymptoticTheoremWithLearning} follows.

A consistent estimate of the asymptotic variance is obtained by substituting $\Bolda_{jt}^{(2)}$, $\B^*$, $\alpha_{jt}^{(1)}$, and $\alpha_{jt}^{(2)}$ with $\A_{jt}^{(2,n^{(1)})}$, $\widehat{\B}$, $n_{jt}^{(1)}/n$, and $n_{jt}^{(2)}/n$, following Consistency Theorem~\ref{theoremConsistencyWithLearning}. Both the asymptotic variance and its consistent estimate coincide with those derived under fixed interventions and independence of $\overline{\mathbf{Y}}^{(1)}$ and $\overline{\mathbf{Y}}^{(2,n^{(1)})}$.

\subsection{Hypothesis Testing}\label{hypothesis_testing}
In a LAGO study, a primary objective is to test the null hypothesis of no overall intervention effect. The null hypothesis is $H_0: \B_A = 0$, which can be tested using Wald-type tests, and these methods are applicable due to the asymptotic normality of $\widehat{\B}$ proven in Section 3.2. Rejection of the null hypothesis $H_0$ indicates evidence of a non-zero intervention effect in the presence of confounding by indication.

Alternatively, in a randomized controlled LAGO trial, let R denote the intervention assignment indicator (R = 1 for intervention, R = 0 for control). The primary hypothesis test is $H_0: \B_A = 0$. Under this null hypothesis, the intervention components have no effect on outcomes, eliminating confounding by indication since the choice of intervention $\A^{(k)}$ becomes irrelevant for the distributions of the outcomes. The observed outcome model under the null hypothesis reduces to $Y_{ijt}^{(k)} = \gamma_j + \eta \cdot \mathbbm{1}[k = 2] + \epsilon_{ijt}^{(k)}$. Thus, we can test whether $\delta = 0$ in the full model $Y_{ijt}^{(k)} = \delta \cdot \text{R} + \gamma_j + \eta \cdot \mathbbm{1}[k = 2] + \epsilon_{ijt}^{(k)}$, for which $\delta$ represents the intervention effect comparing the intervention and the control groups. Under $H_0$ in the full model, we have $\delta = 0$ since the intervention components have no effect. The inclusion of fixed center effects $\gamma_j$ is important as they account for baseline heterogeneity across centers that could otherwise bias the test for an intervention effect. Furthermore, under the null hypothesis, stages 1 and 2 outcomes are conditionally independent given center characteristics $\overline{\BoldZ}$ as the stage 2 intervention cannot influence outcomes when intervention effects are absent. $\widetilde{H}_0: \delta = 0$ can be tested with standard methods for testing any difference between the intervention and the control groups, with rejection indicating that at least some implementation of the intervention package has a statistically significant effect on outcomes.

\subsection{Confidence Sets and Confidence Bands}\label{confidence_tests}
We calculate the confidence interval for the population mean $\bar{\mu}(\boldsymbol{x}) = \sum_{j=1}^{J} w_j \mathbb{E}[Y_{ijt}(\boldsymbol{x})| \mathbf{\beta}^*]$ for each possible value of $\boldsymbol{x}$, where $w_j = \frac{n_j^{(1)} + n_j^{(2)}}{n}$. The 95\% confidence interval for $\bar{\mu}(\boldsymbol{x})$ is
\vspace{-0.3cm}
\begin{align*}
    \text{CI}_{\bar{\mu}}(\boldsymbol{x}) = \left( \boldsymbol{x}^{\top}, w_1, w_2, \ldots, w_J, 1 \right)\widehat{\mathbf{\beta}} \pm 1.96 \, \sqrt{\sigma^2_{\bar{\mu}} \left(\widehat{\mathbf{\beta}}; \boldsymbol{x}\right)},
\end{align*}
where $\sigma^2_{\bar{\mu}} \left(\widehat{\mathbf{\beta}}; \boldsymbol{x}\right) = \left( \boldsymbol{x}^{\top}, w_1, w_2, \ldots, w_J, 1 \right)\, n^{-1}\,\hat{J}^{*-1}\, \hat{V} \,\hat{J}^{*-\top}\left( \boldsymbol{x}^{\top}, w_1, w_2, \ldots, w_J, 1 \right)^{\top} $ is the estimated variance of $\left( \boldsymbol{x}^{\top}, w_1, w_2, \ldots, w_J, 1 \right)\widehat{\mathbf{\beta}}$.

A 95\% confidence set for the optimal intervention is $\text{CS}\left(\boldsymbol{x}^{\text{opt}}\right) = \{\boldsymbol{x}: \text{CI}_{\bar{\mu}}(\boldsymbol{x}) \ni \theta\}$,
that is, those intervention packages $\boldsymbol{x}$ whose confidence interval for the mean outcome $\bar{\mu}(\boldsymbol{x})$ contains $\theta$. From Theorem \ref{AsymptoticTheoremWithLearning}, the probability that $\boldsymbol{x}^{\text{opt}}$ lies in $\text{CS}\left(\boldsymbol{x}^{\text{opt}}\right)$ converges to at least 0.95 as $n \rightarrow \infty$.

LAGO also leads to confidence bands for the mean outcome under a range of intervention packages of interest, $\boldsymbol{x}$, that simultaneously have 95\% coverage asymptotically. Similar to \cite{scheffe1999analysis}, \cite{nevo2021analysis}, and \cite{bing2023learn}, the confidence bands are
\vspace{-0.3cm}
\begin{align*}
    \text{CB}_{\bar{\mu}} = \left( \boldsymbol{x}^{\top}, w_1, w_2, \ldots, w_J, 1 \right) \widehat{\mathbf{\beta}} \pm \sqrt{\chi^2_{0.95,P+J+1} \, \sigma^2_{\bar{\mu}}(\widehat{\mathbf{\beta}};\boldsymbol{x})}, 
\end{align*}
where $\chi^2_{0.95,P+J+1}$ is the upper 95th percentile of the $\chi^2_{P+J+1}$ distribution. $P+J+1$ is the dimension of $\mathbf{\beta}$, with $P$ the number of components in $\A$ and $J+1$ the dimension of $w_1, w_2, \ldots, w_J, 1$. These confidence bands have asymptotic coverage of 95\% for all possible values $\boldsymbol{x}$ of interest simultaneously.

\subsection{Extensions to Binary Outcomes and GLMs}
For binary outcomes modeled with a logistic regression, proofs of unbiased estimating equations and consistency follow the same steps as continuous outcomes modeled with a linear regression. The proof of asymptotic normality additionally uses coupling \citep{lindvall2002lectures}, as in \citet{nevo2021analysis}, to replace the stage 2 outcomes with outcomes coupled with those under the limiting distributions without changing the distribution. The difference between the resulting estimating equations and those under the limiting design is shown to converge to 0 in probability. Full technical details are provided in Appendix~\ref{proofs_binary}.

The framework extends to continuous outcomes modeled with any GLM with a known twice continuously differentiable link function $g(\cdot)$:
\begin{align*}
    Y_{ijt}^{(k)}{(\Bolda)} = g^{-1}\left (\begin{pmatrix} \B^{*}_A \\ 
    \B^{*}_{\ell} \end{pmatrix}^{\top} 
    \begin{pmatrix} \Bolda \\ \boldsymbol{\ell}_j^{(k)} 
    \end{pmatrix} \right) + \epsilon_{ijt}^{(k)}(\Bolda),
\end{align*}
where the distribution of the $\epsilon_{ijt}^{(k)}(\Bolda)$ may depend on $\BoldZ_j$ as in Assumption~\ref{epsilonassumption}. 
Key differences from the linear case are a Donsker class argument for consistency \citep{bing2023learn} 
and the Cram\'{e}r-Wold device with the Lyapunov CLT for non-identically distributed errors; 
see Appendix~\ref{proofs_glm_general_link}.

\section{Simulation}\label{simulations}
The finite-sample properties of the LAGO design were evaluated under three scenarios, each with 2000 simulated datasets with outcomes modeled using linear regression: a null scenario with no intervention effect (Scenario 0), a linear cost function (Scenario 1), and a cubic cost function (Scenario 2), with outcomes modeled using linear regression.

All scenarios examined $J = 6, 10, 20$ centers with $n_{j}^{(k)} = 50, 100$ participants per center in stage 1 and $n_{j}^{(k)} = 100, 200$ in stage 2, with participants randomly assigned to intervention or control in a 1:1 ratio. The intervention comprised two components $\boldsymbol{x} = (x_1, x_2)$ with bounds $[0, 4]$ and $[0, 3]$, and true coefficients $\boldsymbol{\beta}_A = (-1.70, -0.70)$. Center confounders were simulated as $\BoldZ_j \sim N(0,1)$, with results reported for both varying and fixed $\BoldZ_j$ across datasets. The linear and cubic cost functions were $C_1(\boldsymbol{x}) = x_{1} + 0.5x_{2}$ and $C_2(\boldsymbol{x}) = 1.25x_{1} - 0.04x_{1}^{3} + 0.0055x_{1}^{4} + 0.63x_{2} - 0.09x_{2}^{3} + 0.026x_{2}^{4}$, respectively (visualized in Appendix~\ref{sim_visualizations}).

The stage 1 recommended intervention was set to $\boldsymbol{x}^{(1)} = (2, 1.5)$, the midpoint of the feasible region and the stage 2 recommended intervention $\X^{(2,n^{(1)})}$ was obtained via LAGO optimization on the stage 1 data. At each stage $k$, the actual intervention components for participant $i$ in center $j$ were simulated as
\begin{align}
    \A_{jt,1}^{(k)} &= \X_1^{(k)} + \eta_{1} \cdot \BoldZ_j + \xi_{jt,1}^{(k)}, \label{int1_sim}\\
    \A_{jt,2}^{(k)} &= \X_2^{(k)} + \eta_{2} \cdot \BoldZ_j + \xi_{jt,2}^{(k)}, \label{int2_sim}
\end{align}
where $\xi_{jt,1}^{(k)}$ and $\xi_{jt,2}^{(k)}$ are independently sampled from $N(0,1)$, and $\eta_1$, $\eta_2$ control the degree of center-level confounding. Two confounding scenarios were considered: (i) low confounding with target correlations $\rho_1 = 0.05$ and $\rho_2 = 0.07$, achieved by setting $\eta_1 = 0.0501$ and $\eta_2 = 0.0702$; and (ii) high confounding with $\rho_1 = 0.10$ and $\rho_2 = 0.20$, achieved by setting $\eta_1 = 0.1005$ and $\eta_2 = 0.2041$ (calculations in Appendix~\ref{eta_calculations}). These values were held fixed across all scenarios. Simulated outcomes at each stage followed
\begin{align}\label{outcome_sim}
    Y_{ijt}^{(k)} = \beta_1 \cdot \A_{jt,1}^{(k)} + \beta_2 \cdot \A_{jt,2}^{(k)} + \beta_z \cdot \BoldZ_j + \epsilon_{ijt}^{(k)},
\end{align}
where $\beta_z = 2.42$ corresponds to the standard deviation of the month-12 center effects from the EXTRA-CVD study (Table~\ref{tab:regression_results}), and $\epsilon_{ijt}^{(k)} \sim N(0,1)$. At the end of each stage, LAGO optimization was performed targeting a 5-unit reduction in the mean outcome under constraints $0 \leq x_1 \leq 4$ and $0 \leq x_2 \leq 3$, with the combined stages 1 and 2 data used at the end of stage 2. Table~\ref{simulationresults_linear_varying_Zj_lowconf} presents selected results for Scenario 1 with varying $\BoldZ_j$ and low confounding ($\rho_1 = 0.05$, $\rho_2 = 0.07$); full results for Scenario 0 are in Appendix~\ref{sim_null}.

\begin{table}[htbp]
\singlespacing
\caption{Selected simulation study results for Scenario 1 ($\rho_1 = 0.05$ and $\rho_2 = 0.07$), with a linear cost function}
\addtocounter{table}{-1}
\hrule
    \fontsize{8pt}{10pt}\selectfont
    \subfloat[Individual package component effects] {
    \centering
    \begin{tabular}{llllllllll}
    \makecell{$\boldsymbol{\beta}^*=(\beta_{1}^*, \beta_{2}^*)$ \\ (-1.70, -0.70)} & $n_j^{(1)}$ & $n_j^{(2)}$ & $J$ &  & $\widehat{\boldsymbol{\beta}}_{1}$ &  &  & $\widehat{\boldsymbol{\beta}}_{2}$ &  \\
      &  &  &  & $\%$RelBias & \begin{tabular}[c]{@{}l@{}}$\frac{SE}{EMP.SD}$ \\ ($\times 100$)\end{tabular} & CP95 & $\%$RelBias & \begin{tabular}[c]{@{}l@{}}$\frac{SE}{EMP.SD}$\\ ($\times 100$)\end{tabular} & CP95 \\
    \multicolumn{4}{l}{Scenario 1 ($J^{(1)} = J^{(2)} = J$)} &  &  &  &  &  &  \\
     & 50 & 100 & 6 & 0.018 & 95.00 & 94.80 & 0.233 & 101.34 & 95.20 \\
     &  &  & 10 & -0.014 & 96.72 & 94.85 & 0.238 & 100.58 & 94.80 \\
     &  &  & 20 & 0.012 & 98.45 & 95.35 & 0.024 & 99.18 & 94.55 \\
     &  & 200 & 6 & 0.031 & 94.83 & 94.50 & 0.175 & 105.40 & 95.50 \\
     &  &  & 10 & 0.005 & 93.87 & 94.25 & -0.089 & 100.41 & 95.15 \\
     &  &  & 20 & -0.008 & 98.31 & 94.90 & 0.090 & 98.68 & 94.95 \\
     & 100 & 100 & 6 & -0.006 & 96.26 & 94.65 & -0.063 & 98.08 & 94.45 \\
     &  &  & 10 & -0.014 & 101.75 & 95.40 & 0.007 & 97.58 & 94.30 \\
     &  &  & 20 & -0.023 & 96.92 & 94.00 & 0.082 & 100.60 & 94.95 \\
     &  & 200 & 6 & 0.027 & 99.90 & 95.65 & 0.036 & 100.68 & 95.35 \\
     &  &  & 10 & 0.005 & 98.43 & 94.60 & -0.083 & 99.62 & 94.60 \\
     &  &  & 20 & 0.001 & 95.26 & 93.75 & 0.022 & 97.98 & 94.65 \\
    \end{tabular}
    }
    \hrule
    \fontsize{8pt}{10pt}\selectfont
    \subfloat[{Estimated optimal intervention}\label{bias_table_lowconf}]{
    \centering
    \begin{tabular}{llllllllllll}
    \makecell{$\boldsymbol{\beta}^*=(\beta_{1}^*,\beta_{2}^*)$ \\ (-1.70, -0.70)} & \makecell{$\boldsymbol{x}^{opt}$ \\ (2.94, 0)} & $n_j^{(1)}$ & $n_j^{(2)}$ &  & \multicolumn{3}{c}{Stage 1} &  & \multicolumn{3}{c}{Stage 2/LAGO optimized} \\
      &   &  &  &  & \begin{tabular}[c]{@{}l@{}}Bias of\\ $\hat{x}_{1}^{\text{opt},(1)}$\\ \end{tabular} & \begin{tabular}[c]{@{}l@{}}Bias of\\ $\hat{x}_{2}^{\text{opt},(1)}$\\ \end{tabular} & \begin{tabular}[c]{@{}l@{}}rMSE\\ \end{tabular} &  & \begin{tabular}[c]{@{}l@{}}Bias of\\ $\hat{x}_{1}^{\text{opt},(2,n^{(1)})}$\\ \end{tabular} & \begin{tabular}[c]{@{}l@{}}Bias of\\ $\hat{x}_{2}^{\text{opt},(2,n^{(1)})}$\\ \end{tabular} & \begin{tabular}[c]{@{}l@{}}rMSE\\\end{tabular} \\
    \multicolumn{4}{l}{Scenario 1 ($J^{(1)} = J^{(2)} = 6$)} &  &  &  &  &  &  &  &  \\
     & & 50 & 100 &  & -0.074 & 0.207 & 0.603 & & -0.005 & 0.076 & 0.484 \\ 
     &  &  & 200 & & -0.070 & 0.194 & 0.587 & & -0.003 & 0.077 & 0.477 \\
     &  & 100 & 100 &  & -0.018 & 0.100 & 0.505 & & -0.001 & 0.072 & 0.477 \\
      &  & & 200 &  & -0.002 & 0.103 & 0.515 & & 0.011 & 0.078 & 0.490\\
    \end{tabular}
    }
    \hrule
    \fontsize{8pt}{10pt}\selectfont
    \subfloat[{Estimated optimal intervention, confidence set and confidence band}\label{cp_table_lowconf}]{
    \centering
    \begin{tabular}{lllllllll}
    \makecell{$\boldsymbol{\beta}^*=( \beta_{1}^*,\beta_{2}^*)$\\(-1.70, -0.70)} & \makecell{$\boldsymbol{x}^{opt}$\\(2.94, 0)} & $n_j^{(1)}$ & $n_j^{(2)}$ & \begin{tabular}[c]{@{}l@{}}TrueOpt1\\ (Q2.5,Q97.5)\end{tabular} & \begin{tabular}[c]{@{}l@{}}TrueOpt2\\ (Q2.5,Q97.5)\end{tabular} & \begin{tabular}[c]{@{}l@{}}SetCP95\\ $\%$\end{tabular} & \begin{tabular}[c]{@{}l@{}}SetPerc\\ $\%$\end{tabular} & \begin{tabular}[c]{@{}l@{}}BandsCP95\\ $\%$\end{tabular} \\
    \multicolumn{4}{l}{Scenario 1 ($J^{(1)} = J^{(2)} = 6$)} &  &  &  &  &  \\
     &  & 50 & 100 & (-5.331, -4.657) & (-5.189, -4.915) & 95.30 & 4.07 & 96.25 \\
     &  &  & 200 & (-5.323, -4.673) & (-5.152, -4.940) & 94.70 & 3.49 & 96.40 \\
     &  & 100 & 100 & (-5.266, -4.786) & (-5.167, -4.924) & 95.05 & 3.28 & 96.35 \\
     &  &  & 200 & (-5.272, -4.810) & (-5.142, -4.952) & 95.55 & 2.89 & 96.40 \\
    \end{tabular}
    }
\footnotesize{
    \raggedright{
    \smallskip \\
    \fontsize{8pt}{10pt}\selectfont
    $n_j^{(1)}$, $n_j^{(2)}$: number of participants per center $j$ at stages 1 and 2, respectively. \\
    $J$: number of centers for each stage. \\
    \%RelBias: percent relative bias $100(\hat{\beta}-\beta^\star)/\beta^\star$.\\
    SE: mean estimated standard error, 
    EMP.SD: empirical standard deviation.\\
    CP95: empirical coverage rate of 95\% confidence intervals.\\
    Bias of $\hat{x}_{1}^{\text{opt}}$, $\hat{x}_{2}^{\text{opt}}$: bias of the first and second components of the estimated optimal intervention.\\
    rMSE: root mean squared error, $\{\operatorname{mean}(\|\hat{\boldsymbol{x}}^{o p t}-\boldsymbol{x}^{o p t}\|^{2})\}^{1 / 2}$, estimated by average over simulation iterations.\\
    TrueOpt1: mean treatment effect under the stage 2 recommended intervention.\\
    TrueOpt2: mean treatment effect under the final estimated optimal intervention based on all data. \\
    $Q$2.5 and $Q$97.5: 2.5$\%$ and 97.5$\%$ quantiles.\\ 
    SetCP95$\%$: empirical coverage percentage of confidence set for the optimal intervention. \\
    SetPerc$\%$: mean percentage of the size of the confidence set as a percent of the total intervention space.\\
    BandsCP95$\%$: empirical coverage of 95$\%$ confidence band. 
    }
}
\stepcounter{table}
\label{simulationresults_linear_varying_Zj_lowconf}
\end{table}

Table~\ref{simulationresults_linear_varying_Zj_lowconf}a shows minimal relative bias in $\beta_1$ and $\beta_2$ across all settings, with 95\% confidence interval coverage close to the nominal level (94\%--96\%), and standard errors based on the sandwich estimator for the variance (Theorem~\ref{AsymptoticTheoremWithLearning}) decreasing as $J$, $n_j^{(1)}$, and $n_j^{(2)}$ increase.

Table~\ref{simulationresults_linear_varying_Zj_lowconf}b shows that when $J = 6$, bias and rMSE of the estimated optimal intervention were higher using stage 1 data alone than the combined stages 1 and 2 data, reflecting improved precision from incorporating stage 2 data. Table~\ref{simulationresults_linear_varying_Zj_lowconf}c shows that confidence sets for the optimal intervention achieved close to nominal 95\% empirical coverage (94\%--95\%), confidence bands for the mean outcome achieved 95\%--96\% coverage, and confidence set size ranged from 2\%--4\% of the intervention space across sample sizes. Similar patterns were observed for $J = 10, 20$, with confidence set coverage of 94\%--95\%, confidence set size of 1\%--4\%, and confidence band coverage of 95\%--96\% (Appendices~\ref{bias_table_linear_varying_lowconf} and~\ref{cp_table_linear_varying_lowconf}).

For Scenario 1 (linear cost function), results are organized by confounding level and center effect specification (varying or fixed $\BoldZ_j$). Under low confounding ($\rho_1 = 0.05$, $\rho_2 = 0.07$), results with varying $\BoldZ_j$ are in Appendices~\ref{bias_table_linear_varying_lowconf} and~\ref{cp_table_linear_varying_lowconf}, and with fixed $\BoldZ_j$ in Appendix~\ref{simulationresults_linear_fixed_Zj_lowconf}. Under high confounding ($\rho_1 = 0.10$, $\rho_2 = 0.20$), results with varying and fixed $\BoldZ_j$ are in Appendices~\ref{simulationresults_linear_varying_Zj_highconf} and~\ref{simulationresults_linear_fixed_Zj_highconf}, respectively. Scenario 2 (cubic cost function) follows the same structure, with results in Appendices~\ref{simulationresults_cubic_varying_Zj_lowconf}-\ref{simulationresults_cubic_fixed_Zj_highconf}.
\FloatBarrier

\section{Illustrative Example: The EXTRA-CVD Study}\label{real_data}
The LAGO design is illustrated through the A Nurse-Led Intervention to Extend the HIV Treatment Cascade for Cardiovascular Disease Prevention (EXTRA-CVD) study, a randomized clinical trial that recruited people with HIV (PWH) with elevated systolic blood pressure (SBP) at three HIV clinics in the US (University Hospitals Cleveland Medical Center, MetroHealth Medical Center, and Duke Health) from September 2019 to January 2022, with 12 months of follow-up per participant \citep{longenecker2024nurse}. EXTRA-CVD aimed to demonstrate that a multicomponent nurse-led strategy reduces SBP in PWH receiving antiretroviral therapy (ART).

EXTRA-CVD included $n=297$ PWH assigned to intervention ($n=149$) or control ($n=148$). The intervention group received home BP monitoring, counseling, and care coordination from a prevention nurse at four in-person visits and by telephone as needed, while the control group received general health education at each visit \citep{longenecker2024nurse}. The primary outcome was study-measured SBP. To avoid multicollinearity, the LAGO optimization considered two center-level intervention components: average nurse counseling minutes and average home BP measurements per 10 days, calculated from baseline to months 4, 8, and 12.

The outcome of interest was the change in SBP from baseline, $\Delta$SBP, for each participant $i$. Due to confounding by indication and positivity violations at the participant level, the LAGO optimizations considered interventions at the center level rather than at the participant level. The EXTRA-CVD data exhibited confounding by indication at the participant level because nurses typically spent more time with participants who needed care the most. In addition, positivity violations were likely due to an evidence-based algorithm guiding medication changes, which recommended specific medication adjustments for particular groups of PWH.

Participants enrolled from September 2019 to January 2022 (29 months), divided into four intervals ($m = 1,\ldots, 4$): three 7-month and one 8-month intervals, with each month as a 30-day period. Sensitivity analysis showed enrollment interval effects $\beta_m$ had negligible impact, so they were omitted from the main analysis.

First, consider the month-4 outcome as the outcome of interest. For each month, the intervention was defined as the average of each of the two intervention components within the center of participant \(i\) among participants who were in their first 4 months of follow-up. The rationale was that the intervention intensity depended on how long a participant had been enrolled in the study: in the last few months of a participant's follow-up, the intervention intensity was typically lower. For example, if participant $i$ enrolled in February, we calculated the two average intervention components in each of the months from February to May in the center among participants in their first 4 months of follow-up. Specifically, for February, we calculated the average of each of the two intervention components for all participants at the same center as participant $i$ who were in their first 4 months of follow-up. From the monthly averages for February, March, April, and May, we calculated the overall average for each intervention component across these 4 months. 

Let \(\Delta \text{SBP}_{4ij}\) be the change in SBP of participant \(i\) in center $j$ from baseline to month 4 after enrollment. The month-4 counterfactual outcome under treatment $\Bolda$ for participant $i$ in center $j$ enrolled during enrollment interval $k$ was modeled as $\Delta \text{SBP}_{4ij}(\Bolda) = \beta_{A}^{\top} \Bolda + \boldsymbol{\gamma}_j + \epsilon_{ijt}(\Bolda)$, where $i = 1,\ldots, n_{j}$, $j = 1,\ldots,3$, $\Bolda$ is the average of the 2 intervention components from baseline to month 4 and $\boldsymbol{\gamma}_j$ represents the fixed effect of center $j$. For the month-12 outcome $\Delta\text{SBP}_{12ij}$, we modeled the 
counterfactual as $\Delta \text{SBP}_{12ij}(\Bolda) = \beta_{A}^{\top} 
\Bolda + \boldsymbol{\gamma}_j + \epsilon_{ijt}(\Bolda)$, where $\Bolda$ 
represents the average intervention components across three 4-month intervals 
(baseline--month 4, months 5--8, and months 9--12).

\FloatBarrier
Figure~\ref{fig:sbp_combined} in Appendix~\ref{variability_in_outcomes} shows substantial variability in $\Delta\text{SBP}$ from baseline to month 12, with standard deviations of 20.2 and 19.6~mmHg in the intervention and control groups, consistent with other blood pressure intervention trials \citep{kuhmmer2016effectiveness}.
\FloatBarrier
The LAGO optimization outcome goal was set at a reduction of 5 mmHg based on discussions with clinicians. Let $x_{1t}$ denote the average counseling time per 10 days, and $x_{2t}$ the average number of home BP measurements per 10 days, evaluated at months $t = 4, 8, 12$. Two cost functions were considered, where cost per unit is measured in minutes of nursing time: a linear cost function $C_1(\boldsymbol{x}_t) = 1x_{1t} + 0.5x_{2t}$, and a cubic cost function $C_2(\boldsymbol{x}_t) = 1.25x_{1t} - 0.043x_{1t}^{3} + 0.0055x_{1t}^4 + 0.63x_{2t} - 0.09x_{2t}^3 + 0.026x_{2t}^4$, with $t = 4, 8, 12$ visits. Visualizations of the cost functions are presented in Appendix~\ref{fig:cvd_cubic_cost}. The cubic cost function was implemented to reflect economies of scale at lower intervention intensities, while imposing higher costs as the intervention components approached their maximum capacity limits. Constraints were established in close collaboration with the clinicians based on observed data ranges, with bounds decreasing over time that reflect observed changes in intervention intensity: for month 4, $0 \leq x_{1,4} \leq 10$ and $0 \leq x_{2,4} \leq 4$; for month 8, $0 \leq x_{1,8} \leq 8$ and $0 \leq x_{2,8} \leq 3$; and for month 12, $0 \leq x_{1,12} \leq 6.5$ and $0 \leq x_{2,12} \leq 3$. Table~\ref{interventioncompspecifications} provides the LAGO table on the intervention component specifications for the LAGO optimization at month-4, month-8, and month-12 visits. Information on the average $\Delta\text{SBP}$ by center and study arm at months 4, 8, and 12 is reported in Appendix~\ref{AdditionalInfoEXTRA-CVD}.
\begin{table}[H]
\centering
\fontsize{8pt}{10pt}\selectfont
\caption{LAGO table: Intervention component specifications at the center level for the LAGO optimization in the EXTRA-CVD study, per participant per 10 days for outcomes through months 4, 8, and 12}
\label{interventioncompspecifications}
\begin{tabular}{lccccc}
\makecell{Intervention \\ Component} & Unit & Visit Month & Lower bound & Upper bound & \makecell{Cost per unit\\in minutes spent} \\
\hline
Nurse counseling & Minutes & 4 & 0 & 10 & 1 \\
 & & 8 & 0 & 8 & 1\\
 & & 12 & 0 & 6.5 & 1 \\
\hline
Home BP & Number& 4 & 0 & 4 & 0.5 \\
measurements & & 8 & 0 & 3 & 0.5 \\
& & 12 & 0 & 3 & 0.5 
\end{tabular}
\end{table}
Table \ref{tab:regression_results} presents the results of the LAGO optimizations at the end of the EXTRA-CVD study. Table \ref{tab:regression_results} suggests that the first intervention component, average nurse counseling minutes per 10 days, has a potentially meaningful effect on reducing $\text{SBP}$ for all periods: -0.53, -1.43, and -1.59 mmHg per additional minute per 10 days, with the most substantial effect observed later. The second intervention component, average home BP measurements per 10 days, shows a substantial effect in the initial period: -2.20 mmHg per additional SBP measurement per 10 days. Center effects are not significant, indicating that there are no significant reductions in SBP under control in the centers. The confidence intervals in Table \ref{tab:regression_results} for all predictors are calculated based on the sandwich estimator for the variance. 

The LAGO optimization results indicate that resources should have been allocated primarily to counseling minutes, aligning with the shifting pattern of the effectiveness of the observed components. The total cost of the optimal intervention package in month 12 using the linear cost function is 3.15 minutes per 10 days per participant, and the total cost using the cubic cost function is 3.25 minutes per 10 days per participant. At the month-12 recommended interventions with the cubic cost function (3.10 and 0.10 minutes per 10 days, respectively), on average this corresponds to approximately one 10-minute nurse counseling session per month and one home BP measurement every 3 months. The temporal pattern suggests that home BP monitoring can be an essential intervention for immediate improvements in blood pressure, while nurse counseling might be more beneficial for long-term BP control. 

\FloatBarrier
\begin{table}[htbp]
\centering
\fontsize{8pt}{10pt}\selectfont
\caption{The EXTRA-CVD study: package component effect estimates, 95$\%$ confidence intervals, and estimated optimal intervention package for month-4 outcomes, month-8 outcomes, and month-12 outcomes}
\begin{tabular}{llll}
 & \begin{tabular}[c]{@{}l@{}}Month 4\\ $n =231$\\ $\hat{\boldsymbol{\beta}}$ ($95\%$ CI)\end{tabular} & \begin{tabular}[c]{@{}l@{}}Month 8\\ $n = 239$\\ $\hat{\boldsymbol{\beta}}$   ($95\%$ CI)\end{tabular} & \begin{tabular}[c]{@{}l@{}}Month 12\\ $n = 238$\\ $\hat{\boldsymbol{\beta}}$   ($95\%$ CI)\end{tabular} \\
 \hline
Nurse counseling minutes (per 10 days)& -0.53 (-2.30, 1.25) & -1.43 (-3.90, 1.05) & -1.59 (-4.27, 1.10) \\
Home BP measurements (per 10 days) & -2.20 (-6.60, 2.20) & 0.02 (-4.60, 4.63) & -0.59 (-5.79, 4.62) \\
Center: A & 0.94 (-5.74, 7.61) & -3.18 (-8.16, 1.79) & -2.63 (-8.85, 3.59) \\
Center: B & 1.99 (-3.27, 7.25) & 0.57 (-4.23, 5.38) & 0.58 (-4.37, 5.26) \\
Center: C & 3.98 (-0.29, 8.25) & -1.13 (-5.29, 3.03) & 2.11 (-2.59, 6.80) \\
$\hat{\boldsymbol{x}}^{\text{opt}}$ (linear cost) & (0.00, 3.30) & (2.70, 0.00) & (3.10, 0.00) \\
$\hat{\boldsymbol{x}}^{\text{opt}}$ (cubic cost) & (0.00, 3.30) & (2.70, 0.00) & (3.10, 0.10) \\
\end{tabular}
\\
\raggedright{
95\% CI: 95\% confidence interval based on sandwich estimator for VAR($\hat{\boldsymbol{\beta}}$) (see Theorem \ref{AsymptoticTheoremWithLearning}). \\
For (optimal) interventions, the first component is average nurse counseling telephone check-in minutes per 10 days and the second component is average number of home BP measurements per 10 days.\\
Month-4 outcome is $\Delta\text{SBP}_i$ of each participant $i$ from baseline to month 4 (similarly for months 8 and 12).
}
\label{tab:regression_results}
\end{table}
To demonstrate the benefits of LAGO optimizations, a hypothetical stage 2 was simulated under two scenarios: with LAGO optimization (Scenario 1) and without LAGO optimization (Scenario 2), each with $n^{(2)} = 300$ and an outcome goal of a 5.0 mmHg reduction in SBP. This simulation study comprised 2000 simulated datasets.

To design both scenarios, we first fit a regression model for each intervention component in stage 1 on the centers, and the average of the three center effects was taken as the hypothetical stage 1 recommended intervention $\boldsymbol{x}^{(1)}$. All month-12 outcomes were simulated using the coefficients in the third column of Table~\ref{tab:regression_results} as true values, with $\boldsymbol{\beta}^* = (\beta_1, \beta_2) = (-1.59, -0.59)$ and center effects $\gamma_j$ for $j = 1, 2, 3$.
\FloatBarrier
Stage 1 followed the 12-month EXTRA-CVD study with 
observed interventions and re-simulated outcomes 
using equation~\eqref{outcome_sim}, with 
$\epsilon_{ijt}^{(1)}$ sampled from the month-12 
model residuals (Table~\ref{tab:regression_results}). At the end of stage 1, LAGO optimization was performed with an outcome goal of $-5.0$ mmHg under constraints $0 \leq x_1 
\leq 6.5$ and $0 \leq x_2 \leq 3$.

For Scenario 1 (LAGO), the stage 2 recommended intervention $\boldsymbol{x}^{(2,n^{(1)})}$ was obtained from LAGO optimization; for Scenario 2 (non-LAGO), the stage 1 recommended intervention $\boldsymbol{x}^{(1)}$ was used in its place. For stage 2, actual interventions were simulated using equations~\eqref{int1_sim} and~\eqref{int2_sim}, where $\eta_1$ and $\eta_2$ were obtained from the average difference between the observed stage 1 interventions and $\boldsymbol{x}^{(1)}$ in the EXTRA-CVD data to retain center-level confounding by indication, and $\xi_{ij,p}^{(k)}$ were sampled from residuals of models regressing each intervention component on center characteristics. Control group interventions were set to $\mathbf{0}$ in both scenarios. LAGO optimization was performed at the end of each stage for comparison.

Table~\ref{combined_simulations_goal5} in Appendix~\ref{sim_outcome_goal_5} presents results for the $-5$~mmHg outcome goal without a lower bound adjustment, highlighting the vulnerability of LAGO to high outcome variability and limited stage 1 sample size: noisy stage 1 coefficient estimates led LAGO to recommend suboptimal stage 2 packages, substantially worsening expected outcomes ($-3.37$~mmHg for LAGO versus $-5.71$~mmHg for non-LAGO). To address this, Table~\ref{combined_simulations_goal9_adjust} presents results under two modifications: a more ambitious $-9$~mmHg outcome goal and a lower bound adjustment setting intervention lower bounds equal to the stage 1 recommended interventions, which is appropriate when intervention components carry no risk of harm. Appendices~\ref{sim_outcome_5_lowerbound} and~\ref{sim_outcome_goal_9} present additional results for a $-5$~mmHg goal with lower bound adjustment and a $-9$~mmHg goal without adjustment.

\begin{table}[htbp]
\fontsize{8pt}{10pt}\selectfont
\caption{Simulation study results comparing scenarios with and without LAGO optimization for intervention components and optimal intervention package with a linear cost function and an outcome goal of -9 mmHg, with lower bound equal to stage 1 recommended interventions.}
\label{combined_simulations_goal9_adjust}
\begin{tabular}{llllllllll}
\multicolumn{10}{l}{\textbf{(a) Individual intervention components}} \\
 & & & & \multicolumn{3}{c}{Scenario 1 (With LAGO)} & \multicolumn{3}{c}{Scenario 2 (Without LAGO)} \\
$\boldsymbol{\beta}^{*} = (\beta^{*}_1, \beta^{*}_2)$ & $n^{(1)}$ & $n^{(2)}$ & $\hat{\boldsymbol{\beta}}$ & \%RelBias & SE/EMP.SD & CP95 & \%RelBias & SE/EMP.SD & CP95 \\ \hline
\multirow{2}{*}{$(-1.59, -0.59)$} & \multirow{2}{*}{238} & \multirow{2}{*}{300} & $\hat{\beta}_1$ & -1.15 & 96.8 & 94.7 & 1.86 & 100.0 & 95.2 \\
 & &  & $\hat{\beta}_2$ & 33.3 & 99.0 & 94.2 & -8.04 & 101.0 & 95.3 \\
\end{tabular}
\fontsize{8pt}{10pt}\selectfont
\begin{tabular}{lllll}
\multicolumn{5}{l}{\textbf{(b) Estimated optimal intervention}} \\[0.5ex]
$\boldsymbol{x}^{\text{opt}}$ & Metrics &
\multicolumn{1}{c}{\shortstack[c]{Stage 1\\($n^{(1)} = 238$)}} &
\multicolumn{2}{c}{\shortstack[c]{Stage 2\\($n^{(1)} + n^{(2)} = 528$)}} \\
\cmidrule(lr){3-5}
 &  &  & \shortstack[c]{Scenario 1\\(With LAGO)} & \shortstack[c]{Scenario 2\\(Without LAGO)} \\
\midrule
\multirow{3}{*}{(5.67, 0.00)} 
 & Bias of $\hat{\boldsymbol{x}}^{\text{opt}}_1$ & -1.62 & -1.13 & -1.32 \\
 & Bias of $\hat{\boldsymbol{x}}^{\text{opt}}_2$ & 2.05 & 2.12 & 2.14 \\
 & rMSE & 2.86 & 2.66 & 2.77 \\
\end{tabular}
\\
\fontsize{8pt}{10pt}\selectfont
\begin{tabular}{llll}
\multicolumn{4}{l}{\textbf{(c) Estimated optimal intervention, confidence set, and confidence band}} \\[1ex]
Metrics &
\multicolumn{1}{c}{\shortstack[c]{Stage 1}} &
\multicolumn{2}{c}{\shortstack[c]{Stage 2}} \\
\cline{2-4}\\[-0.8ex]
 &  & \shortstack[c]{Scenario 1 \\ (With LAGO)} & \shortstack[c]{Scenario 2\\(Without LAGO)} \\[0.5ex]
\hline\\[-0.8ex]
ExpectedOutActInt & -5.69 & -7.62 & -5.71 \\
ExpectedOutRecInt & -5.69 & -7.61 & -5.69 \\
ExpectedOutEstOptInt & --$^1$ & -8.44 & -8.14 \\
MeanOpt (Q2.5, Q97.5) & (-11.3,\,-5.69) & (-11.4,\,-5.99)$^*$ & (-11.3,\,-5.93)$^*$ \\
AvgObsOut & -5.83 & -7.59 & -5.68\\
MeanCostAct & --$^2$ & 5.07 & 3.81\\
MeanCostRec & 5.06 & 5.59 & 5.06 \\
SetCP95 (\%) & --$^3$ & 94.7 & 95.4 \\
SetPerc (\%) & --$^3$ & 54.0 & 58.1 \\
BandsCP95 (\%) & --$^3$ & 98.9 & 98.8 \\
\hline
\end{tabular}
\begin{tablenotes}[flushleft]
\fontsize{8pt}{10pt}\selectfont
\item See Table \ref{simulationresults_linear_varying_Zj_lowconf} footnotes for definitions.
\item ExpectedOutActInt: Expected $\Delta\text{SBP}$ under actual interventions in stage 1 and stage 2 for the intervention group, using equal center weighting, calculated using true coefficient values.
\item ExpectedOutRecInt: Expected $\Delta\text{SBP}$ under stage 1 and stage 2 recommended interventions for the intervention group, using equal center weighting, calculated using true coefficient values.
\item ExpectedOutEstOptInt: Expected $\Delta\text{SBP}$ under final estimated optimal interventions at the end of stage 2 based on combined stage 1 and stage 2 data, using equal center weighting, calculated using true coefficient values.
\item MeanOpt: Stage 1 reports mean outcome under the stage 2 recommended intervention and stage 2 reports mean outcome under the final estimated optimal intervention, both calculated using true coefficient values.
\item AvgObsOut: Average observed outcome ($\Delta\text{SBP}$) in the stage in the intervention group across all centers. 
\item MeanCostAct/MeanCostRec: Mean cost of actual/recommended interventions in the stage using~$C(\boldsymbol{x}) = 1x_1$~$+ 0.5x_2$.
\item --$^1$: Not applicable for stage 1.
\item --$^2$: Omitted for stage 1 as it is equivalent to the stage 2 non-LAGO value.
\item --$^3$: Omitted for stage 1 as they are computed from combined stage 1 and stage 2 data and are not expected to be meaningful given the limited stage 1 sample size.
\item $^*$: Calculated based on the final estimated optimal intervention at the end of stage 2 using the combined stage 1 and stage 2 data.
\end{tablenotes}
\end{table}

Table~\ref{combined_simulations_goal9_adjust} combines the $-9$ mmHg outcome goal with the lower bound adjustment. The expected outcome under the actual stage 2 interventions improved substantially to $-7.62$ mmHg for LAGO versus $-5.71$ mmHg for non-LAGO, with LAGO recommending higher-intensity interventions at greater cost: mean cost of recommended interventions: 5.59 and 5.06 minutes for LAGO versus non-LAGO. The true expected outcome under the final estimated optimal interventions was $-8.44$ mmHg for LAGO versus $-8.14$ mmHg for non-LAGO. The stage 2 rMSE remained comparable (2.66 for LAGO; 2.77 for non-LAGO), an improvement over stage 1 (2.86) for both methods, with the true optimal intervention being $(5.67, 0.00)$. Confidence set coverage remained appropriate (94.7\% for LAGO; 95.4\% for non-LAGO), with LAGO achieving a smaller confidence set size (54.0\% versus 58.1\%) and similar confidence band coverage (98.9\% and 98.8\%).

To assess robustness to different cost functions, we repeated the analyses and the simulations using a cubic cost function reflecting economies of scale at lower intervention intensities while imposing higher marginal costs as interventions approach capacity limits (Appendix~\ref{extracvdsimulations_cubic}). Results were consistent with the linear cost function across all outcome goals with and without lower bound adjustments, suggesting that the performance of the LAGO design is robust to alternative cost functions.
\FloatBarrier

\section{Discussion}\label{conclusions}
This article addresses confounding by indication in LAGO trials through fixed center effects, yielding consistent and asymptotically normal estimators of intervention effects under both measured and unmeasured confounding. These properties are established across a unified framework covering continuous outcomes modeled with linear regression and general GLMs, binary outcomes modeled with logistic regression, broadening the applicability of LAGO to a wide class of implementation and public health trials.

The goals of the LAGO design are to estimate the cost-effective intervention package that achieves target outcomes, assess component effects, and test the overall intervention effect. Simulations show good finite-sample performance across reasonable sample sizes and numbers of centers. This work demonstrates that LAGO can be applied to multi-center implementation and public health trials by incorporating fixed center effects.

This article employs fixed rather than random center effects to address confounding by indication, which involves important tradeoffs. When $J$ is large and the $\BoldZ_j$ are measured, a random effects model estimating only $\sigma_b^2$, $\B_A$, $\B_z$, and $\beta_0$ may be more efficient than estimating $J$ center-specific parameters $\gamma_j$, and power decreases with $J$ under fixed effects. However, fixed effects offer key advantages: they require neither $J \rightarrow \infty$ nor distributional assumptions on $b_j$, and they adjust for confounding through both measured \emph{and unmeasured} center characteristics not requiring all confounders to be observed. Our simulations (Section~\ref{simulations}) and EXTRA-CVD analysis (Section~\ref{real_data}) demonstrate good performance for small $J$, making this approach well-suited for trials where enrolling many centers is challenging, particularly when per-center sample sizes are large.

The simulation results (Section~\ref{simulations}) demonstrate the challenges of LAGO optimizations under limited sample size and high outcome variability. Additional analyses (Table~\ref{combined_simulations_goal9_adjust}, Appendices~\ref{sim_outcome_5_lowerbound} and~\ref{sim_outcome_goal_9}) show that setting intervention lower bounds equal to the stage 1 recommended interventions improves LAGO performance at the cost of higher resource utilization, and is appropriate when intervention components carry no risk of harm, the noise-to-signal ratio is high, and stage 1 sample sizes are small. Whether to apply lower bound adjustment pre-planned remains important and should be guided by expected precision considerations. Additionally, lower bound choice can be guided by stage 1 performance: when stage 1 falls short of the outcome goal, lower bounds at stage 1 recommendations prevent LAGO optimizations from recommending weaker interventions under noisy estimates; when stage 1 substantially exceeds the goal, clinician-specified minima help identify more cost-effective packages. This strategy extends naturally to multi-stage LAGO designs, where lower bounds can be relaxed as estimation precision improves with accumulated data.

Further investigation across varying noise levels, sample sizes, and outcome goals would establish practical guidelines for balancing intervention effectiveness and cost. Other extensions include incorporating time-varying confounders and developing methods for determining optimal values of $K$, $J^{(k)}$, and $n_j^{(k)}$ for LAGO trial design.

The LAGO design is currently applied in several ongoing trials, including PULESA-Uganda \citep{longenecker2025implementation}, a multicomponent BP intervention in HIV clinical settings in Kampala and Wakiso districts, as well as TASKPEN-Zambia \citep{herce2024evaluating}, MAP-IT-Nigeria \citep{aifah2023study}, and HPTN~096 \citep{nelson2022hptn}. The LAGO design is particularly suited for trials with multicomponent interventions, as is often the case for implementation trials.

\vspace{-0.5cm}
\section*{Acknowledgement}
We thank all participants in the EXTRA-CVD study for their participation. 
\par
\vspace{-0.5cm}
\section*{Funding}
This research was supported by NIH Grant HL167936 to Drs.~Lok~and~Spiegelman. The EXTRA-CVD study was supported by NIH grant HL142099 to Drs.~Longenecker, Webel,~and~Bosworth.
\vspace{-0.5cm}
\section*{Disclosure Statement}
The authors report there are no competing interests to declare.
\vspace{-0.5cm}
\section*{Data Availability Statement}
The data that support the findings of this study are available from the corresponding author, Chris T. Longenecker (ctlongen@uw.edu), upon reasonable request.
\bibliography{reference}   

@book{lindvall2002lectures,
  title={Lectures on the coupling method},
  author={Lindvall, Torgny},
  year={2002},
  publisher={Courier Corporation}
}

@book{allison2009fixed,
  title={Fixed effects regression models},
  author={Allison, Paul D},
  year={2009},
  publisher={SAGE publications}
}

@misc{hausman1981panel,
  title={Panel data and unobservable individual effects},
  author={Hausman, Jerry A and Taylor, William E},
  journal={Econometrica: Journal of the Econometric society},
  pages={1377--1398},
  year={1981},
  publisher={JSTOR}
}

@misc{FDA,  title={{A}daptive {D}esigns for {M}edical {D}evice {C}linical {S}tudies: {G}uidance for {I}ndustry and {F}ood and {D}rug {A}dministration {S}taff},  author = {FDA}, journal={FDA},  year={2016},  month={Jan} }

@article{almirall2014introduction,
  title={Introduction to {SMART} designs for the development of adaptive interventions: with application to weight loss research},
  author={Almirall, Daniel and Nahum-Shani, Inbal and Sherwood, Nancy E and Murphy, Susan A},
  journal={Translational behavioral medicine},
  volume={4},
  number={3},
  pages={260--274},
  year={2014},
  publisher={Oxford University Press}
}

@article{collins2007multiphase,
  title={The multiphase optimization strategy {(MOST)} and the sequential multiple assignment randomized trial {(SMART)}: new methods for more potent e{H}ealth interventions},
  author={Collins, Linda M and Murphy, Susan A and Strecher, Victor},
  journal={American journal of preventive medicine},
  volume={32},
  number={5},
  pages={S112--S118},
  year={2007},
  publisher={Elsevier}
}

@article{bing2023learn,
  title={{L}earn-{A}s-you-{GO} ({LAGO}) {T}rials: {O}ptimizing {T}reatments and {P}reventing {T}rial {F}ailure {T}hrough {O}ngoing {L}earning
},
  author={Bing, Ante and Spiegelman, Donna and Nevo, Daniel and Lok, Judith J},
  journal={Biometrics},
  volume={81},
  number={2},
  pages={ujaf061},
  year={2025},
  publisher={Oxford University Press}
}

@article{nevo2021analysis,
  title={Analysis of “{L}earn-as-{Y}ou-{GO}”{(LAGO)} studies},
  author={Nevo, Daniel and Lok, Judith J and Spiegelman, Donna},
  journal={Annals of statistics},
  volume={49},
  number={2},
  pages={793},
  year={2021},
  publisher={NIH Public Access}
}

@book{van2000asymptotic,
  title={{A}symptotic statistics},
  author={Van der Vaart, Aad W},
  volume={3},
  year={2000},
  publisher={Cambridge university press}
}

@book{billingsley1995,
  author    = {Patrick Billingsley},
  title     = {{P}robability and {M}easure},
  edition   = {3rd},
  year      = {1995},
  publisher = {John Wiley \& Sons},
  address   = {New York},
  pages     = {383}
}

@article{longenecker2024nurse,
  title={{N}urse-{L}ed {S}trategy to {I}mprove {B}lood {P}ressure and {C}holesterol {L}evel {A}mong {P}eople {W}ith {HIV}: {A} {R}andomized {C}linical {T}rial},
  author={Longenecker, Christopher T and Jones, Kelley A and Hileman, Corrilynn O and Okeke, Nwora Lance and Gripshover, Barbara M and Aifah, Angela and Bloomfield, Gerald S and Muiruri, Charles and Smith, Valerie A and Vedanthan, Rajesh and others},
  journal={JAMA Network Open},
  volume={7},
  number={3},
  pages={e2356445--e2356445},
  year={2024},
  publisher={American Medical Association}
}

@article{beidas2022promises,
  title={Promises and pitfalls in implementation science from the perspective of {US}-based researchers: learning from a pre-mortem},
  author={Beidas, Rinad S and Dorsey, Shannon and Lewis, Cara C and Lyon, Aaron R and Powell, Byron J and Purtle, Jonathan and Saldana, Lisa and Shelton, Rachel C and Stirman, Shannon Wiltsey and Lane-Fall, Meghan B},
  journal={Implementation Science},
  volume={17},
  number={1},
  pages={55},
  year={2022},
  publisher={Springer}
}

@article{spiegelman2018evaluating,
  title={{E}valuating public health interventions: 8. {C}ausal inference for time-invariant interventions},
  author={Spiegelman, Donna and Zhou, Xin},
  journal={American journal of public health},
  volume={108},
  number={9},
  pages={1187--1190},
  year={2018},
  publisher={American Public Health Association}
}

@article{pallmann2018adaptive,
  title={{A}daptive designs in clinical trials: why use them, and how to run and report them},
  author={Pallmann, Philip and Bedding, Alun W and Choodari-Oskooei, Babak and Dimairo, Munyaradzi and Flight, Laura and Hampson, Lisa V and Holmes, Jane and Mander, Adrian P and Odondi, Lang’o and Sydes, Matthew R and others},
  journal={BMC medicine},
  volume={16},
  pages={1--15},
  year={2018},
  publisher={Springer}
}

@book{scheffe1999analysis,
  title={The analysis of variance},
  author={Scheffe, Henry},
  volume={72},
  year={1999},
  publisher={John Wiley \& Sons}
}

@article{herce2024evaluating,
  title={Evaluating a multifaceted implementation strategy and package of evidence-based interventions based on {WHO} {PEN} for people living with {HIV} and cardiometabolic conditions in {L}usaka, {Z}ambia: protocol for the {TASKPEN} hybrid effectiveness-implementation stepped wedge cluster randomized trial},
  author={Herce, Michael E and Bosomprah, Samuel and Masiye, Felix and Mweemba, Oliver and Edwards, Jessie K and Mandyata, Chomba and Siame, Mmamulatelo and Mwila, Chilambwe and Matenga, Tulani and Frimpong, Christiana and others},
  journal={Implementation Science Communications},
  volume={5},
  number={1},
  pages={61},
  year={2024},
  publisher={Springer}
}

@article{aifah2023study,
  title={{S}tudy design and protocol of a stepped wedge cluster randomized trial using a practical implementation strategy as a model for hypertension-{HIV} integration—the {MAP-IT} trial},
  author={Aifah, Angela A and Hade, Erinn M and Colvin, Calvin and Henry, Daniel and Mishra, Shivani and Rakhra, Ashlin and Onakomaiya, Deborah and Ekanem, Anyiekere and Shedul, Gabriel and Bansal, Geetha P and others},
  journal={Implementation Science},
  volume={18},
  number={1},
  pages={14},
  year={2023},
  publisher={Springer}
}

@misc{nelson2022hptn,
  author={Nelson, L. and Remien, R. and Beyrer, C.},
  title={Hptn 096: Building equity through advocacy},
  howpublished={\url{https://www.hivcenternyc.org/hptn096}},
  year={2022}
}

@misc{bing2025learn,
  title={{L}earn-{A}s-you-{GO} ({LAGO}) {T}rials: {O}ptimizing {T}rials for {E}ffectiveness and {P}ower to {P}revent {F}ailed {T}rials},
  author={Bing, Ante and Spiegelman, Donna and Lok, Judith J},
  journal={arXiv preprint arXiv:2509.11479},
  year={2025}
}

@article{kuhmmer2016effectiveness,
  title={{E}ffectiveness of multidisciplinary intervention on blood pressure control in primary health care: a randomized clinical trial},
  author={Kuhmmer, Regina and Lazzaretti, Rosmeri Kuhmmer and Guterres, C{\'a}tia Moreira and Raimundo, Fabiana Viegas and Leite, Leni Everson Ara{\'u}jo and Delabary, T{\'a}ssia Scholante and Caon, Suhelen and Bastos, Gisele Alsina Nader and Polanczyk, Carisi Anne},
  journal={BMC health services research},
  volume={16},
  number={1},
  pages={456},
  year={2016},
  publisher={Springer}
}

@article{longenecker2025implementation,
  title={Implementation strategies to integrate HIV and hypertension care in Kampala and Wakiso districts, Uganda: study protocol for a stepped wedge cluster randomized trial (PULESA-Uganda)},
  author={Longenecker, Chris T and Kiggundu, John Baptist and Ayebare, Florence and Muddu, Martin and Kayima, James and Mutungi, Gerald and Ssinabulya, Isaac and Schwartz, Jeremy I and Spiegelman, Donna and Tong, Guangyu and others},
  journal={BMC Health Services Research},
  volume={25},
  number={1},
  pages={1060},
  year={2025},
  publisher={Springer}
}

\centerline{\large\bf Supplementary Materials for ``Addressing Confounding}
\vspace{2pt}
\centerline{\large\bf by Indication Through (Un)Measured Center Characteristics}
\vspace{2pt}
\centerline{\large\bf in Learn-As-you-GO (LAGO) Trials''}
\vspace{.4cm}
\centerline{Minh Thu Bui, 
Chris T. Longenecker, 
Ante Bing, 
Donna Spiegelman,} 
\centerline{Allison R. Webel, 
Hayden B. Bosworth, 
Judith J. Lok}

\textbf{Summary}: Online  Appendix~A provides proofs of Theorems~1--3 (Unbiased Estimating Equations, Consistency, and Asymptotic Normality) for continuous outcomes modeled with linear regression. Online Appendix B provides details on hypothesis testing, and Online Appendix C covers confidence sets and bands. Online Appendices~D and~E cover the theoretical extensions to binary outcomes modeled with logistic regression and to GLM setting with general link functions, respectively. Online  Appendix~F extends the theoretical framework to $K > 2$ stages. Online Appendices~G and~H present additional simulation results under the cubic cost function and additional EXTRA-CVD analyses.

\setcounter{section}{0}
\setcounter{equation}{0}
\renewcommand{\thesection}{\Alph{section}}

\section{Proofs of Theorems \ref{unbiasedEE}, \ref{theoremConsistencyWithLearning}, and \ref{AsymptoticTheoremWithLearning}}\label{proofs}
\subsection{Proof of Unbiased Estimating Equations Theorem \ref{unbiasedEE}.}\label{UnbiasedEEProofAppendix}
\begin{proof}
We show that the expectation of each term in equation (\ref{eqn:OLS_ignore_ee}) evaluated at $\B^{*}$ is 0 conditional on $\overline{\BoldZ}$.
\begin{align*}
    \mathbbm{E}[\,U(\boldsymbol{\beta}^{*})\mid \overline{\BoldZ}] &=  \frac{1}{n} \left \{ \sum_{t=1}^{T^{(1)}} \sum_{j=1}^{J^{(1)}} \sum_{i=1}^{n_{jt}^{(1)}}  \E \left[ \begin{pmatrix} \Aone\\ \BoldLone \end{pmatrix}\left(\Yone - \begin{pmatrix} \boldsymbol{\beta}^{*}_A \\ \B_{\ell}^{*} \end{pmatrix}^T \begin{pmatrix} \Aone\\ \BoldLone \end{pmatrix} \right)\middle| \overline{\BoldZ} \right] \right .\\
    &+\sum_{t=T^{(1)}+1}^{T^{(2)}} \sum_{j=1}^{J^{(2)}} \sum_{i=1}^{n_{jt}^{(2)}} \left . \E \left [\begin{pmatrix} \Atwolago\\ \BoldLtwolago \end{pmatrix} \left(\Ytwolago - \begin{pmatrix} \boldsymbol{\beta}^{*}_A \\ \B^{*}_{\ell} \end{pmatrix}^T \begin{pmatrix} \Atwolago\\ \BoldLtwolago \end{pmatrix} \right)\middle| \overline{\BoldZ}\right ] \right \} \numberthis \label{expectationofEE}.
\end{align*}
For the first term in equation (\ref{expectationofEE}) for center $j$, observe that
\begin{align*}
    &\E \left[\begin{pmatrix} \Aone\\ \BoldLone \end{pmatrix} \left(\Yone - \begin{pmatrix} \boldsymbol{\beta}^{*}_A \\ \B^{*}_{\ell} \end{pmatrix}^T \begin{pmatrix} \Aone\\ \BoldLone \end{pmatrix}\right) \middle| \overline{\BoldZ}\right] \\
    &= \int_{\overline{\Bolda}} \E \left[\begin{pmatrix} \Bolda_{jt} \\ \BoldLone \end{pmatrix} \left(Y_{ijt}^{(1)}\left(\Bolda_{jt} \right) - \begin{pmatrix} \boldsymbol{\beta}^{*}_A \\ \B^{*}_{\ell} \end{pmatrix}^T \begin{pmatrix} \Bolda_{jt} \\ \BoldLone \end{pmatrix} \right) \middle | \BarA^{(1)} = \overline{\Bolda}, \overline{\BoldZ} \right] \cdot f_{\BarA^{(1)}|\overline{\BoldZ}} \left(\overline{\Bolda}\right) \, d\overline{\Bolda}\\
    &=  \int_{\overline{\Bolda}} \begin{pmatrix} \Bolda_{jt} \\ \BoldLone \end{pmatrix} \E \left[ \left(Y_{ijt}^{(1)}\left(\Bolda_{jt} \right) - \begin{pmatrix} \boldsymbol{\beta}^{*}_A \\ \B^{*}_{\ell} \end{pmatrix}^T \begin{pmatrix} \Bolda_{jt} \\ \BoldLone \end{pmatrix} \right) \middle |\overline{\BoldZ} \right] \cdot f_{\BarA^{(1)}|\overline{\BoldZ}} \left(\overline{\Bolda}\right) \, d\overline{\Bolda}\\
    &= \int_{\overline{\Bolda}} \begin{pmatrix} \Bolda_{jt} \\ \BoldLone \end{pmatrix} \cdot (0) \, \cdot f_{\BarA^{(1)}|\overline{\BoldZ}} \left(\overline{\Bolda}\right) \, d\overline{\Bolda} = 0,
\end{align*}
where the first equality follows from Consistency Assumption \ref{asspConsistency}, the second equality follows from Conditional Exchangeability Assumption \ref{asspCondExcha}, and the third equality follows from Main Model Assumption \ref{mainmodel}, Outcome Errors Assumption \ref{epsilonassumption}, and Conditional Exchangeability Between Centers Assumption \ref{asspcenterIndep}. 
\\

Similarly, for the second term of equation (\ref{expectationofEE}), 
\begin{align*}
    &\E \left [\begin{pmatrix} \Atwolago\\ \BoldLtwolago \end{pmatrix} \left(\Ytwolago - \begin{pmatrix} \boldsymbol{\beta}^{*}_A \\ \B^{*}_{\ell} \end{pmatrix}^T \begin{pmatrix} \Atwolago\\ \BoldLtwolago \end{pmatrix} \right) \middle | \overline{\BoldZ} \right ] \\
    &= \int_{\boldsymbol{x}} \int_{\overline{\Bolda}} \E \left [\begin{pmatrix} \Bolda_{jt} \\ \BoldLtwolago \end{pmatrix} \left(Y_{ijt}^{(2)}(\Bolda_{jt}) - \begin{pmatrix} \boldsymbol{\beta}^{*}_A \\ \B^{*}_{\ell} \end{pmatrix}^T \begin{pmatrix} \Bolda_{jt}\\ \BoldLtwolago \end{pmatrix} \right)\middle | \BarA^{(2,n^{(1)})} = \overline{\Bolda}, \X^{(2,n^{(1)})} = \boldsymbol{x}, \overline{\BoldZ}\right ] \\
    &\phantom{=============================}\cdot f_{\BarA^{(2,n^{(1)})}, \X^{(2,n^{(1)})} | \overline{\BoldZ}} (\overline{\Bolda}, \boldsymbol{x}|\overline{\BoldZ}) \,d\overline{\Bolda} \,d\boldsymbol{x}\\
     &= \int_{\boldsymbol{x}} \int_{\overline{\Bolda}} \begin{pmatrix} \Bolda_{jt} \\ \BoldLtwolago \end{pmatrix} \E \left [\left(Y_{ijt}^{(2)}(\Bolda_{jt}) - \begin{pmatrix} \boldsymbol{\beta}^{*}_A \\ \B^{*}_{\ell} \end{pmatrix}^T \begin{pmatrix} \Bolda_{jt}\\ \BoldLtwolago \end{pmatrix} \right)\middle |  \overline{\BoldZ}\right ] \cdot f_{\BarA^{(2,n^{(1)})}, \X^{(2,n^{(1)})} | \overline{\BoldZ}} (\overline{\Bolda}, \boldsymbol{x}|\overline{\BoldZ}) \,d\overline{\Bolda} \,d\boldsymbol{x}\\
    &= \int_{\boldsymbol{x}} \int_{\overline{\Bolda}} \begin{pmatrix} \Bolda_{jt} \\ \BoldLtwolago \end{pmatrix} \cdot (0) \cdot f_{\BarA^{(2,n^{(1)})}, \X^{(2,n^{(1)})} | \overline{\BoldZ}} (\overline{\Bolda}, \boldsymbol{x}|\overline{\BoldZ}) \,d\overline{\Bolda} \,d\boldsymbol{x} = 0, \numberthis \label{expectationstage2zero}
\end{align*}
where the first equality follows from Consistency Assumption \ref{asspConsistency}, the second equality follows from Conditional Exchangeability Assumption \ref{asspCondExcha}, and the third equality follows from Main Model Assumption \ref{mainmodel}, Outcome Errors Assumption \ref{epsilonassumption}, and Conditional Independence Between Centers Assumption \ref{asspcenterIndep}. 

It follows that $\E(U(\B^{*})\mid \overline{\BoldZ}) = 0$. 
\end{proof}

\subsection[Proof of Consistency]{Proof of Consistency of $\hat{\B}$ Theorem \ref{theoremConsistencyWithLearning}.}\label{ConsistencyProofAppendix}
\begin{proof}
To prove Consistency Theorem \ref{theoremConsistencyWithLearning}, we use Theorem 5.9 of \cite{van2000asymptotic}.
First,

\begin{align*}
    U(\B) &= U(\B^{*}) - U(\B^{*}) + U(\B)\\
    &= U(\B^{*}) - \frac{1}{n} \sum_{t=1}^{T^{(1)}} \sum_{j=1}^{J^{(1)}} \sum_{i=1}^{n_{jt}^{(1)}} \begin{pmatrix} \Aone\\ \BoldLone \end{pmatrix} \left [\Yone - \begin{pmatrix} \boldsymbol{\beta}^{*}_A \\ \B^{*}_{\ell} \end{pmatrix}^T \begin{pmatrix} \Aone\\ \BoldLone \end{pmatrix} \right ] \\
    &\phantom{= U(\B^{*})}- \frac{1}{n}\sum_{t=T^{(1)}+1}^{T^{(2)}} \sum_{j=1}^{J^{(2)}} \sum_{i=1}^{n_{jt}^{(2)}} \begin{pmatrix} \Atwolago\\ \BoldLtwolago \end{pmatrix} \left [\Ytwolago - \begin{pmatrix} \boldsymbol{\beta}^{*}_A \\ \B^{*}_{\ell} \end{pmatrix}^T \begin{pmatrix} \Atwolago\\ \BoldLtwolago \end{pmatrix} \right ] \\
    &\phantom{= U(\B^{*})}+ \frac{1}{n} \sum_{t=1}^{T^{(1)}} \sum_{j=1}^{J^{(1)}} \sum_{i=1}^{n_{jt}^{(1)}} \begin{pmatrix} \Aone\\ \BoldLone \end{pmatrix} \left [\Yone - \begin{pmatrix} \boldsymbol{\beta}_A \\ \B_{\ell} \end{pmatrix}^T \begin{pmatrix} \Aone\\ \BoldLone \end{pmatrix} \right ] \\
    &\phantom{= U(\B^{*})}+ \frac{1}{n}\sum_{t=T^{(1)}+1}^{T^{(2)}} \sum_{j=1}^{J^{(2)}} \sum_{i=1}^{n_{jt}^{(2)}} \begin{pmatrix} \Atwolago\\ \BoldLtwolago \end{pmatrix} \left [\Ytwolago - \begin{pmatrix} \boldsymbol{\beta}_A \\ \B_{\ell} \end{pmatrix}^T \begin{pmatrix} \Atwolago\\ \BoldLtwolago \end{pmatrix} \right ]\\
    &= U(\boldsymbol{\beta}^{*}) +  \sum_{t=1}^{T^{(1)}} \sum_{j=1}^{J^{(1)}} \frac{n_{jt}^{(1)}}{n} \begin{pmatrix} \Aone\\ \BoldLone \end{pmatrix} \left [ \begin{pmatrix} \boldsymbol{\beta}^{*}_A \\  \B^{*}_{\ell} \end{pmatrix}^T - \begin{pmatrix} \boldsymbol{\beta}_A \\  \B_{\ell} \end{pmatrix}^T \right ] \begin{pmatrix} \Aone\\ \BoldLone \end{pmatrix} \\
    &\phantom{= U(\boldsymbol{\beta}^{*})}+ \sum_{t=T^{(1)}+1}^{T^{(2)}} \sum_{j=1}^{J^{(2)}} \frac{n_{jt}^{(2)}}{n} \begin{pmatrix} \Atwolago\\ \BoldLtwolago \end{pmatrix} \left [ \begin{pmatrix} \boldsymbol{\beta}^{*}_A \\  \B^{*}_{\ell} \end{pmatrix}^T - \begin{pmatrix} \boldsymbol{\beta}_A \\  \B_{\ell} \end{pmatrix}^T \right ] \begin{pmatrix} \Atwolago\\ \BoldLtwolago\end{pmatrix}. \numberthis \label{U(B)}   
\end{align*}

Using the definition of $u(\B)$ of equation of (\ref{smallu}), it follows that
\begin{align*}
    &U(\B) - u(\B)\\
    &= U(\boldsymbol{\beta}^{*}) +  \sum_{t=1}^{T^{(1)}} \sum_{j=1}^{J^{(1)}} \frac{n_{jt}^{(1)}}{n} \begin{pmatrix} \Aone\\ \BoldLone \end{pmatrix} \left [ \begin{pmatrix} \boldsymbol{\beta}^{*}_A \\  \B^{*}_{\ell} \end{pmatrix}^T - \begin{pmatrix} \boldsymbol{\beta}_A \\  \B_{\ell} \end{pmatrix}^T \right ] \begin{pmatrix} \Aone\\ \BoldLone \end{pmatrix} \\
    &\phantom{U(\B^{*})} + \sum_{t=T^{(1)}+1}^{T^{(2)}} \sum_{j=1}^{J^{(2)}} \frac{n_{jt}^{(2)}}{n} \begin{pmatrix} \Atwolago\\ \BoldLtwolago \end{pmatrix} \left [ \begin{pmatrix} \boldsymbol{\beta}^{*}_A \\  \B^{*}_{\ell} \end{pmatrix}^T - \begin{pmatrix} \boldsymbol{\beta}_A \\  \B_{\ell} \end{pmatrix}^T \right ] \begin{pmatrix} \Atwolago\\ \BoldLtwolago\end{pmatrix} \\
    &\phantom{U(\B^{*})} - \sum_{t=1}^{T^{(1)}} \sum_{j=1}^{J^{(1)}} \alpha_{jt}^{(1)}  \begin{pmatrix} \A_{jt}^{(1)}\\ \BoldLone \end{pmatrix} \left [ \begin{pmatrix} \boldsymbol{\beta}^{*}_A \\  \B^{*}_{\ell} \end{pmatrix}^T - \begin{pmatrix} \boldsymbol{\beta}_A \\  \B_{\ell} \end{pmatrix}^T \right ] \begin{pmatrix} \A_{jt}^{(1)}\\ \BoldLone \end{pmatrix} \\
    &\phantom{U(\B^{*})}-  \sum_{t=T^{(1)}+1}^{T^{(2)}} \sum_{j=1}^{J^{(2)}} \alpha_{jt}^{(2)}  \begin{pmatrix} \Bolda_{jt}^{(2)}\\ \BoldLtwo \end{pmatrix} \left [ \begin{pmatrix} \boldsymbol{\beta}^{*}_A \\  \B^{*}_{\ell} \end{pmatrix}^T - \begin{pmatrix} \boldsymbol{\beta}_A \\  \B_{\ell} \end{pmatrix}^T \right ] \begin{pmatrix} \Bolda_{jt}^{(2)}\\ \BoldLtwo \end{pmatrix}\\
    &= U(\boldsymbol{\beta}^{*}) +  \sum_{t=1}^{T^{(1)}} \sum_{j=1}^{J^{(1)}} \left(\frac{n_{jt}^{(1)}}{n} - \alpha_{jt}^{(1)}\right)\begin{pmatrix} \Aone\\ \BoldLone \end{pmatrix} \left [ \begin{pmatrix} \boldsymbol{\beta}^{*}_A \\  \B^{*}_{\ell} \end{pmatrix}^T - \begin{pmatrix} \boldsymbol{\beta}_A \\  \B_{\ell} \end{pmatrix}^T \right ] \begin{pmatrix} \Aone\\ \BoldLone \end{pmatrix} \\
    &\phantom{= U(\B^{*})} +  \sum_{t=T^{(1)}+1}^{T^{(2)}} \sum_{j=1}^{J^{(2)}} \left(\frac{n_{jt}^{(2)}}{n} - \alpha_{jt}^{(2)}\right)\begin{pmatrix} \Atwo\\ \BoldLtwo\end{pmatrix} \left [ \begin{pmatrix} \boldsymbol{\beta}^{*}_A \\  \B^{*}_{\ell} \end{pmatrix}^T - \begin{pmatrix} \boldsymbol{\beta}_A \\  \B_{\ell} \end{pmatrix}^T \right ] \begin{pmatrix} \Atwo\\ \BoldLtwo \end{pmatrix} \\
    &\phantom{= U(\B^{*})} + \sum_{t=T^{(1)}+1}^{T^{(2)}} \sum_{j=1}^{J^{(2)}}  \alpha_{jt}^{(2)}\left [\begin{pmatrix} \Atwolago\\ \BoldLtwolago \end{pmatrix}\begin{pmatrix} \boldsymbol{\beta}^{*}_A \\ \B^{*}_{\ell} \end{pmatrix}^T\begin{pmatrix} \Atwolago\\ \BoldLtwolago \end{pmatrix} - \begin{pmatrix} \Bolda_{jt}^{(2)}\\ \BoldLtwo \end{pmatrix}\begin{pmatrix} \boldsymbol{\beta}^{*}_A \\ \B^{*}_{\ell} \end{pmatrix}^T\begin{pmatrix} \Bolda_{jt}^{(2)}\\ \BoldLtwo \end{pmatrix} \right ]\\
    &\phantom{= U(\B^{*})} + \sum_{t=T^{(1)}+1}^{T^{(2)}} \sum_{j=1}^{J^{(2)}} \alpha_{jt}^{(2)}\left [\begin{pmatrix} \Bolda_{jt}^{(2)}\\ \BoldLtwo \end{pmatrix}\begin{pmatrix} \boldsymbol{\beta}_A \\ \B_{\ell} \end{pmatrix}^T\begin{pmatrix} \Bolda_{jt}^{(2)}\\ \BoldLtwo \end{pmatrix} - \begin{pmatrix} \Atwolago\\ \BoldLtwolago \end{pmatrix}\begin{pmatrix} \boldsymbol{\beta}_A \\ \B_{\ell} \end{pmatrix}^T\begin{pmatrix} \Atwolago\\ \BoldLtwolago \end{pmatrix} \right ]
\end{align*}
\begin{align*}
    &= U(\boldsymbol{\beta}^{*}) +  \sum_{t=1}^{T^{(1)}} \sum_{j=1}^{J^{(1)}} \left(\frac{n_{jt}^{(1)}}{n} - \alpha_{j1 }\right)\begin{pmatrix} \Aone\\ \BoldLone \end{pmatrix} \left [ \begin{pmatrix} \boldsymbol{\beta}^{*}_A \\  \B^{*}_{\ell} \end{pmatrix}^T - \begin{pmatrix} \boldsymbol{\beta}_A \\  \B_{\ell} \end{pmatrix}^T \right ] \begin{pmatrix} \Aone\\ \BoldLone \end{pmatrix} \\
    &\phantom{= U(\B^{*})} + \sum_{t=T^{(1)}+1}^{T^{(2)}} \sum_{j=1}^{J^{(2)}} \left(\frac{n_{jt}^{(2)}}{n} - \alpha_{jt}^{(2)}\right)\begin{pmatrix} \Atwo\\ \BoldLtwo\end{pmatrix} \left [ \begin{pmatrix} \boldsymbol{\beta}^{*}_A \\  \B^{*}_{\ell} \end{pmatrix}^T - \begin{pmatrix} \boldsymbol{\beta}_A \\  \B_{\ell} \end{pmatrix}^T \right ] \begin{pmatrix} \Atwo\\ \BoldLtwo \end{pmatrix} \\
    &\phantom{= U(\B^{*})} + \sum_{t=T^{(1)}+1}^{T^{(2)}} \sum_{j=1}^{J^{(2)}}  \alpha_{jt}^{(2)}\left [\begin{pmatrix} \Atwolago\\ \BoldLtwolago \end{pmatrix}\begin{pmatrix} \Atwolago\\ \BoldLtwolago \end{pmatrix}^T - \begin{pmatrix} \Bolda_{jt}^{(2)}\\ \BoldLtwo \end{pmatrix}\begin{pmatrix} \Bolda_{jt}^{(2)}\\ \BoldLtwo \end{pmatrix}^T \right ] \begin{pmatrix} \boldsymbol{\beta}^{*}_A \\ \B^{*}_{\ell} \end{pmatrix}\\
    &\phantom{= U(\B^{*})} + \sum_{t=T^{(1)}+1}^{T^{(2)}} \sum_{j=1}^{J^{(2)}} \alpha_{jt}^{(2)}\left [\begin{pmatrix} \Bolda_{jt}^{(2)}\\ \BoldLtwo \end{pmatrix}\begin{pmatrix} \Bolda_{jt}^{(2)}\\ \BoldLtwo \end{pmatrix}^T - \begin{pmatrix} \Atwolago\\ \BoldLtwolago \end{pmatrix}\begin{pmatrix} \Atwolago\\ \BoldLtwolago \end{pmatrix}^T \right ]\begin{pmatrix} \boldsymbol{\beta}_A \\ \B_{\ell} \end{pmatrix}\\
    &= U(\B^{*}) + G_1 + G_2 + G_3 + G_4,
\end{align*}
with
\begin{align}
    G_1 &= \sum_{t=1}^{T^{(1)}} \sum_{j=1}^{J^{(1)}} \left(\frac{n_{jt}^{(1)}}{n} - \alpha_{j1 }\right)\begin{pmatrix} \Aone\\ \BoldLone \end{pmatrix} \left [ \begin{pmatrix} \boldsymbol{\beta}^{*}_A \\  \B^{*}_{\ell} \end{pmatrix}^T - \begin{pmatrix} \boldsymbol{\beta}_A \\  \B_{\ell} \end{pmatrix}^T \right ] \begin{pmatrix} \Aone\\ \BoldLone \end{pmatrix} \label{G_1} \\
    G_2 &= \sum_{t=T^{(1)}+1}^{T^{(2)}} \sum_{j=1}^{J^{(2)}} \left(\frac{n_{jt}^{(2)}}{n} - \alpha_{jt}^{(2)}\right)\begin{pmatrix} \Atwo\\ \BoldLtwo\end{pmatrix} \left [ \begin{pmatrix} \boldsymbol{\beta}^{*}_A \\  \B^{*}_{\ell} \end{pmatrix}^T - \begin{pmatrix} \boldsymbol{\beta}_A \\  \B_{\ell} \end{pmatrix}^T \right ] \begin{pmatrix} \Atwo\\ \BoldLtwo \end{pmatrix} \label{G_2} \\
    G_3 &= \sum_{t=T^{(1)}+1}^{T^{(2)}} \sum_{j=1}^{J^{(2)}}  \alpha_{jt}^{(2)}\left [\begin{pmatrix} \Atwolago\\ \BoldLtwolago \end{pmatrix}\begin{pmatrix} \Atwolago\\ \BoldLtwolago \end{pmatrix}^T - \begin{pmatrix} \Bolda_{jt}^{(2)}\\ \BoldLtwo \end{pmatrix}\begin{pmatrix} \Bolda_{jt}^{(2)}\\ \BoldLtwo \end{pmatrix}^T \right ] \begin{pmatrix} \boldsymbol{\beta}^{*}_A \\ \B^{*}_{\ell} \end{pmatrix} \label{G_3} \\
    G_4 &= \sum_{t=T^{(1)}+1}^{T^{(2)}} \sum_{j=1}^{J^{(2)}} \alpha_{jt}^{(2)}\left [\begin{pmatrix} \Bolda_{jt}^{(2)}\\ \BoldLtwo \end{pmatrix}\begin{pmatrix} \Bolda_{jt}^{(2)}\\ \BoldLtwo \end{pmatrix}^T - \begin{pmatrix} \Atwolago\\ \BoldLtwolago \end{pmatrix}\begin{pmatrix} \Atwolago\\ \BoldLtwolago \end{pmatrix}^T \right ]\begin{pmatrix} \boldsymbol{\beta}_A \\ \B_{\ell} \end{pmatrix}.  \label{G_4}
\end{align}   
Because of the triangle inequality for the supremum norm, it suffices to show convergence in probability to 0 for each term $U(\B^{*})$, $G_1$, $G_2, G_3$, and $G_4$ separately. 

Next, we assess the correlation between the two summands in $U(\B^{*})$. Observe that for any $i, \Tilde{i}, j, \Tilde{j}$,
\begin{align*}
&\E \left [ \begin{pmatrix} \Aone\\ \BoldLone \end{pmatrix} \left [\Yone - \begin{pmatrix} \boldsymbol{\beta}^{*}_A \\ \B^{*}_{\ell} \end{pmatrix}^T \begin{pmatrix} \Aone\\ \BoldLone \end{pmatrix} \right ] \begin{pmatrix} \Atwodifflago\\ \BoldLtwodifflago \end{pmatrix}^T \left [\Ytwodifflago - \begin{pmatrix} \boldsymbol{\beta}^{*}_A \\ \B^{*}_{\ell} \end{pmatrix}^T \begin{pmatrix} \Atwodifflago \\ \BoldLtwodifflago \end{pmatrix} \right ] \middle | \overline{\BoldZ} \right ]\\
&= \E \left[ \E \left[ 
    \begin{pmatrix} \Aone\\ \BoldLone \end{pmatrix} 
    \left[\Yone - 
    \begin{pmatrix} \boldsymbol{\beta}^{*}_A \\ 
    \B^{*}_{\ell} \end{pmatrix}^T 
    \begin{pmatrix} \Aone\\ \BoldLone \end{pmatrix} 
    \right] \right.\right. \nonumber \\
&\phantom{==} \left.\left. \cdot
    \begin{pmatrix} \Atwodifflago\\ 
    \BoldLtwodifflago \end{pmatrix}^T 
    \left[\Ytwodifflago - 
    \begin{pmatrix} \boldsymbol{\beta}^{*}_A \\ 
    \B^{*}_{\ell} \end{pmatrix}^T 
    \begin{pmatrix} \Atwodifflago \\ 
    \BoldLtwodifflago \end{pmatrix} 
    \right] 
    \,\middle|\, \X^{(2,n^{(1)})}, \overline{\BoldZ} 
    \right] 
    \,\middle|\, \overline{\BoldZ} 
\right]\\
&= \E \left . \left[ \E \left[ \begin{pmatrix} \Aone\\ \BoldLone \end{pmatrix} \left [\Yone - \begin{pmatrix} \boldsymbol{\beta}^{*}_A \\ \B^{*}_{\ell} \end{pmatrix}^T \begin{pmatrix} \Aone\\ \BoldLone \end{pmatrix} \right ] \right| \X^{(2,n^{(1)})}, \overline{\BoldZ}\right] \right.  \\
&\quad \cdot \E  \left .\left. \left[ \begin{pmatrix} \Atwodifflago\\ \BoldLtwodifflago \end{pmatrix} \left [\Ytwodifflago - \begin{pmatrix} \boldsymbol{\beta}^{*}_A \\ \B^{*}_{\ell} \end{pmatrix}^T \begin{pmatrix} \Atwodifflago\\ \BoldLtwodifflago \end{pmatrix} \right ] \right| \X^{(2,n^{(1)})}, \overline{\BoldZ} \right] \middle |\overline{\BoldZ}\right]\\
&= \E \left . \left[ \E \left[ \begin{pmatrix} \Aone\\ \BoldLone \end{pmatrix} \left(\Yone - \begin{pmatrix} \boldsymbol{\beta}^{*}_A \\ \B^{*}_{\ell} \end{pmatrix}^T \begin{pmatrix} \Aone\\ \BoldLone \end{pmatrix} \right) \right| \X^{(2,n^{(1)})}, \overline{\BoldZ} \right] \right. \\
&\quad \cdot  \left .  \E \left [\E \left . \left. \left[ \begin{pmatrix} \Atwodifflago\\ \BoldLtwodifflago \end{pmatrix} \left [\Ytwodifflago - \begin{pmatrix} \boldsymbol{\beta}^{*}_A \\ \B^{*}_{\ell} \end{pmatrix}^T \begin{pmatrix} \Atwodifflago\\ \BoldLtwodifflago \end{pmatrix} \right ] \right| \BarA^{(2,n^{(1)})}, \X^{(2,n^{(1)})}, \overline{\BoldZ} \right] \right | \X^{(2,n^{(1)})}, \overline{\BoldZ}\right ] \middle |\overline{\BoldZ}\right], \numberthis \label{correlationProof}
\end{align*}
where the second inequality follows from LAGO Assumption \ref{asspAnte}. Observe that

\begin{align*}
    & \E \left[ \begin{pmatrix} \Atwodifflago\\ \BoldLtwodifflago \end{pmatrix} \left [\Ytwodifflago(\Atwodifflago) - \begin{pmatrix} \boldsymbol{\beta}^{*}_A \\ \B^{*}_{\ell} \end{pmatrix}^T \begin{pmatrix} \Atwodifflago\\ \BoldLtwodifflago \end{pmatrix} \right ] \,\middle|\, \overline{\A}^{(2,n^{(1)})} = \overline{\Bolda}, \X^{(2,n^{(1)})} = \boldsymbol{x}, \overline{\BoldZ} \right] \\
    &= \E \left[ \begin{pmatrix} \Bolda_{\Tilde{j}}\\ \BoldLtwodifflago \end{pmatrix} \left [\Ytwodifflago(\Bolda_{\Tilde{j}}) - \begin{pmatrix} \boldsymbol{\beta}^{*}_A \\ \B^{*}_{\ell} \end{pmatrix}^T \begin{pmatrix} \Bolda_{\Tilde{j}} \\ \BoldLtwodifflago \end{pmatrix} \right ] \,\middle|\, \overline{\A}^{(2,n^{(1)})} = \overline{\Bolda}, \X^{(2,n^{(1)})} = \boldsymbol{x},  \overline{\BoldZ} \right] \\
    &= \E \left [\begin{pmatrix} \Bolda_{\Tilde{j}}\\ \BoldLtwodifflago \end{pmatrix}  \left [\Ytwodifflago(\Bolda_{\Tilde{j}}) - \begin{pmatrix} \boldsymbol{\beta}^{*}_A \\ \B^{*}_{\ell} \end{pmatrix}^T \begin{pmatrix} \Bolda_{\Tilde{j}} \\ \BoldLtwodifflago \end{pmatrix} \right ] \,\middle|\,\overline{\BoldZ} \right ]  = 0
\end{align*}
where the first equality follows from Consistency Assumption \ref{asspConsistency}, the second equality follows from Conditional Exchangeability Assumption \ref{asspCondExcha}, and the last equality follows from Main Model Assumption \ref{mainmodel}, Outcome Errors Assumption \ref{epsilonassumption}, and Conditional Independence Between centers Assumption \ref{asspcenterIndep}. Thus, the correlation between the two terms in equation (\ref{correlationProof}) equals 0, and each stage can be considered separately when evaluating $\Var \left(U(\B^{*}) \right)$.

Because each term has expectation 0 (see the proof of Theorem \ref{unbiasedEE}), the variance of the first entry $U_{2,1}(\B^{*})$ of the stage-2 summand of $U(\B^{*})$ equals 

\begin{align*}
    &\Var[U_{2,1}(\B^{*})\mid \overline{\BoldZ}] \\
    &= \E[U_{2,1}(\B^{*})^2 \mid \overline{\BoldZ}]\\
    &= \E  \left [\frac{1}{n^2} \sum_{t=T^{(1)}+1}^{T^{(2)}} \sum_{j=1}^{J^{(2)}} \sum_{i=1}^{n_{jt}^{(2)}} \sum_{\Tilde{t}=T^{(1)}+1}^{T^{(2)}} \sum_{\Tilde{j}=1}^{J^{(2)}} \sum_{\Tilde{i}=1}^{n_{\Tilde{j}\Tilde{t}}^{(2)}} \A_{jt,1}^{(2,n^{(1)})} \A_{\Tilde{j}\Tilde{t},1}^{(2,n^{(1)})} \left (\Ytwo - \begin{pmatrix} \boldsymbol{\beta}^{*}_A \\ \B^{*}_{\ell} \end{pmatrix}^T \begin{pmatrix} \A_{jt}^{(2,n^{(1)})}\\ \BoldLtwo \end{pmatrix} \right ) \right.\\
    &\phantom{============================} \left. \left (Y_{\Tilde{i}\Tilde{j}}^{(2,n^{(1)})} - \begin{pmatrix} \boldsymbol{\beta}^{*}_A \\ \B^{*}_{\ell} \end{pmatrix}^T \begin{pmatrix} \A_{\Tilde{j}\Tilde{t}}^{(2,n^{(1)})}\\ \BoldL_{\Tilde{j}}^{(2)} \end{pmatrix} \right) \middle | \overline{\BoldZ} \right ]\\
    &= \E \left [ \E \left [ \frac{1}{n^2} \sum_{t=T^{(1)}+1}^{T^{(2)}} \sum_{j=1}^{J^{(2)}} \sum_{i=1}^{n_{jt}^{(2)}} \sum_{\Tilde{t}=T^{(1)}+1}^{T^{(2)}} \sum_{\Tilde{j}=1}^{J^{(2)}} \sum_{\Tilde{i}=1}^{n_{\Tilde{j}\Tilde{t}}^{(2)}} \A_{jt}^{(2,n^{(1)})} \A_{\Tilde{j}\Tilde{t}}^{(2,n^{(1)})} \left(\Ytwo - \begin{pmatrix} \boldsymbol{\beta}^{*}_A \\ \B^{*}_{\ell} \end{pmatrix}^T \begin{pmatrix} \A_{jt}^{(2,n^{(1)})}\\ \BoldLtwo \end{pmatrix} \right) \right. \right.\\
     &\phantom{=====================}\left.\left.\left(Y_{\Tilde{i}\Tilde{j}\Tilde{t}}^{(2,n^{(1)})} - \begin{pmatrix} \boldsymbol{\beta}^{*}_A \\ \B^{*}_{\ell} \end{pmatrix}^T \begin{pmatrix} \A_{\Tilde{j}\Tilde{t},1}^{(2,n^{(1)})}\\ \BoldL_{\Tilde{j}}^{(2)} \end{pmatrix} \right) \middle |  \BarA^{(2,n^{(1)})}, \overline{\BoldZ} \right ] \middle | \overline{\BoldZ}\right ]\\
    &= \frac{1}{n^2}\sum_{t=T^{(1)}+1}^{T^{(2)}} \sum_{j=1}^{J^{(2)}} \sum_{i=1}^{n_{jt}^{(2)}}
    \int_{\overline{\Bolda}}  \E \left [ \left(\Bolda_{jt,1} \right)^2 
    \left(Y_{ijt}^{(2,n^{(1)})}(\Bolda_{jt}) - \begin{pmatrix} \boldsymbol{\beta}^{*}_A \\ \B^{*}_{\ell} \end{pmatrix}^T \begin{pmatrix} \Bolda_{jt} \\ \BoldLtwo \end{pmatrix} \right)^2 \middle | \BarA^{(2,n^{(1)})} = \overline{\Bolda}, \overline{\BoldZ}\right ] \nonumber \\
    &\phantom{==========================================} \cdot f_{\BarA^{(2,n^{(1)})}|\overline{\BoldZ}} (\overline{\Bolda}) \, d\overline{\Bolda}\\
    &= \frac{1}{n^2}\sum_{t=T^{(1)}+1}^{T^{(2)}} \sum_{j=1}^{J^{(2)}} \sum_{i=1}^{n_{jt}^{(2)}}\int_{\overline{\Bolda}}  \E \left [ \left(\Bolda_{jt,1} \right)^2 \left(Y_{ijt}^{(2,n^{(1)})}(\Bolda_{jt}) - \begin{pmatrix} \boldsymbol{\beta}^{*}_A \\ \B^{*}_{\ell} \end{pmatrix}^T \begin{pmatrix} \Bolda_{jt} \\ \BoldLtwo \end{pmatrix} \right)^2 \middle | \overline{\BoldZ}\right ] \cdot f_{\BarA^{(2,n^{(1)})}|\overline{\BoldZ}} (\overline{\Bolda}) \, d\overline{\Bolda}\\
    &=\sum_{t=T^{(1)}+1}^{T^{(2)}} \sum_{j=1}^{J^{(2)}} \sum_{i=1}^{n_{jt}^{(2)}} \int_{\overline{\Bolda}} \frac{1}{n^2}  \left(\Bolda_{j,1} \right)^2  \sigma^2 \left(\BoldZ_j\right) \cdot f_{\BarA^{(2,n^{(1)})}|\overline{\BoldZ}} (\overline{\Bolda}) \, d\overline{\Bolda}\\
    &\leq \frac{1}{n^2}\sum_{t=T^{(1)}+1}^{T^{(2)}} \sum_{j=1}^{J^{(2)}} \sum_{i=1}^{n_{jt}^{(2)}} \int_{\overline{\Bolda}} C_A^2 \, \sigma^2 \cdot f_{\BarA^{(2,n^{(1)})}|\overline{\BoldZ}} (\overline{\Bolda}) \, d\overline{\Bolda}\\
    &= \frac{n^{(2)}}{n^2} C_A^2 \, \sigma^2 \rightarrow 0 \; \text{as} \; n \rightarrow \infty.
\end{align*}
where the fourth equality holds because if $(i, j) \neq (\Tilde{i}, \Tilde{j})$, $\Ytwo$ and $ Y_{\Tilde{i}\Tilde{j}}^{(2,n^{(1)})}$ are independent conditionally on $\bigl(\BarA^{(2,n^{(1)})}, \overline{\BoldZ}\bigr)$ following Conditional Independence Between Outcome Errors Assumption \ref{asspMinh} and Consistency Assumption \ref{asspConsistency}, the fifth equality follows from Conditional Exchangeability Assumption \ref{asspCondExcha}, and the sixth equality follows
from Main Model Assumption \ref{mainmodel}, Outcome Errors Assumption \ref{epsilonassumption}, and Conditional Independence Between Centers Assumption \ref{asspcenterIndep}. The inequality holds because of Outcome Errors Assumption \ref{epsilonassumption} and Compact Space Assumption \ref{asspCompact}. Hence, $\Var(U_{2,1}(\B^{*}))$ converges to 0 as $n \rightarrow \infty$. 

Similarly, it follows that the variance of the other entries in the stage 2 term of $U(\B^{*})$ converges to 0, and that the variance of the stage-1 summand of $U(\B^{*})$ converges to 0. Thus, $\Var(U(\B^{*}))$ converges to 0. By applying Chebyshev's inequality to each entry in $U(\B^{*})$, it follows that each entry in $U(\B^{*})$ converges to 0 in probability, so that $U(\B^{*}) \ConvP 0$. Since $U(\B^{*})$ does not contain $\B$, its supremum over $\B$ converges to 0 in probability. 

Next, we prove that each of the finitely many $J^{(2)}$ summands of $G_2$ converges uniformly to 0. Consider any summand of $G_2$,
\begin{align}\label{summandG1WithLearning}
    \left(\frac{n_{jt}^{(2)}}{n} - \alpha_{jt}^{(2)}\right)\begin{pmatrix} \Atwo\\ \BoldLtwo \end{pmatrix} \left [ \begin{pmatrix} \boldsymbol{\beta}^{*}_A \\  \B^{*}_{\ell} \end{pmatrix}^T - \begin{pmatrix} \boldsymbol{\beta}_A \\  \B_{\ell} \end{pmatrix}^T \right ] \begin{pmatrix} \Atwo\\ \BoldLtwo \end{pmatrix}.
\end{align}
Because of Compact Space Assumption \ref{asspCompact}, $\lvert \A_{jt,p}^{(2,n^{(1)})} \rvert \leq \mathcal{C}_A \; \forall j$ and $\lvert \B_m \rvert \leq \mathcal{C}_{\B} \; \forall m$. Hence,
\begin{align*}
    \left |\left [\begin{pmatrix} \boldsymbol{\beta}^{*}_A \\ \B^{*}_{\ell} \end{pmatrix}^T- \begin{pmatrix} \boldsymbol{\beta}_A \\ \B_{\ell} \end{pmatrix}^T \right ] \begin{pmatrix} \Atwo\\ \BoldLtwo \end{pmatrix} \right | &\leq \left \Vert \left [\begin{pmatrix} \boldsymbol{\beta}^{*}_A \\ \B^{*}_{\ell} \end{pmatrix}^T- \begin{pmatrix} \boldsymbol{\beta}_A \\ \B_{\ell} \end{pmatrix}^T \right ] \right \Vert \left \Vert\begin{pmatrix} \Atwo\\ \BoldLtwo\end{pmatrix} \right \Vert \\
    &\leq 2\,\mathcal{C}_\beta \sqrt{P + J + 1} \left(\sqrt{\mathcal{C}_A^2 \cdot P + 2} \right),
\end{align*}
with $P + J + 1$ the number of components in the vectors $\B^{*}$ and $\B$. It follows that for the first entry of the summand in equation (\ref{summandG1WithLearning}),
\begin{align*}
&\underset{\B \in \mathcal{B}}{\sup} \left \Vert \left(\frac{n_{jt}^{(2)}}{n} - \alpha_{jt}^{(2)}\right) \A_{jt,1}^{(2,n^{(1)})} \left[ \begin{pmatrix} \boldsymbol{\beta}^{*}_A \\  \B^{*}_{\ell} \end{pmatrix}^T - \begin{pmatrix} \boldsymbol{\beta}_A \\  \B_{\ell} \end{pmatrix}^T \right]\begin{pmatrix} \Atwo\\ \BoldLtwo \end{pmatrix} \right \Vert \\
&\leq \left|\frac{n_{jt}^{(2)}}{n} - \alpha_{jt}^{(2)} \right|\mathcal{C}_A \cdot 2\,\mathcal{C}_\beta \, \sqrt{P + J + 1} \left(\sqrt{C_A^2 \cdot P + 2}  \right),     
\end{align*}
which converges to 0 as $n \rightarrow 0$ because of Assumption \ref{ratioConvergence}. Similarly, the supremum of the other entries of the summand of $G_2$ in equation (\ref{summandG1WithLearning}) converges uniformly to 0 as $n \rightarrow 0$. It follows that the supremum of every summand of $G_2$ converges uniformly to 0. Therefore, because $J^{(2)}$ is fixed, $\underset{\B \in \mathcal{B}}{\sup}\|G_2\| \rightarrow 0$. Similarly, $\underset{\B \in \mathcal{B}}{\sup}\|G_1\| \rightarrow 0$.

Next, consider any summand of $G_3$. Since $\Atwolago \ConvP \Bolda_{jt}^{(2)}$ (Remark \ref{convergenceofA}), by the Continuous Mapping Theorem, it follows that $\begin{pmatrix} \Atwolago \\ \BoldLtwolago \end{pmatrix}\begin{pmatrix} \Atwolago \\ \BoldLtwolago \end{pmatrix}^T  \ConvP \begin{pmatrix} \Bolda_{jt}^{(2)} \\ \BoldLtwo \end{pmatrix}\begin{pmatrix} \Bolda_{jt}^{(2)} \\ \BoldLtwo \end{pmatrix}^T$. Thus, \[\underset{\B \in \mathcal{B}}{\sup} \left \|\alpha_{jt}^{(2)} \left [ \begin{pmatrix} \Atwolago\\ \BoldLtwolago \end{pmatrix}\begin{pmatrix} \Atwolago\\ \BoldLtwolago \end{pmatrix}^T - \begin{pmatrix} \Bolda_{jt}^{(2)}\\ \BoldLtwo \end{pmatrix}\begin{pmatrix} \Bolda_{jt}^{(2)}\\ \BoldLtwo \end{pmatrix}^T \right ] \begin{pmatrix} \boldsymbol{\beta}^{*}_A \\ \B^{*}_{\ell} \end{pmatrix} \right \| \ConvP 0.\] Hence, because $J^{(2)}$ is fixed, $\underset{\B \in \mathcal{B}}{\sup}\|G_3\| \ConvP 0$. 

Next, consider $G_4$. Since $\mathcal{B}$ is compact (Compact Space Assumption \ref{asspCompact}), $\underset{\B \in \mathcal{B}}{\sup} \left\|\begin{pmatrix} \boldsymbol{\beta}_A \\ \B_{\ell} \end{pmatrix}\right\| \leq C_\beta \sqrt{P+J+1}$ with $P + J + 1$ the number of components in the vector $\B$. Thus, using similar reasoning as for $G_3$,
$$\underset{\B \in \mathcal{B}}{\sup} \left \|\alpha_{jt}^{(2)} \left [ \begin{pmatrix} \Bolda_{jt}^{(2)}\\ \BoldLtwo \end{pmatrix}\begin{pmatrix} \Bolda_{jt}^{(2)}\\ \BoldLtwo \end{pmatrix}^T - \begin{pmatrix} \Atwolago\\ \BoldLtwolago \end{pmatrix}\begin{pmatrix} \Atwolago\\ \BoldLtwolago \end{pmatrix}^T \right ] \begin{pmatrix} \boldsymbol{\beta}_A \\ \B_{\ell} \end{pmatrix} \right \| \ConvP 0.$$
Hence, because $J^{(2)}$ is fixed, $\underset{\B \in \mathcal{B}}{\sup}\|G_4\| \ConvP 0$. Hence, the first condition in Theorem 5.9 of \cite{van2000asymptotic} holds in our setting. 

The second condition of Theorem 5.9 in \cite{van2000asymptotic} states that for every $\epsilon > 0$,
\[ \underset{\boldsymbol{\beta}:||\boldsymbol{\beta} - \boldsymbol{\beta}^{*}|| > \epsilon}{\inf} ||u(\boldsymbol{\beta})|| > 0 = ||u(\boldsymbol{\beta}^{*})||.\]
First, notice that $u(\boldsymbol{\beta}^{*}) = 0$, implying that $||u(\boldsymbol{\beta}^{*})|| = 0$. Next, we see that the derivative of $u(\boldsymbol{\beta})$ at $\boldsymbol{\beta}$ is
\[ -\sum_{t=1}^{T^{(1)}} \sum_{j=1}^{J^{(1)}} \alpha_{jt}^{(1)} \begin{pmatrix} \Aone\\ \BoldLone \end{pmatrix} \begin{pmatrix} \Aone \\ \BoldLone \end{pmatrix}^T -\sum_{t=T^{(1)}+1}^{T^{(2)}} \sum_{j=1}^{J^{(2)}} \alpha_{jt}^{(2)} \begin{pmatrix} \Bolda_{jt}^{(2)}\\ \BoldLtwo \end{pmatrix} \begin{pmatrix} \Bolda_{jt}^{(2)}\\ \BoldLtwo \end{pmatrix}^T. \] 
In the absence of multicollinearity in the limit, this derivative is negative definite, implying that $u(\boldsymbol{\beta})$ is strictly concave. We conclude that $u(\boldsymbol{\beta})$ has a unique solution at $\boldsymbol{\beta}^{*}$ and the second condition is also satisfied.

Therefore, $\widehat{\B} \ConvP \B^{*}$.
\end{proof}
\subsection[Proof of Asymptotic Normality]{Proof of Asymptotic Normality of $\hat{\B}$ Theorem \ref{AsymptoticTheoremWithLearning}.}\label{AsymptoticProofAppendix}
\begin{proof}
Using the Mean Value Theorem for each of the components of $U(\boldsymbol{\beta})$, we have
$0 = U(\widehat{\boldsymbol{\beta}}) = U(\boldsymbol{\beta}^{*}) + \left( \left .\frac{\partial}{\partial \boldsymbol{\beta}} \right |_{\widetilde{\boldsymbol{\beta}}}U(\boldsymbol{\beta}) \right)(\widehat{\boldsymbol{\beta}} - \boldsymbol{\beta})$ for some $\widetilde{\B}$ between $\hat{\B}$ and $\B^{*}$, possibly different for each row of $\left .\frac{\partial}{\partial \boldsymbol{\beta}} \right |_{\widetilde{\boldsymbol{\beta}}}U(\boldsymbol{\beta})$. Hence,

\begin{align}\label{MVT_Asymptotic_Proof}
- \left[\left( \left .\frac{\partial}{\partial \boldsymbol{\beta}} \right |_{\widetilde{\B}}U(\boldsymbol{\beta}) \right)\right]^{-1} U(\boldsymbol{\beta}^{*}) = (\widehat{\boldsymbol{\beta}} -\boldsymbol{\beta}^{*}).     
\end{align}

It follows that the limiting distribution of $\sqrt{n}\,(\widehat{\boldsymbol{\beta}} -\boldsymbol{\beta}^{*})$ is the limiting distribution of 
$-\sqrt{n}\left[\left( \left .\frac{\partial}{\partial \boldsymbol{\beta}} \right |_{\widetilde{\B}}U(\boldsymbol{\beta}) \right)\right]^{-1} U(\boldsymbol{\beta}^{*})$. 

Define
\begin{align*}
    &\Tilde{J}(\B) = -\frac{\partial}{\partial \boldsymbol{\beta}} U(\boldsymbol{\beta}) \\
    &= -\frac{\partial}{\partial \boldsymbol{\beta}} \frac{1}{n} \left\{  \sum_{t=1}^{T^{(1)}} \sum_{j=1}^{J^{(1)}} \sum_{i=1}^{n_{jt}^{(1)}} \begin{pmatrix} \Aone\\ \BoldLone \end{pmatrix} \left[\Yone - \begin{pmatrix} \boldsymbol{\beta}_A \\ \B_{\ell} \end{pmatrix}^T \begin{pmatrix} \Aone\\ \BoldLone \end{pmatrix} \right] \right.\\
    &\phantom{= }+ \left.\sum_{t=T^{(1)}+1}^{T^{(2)}} \sum_{j=1}^{J^{(2)}} \sum_{i=1}^{n_{jt}^{(2)}} \begin{pmatrix} \Atwolago\\ \BoldLtwolago \end{pmatrix} \left[\Ytwolago - \begin{pmatrix} \boldsymbol{\beta}_A \\ \B_{\ell} \end{pmatrix}^T \begin{pmatrix} \Atwolago\\ \BoldLtwolago \end{pmatrix} \right] \right\} \\
    &= \frac{1}{n} \left\{\sum_{t=1}^{T^{(1)}} \sum_{j=1}^{J^{(1)}} \sum_{i=1}^{n_{jt}^{(1)}} \begin{pmatrix} \Aone\\ \BoldLone \end{pmatrix}\begin{pmatrix} \Aone\\ \BoldLone \end{pmatrix}^T +\sum_{t=T^{(1)}+1}^{T^{(2)}} \sum_{j=1}^{J^{(2)}} \sum_{i=1}^{n_{jt}^{(2)}} \begin{pmatrix} \Atwolago\\ \BoldLtwolago \end{pmatrix}\begin{pmatrix} \Atwolago\\ \BoldLtwolago \end{pmatrix}^T \right\}\\
    &= \sum_{t=1}^{T^{(1)}} \sum_{j=1}^{J^{(1)}} \frac{n_{jt}^{(1)}}{n} \begin{pmatrix} \Aone\\ \BoldLone \end{pmatrix}\begin{pmatrix} \Aone\\ \BoldLone \end{pmatrix}^T + \sum_{t=T^{(1)}+1}^{T^{(2)}} \sum_{j=1}^{J^{(2)}} \frac{n_{jt}^{(2)}}{n} \begin{pmatrix} \Atwolago\\ \BoldLtwolago \end{pmatrix}\begin{pmatrix} \Atwolago \\ \BoldLtwolago \end{pmatrix}^T \\
    &\ConvP J^{*} = \sum_{t=1}^{T^{(1)}} \sum_{j=1}^{J^{(1)}} \alpha_{jt}^{(1)}  \begin{pmatrix} \A_{jt}^{(1)}\\ \BoldLone \end{pmatrix} \begin{pmatrix} \A_{jt}^{(1)}\\ \BoldLone \end{pmatrix}^T + \sum_{t=T^{(1)}+1}^{T^{(2)}} \sum_{j=1}^{J^{(2)}} \alpha_{jt}^{(2)}  \begin{pmatrix} \Bolda_{jt}^{(2)}\\ \BoldLtwo \end{pmatrix} \begin{pmatrix} \Bolda_{jt}^{(2)}\\ \BoldLtwo \end{pmatrix}^T ,\numberthis \label{J(B)}
\end{align*}
because $\Atwolago \ConvP \Bolda_{jt}^{(2)}$ (Remark \ref{convergenceofA}) and using Assumption \ref{ratioConvergence}. From  Slutsky's Theorem, the asymptotic distribution of $\sqrt{n}(\widehat{\boldsymbol{\beta}} -\boldsymbol{\beta}^{*})$ is the same as the asymptotic distribution of $\sqrt{n} \;J^{*-1} \; U(\boldsymbol{\beta}^{*})$. 

Next, decompose $U(\B^{*})$ by adding and subtracting a term that contains $\Bolda_{jt}^{(2)}$:
\begin{align*}
    \sqrt{n}\,U(\boldsymbol{\beta}^{*}) &= \frac{1}{\sqrt{n}} \sum_{t=1}^{T^{(1)}} \sum_{j=1}^{J^{(1)}} \sum_{i=1}^{n_{jt}^{(1)}} \begin{pmatrix} \Aone\\ \BoldLone \end{pmatrix} \left [\Yone - \begin{pmatrix} \boldsymbol{\beta}^{*}_A \\ \B^{*}_{\ell} \end{pmatrix}^T \begin{pmatrix} \Aone\\ \BoldLone \end{pmatrix} \right ] \\
    &+ \frac{1}{\sqrt{n}}\sum_{t=T^{(1)}+1}^{T^{(2)}} \sum_{t=T^{(1)}+1}^{T^{(2)}} \sum_{j=1}^{J^{(2)}} \sum_{i=1}^{n_{jt}^{(2)}} \begin{pmatrix} \Atwolago\\ \BoldLtwolago \end{pmatrix} \left [\Ytwolago - \begin{pmatrix} \boldsymbol{\beta}^{*}_A \\ \B^{*}_{\ell} \end{pmatrix}^T \begin{pmatrix} \Atwolago\\ \BoldLtwolago \end{pmatrix} \right ]\\
    &= \frac{1}{\sqrt{n}} \sum_{t=1}^{T^{(1)}} \sum_{j=1}^{J^{(1)}} \sum_{i=1}^{n_{jt}^{(1)}} \begin{pmatrix} \Aone\\ \BoldLone \end{pmatrix} \epsilon_{ijt}^{(1)}+ \frac{1}{\sqrt{n}}\sum_{t=T^{(1)}+1}^{T^{(2)}} \sum_{t=T^{(1)}+1}^{T^{(2)}} \sum_{j=1}^{J^{(2)}} \sum_{i=1}^{n_{jt}^{(2)}} \begin{pmatrix} \Atwolago\\ \BoldLtwolago \end{pmatrix} \epsilon_{ijt}^{(2,n^{(1)})}\\
    &= \frac{1}{\sqrt{n}} \sum_{t=1}^{T^{(1)}} \sum_{j=1}^{J^{(1)}} \sum_{i=1}^{n_{jt}^{(1)}} \begin{pmatrix} \Aone\\ \BoldLone \end{pmatrix} \epsilon_{ijt}^{(1)} + \frac{1}{\sqrt{n}}\sum_{t=T^{(1)}+1}^{T^{(2)}} \sum_{j=1}^{J^{(2)}} \sum_{i=1}^{n_{jt}^{(2)}} \begin{pmatrix} \Bolda_{jt}^{(2)}\\ \BoldLtwo \end{pmatrix} \epsilon_{ijt}^{(2,n^{(1)})} \\
    &- \frac{1}{\sqrt{n}}\sum_{t=T^{(1)}+1}^{T^{(2)}} \sum_{j=1}^{J^{(2)}} \sum_{i=1}^{n_{jt}^{(2)}} \begin{pmatrix} \Bolda_{jt}^{(2)}\\ \BoldLtwo \end{pmatrix} \epsilon_{ijt}^{(2,n^{(1)})} +  \frac{1}{\sqrt{n}}\sum_{t=T^{(1)}+1}^{T^{(2)}} \sum_{j=1}^{J^{(2)}} \sum_{i=1}^{n_{jt}^{(2)}} \begin{pmatrix} \Atwolago\\ \BoldLtwolago \end{pmatrix} \epsilon_{ijt}^{(2,n^{(1)})}\\
    &= U_{1\_2,n} + U_{2\_2,n},
\end{align*}
with
\begin{align*}
    U_{1\_2,n} &= \frac{1}{\sqrt{n}} \left [\sum_{t=1}^{T^{(1)}} \sum_{j=1}^{J^{(1)}} \sum_{i=1}^{n_{jt}^{(1)}} \begin{pmatrix} \Aone\\ \BoldLone \end{pmatrix} \epsilon_{ijt}^{(1)} +\sum_{t=T^{(1)}+1}^{T^{(2)}}  \sum_{t=T^{(1)}+1}^{T^{(2)}} \sum_{j=1}^{J^{(2)}} \sum_{i=1}^{n_{jt}^{(2)}} \begin{pmatrix} \Bolda_{jt}^{(2)}\\ \BoldLtwo \end{pmatrix} \epsilon_{ijt}^{(2,n^{(1)})} \right ],\\
    U_{2\_2,n} &= \frac{1}{\sqrt{n}}\sum_{t=T^{(1)}+1}^{T^{(2)}} \sum_{j=1}^{J^{(2)}} \sum_{i=1}^{n_{jt}^{(2)}} \left [\begin{pmatrix} \Atwolago\\ \BoldLtwolago \end{pmatrix}  -  \begin{pmatrix} \Bolda_{jt}^{(2)}\\ \BoldLtwo \end{pmatrix} \right ]\epsilon_{ijt}^{(2,n^{(1)})}.\\
\end{align*}
To show that $\sqrt{n}\,U(\B^{*})$ converges in distribution to a normal distribution, we show that $U_{1\_2,n}$ converges to a normal distribution and that $U_{2\_2,n}$ converges to 0 in probability. 

Regarding $U_{1\_2,n}$, we first replace $\epsilon_{j}^{(2,n^{(1)})}$ with $\epsilon_{j}^{(2)}$ without changing the distribution of  $U_{1\_2,n}$ (Outcome Errors Assumption \ref{epsilonassumption}), 
\begin{align}\label{U_1,2}
\widetilde{U}_{1\_2, n} = \frac{1}{\sqrt{n}} \left [\sum_{t=1}^{T^{(1)}} \sum_{j=1}^{J^{(1)}} \sum_{i=1}^{n_{jt}^{(1)}} \begin{pmatrix} \Aone\\ \BoldLone \end{pmatrix} \epsilon_{ijt}^{(1)} + \sum_{t=T^{(1)}+1}^{T^{(2)}} \sum_{j=1}^{J^{(2)}} \sum_{i=1}^{n_{jt}^{(2)}} \begin{pmatrix} \Bolda_{jt}^{(2)}\\ \BoldLtwo \end{pmatrix} \epsilon_{ijt}^{(2)} \right ] \sim U_{1\_2,n}.    
\end{align}
Next, by the linearity of expectation, using the same argument as in equation (\ref{expectationstage2zero}), $\E (\widetilde{U}_{1\_2, n}) = 0$. Because the $\Bolda_{jt}^{(2)}$ are fixed, the two summands of stages 1 and 2 are uncorrelated. Thus, the variance of $\widetilde{U}_{1\_2, n}$ is
\begin{align*}
    &\Var \left [\widetilde{U}_{1\_2, n} \mid \overline{\BoldZ}\right] \\
    &= \Var \left [\frac{1}{\sqrt{n}} \left(\sum_{t=1}^{T^{(1)}} \sum_{j=1}^{J^{(1)}} \sum_{i=1}^{n_{jt}^{(1)}} \begin{pmatrix} \Aone\\ \BoldLone \end{pmatrix} \epsilon_{ijt}^{(1)} + \sum_{t=T^{(1)}+1}^{T^{(2)}} \sum_{j=1}^{J^{(2)}} \sum_{i=1}^{n_{jt}^{(2)}} \begin{pmatrix} \Bolda_{jt}^{(2)}\\ \BoldLtwo \end{pmatrix} \epsilon_{ijt}^{(2)} \right) \middle | \overline{\BoldZ}\right]\\
    &= \E \left [ \frac{1}{n}\sum_{t=1}^{T^{(1)}} \sum_{j=1}^{J^{(1)}} \sum_{i=1}^{n_{jt}^{(1)}} \begin{pmatrix} \Aone\\ \BoldLone \end{pmatrix} \begin{pmatrix} \Aone\\ \BoldLone \end{pmatrix}^T \left(\epsilon_{ijt}^{(1)} \right)^2 \middle | \overline{\BoldZ} \right] + \E \left [\frac{1}{n} \sum_{t=T^{(1)}+1}^{T^{(2)}} \sum_{j=1}^{J^{(2)}} \sum_{i=1}^{n_{jt}^{(2)}} \begin{pmatrix} \Bolda_{jt}^{(2)}\\ \BoldLtwo \end{pmatrix} \begin{pmatrix} \Bolda_{jt}^{(2)}\\ \BoldLtwo \end{pmatrix}^T \left(\epsilon_{ijt}^{(2)}\right)^2 \middle | \overline{\BoldZ}\right]\\
     &= \frac{1}{n}\sum_{t=1}^{T^{(1)}} \sum_{j=1}^{J^{(1)}} \sum_{i=1}^{n_{jt}^{(1)}} \begin{pmatrix} \Aone\\ \BoldLone \end{pmatrix} \begin{pmatrix} \Aone\\ \BoldLone \end{pmatrix}^T \E \left [ \left(\epsilon_{ijt}^{(1)} \right)^2 \middle | \overline{\BoldZ}\right] + \frac{1}{n} \sum_{t=T^{(1)}+1}^{T^{(2)}} \sum_{j=1}^{J^{(2)}} \sum_{i=1}^{n_{jt}^{(2)}} \begin{pmatrix} \Bolda_{jt}^{(2)}\\ \BoldLtwo \end{pmatrix} \begin{pmatrix} \Bolda_{jt}^{(2)}\\ \BoldLtwo \end{pmatrix}^T \E \left [\left(\epsilon_{ijt}^{(2)}\right)^2\middle | \overline{\BoldZ}\right]\\
      &= \sum_{t=1}^{T^{(1)}} \sum_{j=1}^{J^{(1)}} \frac{n_{jt}^{(1)}}{n} \begin{pmatrix} \Aone\\ \BoldLone \end{pmatrix} \begin{pmatrix} \Aone\\ \BoldLone \end{pmatrix}^T \sigma^2(\BoldZ_j) +  \sum_{t=T^{(1)}+1}^{T^{(2)}} \sum_{j=1}^{J^{(2)}}  \frac{n_{jt}^{(2)}}{n}\begin{pmatrix} \Bolda_{jt}^{(2)}\\ \BoldLtwo \end{pmatrix} \begin{pmatrix} \Bolda_{jt}^{(2)}\\ \BoldLtwo \end{pmatrix}^T \sigma^2(\BoldZ_j)\\
    &\rightarrow \sum_{t=1}^{T^{(1)}} \sum_{j=1}^{J^{(1)}} \alpha_{jt}^{(1)}\begin{pmatrix} \A_{jt}^{(1)}\\ \BoldLone \end{pmatrix} \begin{pmatrix} \A_{jt}^{(1)}\\ \BoldLone \end{pmatrix}^T \sigma^2 \left(\BoldZ_j \right) + \sum_{t=T^{(1)}+1}^{T^{(2)}} \sum_{j=1}^{J^{(2)}} \alpha_{jt}^{(2)}\begin{pmatrix} \Bolda_{jt}^{(2)}\\ \BoldLtwo \end{pmatrix} \begin{pmatrix} \Bolda_{jt}^{(2)}\\ \BoldLtwo \end{pmatrix}^T \sigma^2 \left(\BoldZ_j \right) \quad \text{as} \,\, n\rightarrow \infty \numberthis \label{varianceOfFirstTerm}.
\end{align*}
The fourth equality follows from Outcome Errors Assumption \ref{epsilonassumption} and Conditional Independence Between Centers Assumption \ref{asspcenterIndep}. The convergence in the last line follows from Assumption \ref{ratioConvergence}.

Since the $\epsilon_{ijt}^{(1)}$ and the $\epsilon_{ijt}^{(2)}$ are independent (Outcome Errors Assumption \ref{epsilonassumption}), $\widetilde{U}_{1\_2, n}$ converges to a normal distribution because of the Crámer-Wold device and Lyapunov's condition (see page 383 and 362, \cite{billingsley1995}). To see this, first, let a vector of constants $\boldsymbol{t} = \begin{pmatrix}\boldsymbol{t}_A \\ \boldsymbol{t}_\ell \end{pmatrix} \in \boldsymbol{\mathbb{R}}^{P+J+1}$, with $\boldsymbol{t}_A = \begin{pmatrix} t_1\\ \ldots\\ t_P \end{pmatrix}$ and $\boldsymbol{t}_\ell = \begin{pmatrix} t_{P+1} \\ \ldots \\ t_{P+J+1} \end{pmatrix}$. Consider the linear combination of $\widetilde{U}_{1\_2, n}$, $\widetilde{U}_{1\_2, n}^{\text{linear}}$, 
\begin{align*}
    \widetilde{U}_{1\_2, n}^{\text{linear}} &= \frac{1}{\sqrt{n}} \sum_{t=1}^{T^{(1)}} \sum_{j=1}^{J^{(1)}} \sum_{i=1}^{n_{jt}^{(1)}} \boldsymbol{t}^T \begin{pmatrix} \A_{jt}^{(1)} \\ \boldsymbol{\ell}_j^{(1)} \end{pmatrix} \epsilon_{ijt}^{(1)} + \frac{1}{\sqrt{n}}  \sum_{t=T^{(1)}+1}^{T^{(2)}} \sum_{j=1}^{J^{(2)}} \sum_{i=1}^{n_{jt}^{(2)}} \boldsymbol{t}^T  \begin{pmatrix} \Bolda_{jt}^{(2)} \\ \boldsymbol{\ell}_j^{(2)} \end{pmatrix} \epsilon_{ijt}^{(2)} \\
    &= \frac{1}{\sqrt{n}}\left [\sum_{t=1}^{T^{(1)}} \sum_{j=1}^{J^{(1)}} \sum_{i=1}^{n_{jt}^{(1)}} \left( \boldsymbol{t}_a^T \A_{jt}^{(1)} + \boldsymbol{t}_\ell^T \BoldLone \right) \epsilon_{ijt}^{(1)} + \sum_{t=T^{(1)}+1}^{T^{(2)}} \sum_{j=1}^{J^{(2)}} \sum_{i=1}^{n_{jt}^{(2)}} \left(\boldsymbol{t}_a^T\Bolda_{jt}^{(2)} + \boldsymbol{t}_\ell^T\BoldLtwo \right) \epsilon_{ijt}^{(2)} \right ].
\end{align*}
$\widetilde{U}_{1\_2, n}^{\text{linear}}$ converges to a normal distribution if it satisfies the condition of Lyapunov's condition (see page 362, \cite{billingsley1995}), which state that there exists $\delta > 0$ so that as $n \rightarrow \infty$, 
\begin{align*}
    &\frac{1}{\Var \left(\widetilde{U}_{1\_2, n}^{\text{linear}}\right)^{1+\delta/2}} \left(\sum_{t=1}^{T^{(1)}}\sum_{j=1}^{J^{(1)}}\sum_{i=1}^{n_{jt}^{(1)}} \E \left [\left |\frac{1}{\sqrt{n}} \left(\boldsymbol{t}_a^T \A_{jt}^{(1)} + \boldsymbol{t}_\ell^T \BoldLone \right)\epsilon_{ijt}^{(1)} \right |^{2+\delta}\right ] \right.\\
    &+ \left. \sum_{t=T^{(1)}+1}^{T^{(2)}}\sum_{j=1}^{J^{(2)}}\sum_{i=1}^{n_{jt}^{(2)}} \E \left [\left |\frac{1}{\sqrt{n}} \left(\boldsymbol{t}_a^T \Bolda_{jt}^{(2)} + \boldsymbol{t}_\ell^T \BoldLtwo \right)\epsilon_{ijt}^{(2)} \right |^{2+\delta}\right ] \right) \rightarrow 0. \numberthis \label{lyapunov}
\end{align*}
Similar to equation (\ref{varianceOfFirstTerm}),
\begin{align*}
Var \left(\widetilde{U}_{1\_2, n}^{\text{linear}} \right) \rightarrow &\sum_{t=1}^{T^{(1)}} \sum_{j=1}^{J^{(1)}} (\alpha_{jt}^{(1)}) \left (\boldsymbol{t}^T \begin{pmatrix} \A_{jt}^{(1)}\\ \BoldLone \end{pmatrix} \right)^2 \, \sigma^2 \left(\BoldZ_j \right) \\
&+ \sum_{t=T^{(1)}+1}^{T^{(2)}} \sum_{j=1}^{J^{(2)}} (\alpha_{jt}^{(2)}) \left ( \boldsymbol{t}^T \begin{pmatrix} \Bolda_{jt}^{(2)}\\ \BoldLtwo \end{pmatrix} \right)^2 \sigma^2 \left(\BoldZ_j \right) \, \text{as} \, n\rightarrow \infty. 
\end{align*}
Next, from equation (\ref{lyapunov}), using $\delta > 0$ from Outcome Errors Assumption~\ref{epsilonassumption},
\begin{align*}
    &\sum_{t=1}^{T^{(1)}}\sum_{j=1}^{J^{(1)}}\sum_{i=1}^{n_{jt}^{(1)}} \E \left [\left |\frac{1}{\sqrt{n}} \left(\boldsymbol{t}_a^T \A_{jt}^{(1)} + \boldsymbol{t}_\ell^T \BoldLone \right) \epsilon_{ijt}^{(1)}\right |^{2+\delta} \middle | \overline{\BoldZ} \right ]\\
    &\leq  \sum_{t=1}^{T^{(1)}}\sum_{j=1}^{J^{(1)}}\sum_{i=1}^{n_{jt}^{(1)}}\frac{1}{n^{1+\delta/2}} \cdot \max_{j=1,\ldots,J^{(1)}} \left(\left | \boldsymbol{t}_a^T \A_{jt}^{(1)} + \boldsymbol{t}_\ell^T \BoldL_{j}^{(1)} \right | \right)^{2+\delta} \cdot \E \left [ \left(\left |\epsilon_{ijt}^{(1)} \right| \right)^{2+\delta} \middle | \overline{\BoldZ} \right ]\\   
    &\leq  \frac{1}{n^{\delta/2}} \, \max_{j=1,\ldots,J^{(1)}} \, \left(\left | \boldsymbol{t}_a^T \A_{jt}^{(1)} \right | + \left | \boldsymbol{t}_\ell^T \BoldL_{j}^{(1)} \right |\right)^{2+\delta} \cdot \E \left [ \left(\left |\epsilon_{ijt}^{(1)} \right| \right)^{2+\delta} \middle | \overline{\BoldZ} \right ] \\
    &\leq  \frac{1}{n^{\delta/2}} \left(\left|\max_{p=1,\ldots,P}(\boldsymbol{t}_{a,p})\cdot \mathcal{C}_A \cdot P \right|+ \left |\max_{m=1,\ldots,J+1}(\boldsymbol{t}_{\ell,m})\cdot 2\right |\right)^{2+\delta} \cdot \mathcal{C}_\epsilon \rightarrow 0 \quad \text{as} \, \, n \rightarrow \infty, 
\end{align*}
where the third inequality follows from Outcome Errors Assumption \ref{epsilonassumption} and Compact Space Assumption \ref{asspCompact}. Similarly, it follows that 
\[   \sum_{t=T^{(1)}+1}^{T^{(2)}}\sum_{j=1}^{J^{(2)}}\sum_{i=1}^{n_{jt}^{(2)}} \E \left [\left |\frac{1}{\sqrt{n}} \left(\boldsymbol{t}_a^T \Bolda_{jt}^{(2)} + \boldsymbol{t}_\ell^T \BoldLtwo \right)\epsilon_{ijt}^{(2)} \right |^{2+\delta} \middle | \overline{\BoldZ} \right ] \rightarrow 0 \quad \text{as} \, \, n \rightarrow \infty. \]
Since the limit of $\Var(\widetilde{U}_{1\_2, n}^{\text{linear}})^{2+\delta}$ is finite and non-zero and the numerator in equation (\ref{lyapunov}) converges to 0 as $n \rightarrow \infty$, the Lyapunov's condition of equation (\ref{lyapunov}) is satisfied. Therefore, the Lyapunov Central Limit Theorem (CLT) holds for  $\widetilde{U}_{1\_2, n}^{\text{linear}}$. Therefore, by the Crámer-Wold device, $\widetilde{U}_{1\_2, n}$ converges to a normal distribution with mean 0 and variance

\[\sum_{t=1}^{T^{(1)}} \sum_{j=1}^{J^{(1)}} \alpha_{jt}^{(1)} \begin{pmatrix} \A_{jt}^{(1)}\\ \BoldLone \end{pmatrix}  \begin{pmatrix} \A_{jt}^{(1)}\\ \BoldLone \end{pmatrix}^T \sigma^2 \left(\BoldZ_j \right) + \sum_{t=T^{(1)}+1}^{T^{(2)}} \sum_{j=1}^{J^{(2)}} \alpha_{jt}^{(2)} \begin{pmatrix} \Bolda_{jt}^{(2)}\\ \BoldLtwo \end{pmatrix}  \begin{pmatrix} \Bolda_{jt}^{(2)}\\ \BoldLtwo \end{pmatrix}^T \sigma^2 \left(\BoldZ_j \right).\]

Next, we show that $U_{2\_2,n} \ConvP 0$ follows from $\E(U_{2\_2,n}) = 0$ and $\E(U_{2\_2,n}^2) \rightarrow 0$, so that $U_{2\_2,n} \ConvP 0$ by Chebyshev's inequality. By the Law of Total Expectation, 

\begin{align*}
    &\E(U_{2\_2,n} \mid \overline{\BoldZ}) \\
    &= \E \left [\frac{1}{\sqrt{n}}\sum_{t=T^{(1)}+1}^{T^{(2)}} \sum_{j=1}^{J^{(2)}} \sum_{i=1}^{n_{jt}^{(2)}} \left [ \begin{pmatrix} \Atwolago\\ \BoldLtwolago \end{pmatrix} - \begin{pmatrix} \Bolda_{jt}^{(2)}\\ \BoldLtwo \end{pmatrix}\right ] \epsilon_{ijt}^{(2,n^{(1)})}\middle | \overline{\BoldZ} \right ]\\
    &= \E \left [\E \left [\frac{1}{\sqrt{n}}\sum_{t=T^{(1)}+1}^{T^{(2)}} \sum_{j=1}^{J^{(2)}} \sum_{i=1}^{n_{jt}^{(2)}} \left [ \begin{pmatrix} \Atwolago\\ \BoldLtwolago \end{pmatrix} - \begin{pmatrix} \Bolda_{jt}^{(2)}\\ \BoldLtwo \end{pmatrix}\right ] \epsilon_{ijt}^{(2,n^{(1)})}\middle | \overline{\A}^{(2,n^{(1)})}, \overline{\BoldZ} \right ] \middle | \overline{\BoldZ} \right ]\\
    &= \int_{\Bara} \E \left [\frac{1}{\sqrt{n}}\sum_{t=T^{(1)}+1}^{T^{(2)}} \sum_{j=1}^{J^{(2)}} \sum_{i=1}^{n_{jt}^{(2)}} \left [ \begin{pmatrix} \Bolda_{jt} \\ \BoldLtwolago \end{pmatrix} - \begin{pmatrix} \Bolda_{jt}^{(2)}\\ \BoldLtwo \end{pmatrix}\right ] \epsilon_{ijt}^{(2,n^{(1)})}(\Bolda_{jt})\middle |  \overline{\A}^{(2,n^{(1)})} = \overline{\Bolda}, \overline{\BoldZ} \right ] \cdot f_{\BarA^{(2,n^{(1)})}|\overline{\BoldZ}} (\overline{\Bolda})\,d\overline{\Bolda}\\
    &= \int_{\Bara} \frac{1}{\sqrt{n}}\sum_{t=T^{(1)}+1}^{T^{(2)}} \sum_{j=1}^{J^{(2)}} \sum_{i=1}^{n_{jt}^{(2)}} \left [ \begin{pmatrix} \Bolda_{jt} \\ \BoldLtwolago \end{pmatrix} - \begin{pmatrix} \Bolda_{jt}^{(2)}\\ \BoldLtwo \end{pmatrix}\right ] \E \left [\epsilon_{ijt}^{(2,n^{(1)})}(\Bolda_{jt})\middle |  \overline{\BoldZ}\right ] \cdot f_{\overline{\A}^{(2,n^{(1)})}|\overline{\BoldZ}} (\Bolda) \, d\overline{\Bolda}\\
    &= \int_{\Bara} \frac{1}{\sqrt{n}}\sum_{t=T^{(1)}+1}^{T^{(2)}} \sum_{j=1}^{J^{(2)}} \sum_{i=1}^{n_{jt}^{(2)}} \left [ \begin{pmatrix} \Bolda_{jt} \\ \BoldLtwolago \end{pmatrix} - \begin{pmatrix} \Bolda_{jt}^{(2)}\\ \BoldLtwo \end{pmatrix}\right ] \cdot (0) \cdot f_{\overline{\A}^{(2,n^{(1)})}|\overline{\BoldZ}} (\Bolda) \, d\overline{\Bolda} = 0,
\end{align*}
where the third equality follows from Consistency Assumption \ref{asspConsistency}, the fourth equality follows from Outcome Errors Assumption \ref{epsilonassumption} and Conditional Exchangeability Assumption \ref{asspCondExcha}, and the fifth equality follows from Outcome Errors Assumption \ref{epsilonassumption}. Next, the variance of the first entry of $U_{2\_2,n}$, $U_{2\_1,n,1}$, is 
\begin{align}
    &\Var \left(U_{2\_1,n,1} \mid \overline{\BoldZ} \right) \nonumber= \E \left(U_{2\_1,n,1}^2 \mid \overline{\BoldZ} \right)\\ \nonumber
    &= \E \left [ \left(\frac{1}{\sqrt{n}} \sum_{t=T^{(1)}+1}^{T^{(2)}} \sum_{j=1}^{J^{(2)}} \sum_{i=1}^{n_{jt}^{(2)}} \left(\A_{jt,1}^{(2,n^{(1)})} - \Bolda_{jt,1}^{(2)} \right)\epsilon_{ijt}^{(2,n^{(1)})} \right)^2 \middle | \overline{\BoldZ} \right ]\\ \nonumber
    &= \E \left [\E \left [ \left(\frac{1}{\sqrt{n}} \sum_{t=T^{(1)}+1}^{T^{(2)}} \sum_{j=1}^{J^{(2)}} \sum_{i=1}^{n_{jt}^{(2)}} \left(\A_{jt,1}^{(2,n^{(1)})} - \Bolda_{jt,1}^{(2)} \right)\epsilon_{ijt}^{(2,n^{(1)})} \right)^2 \middle | \BarA^{(2,n^{(1)})}, \overline{\BoldZ}\right ] \middle | \overline{\BoldZ} \right]\\ \nonumber
    &= \int_{\Bara} \E \left [ \left(\frac{1}{\sqrt{n}} \sum_{t=T^{(1)}+1}^{T^{(2)}} \sum_{j=1}^{J^{(2)}} \sum_{i=1}^{n_{jt}^{(2)}} \left(\Bolda_{jt,1} - \Bolda_{jt,1}^{(2)} \right)\epsilon_{ijt}^{(2,n^{(1)})}(\Bolda_{jt}) \right)^2 \middle | \BarA^{(2,n^{(1)})} = \overline{\Bolda}, \overline{\BoldZ} \right ] \cdot f_{\BarA^{(2,n^{(1)})}|\overline{\BoldZ}} (\overline{\Bolda})\,d\overline{\Bolda}\\ \nonumber
    &= \int_{\Bara} \frac{1}{n} \sum_{t=T^{(1)}+1}^{T^{(2)}} \sum_{j=1}^{J^{(2)}} \sum_{i=1}^{n_{jt}^{(2)}} \left(\Bolda_{jt,1} - \Bolda_{jt,1}^{(2)} \right)^2 \E \left [ \left(\epsilon_{ijt}^{(2,n^{(1)})}(\Bolda_{jt}) \right)^2 \middle | \BarA^{(2,n^{(1)})} = \overline{\Bolda},\overline{\BoldZ}\right ] \cdot f_{\BarA^{(2,n^{(1)})}|\overline{\BoldZ}} (\overline{\Bolda})\,d\overline{\Bolda}\\ \nonumber
    &= \int_{\Bara} \frac{1}{n} \sum_{t=T^{(1)}+1}^{T^{(2)}} \sum_{j=1}^{J^{(2)}} \sum_{i=1}^{n_{jt}^{(2)}} \left(\Bolda_{jt,1} - \Bolda_{jt,1}^{(2)} \right)^2 \E \left [ \left(\epsilon_{ijt}^{(2,n^{(1)})}(\Bolda_{jt}) \right)^2 \middle | \overline{\BoldZ}\right ] \cdot f_{\BarA^{(2,n^{(1)})}|\overline{\BoldZ}} (\overline{\Bolda})\,d\overline{\Bolda}\\ \nonumber
    &= \int_{\Bara} \frac{1}{n} \sum_{t=T^{(1)}+1}^{T^{(2)}} \sum_{j=1}^{J^{(2)}} \sum_{i=1}^{n_{jt}^{(2)}} \left(\Bolda_{jt,1} - \Bolda_{jt,1}^{(2)} \right)^2 \sigma^2(\BoldZ_j) \cdot f_{\BarA^{(2,n^{(1)})}|\overline{\BoldZ}} (\overline{\Bolda})\,d\overline{\Bolda}\\ 
    &= \E \left [\frac{1}{n} \sum_{t=T^{(1)}+1}^{T^{(2)}} \sum_{j=1}^{J^{(2)}} \sum_{i=1}^{n_{jt}^{(2)}} \left(\A_{jt,1}^{(2,n^{(1)})} - \Bolda_{jt,1}^{(2)} \right)^2 \cdot \sigma^2 (\BoldZ_j) \middle |  \overline{\BoldZ}\right ],
\end{align}
where the fourth equality follows from Consistency Assumption \ref{asspConsistency}, the sixth equality follows from Conditional Exchangeability Assumption \ref{asspCondExcha}, and the seventh equality follows from Outcome Errors Assumption \ref{epsilonassumption} and Conditional Independence Between Centers Assumption \ref{asspcenterIndep}. Since $\Atwolago \ConvP \Bolda_{jt}^{(2)}$ (Remark \ref{convergenceofA}), because of Compact Space Assumption \ref{asspCompact}, the Continuous Mapping Theorem, and Lebesgue's Dominated Convergence Theorem imply that $\Var\left(U_{2\_1,n,1} \right) \rightarrow 0$. By the same argument, the variance of the other entries of $U_{2\_2,n}$ converges to 0. Thus, $\Var(U_{2\_2,n}) \rightarrow 0$. Hence, by Chebyshev's inequality, $U_{2\_2,n} \ConvP 0$.

In conclusion, it follows that $\sqrt{n}\, U(\B^{*})$ has the same asymptotic distribution as
\[U_{1\_2,n} = \frac{1}{\sqrt{n}} \left [ \sum_{t=1}^{T^{(1)}} \sum_{j=1}^{J^{(1)}} \sum_{i=1}^{n_{jt}^{(1)}} \begin{pmatrix} \Aone\\ \BoldLone \end{pmatrix} \epsilon_{ijt}^{(1)} + \sum_{t=T^{(1)}+1}^{T^{(2)}} \sum_{j=1}^{J^{(2)}} \sum_{i=1}^{n_{jt}^{(2)}} \begin{pmatrix} \Bolda_{jt}^{(2)}\\ \BoldLtwo \end{pmatrix} \epsilon_{ijt}^{(2)} \right ],\]
so that $\sqrt{n}\, U(\B^{*})$ converges to a normal distribution with mean 0 and variance 
\[\sum_{t=1}^{T^{(1)}}\sum_{j=1}^{J^{(1)}} \alpha_{jt}^{(1)}\begin{pmatrix} \A_{jt}^{(1)}\\ \BoldLone \end{pmatrix} \begin{pmatrix} \A_{jt}^{(1)}\\ \BoldLone \end{pmatrix}^T \sigma^2 \left(\BoldZ_j \right) + \sum_{t=T^{(1)}+1}^{T^{(2)}} \sum_{j=1}^{J^{(2)}} \alpha_{jt}^{(2)}\begin{pmatrix} \Bolda_{jt}^{(2)}\\ \BoldLtwo \end{pmatrix} \begin{pmatrix} \Bolda_{jt}^{(2)}\\ \BoldLtwo \end{pmatrix}^T \sigma^2 \left(\BoldZ_j \right).\]

Combining this with equation (\ref{J(B)}) implies that Theorem \ref{AsymptoticTheoremWithLearning} holds.
\end{proof}

\section{Proofs of Theorems \ref{unbiasedEE}, 
\ref{theoremConsistencyWithLearning}, and 
\ref{AsymptoticTheoremWithLearning} for Binary Outcomes}\label{proofs_binary}

From Main Model Assumption~\ref{mainmodel}, the counterfactual 
outcome of participant $i$ in center $j$ with center 
characteristics $\BoldZ_j$ under intervention 
package $\Bolda$ has the following success probability:
\begin{align*}
    \mathrm{logit}\!\bigl(p\!\left(\Bolda; 
    \boldsymbol{\beta}^*, \BoldZ_j\right)\bigr)
    = \begin{pmatrix} \boldsymbol{\beta}^*_A \\ 
    \B^*_{\ell} \end{pmatrix}^{\top} 
    \begin{pmatrix} \Bolda \\ \boldsymbol{\ell}_j^{(k)} 
    \end{pmatrix}.
\end{align*}

The estimating equations used to estimate $\boldsymbol{\beta}$ are (compare with equation~\eqref{eqn:OLS_ignore_ee}):
\begin{align}\label{eqn:binary_ee}
    0 &= U(\boldsymbol{\beta}) \nonumber\\
    &= \frac{1}{n} \left\{ 
        \sum_{t=1}^{T^{(1)}} \sum_{j=1}^{J^{(1)}} \sum_{i=1}^{n_{jt}^{(1)}} 
        \begin{pmatrix} \Aone\\ \BoldLone \end{pmatrix} 
        \left[ Y_{ijt}^{(1)} - 
        \text{expit}\!\left(
        \begin{pmatrix} \boldsymbol{\beta}_A \\ \B_{\ell} \end{pmatrix}^{\top} 
        \begin{pmatrix} \Aone\\ \BoldLone \end{pmatrix} \right) \right] 
        \right. \nonumber\\
    &\left. \quad + 
        \sum_{t=T^{(1)}+1}^{T^{(2)}} \sum_{j=1}^{J^{(2)}} \sum_{i=1}^{n_{jt}^{(2)}} 
        \begin{pmatrix} \Atwolago\\ \BoldLtwolago \end{pmatrix} 
        \left[ Y_{ijt}^{(2,n^{(1)})} - 
        \text{expit}\!\left(
        \begin{pmatrix} \boldsymbol{\beta}_A \\ \B_{\ell} \end{pmatrix}^{\top} 
        \begin{pmatrix} \Atwolago\\ \BoldLtwolago \end{pmatrix} \right) \right] 
    \right\},
\end{align}
where $\text{expit}(x) = e^x / (1 + e^x)$.

Define 
\begin{align}\label{smallu_binary}
    u(\boldsymbol{\beta}) 
    &= \sum_{t=1}^{T^{(1)}} \sum_{j=1}^{J^{(1)}} \alpha_{jt}^{(1)} 
    \begin{pmatrix} \Aone \\ \BoldLone \end{pmatrix}
    \left[
        \mathrm{expit}\!\left(
        \begin{pmatrix} \boldsymbol{\beta}^*_A \\ \B^*_{\ell} 
        \end{pmatrix}^{\top}
        \begin{pmatrix} \Aone \\ \BoldLone \end{pmatrix}
        \right)
        -
        \mathrm{expit}\!\left(
        \begin{pmatrix} \boldsymbol{\beta}_A \\ \B_{\ell} 
        \end{pmatrix}^{\top}
        \begin{pmatrix} \Aone \\ \BoldLone \end{pmatrix}
        \right)
    \right] \nonumber \\
    &+ \sum_{t=T^{(1)}+1}^{T^{(2)}} \sum_{j=1}^{J^{(2)}} \alpha_{jt}^{(2)} 
    \begin{pmatrix} \Bolda_{jt}^{(2)} \\ \BoldLtwo \end{pmatrix}
    \left[
        \mathrm{expit}\!\left(
        \begin{pmatrix} \boldsymbol{\beta}^*_A \\ \B^*_{\ell} 
        \end{pmatrix}^{\top}
        \begin{pmatrix} \Bolda_{jt}^{(2)} \\ \BoldLtwo \end{pmatrix}
        \right)
        -
        \mathrm{expit}\!\left(
        \begin{pmatrix} \boldsymbol{\beta}_A \\ \B_{\ell} 
        \end{pmatrix}^{\top}
        \begin{pmatrix} \Bolda_{jt}^{(2)} \\ \BoldLtwo \end{pmatrix}
        \right)
    \right].
\end{align}

While Unbiased Estimating Equations and Consistency Theorems~\ref{unbiasedEE} and \ref{theoremConsistencyWithLearning} for binary outcomes remain the same as those for continuous outcomes, Asymptotic Normality Theorem~\ref{AsymptoticTheoremWithLearning} needs to be updated as follows.

\begin{theorem}[Asymptotic Normality] \label{AsymptoticTheoremBinary}
Under Assumptions \ref{ratioConvergence}-\ref{asspCompact}, 
conditional on the fixed center characteristics 
$\overline{\BoldZ}$,
\[
    \sqrt{n}(\widehat{\boldsymbol{\beta}} - 
    \boldsymbol{\beta}^{*}) \overset{\mathcal{D}}{\rightarrow} 
    \mathcal{N}\!\left(0,\, J^{*-1}\right),
\]
where
\begin{align*}
    J^{*} 
    &= \sum_{t=1}^{T^{(1)}} \sum_{j=1}^{J^{(1)}} \alpha_{jt}^{(1)} 
    \begin{pmatrix} \A_{jt}^{(1)}\\ \BoldLone \end{pmatrix}
    \begin{pmatrix} \A_{jt}^{(1)}\\ \BoldLone \end{pmatrix}^{\top}
    p\!\left(\A_{jt}^{(1)}; \boldsymbol{\beta}^*, \right)
    \left(1 - p\!\left(\A_{jt}^{(1)}; 
    \boldsymbol{\beta}^*\right)\right) \\
    &+ \sum_{t=T^{(1)}+1}^{T^{(2)}} \sum_{j=1}^{J^{(2)}} \alpha_{jt}^{(2)} 
    \begin{pmatrix} \Bolda_{jt}^{(2)}\\ \BoldLtwo \end{pmatrix}
    \begin{pmatrix} \Bolda_{jt}^{(2)}\\ \BoldLtwo \end{pmatrix}^{\top}
    p\!\left(\Bolda_{jt}^{(2)}; \boldsymbol{\beta}^*\right)
    \left(1 - p\!\left(\Bolda_{jt}^{(2)}; 
    \boldsymbol{\beta}^*\right)\right),
\end{align*}
with $p\!\left(\Bolda; \boldsymbol{\beta}^*\right) 
= \mathrm{expit}\!\left(\begin{pmatrix} 
\boldsymbol{\beta}^*_A \\ \B^*_{\ell} \end{pmatrix}^{\top} 
\begin{pmatrix} \Bolda \\ \boldsymbol{\ell}_j \end{pmatrix}
\right)$. 
\end{theorem}

\subsection{Proof of Unbiased Estimating Equations Theorem \ref{unbiasedEE}}\label{UnbiasedEETheorem_Binary}

\begin{proof}
Following the same steps as for continuous outcomes modeled with linear regression in Online Appendix \ref{UnbiasedEEProofAppendix}, we show that $\E[U(\boldsymbol{\beta}^{*}) \mid \overline{\BoldZ}] 
= 0$ by showing that, conditional on $\overline{\BoldZ}$, the 
expectation of each term in equation~\eqref{eqn:binary_ee} 
equals zero at $\B = \B^{*}$.

For stage 1 in equation~\eqref{eqn:binary_ee}, by the Law of Total Expectations, conditioning first on $\BarA^{(1)}$ and 
$\overline{\BoldZ}$,
\begin{align*}
    &\E\left[
        \begin{pmatrix} \Aone \\ \BoldLone \end{pmatrix}
        \left(Y_{ijt}^{(1)} - 
        \mathrm{expit}\!\left(
        \begin{pmatrix} \boldsymbol{\beta}^*_A \\ \B^*_{\ell} 
        \end{pmatrix}^{\top}
        \begin{pmatrix} \Aone \\ \BoldLone \end{pmatrix}
        \right)\right)
        \;\middle|\; \overline{\BoldZ}
    \right] \nonumber \\
    = &\, \E\left[
        \begin{pmatrix} \Aone \\ \BoldLone \end{pmatrix}
        \E\!\left[
            Y_{ijt}^{(1)} - 
            \mathrm{expit}\!\left(
            \begin{pmatrix} \boldsymbol{\beta}^*_A \\ \B^*_{\ell} 
            \end{pmatrix}^{\top}
            \begin{pmatrix} \Aone \\ \BoldLone \end{pmatrix}
            \right)
        \;\middle|\; \BarA^{(1)}, \overline{\BoldZ}
        \right]
        \;\middle|\; \overline{\BoldZ}
    \right] = 0,
\end{align*}
where the inner conditional expectation equals zero because of Conditional Exchangeability Assumption~\ref{asspCondExcha} and Main Model Assumption~\ref{mainmodel}.

For stage 2 in equation~\eqref{eqn:binary_ee}, the argument is analogous, 
with the additional step of conditioning on 
$\X^{(2,n^{(1)})}$ to account for the fact that the stage 2 
recommended intervention $\X^{(2,n^(1))}$ depends on stage 1 outcomes. By the 
Law of Total Expectations,
\begin{align*}
    &\E\left[
        \begin{pmatrix} \Atwolago \\ \BoldLtwolago \end{pmatrix}
        \left(Y_{ijt}^{(2,n^{(1)})} - 
        \mathrm{expit}\!\left(
        \begin{pmatrix} \boldsymbol{\beta}^*_A \\ \B^*_{\ell} 
        \end{pmatrix}^{\top}
        \begin{pmatrix} \Atwolago \\ \BoldLtwolago \end{pmatrix}
        \right)\right)
        \;\middle|\; \overline{\BoldZ}
    \right] \nonumber \\
    &\quad= \E\left[
        \begin{pmatrix} \Atwolago \\ \BoldLtwolago \end{pmatrix}
        \E\!\left[
            Y_{ijt}^{(2,n^{(1)})} - 
            \mathrm{expit}\!\left(
            \begin{pmatrix} \boldsymbol{\beta}^*_A \\ \B^*_{\ell} 
            \end{pmatrix}^{\top}
            \begin{pmatrix} \Atwolago \\ \BoldLtwolago \end{pmatrix}
            \right)
        \;\middle|\; \BarA^{(2,n^{(1)})}, \X^{(2,n^{(1)})}, \overline{\BoldZ}
        \right]
        \;\middle|\; \overline{\BoldZ}
    \right] = 0,
\end{align*}
where the inner conditional expectation equals zero because of Conditional Exchangeability Assumption~\ref{asspCondExcha} and Main Model Assumption~\ref{mainmodel}. Since the expectation of each term for each stage equals zero, it follows that $\E[U(\boldsymbol{\beta}^{*}) \mid 
\overline{\BoldZ}] = 0$.
\end{proof}

\subsection{Proof of Consistency of \texorpdfstring{$\widehat{\B}$}{Beta-hat} Theorem \ref{theoremConsistencyWithLearning}}\label{ConsistencyTheorem_Binary}
\begin{proof}
We follow the same steps as for continuous outcomes modeled with linear regression in Online Appendix \ref{ConsistencyProofAppendix}. From equations~\eqref{eqn:binary_ee} and \eqref{smallu_binary}, 
it follows that $U(\boldsymbol{\beta}) - u(\boldsymbol{\beta})$ 
can be decomposed as
\begin{align}
    U(\boldsymbol{\beta}) - u(\boldsymbol{\beta}) 
    = U(\boldsymbol{\beta}^*) + G_1 + G_2 + G_3 + G_4,
\end{align}
where
\begin{align}
    G_1 &= \sum_{t=1}^{T^{(1)}} \sum_{j=1}^{J^{(1)}} 
    \left(\frac{n_{jt}^{(1)}}{n} - \alpha_{jt}^{(1)}\right)
    \begin{pmatrix} \Aone \\ \BoldLone \end{pmatrix} \nonumber \\
    &\times
    \left[
        \mathrm{expit}\!\left(
        \begin{pmatrix} \boldsymbol{\beta}^*_A \\ \B^*_{\ell} 
        \end{pmatrix}^{\top}
        \begin{pmatrix} \Aone \\ \BoldLone \end{pmatrix}
        \right)
        -
        \mathrm{expit}\!\left(
        \begin{pmatrix} \boldsymbol{\beta}_A \\ \B_{\ell} 
        \end{pmatrix}^{\top}
        \begin{pmatrix} \Aone \\ \BoldLone \end{pmatrix}
        \right)
    \right], \label{G1_binary} \\
    G_2 &= \sum_{t=T^{(1)}+1}^{T^{(2)}} \sum_{j=1}^{J^{(2)}} 
    \left(\frac{n_{jt}^{(2)}}{n} - \alpha_{jt}^{(2)}\right)
    \begin{pmatrix} \Atwolago \\ \BoldLtwolago \end{pmatrix} \nonumber \\
    &\times
    \left[
        \mathrm{expit}\!\left(
        \begin{pmatrix} \boldsymbol{\beta}^*_A \\ \B^*_{\ell} 
        \end{pmatrix}^{\top}
        \begin{pmatrix} \Atwolago \\ \BoldLtwolago \end{pmatrix}
        \right)
        -
        \mathrm{expit}\!\left(
        \begin{pmatrix} \boldsymbol{\beta}_A \\ \B_{\ell} 
        \end{pmatrix}^{\top}
        \begin{pmatrix} \Atwolago \\ \BoldLtwolago \end{pmatrix}
        \right)
    \right], \label{G2_binary} \\
    G_3 &= \sum_{t=T^{(1)}+1}^{T^{(2)}} \sum_{j=1}^{J^{(2)}} \alpha_{jt}^{(2)}
    \left[
        \begin{pmatrix} \Atwolago \\ \BoldLtwolago \end{pmatrix}
        \mathrm{expit}\!\left(
        \begin{pmatrix} \boldsymbol{\beta}^*_A \\ \B^*_{\ell} 
        \end{pmatrix}^{\top}
        \begin{pmatrix} \Atwolago \\ \BoldLtwolago \end{pmatrix}
        \right) \right. \nonumber \\
    &\left.
        -
        \begin{pmatrix} \Bolda_{jt}^{(2)} \\ \BoldLtwo \end{pmatrix}
        \mathrm{expit}\!\left(
        \begin{pmatrix} \boldsymbol{\beta}^*_A \\ \B^*_{\ell} 
        \end{pmatrix}^{\top}
        \begin{pmatrix} \Bolda_{jt}^{(2)} \\ \BoldLtwo \end{pmatrix}
        \right)
    \right], \label{G3_binary} \\
    G_4 &= \sum_{t=T^{(1)}+1}^{T^{(2)}} \sum_{j=1}^{J^{(2)}} \alpha_{jt}^{(2)}
    \left[
        \begin{pmatrix} \Bolda_{jt}^{(2)} \\ \BoldLtwo \end{pmatrix}
        \mathrm{expit}\!\left(
        \begin{pmatrix} \boldsymbol{\beta}_A \\ \B_{\ell} 
        \end{pmatrix}^{\top}
        \begin{pmatrix} \Bolda_{jt}^{(2)} \\ \BoldLtwo \end{pmatrix}
        \right) \right. \nonumber \\
    &\left.
        -
        \begin{pmatrix} \Atwolago \\ \BoldLtwolago \end{pmatrix}
        \mathrm{expit}\!\left(
        \begin{pmatrix} \boldsymbol{\beta}_A \\ \B_{\ell} 
        \end{pmatrix}^{\top}
        \begin{pmatrix} \Atwolago \\ \BoldLtwolago \end{pmatrix}
        \right)
    \right], \label{G4_binary} 
\end{align}
Regarding $U(\boldsymbol{\beta}^*)$, the proof that 
$U(\boldsymbol{\beta}^*) \ConvP 0$ follows identically 
to the continuous outcomes modeled with linear regression in 
Online Appendix~\ref{ConsistencyProofAppendix}, replacing 
the continuous error terms $\epsilon_{ijt}^{(k)}$ with 
the binary error terms $Y_{ijt}^{(k)} - 
\mathrm{expit}\!\left(\boldsymbol{\beta}^{*\top}
\begin{pmatrix} \A_{jt}^{(k)} \\ \ell_j^{(k)} 
\end{pmatrix}\right)$, which satisfy 
$\E \left [Y_{ijt}^{(k)} - \mathrm{expit}\!\left(\boldsymbol{\beta}^{*\top}
\begin{pmatrix} \A_{jt}^{(k)} \\ \ell_j^{(k)} 
\end{pmatrix} \right)\middle | 
\A_{jt}^{(k)}, \overline{\BoldZ} \right] = 0$. The variance of 
the binary error terms is
\[
    \mathrm{Var}\!\left(Y_{ijt}^{(k)} \mid \A_{jt}^{(k)}, 
    \overline{\BoldZ}\right) = 
    p_{ij}^{(k)}\left(1 - p_{ij}^{(k)}\right),
\]
where $p_{ij}^{(k)} = \mathrm{expit}\!\left(
\boldsymbol{\beta}^{*\top}\begin{pmatrix} \A_{jt}^{(k)} \\ 
\ell_j^{(k)} \end{pmatrix}\right)$. Since 
$p_{ij}^{(k)}\left(1 - p_{ij}^{(k)}\right) \leq 1$ 
for all $p_{ij}^{(k)} \in (0,1)$, the variance is 
bounded, and the Chebyshev argument goes through 
unchanged, so $U(\boldsymbol{\beta}^*) \ConvP 0$.

Similar to the proof of $G_1$ and $G_2$ in the continuous outcomes modeled with linear regression in Online Appendix \ref{ConsistencyProofAppendix}, we prove that each of the finitely many $J^{(2)}$ 
summands of $G_2$ converges uniformly to 0. Consider any 
summand of $G_2$,
\begin{align}\label{summandG2_binary}
    \left(\frac{n_{jt}^{(2)}}{n} - \alpha_{jt}^{(2)}\right)
    \begin{pmatrix} \Atwolago\\ \BoldLtwolago \end{pmatrix}
    \left[
        \mathrm{expit}\!\left(
        \begin{pmatrix} \boldsymbol{\beta}^*_A \\ \B^*_{\ell} 
        \end{pmatrix}^{\top}
        \begin{pmatrix} \Atwolago \\ \BoldLtwolago \end{pmatrix}
        \right)
        -
        \mathrm{expit}\!\left(
        \begin{pmatrix} \boldsymbol{\beta}_A \\ \B_{\ell} 
        \end{pmatrix}^{\top}
        \begin{pmatrix} \Atwolago \\ \BoldLtwolago \end{pmatrix}
        \right)
    \right].
\end{align}
Following from the Compact Space Assumption~\ref{asspCompact}, 
$|\A_{j,p}^{(2,n^{(1)})}| \leq \mathcal{C}_A \; \forall j$, the difference 
of two expit terms is bounded by 1 uniformly in 
$\boldsymbol{\beta} \in \mathcal{B}$ such that
\begin{align*}
    \left|
        \mathrm{expit}\!\left(
        \begin{pmatrix} \boldsymbol{\beta}^*_A \\ \B^*_{\ell} 
        \end{pmatrix}^{\top}
        \begin{pmatrix} \Atwolago \\ \BoldLtwolago \end{pmatrix}
        \right)
        -
        \mathrm{expit}\!\left(
        \begin{pmatrix} \boldsymbol{\beta}_A \\ \B_{\ell} 
        \end{pmatrix}^{\top}
        \begin{pmatrix} \Atwolago \\ \BoldLtwolago \end{pmatrix}
        \right)
    \right| \leq 1.
\end{align*}
It follows that for the first entry of the summand in 
equation~\eqref{summandG2_binary},
\begin{align*}
    &\sup_{\boldsymbol{\beta} \in \mathcal{B}} 
    \left\| 
    \left(\frac{n_{jt}^{(2)}}{n} - \alpha_{jt}^{(2)}\right)
    \A_{j,1}^{(2,n^{(1)})}
    \left[
        \mathrm{expit}\!\left(
        \begin{pmatrix} \boldsymbol{\beta}^*_A \\ \B^*_{\ell} 
        \end{pmatrix}^{\top}
        \begin{pmatrix} \Atwolago \\ \BoldLtwolago \end{pmatrix}
        \right)
        -
        \mathrm{expit}\!\left(
        \begin{pmatrix} \boldsymbol{\beta}_A \\ \B_{\ell} 
        \end{pmatrix}^{\top}
        \begin{pmatrix} \Atwolago \\ \BoldLtwolago \end{pmatrix}
        \right)
    \right]
    \right\| \\
    &\leq \left|\frac{n_{jt}^{(2)}}{n} - \alpha_{jt}^{(2)}\right| 
    \cdot \mathcal{C}_A \cdot 1,
\end{align*}
which converges to 0 as $n \rightarrow \infty$ because of Assumption~\ref{ratioConvergence}. Similarly, the supremum of 
the other entries of the summand of $G_2$ in 
equation~\eqref{summandG2_binary} converges uniformly to 0 
as $n \rightarrow \infty$. It follows that the supremum of 
every summand of $G_2$ converges uniformly to 0. Therefore, 
because $J^{(2)}$ is fixed, 
$\sup_{\boldsymbol{\beta} \in \mathcal{B}}\|G_2\| 
\rightarrow 0$. Similarly, 
$\sup_{\boldsymbol{\beta} \in \mathcal{B}}\|G_1\| 
\rightarrow 0$.

Next, consider any summand of $G_3$. Since 
$\Atwolago \ConvP \Bolda_{jt}^{(2)}$ 
(Remark~\ref{convergenceofA}), by the Continuous Mapping 
Theorem, it follows that
\begin{align*}
    \begin{pmatrix} \Atwolago \\ \BoldLtwolago \end{pmatrix}
    \mathrm{expit}\!\left(
    \begin{pmatrix} \boldsymbol{\beta}^*_A \\ \B^*_{\ell} 
    \end{pmatrix}^{\top}
    \begin{pmatrix} \Atwolago \\ \BoldLtwolago \end{pmatrix}
    \right)
    \ConvP
    \begin{pmatrix} \Bolda_{jt}^{(2)} \\ \BoldLtwo \end{pmatrix}
    \mathrm{expit}\!\left(
    \begin{pmatrix} \boldsymbol{\beta}^*_A \\ \B^*_{\ell} 
    \end{pmatrix}^{\top}
    \begin{pmatrix} \Bolda_{jt}^{(2)} \\ \BoldLtwo \end{pmatrix}
    \right).
\end{align*}
Thus,
\begin{align*}
    \sup_{\boldsymbol{\beta} \in \mathcal{B}}
    \left\|
    \alpha_{jt}^{(2)}
    \left[
        \begin{pmatrix} \Atwolago \\ \BoldLtwolago \end{pmatrix}
        \mathrm{expit}\!\left(
        \begin{pmatrix} \boldsymbol{\beta}^*_A \\ \B^*_{\ell} 
        \end{pmatrix}^{\top}
        \begin{pmatrix} \Atwolago \\ \BoldLtwolago \end{pmatrix}
        \right)
        -
        \begin{pmatrix} \Bolda_{jt}^{(2)} \\ \BoldLtwo \end{pmatrix}
        \mathrm{expit}\!\left(
        \begin{pmatrix} \boldsymbol{\beta}^*_A \\ \B^*_{\ell} 
        \end{pmatrix}^{\top}
        \begin{pmatrix} \Bolda_{jt}^{(2)} \\ \BoldLtwo \end{pmatrix}
        \right)
    \right]
    \right\| \ConvP 0.
\end{align*}
Hence, because $J^{(2)}$ is fixed, 
$\sup_{\boldsymbol{\beta} \in \mathcal{B}}\|G_3\| \ConvP 0$.

Next, consider any summand of $G_4$. By the Mean Value 
Theorem, there exists $\widetilde{\A}_j^{(2)}$ 
on the line segment between $\A_{jt}^{(2,n^{(1)})}$ and 
$\Bolda_{jt}^{(2)}$ such that
\begin{align*}
    &\mathrm{expit}\!\left(
    \begin{pmatrix} \boldsymbol{\beta}_A \\ \B_{\ell} 
    \end{pmatrix}^{\top}
    \begin{pmatrix} \Bolda_{jt}^{(2)} \\ \BoldLtwo \end{pmatrix}
    \right)
    -
    \mathrm{expit}\!\left(
    \begin{pmatrix} \boldsymbol{\beta}_A \\ \B_{\ell} 
    \end{pmatrix}^{\top}
    \begin{pmatrix} \Atwolago \\ \BoldLtwolago \end{pmatrix}
    \right) \\
    &= \left.\frac{\partial}{\partial \Bolda}\right|_{\widetilde{\A}_j^{(2)}}
    \mathrm{expit}\!\left(
    \begin{pmatrix} \boldsymbol{\beta}_A \\ \B_{\ell} 
    \end{pmatrix}^{\top}
    \begin{pmatrix} \Bolda \\ \BoldLtwo \end{pmatrix}
    \right)
    \left(\Bolda_{jt}^{(2)} - \A_{jt}^{(2,n^{(1)})}\right). \numberthis \label{deriveG4}
\end{align*}
Since $|\boldsymbol{\beta}_{A}| \leq 
\mathcal{C}_{\boldsymbol{\beta}}$ by Compact Space
Assumption~\ref{asspCompact}, the derivative in equation~\eqref{deriveG4} is 
bounded by some finite constant $M < \infty$ uniformly 
in $\boldsymbol{\beta} \in \mathcal{B}$. Therefore, for the first 
entry of the summand of $G_4$,
\begin{align*}
    &\sup_{\boldsymbol{\beta} \in \mathcal{B}}
    \left\|
    \alpha_{jt}^{(2)}
    \left[
        \Bolda_{j,1}^{(2)}
        \,\mathrm{expit}\!\left(
        \begin{pmatrix} \boldsymbol{\beta}_A \\ \B_{\ell} 
        \end{pmatrix}^{\top}
        \begin{pmatrix} \Bolda_{jt}^{(2)} \\ \BoldLtwo \end{pmatrix}
        \right)
        -
        \A_{j,1}^{(2,n^{(1)})}
        \,\mathrm{expit}\!\left(
        \begin{pmatrix} \boldsymbol{\beta}_A \\ \B_{\ell} 
        \end{pmatrix}^{\top}
        \begin{pmatrix} \Atwolago \\ \BoldLtwolago \end{pmatrix}
        \right)
    \right]
    \right\| \\
    &\leq \alpha_{jt}^{(2)} \cdot M \cdot 
    \left\|\Bolda_{jt}^{(2)} - 
    \A_{jt}^{(2,n^{(1)})}\right\| \ConvP 0,
\end{align*}
where the convergence to zero follows from $ \A_{jt}^{(2,n^{(1)})} \ConvP \Bolda_{jt}^{(2)}$ (Remark~\ref{convergenceofA}). Similarly, the supremum of the other entries 
of the summand of $G_4$ converges uniformly to 0. 
Therefore, because $J^{(2)}$ is fixed, 
$\sup_{\boldsymbol{\beta} \in \mathcal{B}}\|G_4\| 
\ConvP 0$.

Thus, that finishes the proof of Theorem~\ref{theoremConsistencyWithLearning} for binary outcomes modeled with logistic regression.
\end{proof}

\subsection{Proof of Asymptotic Normality of \texorpdfstring{$\widehat{\B}$}{Beta-hat} Theorem \ref{AsymptoticTheoremBinary}}\label{AsymptoticNormalityTheorem_Binary}

\begin{proof}
We follow the same steps as for continuous outcomes modeled with linear regression in Online Appendix \ref{AsymptoticProofAppendix}, combined with the proofs in \citet{nevo2021analysis} for binary outcomes in the absence of confounding by indication. By the Mean Value Theorem applied to each component of $U(\boldsymbol{\beta})$ as in equation~\eqref{MVT_Asymptotic_Proof}, there exists 
$\widetilde{\boldsymbol{\beta}}$ on the line segment 
between $\widehat{\boldsymbol{\beta}}$ and 
$\boldsymbol{\beta}^*$ such that equation~\eqref{MVT_Asymptotic_Proof} holds. Define
\begin{align*}
    J(\B) &= -\frac{\partial}{\partial \boldsymbol{\beta}} U(\B) \\
    &= \frac{1}{n} \left\{
        \sum_{t=1}^{T^{(1)}} \sum_{j=1}^{J^{(1)}} \sum_{i=1}^{n_{jt}^{(1)}} 
        \begin{pmatrix} \Aone\\ \BoldLone \end{pmatrix}
        \begin{pmatrix} \Aone\\ \BoldLone \end{pmatrix}^T \right. \\
    &\qquad \times \left.
        \mathrm{expit}\!\left(
        \begin{pmatrix} \B_A \\ \B_{\ell} 
        \end{pmatrix}^{\top}
        \begin{pmatrix} \Aone \\ \BoldLone \end{pmatrix}
        \right)
        \left(1 - \mathrm{expit}\!\left(
        \begin{pmatrix} \B_A \\ \B_{\ell} 
        \end{pmatrix}^{\top}
        \begin{pmatrix} \Aone \\ \BoldLone \end{pmatrix}
        \right)\right) \right.\\
    &\quad \left. + 
        \sum_{t=T^{(1)}+1}^{T^{(2)}} \sum_{j=1}^{J^{(2)}} \sum_{i=1}^{n_{jt}^{(2)}} 
        \begin{pmatrix} \Atwolago\\ \BoldLtwolago \end{pmatrix}
        \begin{pmatrix} \Atwolago\\ \BoldLtwolago \end{pmatrix}^T \right. \\
    &\qquad \times \left.
        \mathrm{expit}\!\left(
        \begin{pmatrix} \B_A \\ \B_{\ell} 
        \end{pmatrix}^{\top}
        \begin{pmatrix} \Atwolago \\ \BoldLtwolago \end{pmatrix}
        \right)
        \left(1 - \mathrm{expit}\!\left(
        \begin{pmatrix} \B_A \\ \B_{\ell} 
        \end{pmatrix}^{\top}
        \begin{pmatrix} \Atwolago \\ \BoldLtwolago \end{pmatrix}
        \right)\right)
    \right\}.
\end{align*}
Since $\widetilde{\B}$ is on the line segment between $\widehat{\B}$ and $\B^*$ and $\widehat{\B} \ConvP \B^*$ (Consistency Theorem \ref{theoremConsistencyWithLearning}), $\Atwolago \ConvP \Bolda_{jt}^{(2)}$ (Remark~\ref{convergenceofA}), and the Continuous Mapping Theorem, 
\[
    J(\widetilde{\boldsymbol{\beta}}) 
    \ConvP J^{*},
\]
where
\begin{align}
    J^{*} &= 
        \sum_{t=1}^{T^{(1)}} \sum_{j=1}^{J^{(1)}} \alpha_{jt}^{(1)}  
        \begin{pmatrix} \A_{jt}^{(1)}\\ \BoldLone \end{pmatrix} 
        \begin{pmatrix} \A_{jt}^{(1)}\\ \BoldLone \end{pmatrix}^T \nonumber \\
    &\qquad \times
        \mathrm{expit}\!\left(
        \begin{pmatrix} \boldsymbol{\beta}^*_A \\ \B^*_{\ell} 
        \end{pmatrix}^{\top}
        \begin{pmatrix} \A_{jt}^{(1)} \\ \BoldLone \end{pmatrix}
        \right)
        \left(1 - \mathrm{expit}\!\left(
        \begin{pmatrix} \boldsymbol{\beta}^*_A \\ \B^*_{\ell} 
        \end{pmatrix}^{\top}
        \begin{pmatrix} \A_{jt}^{(1)} \\ \BoldLone \end{pmatrix}
        \right)\right) \nonumber \\
    &\quad + \sum_{t=T^{(1)}+1}^{T^{(2)}} \sum_{j=1}^{J^{(2)}} \alpha_{jt}^{(2)}  
        \begin{pmatrix} \Bolda_{jt}^{(2)}\\ \BoldLtwo \end{pmatrix} 
        \begin{pmatrix} \Bolda_{jt}^{(2)}\\ \BoldLtwo \end{pmatrix}^T \nonumber \\
    &\qquad \times
        \mathrm{expit}\!\left(
        \begin{pmatrix} \boldsymbol{\beta}^*_A \\ \B^*_{\ell} 
        \end{pmatrix}^{\top}
        \begin{pmatrix} \Bolda_{jt}^{(2)} \\ \BoldLtwo \end{pmatrix}
        \right)
        \left(1 - \mathrm{expit}\!\left(
        \begin{pmatrix} \boldsymbol{\beta}^*_A \\ \B^*_{\ell} 
        \end{pmatrix}^{\top}
        \begin{pmatrix} \Bolda_{jt}^{(2)} \\ \BoldLtwo \end{pmatrix}
        \right)\right).
        \label{J_binary}
\end{align}
From Slutsky's Theorem applied to equation~\eqref{MVT_Asymptotic_Proof}, the asymptotic 
distribution of $\sqrt{n}(\widehat{\boldsymbol{\beta}} 
-\boldsymbol{\beta}^{*})$ is the same as the asymptotic 
distribution of $\sqrt{n} \;J^{*-1} \; U(\boldsymbol{\beta}^{*})$.

It remains to show that $\sqrt{n}\,U(\boldsymbol{\beta}^*)$ 
converges in distribution to a normal distribution. The 
proof follows the coupling argument detailed in \citet{lindvall2002lectures} that was also used in \citet{nevo2021analysis}. Observe that $\sqrt{n}\,U(\boldsymbol{\beta}^*)$ 
can be written as
\begin{align} \label{squarerootU}
    &\sqrt{n}\,U(\boldsymbol{\beta}^*) \nonumber\\
    &= \frac{1}{\sqrt{n}} \left\{
        \sum_{t=1}^{T^{(1)}} \sum_{j=1}^{J^{(1)}} \sum_{i=1}^{n_{jt}^{(1)}} 
        \begin{pmatrix} \Aone\\ \BoldLone \end{pmatrix} 
        \left[Y_{ijt}^{(1)} - 
        \mathrm{expit}\!\left(
        \begin{pmatrix} \boldsymbol{\beta}^*_A \\ \B^*_{\ell} 
        \end{pmatrix}^{\top}
        \begin{pmatrix} \Aone \\ \BoldLone \end{pmatrix}
        \right)\right] \right. \nonumber \\
    &\phantom{=}\left. + 
        \sum_{t=T^{(1)}+1}^{T^{(2)}} \sum_{j=1}^{J^{(2)}} \sum_{i=1}^{n_{jt}^{(2)}} 
        \begin{pmatrix} \Atwolago\\ \BoldLtwolago \end{pmatrix} 
        \left[Y_{ijt}^{(2,n^{(1)})} - 
        \mathrm{expit}\!\left(
        \begin{pmatrix} \boldsymbol{\beta}^*_A \\ \B^*_{\ell} 
        \end{pmatrix}^{\top}
        \begin{pmatrix} \Atwolago \\ \BoldLtwolago \end{pmatrix}
        \right)\right]
    \right\}.
\end{align}
The two summands are not independent because 
$\A_{jt}^{(2,n^{(1)})}$ depends on the stage 1 outcomes. 
To handle this dependence, we apply the coupling argument 
of \citet{nevo2021analysis}. For each $j = 1,\ldots,J^{(2)}$, 
let $Y_{ijt}^{(2)}$ be independent Bernoulli random 
variables, independent of all stage 1 data, with success 
probability $\mathrm{expit}\!\left(
    \begin{pmatrix} \boldsymbol{\beta}^*_A \\ \B^*_{\ell} 
    \end{pmatrix}^{\top}
    \begin{pmatrix} \Bolda_{jt}^{(2)} \\ \BoldLtwo \end{pmatrix}
    \right)$.
By the coupling construction of \citet{nevo2021analysis}, 
we construct variables $\widetilde{Y}_{ij}^{(2,n^{(1)})}$ 
which, given the stage 1 data and $\A_{jt}^{(2,n^{(1)})}$, 
have the same distribution as the original 
$Y_{ijt}^{(2,n^{(1)})}$, but are coupled with the $Y_{ijt}^{(2)}$, the outcomes under the limiting interventions $\Bolda_{jt}^{(2)}$. The key property of the coupling is that replacing $Y_{ijt}^{(2,n^{(1)})}$ with $\widetilde{Y}_{ij}^{(2,n^{(1)})}$ does not affect the distribution of $\sqrt{n}\,U(\boldsymbol{\beta}^*)$. After using the coupling argument, $\sqrt{n}\,U(\boldsymbol{\beta}^*)$ has the same asymptotic distribution as
\begin{align}\label{aftercoupling_squarerootU}
    \frac{1}{\sqrt{n}} \left\{
        \sum_{t=1}^{T^{(1)}} \sum_{j=1}^{J^{(1)}} \sum_{i=1}^{n_{jt}^{(1)}} 
        \begin{pmatrix} \Aone\\ \BoldLone \end{pmatrix} 
        \left[Y_{ijt}^{(1)} - 
        \mathrm{expit}\!\left(
        \begin{pmatrix} \boldsymbol{\beta}^*_A \\ \B^*_{\ell} 
        \end{pmatrix}^{\top}
        \begin{pmatrix} \Aone \\ \BoldLone \end{pmatrix}
        \right)\right] \right.\nonumber \\
    \left. + 
        \sum_{t=T^{(1)}+1}^{T^{(2)}} \sum_{j=1}^{J^{(2)}} \sum_{i=1}^{n_{jt}^{(2)}} 
        \begin{pmatrix} \Bolda_{jt}^{(2)}\\ \BoldLtwo \end{pmatrix} 
        \left[Y_{ijt}^{(2)} - 
        \mathrm{expit}\!\left(
        \begin{pmatrix} \boldsymbol{\beta}^*_A \\ \B^*_{\ell} 
        \end{pmatrix}^{\top}
        \begin{pmatrix} \Bolda_{jt}^{(2)} \\ \BoldLtwo \end{pmatrix}
        \right)\right] 
    \right\},
\end{align}
because as in \citet{nevo2021analysis}, the difference with equation~\eqref{squarerootU} converges to zero in probability. The two terms in equation~\eqref{aftercoupling_squarerootU} are independent, since the $\Bolda_{jt}^{(2)}$ are fixed and the $\A_{jt}^{(1)}$ are also fixed pre-trial.

We now show that each term converges to a normal 
distribution. For the stage 1 term, by the 
Cram\'{e}r-Wold device, it suffices to show that 
for any vector $\boldsymbol{c}$,
\[
    \frac{1}{\sqrt{n}}
    \sum_{t=1}^{T^{(1)}} \sum_{j=1}^{J^{(1)}} \sum_{i=1}^{n_{jt}^{(1)}} 
    \boldsymbol{c}^{\top}
    \begin{pmatrix} \Aone\\ \BoldLone \end{pmatrix} 
    \epsilon_{ijt}^{(1)}
\]
converges in distribution to a normal random variable, 
where $\epsilon_{ijt}^{(1)} = Y_{ijt}^{(1)} - 
\mathrm{expit}\!\left(\boldsymbol{\beta}^{*\top}
\begin{pmatrix} \Aone \\ \BoldLone \end{pmatrix}
\right)$ (Main Model Assumption \ref{mainmodel}). Since $Y_{ijt}^{(1)} \in \{0,1\}$, we have $|\epsilon_{ijt}^{(1)}| \leq 1$, so the Lyapunov condition is satisfied with 
$\delta = 1$:
\begin{align*}
    &\frac{1}{n^{1+\delta/2}}
\sum_{t=1}^{T^{(1)}} \sum_{j=1}^{J^{(1)}} \sum_{i=1}^{n_{jt}^{(1)}} 
\E\!\left[
\left|\boldsymbol{c}^{\top}
\begin{pmatrix} \Aone\\ \BoldLone \end{pmatrix} 
\epsilon_{ijt}^{(1)}\right|^{2+\delta}
\mid \overline{\BoldZ}
\right] \\
&\leq 
\frac{1}{n^{\delta/2}}
\left(
\left|\max_{p=1,\ldots,P}(c_p)\right| \cdot \mathcal{C}_A \cdot P 
+ 
\left|\max_{m=1,\ldots,J+1}(c_m)\right| \cdot 2
\right)^{2+\delta}
\rightarrow 0
\end{align*}
since $|\A_{jt,p}^{(1)}| \leq \mathcal{C}_A$ by 
Compact Space Assumption~\ref{asspCompact} and 
$|\epsilon_{ijt}^{(1)}| \leq 1$. By the Lyapunov 
Central Limit Theorem, the stage 1 term converges 
in distribution to
\[
    N\!\left(0, \sum_{t=1}^{T^{(1)}} \sum_{j=1}^{J^{(1)}} \alpha_{jt}^{(1)} 
    \begin{pmatrix} \A_{jt}^{(1)}\\ \BoldLone \end{pmatrix}
    \begin{pmatrix} \A_{jt}^{(1)}\\ \BoldLone \end{pmatrix}^{\top}
    p\!\left(\A_{jt}^{(1)}; \boldsymbol{\beta}^* \right)
    \left(1 - p\!\left(\A_{jt}^{(1)}; 
    \boldsymbol{\beta}^*\right)\right)
    \right).
\]
By the same argument, conditionally on 
$\overline{\BoldZ}$, the stage 2 term converges in 
distribution to
\[
    N\!\left(0, \sum_{t=T^{(1)}+1}^{T^{(2)}} \sum_{j=1}^{J^{(2)}} \alpha_{jt}^{(2)} 
    \begin{pmatrix} \Bolda_{jt}^{(2)}\\ \BoldLtwo \end{pmatrix}
    \begin{pmatrix} \Bolda_{jt}^{(2)}\\ \BoldLtwo \end{pmatrix}^{\top}
    p\!\left(\Bolda_{jt}^{(2)}; \boldsymbol{\beta}^* \right)
    \left(1 - p\!\left(\Bolda_{jt}^{(2)}; 
    \boldsymbol{\beta}^*\right)\right)
    \right).
\]
Since the stage 1 and stage 2 terms in equation~\eqref{aftercoupling_squarerootU} are independent, their sum converges in distribution 
to $N(0, J^*)$, where $J^*$ is as defined in 
Theorem~\ref{AsymptoticTheoremBinary}. Combining 
with $J(\widetilde{\boldsymbol{\beta}}) 
\ConvP J^*$ and Slutsky's Theorem give
\[
    \sqrt{n}\,(\widehat{\boldsymbol{\beta}} - 
    \boldsymbol{\beta}^*) \ConvD
    \mathcal{N}\!\left(0,\, J^{*-1}\right).
\]
Therefore, Asymptotic Normality Theorem~\ref{AsymptoticTheoremBinary} for binary outcomes follows.
\end{proof}

\section{Proofs of Theorems \ref{unbiasedEE}, \ref{theoremConsistencyWithLearning}, and \ref{AsymptoticTheoremWithLearning} with Generalized Linear Model (GLM) with a general link function}\label{proofs_glm_general_link}
\begin{assumption}[Main Model]\label{mainmodel_GLM}
The counterfactual outcome $Y_{ijt}^{(k)}(\Bolda)$ 
of participant $i$ in center $j$ at stage $k$ with 
center characteristics $\BoldZ_j$ under intervention 
package $\Bolda$ satisfies
\[
    Y_{ijt}^{(k)}{(\Bolda)} = g^{-1}\left (\begin{pmatrix} \B^{*}_A \\ 
    \B^{*}_{\ell} \end{pmatrix}^{\top} 
    \begin{pmatrix} \Bolda \\ \boldsymbol{\ell}_j^{(k)} 
    \end{pmatrix} \right) + \epsilon_{ijt}^{(k)}(\Bolda),
\]
where $g(\cdot)$ is a known twice continuously 
differentiable link function whose inverse $g^{-1}$ 
has bounded first and second derivatives. The linear model case 
corresponds to $g^{-1}(x) = x$, and the binary 
outcome modeled with logistic regression corresponds to~$g^{-1}(x) = \mathrm{expit}(x)$.
\end{assumption}

The estimating equations used to estimate $\boldsymbol{\beta}$ 
are
\begin{align}\label{eqn:GLM_ee}
    0 &= U(\boldsymbol{\beta}) \nonumber\\
    &= \frac{1}{n} \left\{ 
        \sum_{t=1}^{T^{(1)}} \sum_{j=1}^{J^{(1)}} \sum_{i=1}^{n_{jt}^{(1)}} 
        \frac{\partial}{\partial \boldsymbol{\beta}}
        g^{-1}\!\left(
        \begin{pmatrix} \boldsymbol{\beta}_A \\ \B_{\ell} 
        \end{pmatrix}^{\top}
        \begin{pmatrix} \Aone\\ \BoldLone \end{pmatrix} \right) \right. \nonumber\\
    &\qquad \times \left.
        \left[ Y_{ijt}^{(1)} - 
        g^{-1}\!\left(
        \begin{pmatrix} \boldsymbol{\beta}_A \\ \B_{\ell} 
        \end{pmatrix}^{\top}
        \begin{pmatrix} \Aone\\ \BoldLone \end{pmatrix} \right) \right] 
        \right. \nonumber\\
    &\quad \left. + 
        \sum_{t=T^{(1)}+1}^{T^{(2)}} \sum_{j=1}^{J^{(2)}} \sum_{i=1}^{n_{jt}^{(2)}} 
        \frac{\partial}{\partial \boldsymbol{\beta}}
        g^{-1}\!\left(
        \begin{pmatrix} \boldsymbol{\beta}_A \\ \B_{\ell} 
        \end{pmatrix}^{\top}
        \begin{pmatrix} \Atwolago\\ \BoldLtwolago \end{pmatrix} \right) \right. \nonumber\\
    &\qquad \times \left.
        \left[ Y_{ijt}^{(2,n^{(1)})} - 
        g^{-1}\!\left(
        \begin{pmatrix} \boldsymbol{\beta}_A \\ \B_{\ell} 
        \end{pmatrix}^{\top}
        \begin{pmatrix} \Atwolago\\ \BoldLtwolago \end{pmatrix} \right) \right] 
    \right\}.
\end{align}
Define
\begin{align}\label{smallu_GLM}
    u(\boldsymbol{\beta}) 
    &= \sum_{t=1}^{T^{(1)}} \sum_{j=1}^{J^{(1)}} \alpha_{jt}^{(1)} 
    \frac{\partial}{\partial \boldsymbol{\beta}}
    g^{-1}\!\left(
    \begin{pmatrix} \boldsymbol{\beta}_A \\ \B_{\ell} 
    \end{pmatrix}^{\top}
    \begin{pmatrix} \Aone \\ \BoldLone \end{pmatrix}
    \right) \nonumber \\
    &\qquad \times
    \left[
        g^{-1}\!\left(
        \begin{pmatrix} \boldsymbol{\beta}^*_A \\ \B^*_{\ell} 
        \end{pmatrix}^{\top}
        \begin{pmatrix} \Aone \\ \BoldLone \end{pmatrix}
        \right)
        -
        g^{-1}\!\left(
        \begin{pmatrix} \boldsymbol{\beta}_A \\ \B_{\ell} 
        \end{pmatrix}^{\top}
        \begin{pmatrix} \Aone \\ \BoldLone \end{pmatrix}
        \right)
    \right] \nonumber \\
    &+ \sum_{t=T^{(1)}+1}^{T^{(2)}} \sum_{j=1}^{J^{(2)}} \alpha_{jt}^{(2)} 
    \frac{\partial}{\partial \boldsymbol{\beta}}
    g^{-1}\!\left(
    \begin{pmatrix} \boldsymbol{\beta}_A \\ \B_{\ell} 
    \end{pmatrix}^{\top}
    \begin{pmatrix} \Bolda_{jt}^{(2)} \\ \BoldLtwo \end{pmatrix}
    \right) \nonumber \\
    &\qquad \times
    \left[
        g^{-1}\!\left(
        \begin{pmatrix} \boldsymbol{\beta}^*_A \\ \B^*_{\ell} 
        \end{pmatrix}^{\top}
        \begin{pmatrix} \Bolda_{jt}^{(2)} \\ \BoldLtwo \end{pmatrix}
        \right)
        -
        g^{-1}\!\left(
        \begin{pmatrix} \boldsymbol{\beta}_A \\ \B_{\ell} 
        \end{pmatrix}^{\top}
        \begin{pmatrix} \Bolda_{jt}^{(2)} \\ \BoldLtwo \end{pmatrix}
        \right)
    \right].
\end{align}

While Unbiased Estimating Equations and Consistency Theorems~\ref{unbiasedEE} and \ref{theoremConsistencyWithLearning} in the GLM setting with a general link function remain the same as those in the main text, Asymptotic Normality Theorem~\ref{AsymptoticTheoremWithLearning} in the main text needs to be updated as follows.

\begin{theorem}[Asymptotic Normality]\label{AsymptoticTheoremGLM} Under 
Assumptions~\ref{ratioConvergence}-\ref{mainmodel_GLM}, 
conditional on the fixed center characteristics 
$\overline{\BoldZ}$,
\[
    \sqrt{n}(\widehat{\boldsymbol{\beta}} - 
    \boldsymbol{\beta}^{*}) \overset{\mathcal{D}}{\rightarrow} 
    \mathcal{N}\!\left(0,\, J(\boldsymbol{\beta}^*)^{-1} 
    \cdot V(\boldsymbol{\beta}^*) \cdot 
    J(\boldsymbol{\beta}^*)^{-T}\right),
\]
where
\begin{align*}
    J(\boldsymbol{\beta}^*) 
    &= \sum_{t=1}^{T^{(1)}} \sum_{j=1}^{J^{(1)}} \alpha_{jt}^{(1)} 
    \left(
    \left.\frac{\partial}{\partial\boldsymbol{\beta}}\right|_{\boldsymbol{\beta}^*}
    g^{-1}\!\left(
    \begin{pmatrix} \boldsymbol{\beta}_A \\ \B_{\ell} 
    \end{pmatrix}^{\top}
    \begin{pmatrix} \A_{jt}^{(1)} \\ \BoldLone \end{pmatrix}
    \right)
    \right)^{\otimes 2} \\
    &+ \sum_{t=T^{(1)}+1}^{T^{(2)}} \sum_{j=1}^{J^{(2)}} \alpha_{jt}^{(2)} 
    \left(
    \left.\frac{\partial}{\partial\boldsymbol{\beta}}\right|_{\boldsymbol{\beta}^*}
    g^{-1}\!\left(
    \begin{pmatrix} \boldsymbol{\beta}_A \\ \B_{\ell} 
    \end{pmatrix}^{\top}
    \begin{pmatrix} \Bolda_{jt}^{(2)} \\ \BoldLtwo \end{pmatrix}
    \right)
    \right)^{\otimes 2}\\
    V(\boldsymbol{\beta}^*) 
    &= \sum_{t=1}^{T^{(1)}} \sum_{j=1}^{J^{(1)}} \alpha_{jt}^{(1)} 
    \left(
    \left.\frac{\partial}{\partial\boldsymbol{\beta}}\right|_{\boldsymbol{\beta}^*}
    g^{-1}\!\left(
    \begin{pmatrix} \boldsymbol{\beta}_A \\ \B_{\ell} 
    \end{pmatrix}^{\top}
    \begin{pmatrix} \A_{jt}^{(1)} \\ \BoldLone \end{pmatrix}
    \right)
    \right)^{\otimes 2}
    \sigma^2(\BoldZ_j)\\
    &+ \sum_{t=T^{(1)}+1}^{T^{(2)}} \sum_{j=1}^{J^{(2)}} \alpha_{jt}^{(2)} 
    \left(
    \left.\frac{\partial}{\partial\boldsymbol{\beta}}\right|_{\boldsymbol{\beta}^*}
    g^{-1}\!\left(
    \begin{pmatrix} \boldsymbol{\beta}_A \\ \B_{\ell} 
    \end{pmatrix}^{\top}
    \begin{pmatrix} \Bolda_{jt}^{(2)} \\ \BoldLtwo \end{pmatrix}
    \right)
    \right)^{\otimes 2}
    \sigma^2(\BoldZ_j).
\end{align*}
\end{theorem}

\subsection{Proof of Unbiased Estimating Equations Theorem \ref{unbiasedEE}}\label{UnbiasedEE_GLM_General}

\begin{proof}
Following the same steps as for continuous outcomes modeled with linear regression in Online Appendix \ref{UnbiasedEEProofAppendix}, we show that $\E[U(\boldsymbol{\beta}^{*}) \mid \overline{\BoldZ}] = 0$ by showing that, conditional on $\overline{\BoldZ}$, the expectation of each term in equation~\eqref{eqn:GLM_ee} equals zero.

For stage 1, by the Law of Total Expectations, conditioning first on $\BarA^{(1)}$ and 
$\overline{\BoldZ}$,
\begin{align}
    &\E\left[
        \frac{\partial}{\partial \boldsymbol{\beta}}
        g^{-1}\!\left(
        \begin{pmatrix} \boldsymbol{\beta}^*_A \\ \B^*_{\ell} 
        \end{pmatrix}^{\top}
        \begin{pmatrix} \Aone \\ \BoldLone \end{pmatrix}
        \right)
        \left(Y_{ijt}^{(1)} - 
        g^{-1}\!\left(
        \begin{pmatrix} \boldsymbol{\beta}^*_A \\ \B^*_{\ell} 
        \end{pmatrix}^{\top}
        \begin{pmatrix} \Aone \\ \BoldLone \end{pmatrix}
        \right)\right)
        \;\middle|\; \overline{\BoldZ}
    \right] \nonumber \\
    = &\E\left[
        \frac{\partial}{\partial \boldsymbol{\beta}}
        g^{-1}\!\left(
        \begin{pmatrix} \boldsymbol{\beta}^*_A \\ \B^*_{\ell} 
        \end{pmatrix}^{\top}
        \begin{pmatrix} \Aone \\ \BoldLone \end{pmatrix}
        \right)
        \E\!\left[
            Y_{ijt}^{(1)} - 
            g^{-1}\!\left(
            \begin{pmatrix} \boldsymbol{\beta}^*_A \\ \B^*_{\ell} 
            \end{pmatrix}^{\top}
            \begin{pmatrix} \Aone \\ \BoldLone \end{pmatrix}
            \right)
        \;\middle|\; \BarA^{(1)}, \overline{\BoldZ}
        \right]
        \;\middle|\; \overline{\BoldZ}
    \right] = 0,
\end{align}
where we used the fact that 
$\frac{\partial}{\partial \boldsymbol{\beta}}
g^{-1}\!\left(\boldsymbol{\beta}^{*\top}
\begin{pmatrix} \Aone \\ \BoldLone \end{pmatrix}\right)
= (g^{-1})'\!\left(\boldsymbol{\beta}^{*\top}
\begin{pmatrix} \Aone \\ \BoldLone \end{pmatrix}\right)
\begin{pmatrix} \Aone \\ \BoldLone \end{pmatrix}$
is measurable with respect to $\BarA^{(1)}$ and 
$\overline{\BoldZ}$, and can therefore be pulled outside 
the inner conditional expectation. The inner expectation 
equals zero because, conditional on $\BarA^{(1)}$ and 
$\overline{\BoldZ}$, Conditional Exchangeability Assumption~\ref{asspCondExcha} and Main Model Assumption~\ref{mainmodel_GLM} imply that
\[
    \E\!\left[Y_{ijt}^{(1)} \mid \BarA^{(1)}, 
    \overline{\BoldZ}\right] 
    = g^{-1}\!\left(
    \boldsymbol{\beta}^{*\top} 
    \begin{pmatrix} \Aone \\ \BoldLone \end{pmatrix}\right).
\]
The outer expectation over the distribution of $\BarA^{(1)}$ 
given $\overline{\BoldZ}$ is then also zero.

For stage 2, the argument is 
analogous, but as for the linear link additionally conditions on the stage 2 recommended intervention $\X^{(2,n^{(1)})}$. By the Law of Total Expectations,
\begin{align}
    &\E\left[
        \frac{\partial}{\partial \boldsymbol{\beta}}
        g^{-1}\!\left(
        \begin{pmatrix} \boldsymbol{\beta}^*_A \\ \B^*_{\ell} 
        \end{pmatrix}^{\top}
        \begin{pmatrix} \Atwolago \\ \BoldLtwolago \end{pmatrix}
        \right)
        \left(Y_{ijt}^{(2,n^{(1)})} - 
        g^{-1}\!\left(
        \begin{pmatrix} \boldsymbol{\beta}^*_A \\ \B^*_{\ell} 
        \end{pmatrix}^{\top}
        \begin{pmatrix} \Atwolago \\ \BoldLtwolago \end{pmatrix}
        \right)\right)
        \;\middle|\; \overline{\BoldZ}
    \right] \nonumber \\
    &= \E\left[
        \frac{\partial}{\partial \boldsymbol{\beta}}
        g^{-1}\!\left(
        \begin{pmatrix} \boldsymbol{\beta}^*_A \\ \B^*_{\ell} 
        \end{pmatrix}^{\top}
        \begin{pmatrix} \Atwolago \\ \BoldLtwolago \end{pmatrix}
        \right) \cdot \right. \nonumber \\
    &\left.
        \E\!\left[
            Y_{ijt}^{(2,n^{(1)})} - 
            g^{-1}\!\left(
            \begin{pmatrix} \boldsymbol{\beta}^*_A \\ \B^*_{\ell} 
            \end{pmatrix}^{\top}
            \begin{pmatrix} \Atwolago \\ \BoldLtwolago \end{pmatrix}
            \right)
        \;\middle|\; \BarA^{(2,n^{(1)})}, \X^{(2,n^{(1)})}, \overline{\BoldZ}
        \right]
        \;\middle|\; \overline{\BoldZ}
    \right] = 0,
\end{align}
where the derivative 
$(g^{-1})'\!\left(\boldsymbol{\beta}^{*\top}
\begin{pmatrix} \Atwolago \\ \BoldLtwolago \end{pmatrix}\right)
\begin{pmatrix} \Atwolago \\ \BoldLtwolago \end{pmatrix}$
is measurable with respect to $\BarA^{(2,n^{(1)})}$ and 
$\overline{\BoldZ}$ and can therefore be pulled outside 
the inner conditional expectation. The inner expectation 
equals zero because, following from LAGO, Conditional Exchangeability and  Model Model Assumptions~\ref{asspAnte}, \ref{asspCondExcha}, \ref{mainmodel_GLM}, conditional on $\BarA^{(2,n^{(1)})}$, $\X^{(2,n^{(1)})}$, 
and $\overline{\BoldZ}$, the GLM model gives
\[
    \E\!\left[Y_{ijt}^{(2,n^{(1)})} \mid \BarA^{(2,n^{(1)})}, \X^{(2,n^{(1)})}, 
    \overline{\BoldZ}\right] 
    = g^{-1}\!\left(
    \boldsymbol{\beta}^{*\top} 
    \begin{pmatrix} \Atwolago \\ \BoldLtwolago \end{pmatrix}
    \right).
\]
Since the expectation for each term in each stage equals 
zero, we conclude that $\E[U(\boldsymbol{\beta}^{*}) \mid 
\overline{\BoldZ}] = 0$, establishing that 
$U(\boldsymbol{\beta})$ is an unbiased estimating equation 
for $\boldsymbol{\beta}^{*}$.
\end{proof}

\subsection{Proof of Consistency of \texorpdfstring{$\widehat{\B}$}{Beta-hat} Theorem \ref{theoremConsistencyWithLearning}}\label{ConsistencyTheorem_GLM_General}
\begin{proof}
We follow the same steps as for continuous outcomes modeled with linear regression in Online Appendix \ref{ConsistencyProofAppendix}. From equations~\eqref{eqn:GLM_ee} and \eqref{smallu_GLM}, 
it follows that $U(\boldsymbol{\beta}) - u(\boldsymbol{\beta})$ 
can be decomposed as
\begin{align}
    U(\boldsymbol{\beta}) - u(\boldsymbol{\beta}) 
    = G_{1,n}^{(g)} + G_{2,n}^{(g)} 
    + G_{3,n}^{(g)} + G_{4,n}^{(g)} + G_{5,n}^{(g)},
\end{align}
where
\begin{align}
    G_{1,n}^{(g)} &= \frac{1}{n}
    \sum_{t=1}^{T^{(1)}} \sum_{j=1}^{J^{(1)}}\sum_{i=1}^{n_{jt}^{(1)}}
    \frac{\partial}{\partial \boldsymbol{\beta}}
    g^{-1}\!\left(
    \begin{pmatrix} \boldsymbol{\beta}_A \\ \B_{\ell} 
    \end{pmatrix}^{\top}
    \begin{pmatrix} \Aone \\ \BoldLone \end{pmatrix}
    \right)
    \left[
        Y_{ijt}^{(1)} - 
        g^{-1}\!\left(
        \begin{pmatrix} \boldsymbol{\beta}^*_A \\ \B^*_{\ell} 
        \end{pmatrix}^{\top}
        \begin{pmatrix} \Aone \\ \BoldLone \end{pmatrix}
        \right)
    \right], \label{G1_GLM} 
\end{align}
\begin{align}
    G_{2,n}^{(g)} &= \frac{1}{n}
    \sum_{t=T^{(1)}+1}^{T^{(2)}} \sum_{j=1}^{J^{(2)}}\sum_{i=1}^{n_j^{(2)}}
    \frac{\partial}{\partial \boldsymbol{\beta}}
    g^{-1}\!\left(
    \begin{pmatrix} \boldsymbol{\beta}_A \\ \B_{\ell} 
    \end{pmatrix}^{\top}
    \begin{pmatrix} \Atwolago \\ \BoldLtwolago \end{pmatrix}
    \right)
    \left[
        Y_{ijt}^{(2,n^{(1)})} - 
        g^{-1}\!\left(
        \begin{pmatrix} \boldsymbol{\beta}^*_A \\ \B^*_{\ell} 
        \end{pmatrix}^{\top}
        \begin{pmatrix} \Atwolago \\ \BoldLtwolago \end{pmatrix}
        \right)
    \right], \label{G2_GLM} 
\end{align}
\begin{align}
    G_{3,n}^{(g)} &= \frac{1}{n}
    \sum_{t=T^{(1)}+1}^{T^{(2)}} \sum_{j=1}^{J^{(2)}}\sum_{i=1}^{n_j^{(2)}}
    \left[
    \frac{\partial}{\partial \boldsymbol{\beta}}
    g^{-1}\!\left(
    \begin{pmatrix} \boldsymbol{\beta}_A \\ \B_{\ell} 
    \end{pmatrix}^{\top}
    \begin{pmatrix} \Atwolago \\ \BoldLtwolago \end{pmatrix}
    \right)
    g^{-1}\!\left(
    \begin{pmatrix} \boldsymbol{\beta}^*_A \\ \B^*_{\ell} 
    \end{pmatrix}^{\top}
    \begin{pmatrix} \Atwolago \\ \BoldLtwolago \end{pmatrix}
    \right) \right. \nonumber \\
    &\phantom{=\frac{1}{n}\sum} \left.
    -
    \frac{\partial}{\partial \boldsymbol{\beta}}
    g^{-1}\!\left(
    \begin{pmatrix} \boldsymbol{\beta}_A \\ \B_{\ell} 
    \end{pmatrix}^{\top}
    \begin{pmatrix} \Bolda_{jt}^{(2)} \\ \BoldLtwo \end{pmatrix}
    \right)
    g^{-1}\!\left(
    \begin{pmatrix} \boldsymbol{\beta}^*_A \\ \B^*_{\ell} 
    \end{pmatrix}^{\top}
    \begin{pmatrix} \Bolda_{jt}^{(2)} \\ \BoldLtwo \end{pmatrix}
    \right)
    \right], \label{G3_GLM} 
\end{align}
\begin{align}
    G_{4,n}^{(g)} &= \frac{1}{n}
    \sum_{t=T^{(1)}+1}^{T^{(2)}} \sum_{j=1}^{J^{(2)}}\sum_{i=1}^{n_j^{(2)}}
    \left[
    \frac{\partial}{\partial \boldsymbol{\beta}}
    g^{-1}\!\left(
    \begin{pmatrix} \boldsymbol{\beta}_A \\ \B_{\ell} 
    \end{pmatrix}^{\top}
    \begin{pmatrix} \Bolda_{jt}^{(2)} \\ \BoldLtwo \end{pmatrix}
    \right)
    g^{-1}\!\left(
    \begin{pmatrix} \boldsymbol{\beta}_A \\ \B_{\ell} 
    \end{pmatrix}^{\top}
    \begin{pmatrix} \Bolda_{jt}^{(2)} \\ \BoldLtwo \end{pmatrix}
    \right) \right. \nonumber \\
    &\phantom{=\frac{1}{n}\sum} \left.
    -
    \frac{\partial}{\partial \boldsymbol{\beta}}
    g^{-1}\!\left(
    \begin{pmatrix} \boldsymbol{\beta}_A \\ \B_{\ell} 
    \end{pmatrix}^{\top}
    \begin{pmatrix} \Atwolago \\ \BoldLtwolago \end{pmatrix}
    \right)
    g^{-1}\!\left(
    \begin{pmatrix} \boldsymbol{\beta}_A \\ \B_{\ell} 
    \end{pmatrix}^{\top}
    \begin{pmatrix} \Atwolago \\ \BoldLtwolago \end{pmatrix}
    \right)
    \right], \label{G4_GLM} 
\end{align}
\begin{align}
    G_{5,n}^{(g)} &= 
    \sum_{t=1}^{T^{(1)}} \sum_{j=1}^{J^{(1)}} 
    \left(\alpha_{jt}^{(1)} - \frac{n_{jt}^{(1)}}{n}\right)
    \frac{\partial}{\partial \boldsymbol{\beta}}
    g^{-1}\!\left(
    \begin{pmatrix} \boldsymbol{\beta}_A \\ \B_{\ell} 
    \end{pmatrix}^{\top}
    \begin{pmatrix} \Aone \\ \BoldLone \end{pmatrix}
    \right)
    g^{-1}\!\left(
    \begin{pmatrix} \boldsymbol{\beta}_A \\ \B_{\ell} 
    \end{pmatrix}^{\top}
    \begin{pmatrix} \Aone \\ \BoldLone \end{pmatrix}
    \right) \nonumber \\
    &+ \sum_{t=1}^{T^{(1)}} \sum_{j=1}^{J^{(1)}} 
    \left(\frac{n_{jt}^{(1)}}{n} - \alpha_{jt}^{(1)}\right)
    \frac{\partial}{\partial \boldsymbol{\beta}}
    g^{-1}\!\left(
    \begin{pmatrix} \boldsymbol{\beta}_A \\ \B_{\ell} 
    \end{pmatrix}^{\top}
    \begin{pmatrix} \Aone \\ \BoldLone \end{pmatrix}
    \right)
    g^{-1}\!\left(
    \begin{pmatrix} \boldsymbol{\beta}^*_A \\ \B^*_{\ell} 
    \end{pmatrix}^{\top}
    \begin{pmatrix} \Aone \\ \BoldLone \end{pmatrix}
    \right) \nonumber \\
    &+ \sum_{t=T^{(1)}+1}^{T^{(2)}} \sum_{j=1}^{J^{(2)}} 
    \left(\alpha_{jt}^{(2)} - \frac{n_{jt}^{(2)}}{n}\right)
    \frac{\partial}{\partial \boldsymbol{\beta}}
    g^{-1}\!\left(
    \begin{pmatrix} \boldsymbol{\beta}_A \\ \B_{\ell} 
    \end{pmatrix}^{\top}
    \begin{pmatrix} \Bolda_{jt}^{(2)} \\ \BoldLtwolago \end{pmatrix}
    \right)
    g^{-1}\!\left(
    \begin{pmatrix} \boldsymbol{\beta}_A \\ \B_{\ell} 
    \end{pmatrix}^{\top}
    \begin{pmatrix} \Bolda_{jt}^{(2)} \\ \BoldLtwolago \end{pmatrix}
    \right) \nonumber \\
    &+\sum_{t=T^{(1)}+1}^{T^{(2)}}  \sum_{j=1}^{J^{(2)}} 
    \left(\frac{n_{jt}^{(2)}}{n} - \alpha_{jt}^{(2)}\right)
    \frac{\partial}{\partial \boldsymbol{\beta}}
    g^{-1}\!\left(
    \begin{pmatrix} \boldsymbol{\beta}_A \\ \B_{\ell} 
    \end{pmatrix}^{\top}
    \begin{pmatrix} \Bolda_{jt}^{(2)} \\ \BoldLtwolago \end{pmatrix}
    \right)
    g^{-1}\!\left(
    \begin{pmatrix} \boldsymbol{\beta}^*_A \\ \B^*_{\ell} 
    \end{pmatrix}^{\top}
    \begin{pmatrix} \Bolda_{jt}^{(2)} \\ \BoldLtwolago \end{pmatrix}
    \right). \label{G5_GLM}
\end{align}

We show that the supremum over $\B$ of $G_{1,n}^{(g)}, G_{2,n}^{(g)}, G_{3,n}^{(g)}, G_{4,n}^{(g)}, G_{5,n}^{(g)}$ converges to 0 in probability.

For $G_{1,n}^{(g)}$, the proof follows the Donsker 
class argument of \citet{bing2023learn}, with one 
modification to account for the random $\A_{jt}^{(1)}$ 
due to confounding by indication. Define
\[
    \Psi_{\boldsymbol{\beta}}(O_{ij}^{(1)}) = 
    \frac{\partial}{\partial \boldsymbol{\beta}}
    g^{-1}\!\left(
    \boldsymbol{\beta}^{\top}
    \begin{pmatrix} \Aone \\ \BoldLone \end{pmatrix}
    \right)
    \left[
        Y_{ijt}^{(1)} - 
        g^{-1}\!\left(
        \boldsymbol{\beta}^{*\top}
        \begin{pmatrix} \Aone \\ \BoldLone \end{pmatrix}
        \right)
    \right],
\]
where $O_{ij}^{(1)} = (Y_{ijt}^{(1)}, \A_{jt}^{(1)})$, 
so that $G_{1,n}^{(g)} = \frac{1}{n}
\sum_{t=1}^{T^{(1)}} \sum_{j=1}^{J^{(1)}}\sum_{i=1}^{n_{jt}^{(1)}} 
\Psi_{\boldsymbol{\beta}}(O_{ij}^{(1)})$.
Conditioning on $\BarA^{(1)}$ and $\overline{\BoldZ}$, 
the observations $O_{ij}^{(1)}$ are i.i.d. within 
each center $j$ and independent between centers by Conditional Independence Between 
Centers Assumption~\ref{asspcenterIndep}, and the 
function class $\mathcal{F} = 
\{\Psi_{\boldsymbol{\beta}} : \boldsymbol{\beta} 
\in \mathcal{B}\}$ satisfies the Lipschitz condition 
by the Mean Value Theorem, the $C^2$ property of 
$g^{-1}$ (Main Model Assumption~\ref{mainmodel_GLM}), and 
compactness of $\mathcal{B}$ 
(Compact Space Assumption~\ref{asspCompact}), same as in 
\citet{bing2023learn}. By Example 19.7 and Theorem 
19.4 of \citet{van2000asymptotic}, $\mathcal{F}$ 
is Donsker and Glivenko-Cantelli. Since 
$\E[\Psi_{\boldsymbol{\beta}}(O_{ij}^{(1)}) \mid 
\BarA^{(1)}, \overline{\BoldZ}] = 0$ by Unbiased Estimating Equations Theorem~\ref{unbiasedEE} and $n_j^{(1)}/n < 1$ 
with $J^{(1)}$ finite, 
$\sup_{\boldsymbol{\beta} \in \mathcal{B}} 
\|G_{1,n}^{(g)}\| \ConvP 0$ conditionally on 
$\BarA^{(1)}$ and $\overline{\BoldZ}$. The 
unconditional result follows for any $\epsilon > 0$ 
by the Dominated Convergence Theorem, since
\[
    P\!\left(\sup_{\boldsymbol{\beta} \in \mathcal{B}} 
    \|G_{1,n}^{(g)}\| > \epsilon\right) 
    = \E\!\left[P\!\left(
    \sup_{\boldsymbol{\beta} \in \mathcal{B}} 
    \|G_{1,n}^{(g)}\| > \epsilon 
    \;\middle|\; \BarA^{(1)}, \overline{\BoldZ}
    \right)\right] \rightarrow 0,
\]
as the inner conditional probability is a random 
variable bounded by 1 that converges to 0 for all 
values of $\BarA^{(1)}$ and $\overline{\BoldZ}$. 
The same argument can be used to adapt the proof 
from \citet{bing2023learn} for $G_{5,n}^{(g)}$, 
giving $\sup_{\boldsymbol{\beta} \in \mathcal{B}} 
\|G_{5,n}^{(g)}\| \ConvP 0$.

For the terms $G_{2,n}^{(g)}$, $G_{3,n}^{(g)}$, and $G_{4,n}^{(g)}$, the proofs follow exactly as in Appendix C in \citet{bing2023learn}, since these terms involve only stage 2 quantities where $\A_{jt}^{(2,n^{(1)})}$ converges in 
probability to $\Bolda_{jt}^{(2)}$ by Remark~\ref{convergenceofA}, and $g^{-1}$ is $C^2$ (Main Model Assumption~\ref{mainmodel_GLM}) with bounded 
derivatives on compact sets (Compact Set Assumption~\ref{asspCompact}). The first condition of Theorem 5.9 of \citet{van2000asymptotic} follows. 

The second condition of Theorem 5.9 of 
\citet{van2000asymptotic} requires that the
$u(\boldsymbol{\beta}^*) = 0$ and $u(\boldsymbol{\beta})$ 
has a unique zero at $\boldsymbol{\beta}^*$. The 
first requirement $u(\boldsymbol{\beta}^*) = 0$ follows from 
equation~\eqref{smallu_GLM}. The second requirement follows 
from the same argument as in \citet{bing2023learn}, 
using the strict monotonicity of $g^{-1}$ implied 
by Main Model Assumption~\ref{mainmodel_GLM}, the 
compactness of $\mathcal{B}$ from Compact Set 
Assumption~\ref{asspCompact}, and the variation in the intervention components to identify $\boldsymbol{\beta}^*$. Hence, the second condition of Theorem 5.9 of 
\citet{van2000asymptotic} is satisfied.

Consistency Theorem~\ref{theoremConsistencyWithLearning} in the 
GLM setting for continuous outcomes with a general link function follows.
\end{proof}

\subsection{Proof of Asymptotic Normality of \texorpdfstring{$\widehat{\B}$}{Beta-hat} Theorem \ref{AsymptoticTheoremGLM} }\label{AsymptoticNormalityTheorem_GLM_General}

\begin{proof}
We follow the same steps as for continuous outcomes modeled with linear regression in Online  Appendix \ref{AsymptoticProofAppendix}. By the mean value theorem applied to each component of $U(\boldsymbol{\beta})$, there exists $\widetilde{\boldsymbol{\beta}}$ on the line segment between $\widehat{\boldsymbol{\beta}}$ and $\boldsymbol{\beta}^*$ such that
\[
    0 = U(\widehat{\boldsymbol{\beta}}) = 
    U(\boldsymbol{\beta}^*) + 
    \left.\frac{\partial}{\partial\boldsymbol{\beta}}\right|_{\widetilde{\boldsymbol{\beta}}}
    U(\boldsymbol{\beta})
    (\widehat{\boldsymbol{\beta}} - \boldsymbol{\beta}^*),
\]
which gives
\[
    \sqrt{n}(\widehat{\boldsymbol{\beta}} - \boldsymbol{\beta}^*) 
    = -\left[\left.\frac{\partial}{\partial\boldsymbol{\beta}}\right|_{\widetilde{\boldsymbol{\beta}}}
    U(\boldsymbol{\beta})
    \right]^{-1} \sqrt{n}\, U(\boldsymbol{\beta}^*).
\]
We first show that 
$-\left.\frac{\partial}{\partial\boldsymbol{\beta}}\right|_{\widetilde{\boldsymbol{\beta}}}
U(\boldsymbol{\beta})
\ConvP J(\boldsymbol{\beta}^*)$, then that 
$\sqrt{n}\,U(\boldsymbol{\beta}^*)$ converges in 
distribution to $N(0, V(\boldsymbol{\beta}^*))$.

The proof for $-\left.\frac{\partial}{\partial\boldsymbol{\beta}}\right|_{\widetilde{\boldsymbol{\beta}}}
U(\boldsymbol{\beta})
\ConvP J(\boldsymbol{\beta}^*)$ follows 
\citet{bing2023learn} exactly, with 
one modification for the stage~1 terms. In 
\citet{bing2023learn}, $\Bolda_{jt}^{(1)}$ is fixed before 
the trial, so the convergence of the stage~1 terms 
follows directly. In our setting, $\A_{jt}^{(1)}$ is 
random due to potential confounding by indication. We therefore 
condition on $\overline{\BoldZ}$ 
throughout the stage 1 argument and use the 
Law of Total Expectations. 

Regarding $\sqrt{n}\,U(\boldsymbol{\beta}^*)$, the 
proof follows the same structure as for the case of continuous 
outcomes modeled with linear regression case in Online 
Appendix~\ref{AsymptoticProofAppendix}. Decompose 
$\sqrt{n}\,U(\boldsymbol{\beta}^*)$ into three terms 
$U_{1,n}$, $U_{2\_1,n}$, and $U_{2\_2,n}$. First, $U_{1,n}$ converges in 
distribution to
\[
    N\!\biggl(0, \sum_{t=1}^{T^{(1)}} \sum_{j=1}^{J^{(1)}} \alpha_{jt}^{(1)} 
    \left(
    \left.\frac{\partial}{\partial\boldsymbol{\beta}}\right|_{\boldsymbol{\beta}^*}
    g^{-1}\!\left(
    \begin{pmatrix} \boldsymbol{\beta}_A \\ \B_{\ell} 
    \end{pmatrix}^{\top}
    \begin{pmatrix} \A_{jt}^{(1)} \\ \BoldLone 
    \end{pmatrix}
    \right)\right)^{\otimes 2}
    \sigma^2(\BoldZ_j)\biggr)
\]
by the Cram\'{e}r-Wold device and the Lyapunov 
Central Limit Theorem, conditionally on $\overline{\BoldZ}$. Second, $U_{2\_1,n} 
\ConvP 0$ by Chebyshev's inequality, using that
$\max_j\|\A_{jt}^{(2,n^{(1)})} - \Bolda_{jt}^{(2)}\| 
\ConvP 0$ from Remark~\ref{convergenceofA} and the 
boundedness of $(g^{-1})'$ on compact sets 
(Main Model Assumption~\ref{mainmodel_GLM}), exactly as in 
\citet{bing2023learn}. Third, 
$U_{2\_2,n}$ converges in distribution to
\[
    N\!\biggl(0, \sum_{t=T^{(1)}+1}^{T^{(2)}} s\sum_{j=1}^{J^{(2)}} \alpha_{jt}^{(2)} 
    \left(
    \left.\frac{\partial}{\partial\boldsymbol{\beta}}\right|_{\boldsymbol{\beta}^*}
    g^{-1}\!\left(
    \begin{pmatrix} \boldsymbol{\beta}_A \\ \B_{\ell} 
    \end{pmatrix}^{\top}
    \begin{pmatrix} \Bolda_{jt}^{(2)} \\ \BoldLtwo 
    \end{pmatrix}
    \right)\right)^{\otimes 2}
    \sigma^2(\BoldZ_j)\biggr)
\]
by the same argument as $U_{1,n}$, after replacing 
$\epsilon_{ijt}^{(2,n^{(1)})}$ with $\epsilon_{ijt}^{(2)}$ under the limiting intervention $\Bolda_{jt}^{(2)}$ without changing the distribution of $U_{2\_2,n}$. 
Since $U_{1,n}$ and $U_{2\_2,n}$ are independent 
as $\Bolda_{jt}^{(2)}$ are constants, their 
sum converges in distribution to 
$N(0, V(\boldsymbol{\beta}^*))$. Combined with 
$U_{2\_1,n} \ConvP 0$ and 
$-\frac{\partial}{\partial\boldsymbol{\beta}}
U(\widetilde{\boldsymbol{\beta}}) \ConvP 
J(\boldsymbol{\beta}^*)$, Slutsky's Theorem leads to
\[
    \sqrt{n}\,(\widehat{\boldsymbol{\beta}} - 
    \boldsymbol{\beta}^*) \xrightarrow{D} 
    \mathcal{N}\!\left(0,\, J(\boldsymbol{\beta}^*)^{-1} 
    \cdot V(\boldsymbol{\beta}^*) \cdot 
    J(\boldsymbol{\beta}^*)^{-T}\right),
\]
where $J(\boldsymbol{\beta}^*)$ and 
$V(\boldsymbol{\beta}^*)$ are as defined in Asymptotic Normality
Theorem~\ref{AsymptoticTheoremGLM}.
\end{proof}

\section{Extension to the number of stages: \texorpdfstring{$K > 2$}{K > 2}}\label{extensionK2}

This section extends LAGO with confounding by indication to $K > 2$ stages, with 
notation consistent with the main text. Let $\X_j^{(k, n^{(k-1)})}$ and 
$\A_{jt}^{(k, n^{(k-1)})}$ denote the recommended and actual intervention 
package for center $j$ during visit $t$ at stage $k$, respectively, where 
$h_{jt}^{(k)}$ is a deterministic continuous function such that 
$\A_{jt}^{(k, n^{(k-1)})} = h_{jt}^{(k)}(\boldsymbol{x}_j^{(k, n^{(k-1)})})$. 
Let $Y_{ijt}^{(k, n^{(k-1)})}$ denote the outcome of participant $i$ in center $j$ 
at visit $t$ in stage $k$, and let 
$\boldsymbol{Y}_{jt}^{(k,n^{(k-1)})} = (Y_{1jt}^{(k,n^{(k-1)})}, \ldots, Y_{n_{jt}^{(k)}jt}^{(k,n^{(k-1)})})$ 
collect the outcomes across all participants in center $j$ at visit $t$. 
The stage-$k$ recommended interventions, actual interventions, outcomes, and errors 
are first aggregated across centers at each visit $t$ as follows:
$\overline{\X}_t^{(k,n^{(k-1)})} = (\X_{1t}^{(k,n^{(k-1)})}, \ldots, \X_{J^{(k)}t}^{(k,n^{(k-1)})})$,
$\overline{\A}_t^{(k,n^{(k-1)})} = (\A_{1t}^{(k,n^{(k-1)})}, \ldots, \A_{J^{(k)}t}^{(k,n^{(k-1)})})$,
$\overline{\boldsymbol{Y}}_t^{(k,n^{(k-1)})} = (\boldsymbol{Y}_{1t}^{(k,n^{(k-1)})}, \ldots, \boldsymbol{Y}_{J^{(k)}t}^{(k,n^{(k-1)})})$,
and $\overline{\boldsymbol{\epsilon}}_t^{(k,n^{(k-1)})} = (\boldsymbol{\epsilon}_{1t}^{(k,n^{(k-1)})}, \ldots, \boldsymbol{\epsilon}_{J^{(k)}t}^{(k,n^{(k-1)})})$.
Hence, we denote
$\overline{\X}^{(k,n^{(k-1)})} = \bigl(\overline{\X}_1^{(k,n^{(k-1)})}, \ldots, \overline{\X}_t^{(k,n^{(k-1)})}\bigr)$,
$\overline{\A}^{(k,n^{(k-1)})} = \bigl(\overline{\A}_1^{(k,n^{(k-1)})}, \ldots, \overline{\A}_t^{(k,n^{(k-1)})}\bigr)$,
$\overline{\boldsymbol{Y}}^{(k,n^{(k-1)})} = \bigl(\overline{\boldsymbol{Y}}_1^{(k,n^{(k-1)})}, \ldots, \overline{\boldsymbol{Y}}_t^{(k,n^{(k-1)})}\bigr)$,
and $\overline{\boldsymbol{\epsilon}}^{(k,n^{(k-1)})} = \bigl(\overline{\boldsymbol{\epsilon}}_1^{(k,n^{(k-1)})}, \ldots, \overline{\boldsymbol{\epsilon}}_t^{(k,n^{(k-1)})}\bigr)$. The superscript $(k, \PastN)$ indicates that $Y_{ijt}^{(k, \PastN)}$ and $\overline{\boldsymbol{Y}}^{(k, \PastN)}$ depend on the data of participants from all previous $k - 1$ stages through the recommended intervention for stage $k$. 

Let $\widetilde{\X}^{(k, \PastN)} = (\boldsymbol{x}^{(1)}, \overline{\X}^{(2, n^{(1)})}, \hdots, \overline{\X}^{(k, \PastN)})$, and $\widetilde{\A}^{(k, \PastN)} = (\overline{\A}^{(1)}, \hdots, \overline{\A}^{(k, \PastN)})$, and $\widetilde{\boldsymbol{Y}}^{(k, \PastN)} = (\overline{\boldsymbol{Y}}^{(1)}, \hdots, \overline{\boldsymbol{Y}}^{(k, \PastN)})$ be the recommended intervention package, the actual intervention package, and the outcomes in stages $1,2,\ldots,k$ combined.

The estimating equations used, similar to those of equation (\ref{eqn:OLS_ignore_ee}) in the main text, are
\begin{align*}
    0 &= U_K(\boldsymbol{\beta})\\
    &= \frac{1}{n} \left\{ 
        \sum_{t=1}^{T^{(1)}} \sum_{j=1}^{J^{(1)}} \sum_{i=1}^{n_{jt}^{(1)}} 
        \begin{pmatrix} \Aone\\ \BoldLone \end{pmatrix} 
        \left[Y_{ijt}^{(1)} - \begin{pmatrix} \boldsymbol{\beta}_A \\ \B_{\ell} \end{pmatrix}^T 
        \begin{pmatrix} \Aone\\ \BoldLone \end{pmatrix} \right] \right.\\
    &\phantom{= }+ \left.
        \sum_{k=2}^{K} \sum_{t=T^{(1)}+1}^{T^{(k)}} \sum_{j=1}^{J^{(k)}} \sum_{i=1}^{n_{j}^{(k)}} 
        \begin{pmatrix} \A_{jt}^{(k, \PastN)}\\ \boldsymbol{\ell}_j^{(k)} \end{pmatrix} 
        \left[Y_{ijt}^{(k,\PastN)} - \begin{pmatrix} \boldsymbol{\beta}_A \\ \B_{\ell} \end{pmatrix}^T 
        \begin{pmatrix} \A_{jt}^{(k, \PastN)}\\ \boldsymbol{\ell}_j^{(k)} \end{pmatrix}
        \right]
    \right\}, \numberthis \label{multiplestagesEE}
\end{align*}
where the second term includes a summation from $k = 2, \hdots, K$. Next, the following assumptions need to be adjusted for Theorems \ref{unbiasedEE}, \ref{theoremConsistencyWithLearning}, and \ref{AsymptoticTheoremWithLearning} to hold when $K > 2$.

\begin{assumption}\label{multipleasspAnte}
(LAGO Assumption). Conditional on $\overline{\X}^{(k, \PastN)}$ and $\overline{\BoldZ}$, $\bigl(\overline{\A}^{(k, \PastN)}, \overline{\boldsymbol{Y}}^{(k, \PastN)}\bigr)$ is independent of the data from the previous stages, $\bigl(\widetilde{\A}^{(k-1, \overline{n}_{(k-1)-})}, \widetilde{\boldsymbol{Y}}^{(k-1, \overline{n}_{(k-1)-})} \bigr)$. 
\end{assumption}

\begin{assumption}\label{multipleasspRecommendedIntervention}
(Limiting Recommendations). The stage $k$ recommended intervention $\widehat{\boldsymbol{x}}^{\text{opt}, (k,\PastN)}$ converges in probability to a limit $\boldsymbol{x}^{(k)}$.
\end{assumption}

For the case where $K = 2$, conditions for Limiting Recommendations Assumption \ref{multipleasspRecommendedIntervention} to be satisfied are established in the main text. At stage 3, when the recommended intervention equals the estimated optimal intervention based on the data of stage 1 and stage 2, $\hat{\boldsymbol{x}}^{\text{opt},(3, (n^{(1)}, n^{(2)}))} \rightarrow \hat{\boldsymbol{x}}^{\text{opt}, (3)}$ under the conditions for $K = 2$ of the main text. Through forward induction, this extends the Limiting Recommendations Assumption \ref{multipleasspRecommendedIntervention} to cases where $K > 2$. Under Limiting Recommendations Assumption \ref{multipleasspRecommendedIntervention}, the definition of $h_{jt}^{(k)}$, and the Continuous Mapping Theorem, $\A_{jt}^{(k, \PastN)} = h_{jt}^{(k)} \left(\X^{(k, \PastN)} \right)$ converges in probability to $\Bolda_{jt}^{(k)} = h_{jt}^{(k)}\left(\boldsymbol{x}^{(k)} \right)$.

Define
\begin{align*}
u_K(\boldsymbol{\beta}) &= \sum_{t=1}^{T^{(1)}} \sum_{j=1}^{J^{(1)}} \alpha_{jt}^{(1)}  \begin{pmatrix} \A_{jt}^{(1)}\\ \BoldLone \end{pmatrix} \left [ \begin{pmatrix} \boldsymbol{\beta}^{*}_A \\  \B^{*}_{\ell} \end{pmatrix}^T - \begin{pmatrix} \boldsymbol{\beta}_A \\  \B_{\ell} \end{pmatrix}^T \right ] \begin{pmatrix} \A_{jt}^{(1)}\\ \BoldLone \end{pmatrix} \\
&+  \sum_{k=2}^{K} \sum_{t=T^{(1)}+1}^{T^{(k)}} \sum_{j=1}^{J^{(k)}} \alpha_{jt}^{(k)}  \begin{pmatrix} \Bolda_{jt}^{(k)}\\ \boldsymbol{\ell}_j^{(k)} \end{pmatrix} \left [ \begin{pmatrix} \boldsymbol{\beta}^{*}_A \\  \B^{*}_{\ell} \end{pmatrix}^T - \begin{pmatrix} \boldsymbol{\beta}_A \\  \B_{\ell} \end{pmatrix}^T \right ] \begin{pmatrix} \Bolda_{jt}^{(k)}\\ \boldsymbol{\ell}_j^{(k)} \end{pmatrix}.
\end{align*}

As in the proof of Unbiased estimating equations Theorem \ref{unbiasedEE} for the terms for stage 2, the conditional expectation of the terms in equation (\ref{multiplestagesEE}) for each stage $k$ given the fixed center characteristics $\overline{\BoldZ}$ equals 0. It follows that $U(\B)$ are unbiased estimating equations when $K > 2$.

The proof of Consistency Theorem \ref{theoremConsistencyWithLearning} can be adapted in the following way. Consider $k = 3$ first. The term $U_3(\B) - u_3(\B)$ can be decomposed into five components, and their supremum over $\B$ can be proven to converge to 0 in probability separately. The most difficult term is $U_3(\B^{*})$, which does not contain $\B$. Using Consistency Assumption \ref{asspConsistency}, Main Model Assumption \ref{mainmodel}, Outcome Errors Assumption \ref{epsilonassumption}, and Conditional Independence Between Centers Assumption \ref{asspcenterIndep}, the two summands of stage 2 and stage 3 in $U_3(\B^{*})$ are uncorrelated by conditioning on the recommended interventions of stage 3, $\overline{\X}^{(3,n^{(2)})}$, and the center characteristics $\overline{\BoldZ}$. Then, following the same proof as in the main text, the variance of each entry of each stage in $U_3(\B^{*})$ converges to 0 as $n \rightarrow \infty$. Hence, each entry in $U_3(\B^{*})$ converges to 0 in probability by applying Chebyshev's inequality, so that $U_3(\B^{*}) \ConvP 0$. Hence, by forward induction, it follows that $\sup_{\beta}||U_{K}(\B^{*})|| \rightarrow 0$. The rest of the proof follows as before, and Consistency Theorem \ref{theoremConsistencyWithLearning} follows for settings with $K > 2$ stages.

The proof of  Asymptotic Normality Theorem \ref{AsymptoticTheoremGLM} with $K > 2$ stages extends by forward induction. The base case $K = 2$ is established in the main text, where we showed that $\sqrt{n}\,U(\B^{*})$ decomposes into stage 1 and stage 2 terms that are uncorrelated conditionally on $\overline{\BoldZ}$ that jointly converge to a normal distribution. For $K = 3$, we decompose $\sqrt{n}\,U(\B^{*})$ into terms for each stage. The stage 3 term is handled by first replacing the conditional errors $\epsilon_{ijt}^{(3,n^{(2)})}$ with unconditional errors $\epsilon_{ijt}^{(3)}$ without changing the distribution (Outcome Errors Assumption \ref{epsilonassumption}), then decomposing by adding and subtracting the limiting term $\Bolda_{jt}^{(3)}$ to separate the term into a main component involving $\Bolda_{jt}^{(3)}$ and $\epsilon_{ijt}^{(3)}$, plus a remainder term involving the difference between $\A_{jt}^{(3,n^{(2)})}$ and $\Bolda_{jt}^{(3)}$ that converges in probability to 0 as shown for the second term in the main text. For the stage-2 term, the conditional errors $\epsilon_{ijt}^{(2,n^{(1)})}$ is replaced by the unconditional errors $\epsilon_{ijt}^{(2)}$. Then, the terms for three stages are independent and using the same argument as in the main text, the sum converges to a normal distribution. The general case $K > 3$ follows by induction. Hence, Asymptotic Normality Theorem \ref{AsymptoticTheoremWithLearning} holds in the settings with $K > 2$.

Extension to $K > 2$ stages for binary outcomes modeled with logistic regression and for GLM with a general link function follow the same steps as for continuous outcomes modeled with linear regression outlined above. For further technical details, refer to Section 3 of the Supplementary Materials of \citet{nevo2021analysis} for the binary outcomes with logistic regression and to Appendix D of \citet{bing2023learn} for the GLM setting with a general link function.

\section{Additional Simulation Results}
\subsection{Parameter values used in simulation}\label{Parameter_reporting}
Twenty centers were simulated with characteristics drawn from a standard normal distribution $\BoldZ_j \sim N(0,1)$ and used across all scenarios. Random variations around the recommended intervention, $\eta_{1}$ and $\eta_{2}$, were also fixed across scenarios.

The 20 center characteristics $\BoldZ_j$ were:
-1.238, -0.456, -0.830, 0.340, 1.066, 1.216, 0.736, -0.481, 0.563, -1.246, 0.381, -1.430, -1.048, -0.219, -1.490, 1.173, -1.480, -0.430, -1.052, 1.523.

The corresponding $\eta_{1}$ values were:
-0.428, -1.569, -0.485, -1.521, -0.908, -0.268, -0.494, -0.626, -0.073, -0.182, -1.331, -0.056, 0.695, 0.133, -0.991, -1.031, 0.260, -1.263, -0.233, -0.468.

The corresponding $\eta_{2}$ values were:
0.315, -0.221, -0.240, 0.242, -0.280, -0.008, 0.067, -0.087, -0.196, 0.190, -0.116, -0.097, 0.396, 0.619, 0.146, 0.228, -0.261, 0.154, -0.187, 0.084.

\subsection{\texorpdfstring{$\eta_1$ and $\eta_2$}{eta1 and eta2} calculations}\label{eta_calculations}
For $m = 1, 2$ component, since $\xi_{ij,m}^{(1)} \sim N(0,1)$ and $\BoldZ_j \sim N(0,1)$ independently, the correlation between $\A_{ij,m}^{(1)}$ and $\BoldZ_j$ is:

\begin{align*}
\rho_k = \text{Cor}(\A_{ij,m}^{(1)}, \BoldZ_j)
&= \frac{\text{Cov}(x_m^{(1)} + \eta_m \BoldZ_j + \xi_{ij,m}^{(1)},\ \BoldZ_j)}
        {\sqrt{\text{Var}(x_m^{(1)} + \eta_m \BoldZ_j + \xi_{ij,m}^{(1)})} 
         \cdot \sqrt{\text{Var}(\BoldZ_j)}} \\
&= \frac{\eta_m \text{Var}(\BoldZ_j)}
        {\sqrt{\eta_m^2 \text{Var}(\BoldZ_j) + \text{Var}(\xi_{ij,m}^{(1)})} 
         \cdot \sqrt{\text{Var}(\BoldZ_j)}} \\
&= \frac{\eta_m}{\sqrt{\eta_m^2 + 1}},
\end{align*}

where the last equality uses $\text{Var}(\BoldZ_j) = \text{Var}(\xi_{ij,m}^{(1)}) = 1$. 
Solving for $\eta_m$ gives:
\begin{align*}
\eta_m = \frac{\rho_m}{\sqrt{1 - \rho_m^2}}, \quad m = 1, 2.
\end{align*}
\subsection{Simulation results for Scenario 0 under the null hypothesis and \texorpdfstring{$\rho_1 = 0.1, \rho_2 = 0.2$}{rho1 = 0.1, rho2 = 0.2}}\label{sim_null}
Tables \ref{sim_null_table_varying_highconf} and \ref{sim_null_table_fixed_highconf} present simulation results under a scenario where the null hypothesis holds: the effect of both intervention components is 0 ($H_0: \beta_1 = \beta_2 = 0$). Tables report bias (times 1000), the ratio of mean estimated standard errors to empirical standard deviations, and Type I error rates for individual tests of $H_0:\beta_1 = 0$ and $H_0: \beta_2 = 0$, as well as the joint Type I error rate for simultaneously testing both hypotheses $H_0: \beta_1 = \beta_2 = 0$. The individual error rates $\alpha_1$ and $\alpha_2$ represent the proportions of simulations incorrectly rejecting each null hypothesis, calculated based on the test statistic $\hat{\beta_1} / \text{SE}(\hat{\beta_1})$ and $\hat{\beta_2} / \text{SE}(\hat{\beta_2})$ using sandwich standard errors based on the sandwich estimation of the variance. The combined error rate ($\alpha_{\text{combined}}$) represents the proportion of simulations where the p-values for testing $H_0: \beta_1 = \beta_2 = 0$ fall below 0.05, where the Wald statistic follows a chi-squared distribution with 2 degrees of freedom under the null hypothesis.

Table \ref{sim_null_table_varying_highconf} reports the Type I errors for high confounding case in which $\rho_1 = 0.1$ and $\rho_2 = 0.2$ and $\BoldZ_j$ varies per simulated dataset. Table \ref{sim_null_table_fixed_highconf} reports the Type I errors for high confounding case in which $\rho_1 = 0.1$ and $\rho_2 = 0.2$ and $\BoldZ_j$ were simulated from the beginning and remained fixed for the all simulations.

Overall, for the individual tests, $\alpha_1$ and $\alpha_2$ were generally close to the desired 0.05 level (between 0.04 and 0.06 as expected given the number of simulation iterations (2000)). For the combined test, $\alpha_{\text{combined}}$ was between 0.04 to 0.06 across all scenarios, confirming that Type I errors were under control under the null hypothesis.

\FloatBarrier
\begin{table}[htbp]
\caption{Simulation results under the null $\boldsymbol{\beta} = \mathbf{0}$ for Scenario 0 ($\rho_1 = 0.1$ and $\rho_2 = 0.2$, varying $Z_j$)}
\label{sim_null_table_varying_highconf}
\fontsize{8pt}{10pt}\selectfont
\centering
\begin{tabular}{llllllllllll}
$n_j^{(1)}$ & $n_j^{(2)}$ & $J$ & & \multicolumn{3}{c}{$\hat{\beta}_{1}$} & & \multicolumn{3}{c}{$\hat{\beta}_{2}$} \\
 &  &  & & \makecell{Bias\\($\times 1000$)} & SE/EMP.SD & $\alpha_1$ & & \makecell{Bias\\($\times 1000$)} & SE/EMP.SD & $\alpha_2$ & $\alpha_{\text{combined}}$\\ \hline
 \\
\multicolumn{11}{l}{Scenario 0 ($J^{(1)} = J^{(2)} = J$)} \\
50 & 100 & 6 & & -0.26 & 98.62 & 0.049 & & -0.14 & 98.62 & 0.059 & 0.049\\
 & 100 & 10 & & -0.65 & 99.95 & 0.046 & & 1.05 & 98.75 & 0.055 & 0.053\\
 & 100 & 20 & & -0.03 & 98.70 & 0.056 & & 0.40 & 99.41 & 0.056 & 0.059\\
50 & 200 & 6 & & -0.77 & 99.63 & 0.056 & & 0.89 & 104.19 & 0.041 & 0.048\\
 & 200 & 10 & & 0.36 & 98.80 & 0.052 & & -1.18 & 102.36 & 0.046 & 0.046\\
 & 200 & 20 & & -0.20 & 101.63 & 0.047 & & 0.30 & 99.44 & 0.049 & 0.040\\
100 & 100 & 6 & & -0.42 & 97.33 & 0.051 & & -0.37 & 97.13 & 0.056 & 0.060\\
 & 100 & 10 & & -0.73 & 100.05 & 0.062 & & 0.58 & 97.34 & 0.061 & 0.063\\
 & 100 & 20 & & -0.76 & 99.70 & 0.053 & & 0.81 & 100.88 & 0.048 & 0.047\\
100 & 200 & 6 & & 0.70 & 105.35 & 0.040 & & -0.30 & 102.20 & 0.048 & 0.049\\
 & 200 & 10 & & 0.01 & 100.73 & 0.041 & & -0.60 & 100.12 & 0.049 & 0.051\\
 & 200 & 20 & & -0.00 & 95.95 & 0.065 & & 0.21 & 98.83 & 0.057 & 0.065\\
\end{tabular}
\footnotesize{
    \fontsize{8pt}{10pt}\selectfont
    \raggedright{
    \medskip \\
    $n_j^{(1)}$: number of participants per center $j$ at stage 1.\\
    $n_j^{(2)}$: number of participants per center $j$ at stage 2.\\
    SE/EMP.SD: mean estimated standard error to empirical standard deviation ratio, multiplied by 100.\\
    $J$: number of centers for each stage. \\
    Bias of $\hat{\beta}_{1}$: mean bias $\widehat{\beta}_1 - \beta_1^{*}$.\\
    Bias of $\hat{\beta}_{2}$: mean bias $\widehat{\beta}_2 - \beta_2^{*}$.\\
    $\alpha_1$: proportion of tests incorrectly rejecting the null hypothesis $H_0: \beta_1 = 0$ at $\alpha = 0.05$.\\
    $\alpha_2$: proportion of tests incorrectly rejecting the null hypothesis $H_0: \beta_2 = 0$ at $\alpha = 0.05$.\\
    $\alpha_{\text{combined}}$: proportion of tests incorrectly rejecting the null hypothesis $H_0: \beta_1 = \beta_2 = 0$ at $\alpha = 0.05$.\\
    }
}
\end{table}

\begin{table}[htbp]
\caption{Simulation results under the null $\boldsymbol{\beta} = \mathbf{0}$ for Scenario 0 ($\rho_1 = 0.1$ and $\rho_2 = 0.2$, fixed $Z_j$)}
\label{sim_null_table_fixed_highconf}
\fontsize{8pt}{10pt}\selectfont
\centering
\begin{tabular}{llllllllllll}
$n_j^{(1)}$ & $n_j^{(2)}$ & $J$ & & \multicolumn{3}{c}{$\hat{\beta}_{1}$} & & \multicolumn{3}{c}{$\hat{\beta}_{2}$} \\
 &  &  & & \makecell{Bias\\($\times 1000$)} & SE/EMP.SD & $\alpha_1$ & & \makecell{Bias\\($\times 1000$)} & SE/EMP.SD & $\alpha_2$ & $\alpha_{\text{combined}}$\\ \hline
 \\
\multicolumn{11}{l}{Scenario 0 ($J^{(1)} = J^{(2)} = J$)} \\
50 & 100 & 6 & & -0.99 & 99.87 & 0.057 & & 0.29 & 98.00 & 0.059 & 0.054\\
 & 100 & 10 & & -1.02 & 98.78 & 0.055 & & 0.88 & 101.18 & 0.051 & 0.054\\
 & 100 & 20 & & 0.96 & 97.86 & 0.051 & & -0.17 & 99.34 & 0.055 & 0.049\\
50 & 200 & 6 & & -1.08 & 100.52 & 0.056 & & 1.11 & 103.89 & 0.049 & 0.050\\
 & 200 & 10 & & -0.08 & 100.43 & 0.056 & & -0.64 & 100.94 & 0.047 & 0.057\\
 & 200 & 20 & & 0.01 & 102.51 & 0.043 & & 0.59 & 101.58 & 0.051 & 0.042\\
100 & 100 & 6 & & -0.45 & 97.59 & 0.058 & & -0.23 & 96.88 & 0.058 & 0.054\\
 & 100 & 10 & & -1.19 & 99.46 & 0.055 & & 0.65 & 98.81 & 0.052 & 0.060\\
 & 100 & 20 & & -0.02 & 99.60 & 0.058 & & 0.13 & 98.34 & 0.062 & 0.056\\
100 & 200 & 6 & & 0.78 & 103.71 & 0.043 & & -0.04 & 102.48 & 0.041 & 0.043\\
 & 200 & 10 & & -0.11 & 100.85 & 0.043 & & -0.04 & 99.52 & 0.042 & 0.047\\
 & 200 & 20 & & 0.13 & 98.53 & 0.057 & & -0.03 & 98.10 & 0.057 & 0.054\\
\end{tabular}
\footnotesize{
    \fontsize{8pt}{10pt}\selectfont
    \raggedright{
    \medskip \\
    $n_j^{(1)}$: number of participants per center $j$ at stage 1.\\
    $n_j^{(2)}$: number of participants per center $j$ at stage 2.\\
    SE/EMP.SD: mean estimated standard error to empirical standard deviation ratio, multiplied by 100.\\
    $J$: number of centers for each stage. \\
    Bias of $\hat{\beta}_{1}$: mean bias $\widehat{\beta}_1 - \beta_1^{*}$.\\
    Bias of $\hat{\beta}_{2}$: mean bias $\widehat{\beta}_2 - \beta_2^{*}$.\\
    $\alpha_1$: proportion of tests incorrectly rejecting the null hypothesis $H_0: \beta_1 = 0$ at $\alpha = 0.05$.\\
    $\alpha_2$: proportion of tests incorrectly rejecting the null hypothesis $H_0: \beta_2 = 0$ at $\alpha = 0.05$.\\
    $\alpha_{\text{combined}}$: proportion of tests incorrectly rejecting the null hypothesis $H_0: \beta_1 = \beta_2 = 0$ at $\alpha = 0.05$.\\
    }
}
\end{table}
\FloatBarrier

\subsection{Extension of Table \ref{simulationresults_linear_varying_Zj_lowconf}b in the main text} \label{bias_table_linear_varying_lowconf}
\FloatBarrier
\begin{table}[htbp]
\caption{Simulation study results for estimated optimal intervention with a linear cost function ($\rho_1 = 0.05$ and $\rho_2 = 0.07$, varying $\BoldZ_j$)}
\centering
\fontsize{8pt}{10pt}\selectfont
\begin{tabular}{lllllllllllll}
\makecell{$\boldsymbol{\beta}^*=(\beta_1^*,\beta_2^*)$\\(-1.70, -0.70)} & \makecell{$\boldsymbol{x}^{\text{opt}}$\\(2.94, 0)} & $n_j^{(1)}$ & $n_j^{(2)}$ & $J$ & & 
\multicolumn{3}{c}{Stage 1} & & \multicolumn{3}{c}{Stage 2/LAGO optimized} \\
 &  &  &  &  & & \begin{tabular}[c]{@{}l@{}}Bias of\\ $\hat{x}_1^{\text{opt},(1)}$\\ \end{tabular} & \begin{tabular}[c]{@{}l@{}}Bias of\\ $\hat{x}_2^{\text{opt},(1)}$\\ \end{tabular} & \begin{tabular}[c]{@{}l@{}}rMSE\\ \end{tabular} & & \begin{tabular}[c]{@{}l@{}}Bias of\\ $\hat{x}_1^{\text{opt},(2,n^{(1)})}$\\ \end{tabular} & \begin{tabular}[c]{@{}l@{}}Bias of\\ $\hat{x}_2^{\text{opt},(2,n^{(1)})}$\\ \end{tabular} & \begin{tabular}[c]{@{}l@{}}rMSE\\ \end{tabular} \\ \hline
\multicolumn{13}{l}{Scenario 1 ($J^{(1)} = J^{(2)} = J$)} \\
 &  & 50 & 100 & 6 & & -0.074 & 0.207 & 0.603 & & -0.005 & 0.076 & 0.484 \\
 &  &  & 100 & 10 & & -0.024 & 0.102 & 0.426 & & 0.002 & 0.054 & 0.380 \\
 &  &  & 100 & 20 & & 0.000 & 0.067 & 0.294 & & 0.008 & 0.054 & 0.281 \\
 &  & 50 & 200 & 6 & & -0.070 & 0.194 & 0.587 & & -0.003 & 0.077 & 0.477 \\
 &  &  & 200 & 10 & & 0.006 & 0.100 & 0.431 & & 0.022 & 0.065 & 0.396 \\
 &  &  & 200 & 20 & & -0.006 & 0.062 & 0.290 & & 0.000 & 0.050 & 0.277 \\
 &  & 100 & 100 & 6 & & -0.018 & 0.100 & 0.505 & & -0.001 & 0.072 & 0.477 \\
 &  & & 100 & 10 & & -0.006 & 0.064 & 0.385 & & -0.003 & 0.057 & 0.373 \\
 &  & & 100 & 20 & & -0.001 & 0.054 & 0.279 & & 0.000 & 0.052 & 0.277 \\
 &  & 100 & 200 & 6 & & -0.002 & 0.103 & 0.515 & & 0.011 & 0.078 & 0.490 \\
 &  & & 200 & 10 & & -0.006 & 0.069 & 0.387 & & -0.002 & 0.059 & 0.377 \\
 &  & & 200 & 20 & & 0.005 & 0.056 & 0.283 & & 0.006 & 0.050 & 0.275 \\
\end{tabular}
\footnotesize{
    \fontsize{8pt}{10pt}\selectfont
    \raggedright{
    \medskip \\
    $n_j^{(1)}$: number of participants per center $j$ at stage 1, 
    $n_j^{(2)}$: number of participants per center $j$ at stage 2.\\
    $J$: number of centers for each stage. \\
    Bias of $\hat{x}_{1}^{\text{opt}}$, $\hat{x}_{2}^{\text{opt}}$: bias of the first and second components of the estimated optimal intervention.\\
    rMSE: root mean squared error, $\{\operatorname{mean}(\|\hat{\boldsymbol{x}}^{o p t}-\boldsymbol{x}^{o p t}\|^{2})\}^{1 / 2}$, estimated by average over simulation iterations.\\
    }
}
\label{tab:optimization_performance_lowconf}
\end{table}
\FloatBarrier
\subsection{Extension of Table \ref{simulationresults_linear_varying_Zj_lowconf}c in the main text}\label{cp_table_linear_varying_lowconf}
\FloatBarrier
\begin{table}[htbp]
\caption{Simulation study results for estimated optimal intervention, confidence set, and confidence band with a linear cost function ($\rho_1 = 0.05$ and $\rho_2 = 0.07$, varying $\BoldZ_j$)}
\centering
\fontsize{8pt}{10pt}\selectfont
\begin{tabular}{llllllllll}
\makecell{$\boldsymbol{\beta}^*=(\beta_1^*,\beta_2^*)$\\(-1.70, -0.70)} & \makecell{$\boldsymbol{x}^{\text{opt}}$\\(2.94, 0)} & $n_j^{(1)}$ & $n_j^{(2)}$ & $J$ & \begin{tabular}[c]{@{}l@{}}TrueOpt1\\ (Q2.5,Q97.5)\end{tabular} & \begin{tabular}[c]{@{}l@{}}TrueOpt2\\ (Q2.5,Q97.5)\end{tabular} & \begin{tabular}[c]{@{}l@{}}SetCP95\\ $\%$\end{tabular} & \begin{tabular}[c]{@{}l@{}}SetPerc\\ $\%$\end{tabular} & \begin{tabular}[c]{@{}l@{}}BandsCP95\\ $\%$\end{tabular} \\ \hline
\multicolumn{10}{l}{Scenario 1 ($J^{(1)} = J^{(2)} = J$)} \\
 &  & 50 & 100 & 6 & (-5.331, -4.657) & (-5.189, -4.915) & 95.30 & 4.07 & 96.25 \\
 &  & & 100 & 10 & (-5.297, -4.760) & (-5.145, -4.942) & 95.00 & 3.17 & 95.70 \\
 &  & & 100 & 20 & (-5.225, -4.863) & (-5.125, -4.966) & 95.30 & 2.24 & 96.60 \\
 &  & 50 & 200 & 6 & (-5.323, -4.673) & (-5.152, -4.940) & 94.70 & 3.49 & 96.40 \\
 &  &  & 200 & 10 & (-5.295, -4.771) & (-5.130, -4.960) & 94.20 & 2.72 & 96.50 \\
 &  &  & 200 & 20 & (-5.247, -4.857) & (-5.110, -4.977) & 94.60 & 1.94 & 96.85 \\
 &  & 100 & 100 & 6 & (-5.266, -4.786) & (-5.167, -4.924) & 95.05 & 3.28 & 96.35 \\
 &  &  & 100 & 10 & (-5.231, -4.861) & (-5.136, -4.943) & 95.35 & 2.54 & 96.25 \\
 &  &  & 100 & 20 & (-5.174, -4.908) & (-5.121, -4.967) & 94.40 & 1.80 & 96.40 \\
 &  & 100 & 200 & 6 & (-5.272, -4.810) & (-5.142, -4.952) & 95.55 & 2.89 & 96.40 \\
 &  &  & 200 & 10 & (-5.237, -4.861) & (-5.121, -4.966) & 94.70 & 2.25 & 96.10 \\
 &  &  & 200 & 20 & (-5.189, -4.904) & (-5.111, -4.976) & 93.60 & 1.60 & 96.25 \\
\end{tabular}
\footnotesize{
    \fontsize{8pt}{10pt}\selectfont
    \raggedright{
    \medskip \\
    $n_j^{(1)}$: number of participants per center $j$ at stage 1, 
    $n_j^{(2)}$: number of participants per center $j$ at stage 2.\\
    $J$: number of centers for each stage. \\
    TrueOpt1: mean treatment effect under the stage 2 recommended intervention.\\
    TrueOpt2: mean treatment effect under the final estimated optimal intervention based on all data. \\
    $Q$2.5 and $Q$97.5: 2.5$\%$ and 97.5$\%$ quantiles.\\ 
    SetCP95$\%$: empirical coverage percentage of confidence set for the optimal intervention. \\
    SetPerc$\%$: mean percentage of the size of the confidence set as a percent of the total intervention space. \\
    BandsCP95$\%$: empirical coverage of 95$\%$ confidence band. \\
    }
}
\label{tab:confidence_metrics_lowconf}
\end{table}
\FloatBarrier

\normalsize

\subsection{Simulation results for Scenario 1 with varying \texorpdfstring{$\BoldZ_j$}{Z\_j} and \texorpdfstring{$\rho_1 = 0.1, \rho_2 = 0.2$}{rho1 = 0.1, rho2 = 0.2}}\label{simulationresults_linear_varying_Zj_highconf}
\FloatBarrier
\begin{table}[htbp]
\caption{Simulation study results for individual package component effects with a linear cost function ($\rho_1 = 0.1, \rho_2 = 0.2$, varying $Z_j$)}
\centering
\fontsize{8pt}{10pt}\selectfont
\begin{tabular}{llllllllll}
\makecell{$\boldsymbol{\beta}^*=(\beta_{1}^*, \beta_{2}^*)$ \\ (-1.70, -0.70)} & $n_j^{(1)}$ & $n_j^{(2)}$ & $J$ &  & $\widehat{\boldsymbol{\beta}}_{1}$ &  &  & $\widehat{\boldsymbol{\beta}}_{2}$ &  \\
  &  &  &  & $\%$RelBias & \begin{tabular}[c]{@{}l@{}}$\frac{SE}{EMP.SD}$ \\ ($\times 100$)\end{tabular} & CP95 & $\%$RelBias & \begin{tabular}[c]{@{}l@{}}$\frac{SE}{EMP.SD}$\\ ($\times 100$)\end{tabular} & CP95 \\ \hline
\multicolumn{4}{l}{Scenario 1 ($J^{(1)} = J^{(2)} = J$)} &  &  &  &  &  &  \\
 & 50 & 100 & 6 & 0.018 & 95.16 & 94.80 & 0.224 & 101.20 & 94.90 \\
 &  &  & 10 & -0.015 & 96.69 & 94.65 & 0.234 & 100.38 & 94.70 \\
 &  &  & 20 & 0.012 & 98.51 & 95.15 & 0.018 & 99.46 & 94.80 \\
 &  & 200 & 6 & 0.030 & 95.03 & 94.40 & 0.165 & 105.12 & 95.75 \\
 &  &  & 10 & 0.006 & 94.24 & 94.05 & -0.097 & 100.23 & 95.15 \\
 &  &  & 20 & -0.009 & 98.34 & 95.00 & 0.086 & 98.56 & 95.15 \\
 & 100 & 100 & 6 & -0.007 & 96.19 & 94.55 & -0.051 & 98.14 & 94.80 \\
 &  &  & 10 & -0.015 & 101.57 & 95.45 & 0.011 & 97.42 & 94.45 \\
 &  &  & 20 & -0.024 & 97.06 & 94.05 & 0.087 & 100.84 & 95.25 \\
 &  & 200 & 6 & 0.028 & 99.55 & 95.65 & 0.032 & 100.58 & 95.10 \\
 &  &  & 10 & 0.005 & 98.67 & 94.55 & -0.086 & 99.67 & 94.25 \\
 &  &  & 20 & 0.001 & 95.20 & 93.60 & 0.024 & 98.00 & 94.55 \\
\end{tabular}
\footnotesize{
    \fontsize{8pt}{10pt}\selectfont
    \raggedright{
    \medskip \\
    $n_j^{(1)}$: number of participants per center $j$ at stage 1, 
    $n_j^{(2)}$: number of participants per center $j$ at stage 2.\\
    $J$: number of centers for each stage. \\
    \%RelBias: percent relative bias $100(\hat{\beta}-\beta^\star)/\beta^\star$.\\
    SE: mean estimated standard error, 
    EMP.SD: empirical standard deviation.\\
    CP95: empirical coverage rate of 95\% confidence intervals.\\
    }
}
\label{tab:intervention_effects_highconf}
\end{table}
\FloatBarrier
\begin{table}[H]
\singlespacing
\caption{Simulation study results for estimated optimal intervention with a linear cost function ($\rho_1 = 0.1, \rho_2 = 0.2$, varying $Z_j$)}
\centering
\fontsize{8pt}{10pt}\selectfont
\begin{tabular}{lllllllllllll}
\makecell{$\boldsymbol{\beta}^*=(\beta_1^*,\beta_2^*)$\\(-1.70, -0.70)} & \makecell{$\boldsymbol{x}^{\text{opt}}$\\(2.94, 0)} & $n_j^{(1)}$ & $n_j^{(2)}$ & $J$ & & 
\multicolumn{3}{c}{Stage 1} & & \multicolumn{3}{c}{Stage 2/LAGO optimized} \\
 &  &  &  &  & & \begin{tabular}[c]{@{}l@{}}Bias of\\ $\hat{x}_1^{\text{opt},(1)}$\\ \end{tabular} & \begin{tabular}[c]{@{}l@{}}Bias of\\ $\hat{x}_2^{\text{opt},(1)}$\\ \end{tabular} & \begin{tabular}[c]{@{}l@{}}rMSE\\ \end{tabular} & & \begin{tabular}[c]{@{}l@{}}Bias of\\ $\hat{x}_1^{\text{opt},(2,n^{(1)})}$\\ \end{tabular} & \begin{tabular}[c]{@{}l@{}}Bias of\\ $\hat{x}_2^{\text{opt},(2,n^{(1)})}$\\ \end{tabular} & \begin{tabular}[c]{@{}l@{}}rMSE\\ \end{tabular} \\ \hline
\multicolumn{13}{l}{Scenario 1 ($J^{(1)} = J^{(2)} = J$)} \\
 &  & 50 & 100 & 6 & & -0.075 & 0.210 & 0.608 & & -0.006 & 0.077 & 0.485 \\
 &  &  & 100 & 10 & & -0.021 & 0.097 & 0.422 & & 0.002 & 0.054 & 0.380 \\
 &  &  & 100 & 20 & & -0.000 & 0.068 & 0.296 & & 0.008 & 0.052 & 0.280 \\
 &  & 50 & 200 & 6 & & -0.068 & 0.191 & 0.586 & & -0.003 & 0.077 & 0.477 \\
 &  &  & 200 & 10 & & 0.004 & 0.101 & 0.432 & & 0.022 & 0.065 & 0.395 \\
 &  &  & 200 & 20 & & -0.006 & 0.063 & 0.291 & & 0.000 & 0.049 & 0.276 \\
 &  & 100 & 100 & 6 & & -0.017 & 0.100 & 0.504 & & -0.002 & 0.072 & 0.478 \\
 &  & & 100 & 10 & & -0.005 & 0.062 & 0.382 & & -0.003 & 0.058 & 0.374 \\
 &  & & 100 & 20 & & -0.001 & 0.055 & 0.280 & & 0.000 & 0.051 & 0.277 \\
 &  & 100 & 200 & 6 & & -0.002 & 0.104 & 0.514 & & 0.011 & 0.078 & 0.490 \\
 &  & & 200 & 10 & & -0.005 & 0.068 & 0.386 & & -0.001 & 0.058 & 0.377 \\
 &  & & 200 & 20 & & 0.005 & 0.058 & 0.283 & & 0.006 & 0.050 & 0.275 \\
\end{tabular}
\footnotesize{
    \fontsize{8pt}{10pt}\selectfont
    \raggedright{
    \medskip \\
    $n_j^{(1)}$: number of participants per center $j$ at stage 1, 
    $n_j^{(2)}$: number of participants per center $j$ at stage 2.\\
    $J$: number of centers for each stage. \\
    Bias of $\hat{x}_{1}^{\text{opt}}$, $\hat{x}_{2}^{\text{opt}}$: bias of the first and second components of the estimated optimal intervention.\\
    rMSE: root mean squared error, $\{\operatorname{mean}(\|\hat{\boldsymbol{x}}^{o p t}-\boldsymbol{x}^{o p t}\|^{2})\}^{1 / 2}$, estimated by average over simulation iterations.\\
    }
}
\label{tab:optimization_performance_highconf}
\end{table}

\begin{table}[H]
\singlespacing
\caption{Simulation study results for estimated optimal intervention, confidence set, and confidence band with a linear cost function ($\rho_1 = 0.1, \rho_2 = 0.2$, varying $Z_j$)}
\centering
\fontsize{8pt}{10pt}\selectfont
\begin{tabular}{llllllllll}
\makecell{$\boldsymbol{\beta}^*=(\beta_1^*,\beta_2^*)$\\(-1.70, -0.70)} & \makecell{$\boldsymbol{x}^{\text{opt}}$\\(2.94, 0)} & $n_j^{(1)}$ & $n_j^{(2)}$ & $J$ & \begin{tabular}[c]{@{}l@{}}TrueOpt1\\ (Q2.5,Q97.5)\end{tabular} & \begin{tabular}[c]{@{}l@{}}TrueOpt2\\ (Q2.5,Q97.5)\end{tabular} & \begin{tabular}[c]{@{}l@{}}SetCP95\\ $\%$\end{tabular} & \begin{tabular}[c]{@{}l@{}}SetPerc\\ $\%$\end{tabular} & \begin{tabular}[c]{@{}l@{}}BandsCP95\\ $\%$\end{tabular} \\ \hline
\multicolumn{10}{l}{Scenario 1 ($J^{(1)} = J^{(2)} = J$)} \\
 &  & 50 & 100 & 6 & (-5.331, -4.657) & (-5.188, -4.915) & 94.90 & 4.07 & 96.10 \\
 &  & & 100 & 10 & (-5.303, -4.761) & (-5.144, -4.942) & 95.05 & 3.16 & 95.80 \\
 &  & & 100 & 20 & (-5.224, -4.860) & (-5.125, -4.966) & 95.05 & 2.24 & 96.55 \\
 &  & 50 & 200 & 6 & (-5.325, -4.666) & (-5.154, -4.941) & 94.95 & 3.49 & 96.45 \\
 &  &  & 200 & 10 & (-5.291, -4.775) & (-5.131, -4.959) & 94.15 & 2.71 & 96.40 \\
 &  &  & 200 & 20 & (-5.246, -4.855) & (-5.110, -4.977) & 95.00 & 1.93 & 96.80 \\
 &  & 100 & 100 & 6 & (-5.266, -4.780) & (-5.166, -4.924) & 94.95 & 3.28 & 96.30 \\
 &  &  & 100 & 10 & (-5.234, -4.861) & (-5.136, -4.942) & 95.25 & 2.54 & 96.35 \\
 &  &  & 100 & 20 & (-5.176, -4.907) & (-5.121, -4.967) & 94.50 & 1.80 & 96.25 \\
 &  & 100 & 200 & 6 & (-5.272, -4.809) & (-5.144, -4.949) & 95.45 & 2.89 & 96.40 \\
 &  &  & 200 & 10 & (-5.237, -4.863) & (-5.121, -4.966) & 94.75 & 2.25 & 96.10 \\
 &  &  & 200 & 20 & (-5.189, -4.902) & (-5.112, -4.976) & 93.50 & 1.59 & 96.15 \\
\end{tabular}
\footnotesize{
    \fontsize{8pt}{10pt}\selectfont
    \raggedright{
    \medskip \\
    $n_j^{(1)}$: number of participants per center $j$ at stage 1, 
    $n_j^{(2)}$: number of participants per center $j$ at stage 2.\\
    $J$: number of centers for each stage. \\
    TrueOpt1: mean treatment effect under the stage 2 recommended intervention.\\
    TrueOpt2: mean treatment effect under the final estimated optimal intervention based on all data. \\
    $Q$2.5 and $Q$97.5: 2.5$\%$ and 97.5$\%$ quantiles.\\ 
    SetCP95$\%$: empirical coverage percentage of confidence set for the optimal intervention. \\
    SetPerc$\%$: mean percentage of the size of the confidence set as a percent of the total intervention space.\\
    BandsCP95$\%$: empirical coverage of 95$\%$ confidence band. \\
    }
}
\label{tab:confidence_metrics_highconf}
\end{table}

\subsection{Simulation results for Scenario 1 with fixed \texorpdfstring{$\BoldZ_j$}{Z\_j} and \texorpdfstring{$\rho_1 = 0.05, \rho_2 = 0.07$}{rho1 = 0.05, rho2 = 0.07}}\label{simulationresults_linear_fixed_Zj_lowconf}
\FloatBarrier
\begin{table}[t]
\caption{Simulation study results for individual package component effects with a linear cost function ($\rho_1 = 0.05, \rho_2 = 0.07$, fixed $Z_j$)}
\centering
\fontsize{8pt}{10pt}\selectfont
\begin{tabular}{llllllllll}
\makecell{$\boldsymbol{\beta}^*=(\beta_{1}^*, \beta_{2}^*)$ \\ (-1.70, -0.70)} & $n_j^{(1)}$ & $n_j^{(2)}$ & $J$ &  & $\widehat{\boldsymbol{\beta}}_{1}$ &  &  & $\widehat{\boldsymbol{\beta}}_{2}$ &  \\
  &  &  &  & $\%$RelBias & \begin{tabular}[c]{@{}l@{}}$\frac{SE}{EMP.SD}$ \\ ($\times 100$)\end{tabular} & CP95 & $\%$RelBias & \begin{tabular}[c]{@{}l@{}}$\frac{SE}{EMP.SD}$\\ ($\times 100$)\end{tabular} & CP95 \\ \hline
\multicolumn{4}{l}{Scenario 1 ($J^{(1)} = J^{(2)} = J$)} &  &  &  &  &  &  \\
 & 50 & 100 & 6 & -0.039 & 98.54 & 94.75 & 0.323 & 101.61 & 94.15 \\
 &  &  & 10 & -0.010 & 96.95 & 94.35 & 0.155 & 101.56 & 95.40 \\
 &  &  & 20 & 0.022 & 95.29 & 94.80 & 0.047 & 99.41 & 94.95 \\
 &  & 200 & 6 & 0.005 & 97.43 & 94.65 & 0.179 & 104.33 & 95.80 \\
 &  &  & 10 & 0.000 & 98.12 & 94.85 & 0.059 & 99.31 & 94.20 \\
 &  &  & 20 & -0.010 & 98.43 & 94.75 & 0.189 & 98.18 & 94.20 \\
 & 100 & 100 & 6 & -0.009 & 95.87 & 94.30 & -0.029 & 97.35 & 94.30 \\
 &  &  & 10 & -0.047 & 101.58 & 95.05 & 0.101 & 99.47 & 94.60 \\
 &  &  & 20 & 0.008 & 98.30 & 94.70 & -0.024 & 97.72 & 95.10 \\
 &  & 200 & 6 & 0.024 & 99.88 & 95.85 & 0.101 & 101.25 & 94.75 \\
 &  &  & 10 & 0.005 & 101.65 & 95.90 & 0.005 & 100.44 & 94.65 \\
 &  &  & 20 & 0.004 & 99.55 & 95.20 & -0.014 & 97.69 & 94.30 \\
\end{tabular}
\footnotesize{
    \fontsize{8pt}{10pt}\selectfont
    \raggedright{
    \medskip \\
    $n_j^{(1)}$: number of participants per center $j$ at stage 1, 
    $n_j^{(2)}$: number of participants per center $j$ at stage 2.\\
    $J$: number of centers for each stage. \\
    \%RelBias: percent relative bias $100(\hat{\beta}-\beta^\star)/\beta^\star$.\\
    SE: mean estimated standard error, 
    EMP.SD: empirical standard deviation.\\
    CP95: empirical coverage rate of 95\% confidence intervals.\\
    }
}
\label{tab:intervention_effects_fixed_lowconf}
\end{table}
\FloatBarrier
\begin{table}[H]
\singlespacing
\caption{Simulation study results for estimated optimal intervention with a linear cost function ($\rho_1 = 0.05, \rho_2 = 0.07$, fixed $Z_j$)}
\centering
\fontsize{8pt}{10pt}\selectfont
\begin{tabular}{lllllllllllll}
\makecell{$\boldsymbol{\beta}^*=(\beta_1^*,\beta_2^*)$\\(-1.70, -0.70)} & \makecell{$\boldsymbol{x}^{\text{opt}}$\\(2.94, 0)} & $n_j^{(1)}$ & $n_j^{(2)}$ & $J$ & & 
\multicolumn{3}{c}{Stage 1} & & \multicolumn{3}{c}{Stage 2/LAGO optimized} \\
 &  &  &  &  & & \begin{tabular}[c]{@{}l@{}}Bias of\\ $\hat{x}_1^{\text{opt},(1)}$\\ \end{tabular} & \begin{tabular}[c]{@{}l@{}}Bias of\\ $\hat{x}_2^{\text{opt},(1)}$\\ \end{tabular} & \begin{tabular}[c]{@{}l@{}}rMSE\\ \end{tabular} & & \begin{tabular}[c]{@{}l@{}}Bias of\\ $\hat{x}_1^{\text{opt},(2,n^{(1)})}$\\ \end{tabular} & \begin{tabular}[c]{@{}l@{}}Bias of\\ $\hat{x}_2^{\text{opt},(2,n^{(1)})}$\\ \end{tabular} & \begin{tabular}[c]{@{}l@{}}rMSE\\ \end{tabular} \\ \hline
\multicolumn{13}{l}{Scenario 1 ($J^{(1)} = J^{(2)} = J$)} \\
 &  & 50 & 100 & 6 & & -0.397 & 0.157 & 0.468 & & -0.340 & 0.045 & 0.348 \\
 &  &  & 100 & 10 & & -0.745 & 0.146 & 0.788 & & -0.736 & 0.138 & 0.766 \\
 &  &  & 100 & 20 & & -0.081 & 0.051 & 0.112 & & -0.064 & 0.025 & 0.073 \\
 &  & 50 & 200 & 6 & & -0.393 & 0.142 & 0.457 & & -0.338 & 0.042 & 0.344 \\
 &  &  & 200 & 10 & & -0.743 & 0.148 & 0.786 & & -0.744 & 0.156 & 0.779 \\
 &  &  & 200 & 20 & & -0.079 & 0.047 & 0.107 & & -0.057 & 0.016 & 0.062 \\
 &  & 100 & 100 & 6 & & -0.352 & 0.073 & 0.375 & & -0.338 & 0.045 & 0.345 \\
 &  & & 100 & 10 & & -0.724 & 0.110 & 0.748 & & -0.736 & 0.138 & 0.766 \\
 &  & & 100 & 20 & & -0.073 & 0.036 & 0.090 & & -0.062 & 0.022 & 0.070 \\
 &  & 100 & 200 & 6 & & -0.348 & 0.068 & 0.369 & & -0.338 & 0.043 & 0.345 \\
 &  & & 200 & 10 & & -0.721 & 0.106 & 0.743 & & -0.742 & 0.154 & 0.777 \\
 &  & & 200 & 20 & & -0.073 & 0.038 & 0.091 & & -0.056 & 0.014 & 0.060 \\
\end{tabular}
\footnotesize{
    \fontsize{8pt}{10pt}\selectfont
    \raggedright{
    \medskip \\
    $n_j^{(1)}$: number of participants per center $j$ at stage 1, 
    $n_j^{(2)}$: number of participants per center $j$ at stage 2.\\
    $J$: number of centers for each stage. \\
    Bias of $\hat{x}_{1}^{\text{opt}}$, $\hat{x}_{2}^{\text{opt}}$: bias of the first and second components of the estimated optimal intervention.\\
    rMSE: root mean squared error, $\{\operatorname{mean}(\|\hat{\boldsymbol{x}}^{o p t}-\boldsymbol{x}^{o p t}\|^{2})\}^{1 / 2}$, estimated by average over simulation iterations.\\
    }
}
\label{tab:optimization_performance_fixed_lowconf}
\end{table}

\begin{table}[H]
\singlespacing
\caption{Simulation study results for estimated optimal intervention, confidence set, and confidence band with a linear cost function ($\rho_1 = 0.05, \rho_2 = 0.07$, fixed $Z_j$)}
\centering
\fontsize{8pt}{10pt}\selectfont
\begin{tabular}{llllllllll}
\makecell{$\boldsymbol{\beta}^*=(\beta_1^*,\beta_2^*)$\\(-1.70, -0.70)} & \makecell{$\boldsymbol{x}^{\text{opt}}$\\(2.94, 0)} & $n_j^{(1)}$ & $n_j^{(2)}$ & $J$ & \begin{tabular}[c]{@{}l@{}}TrueOpt1\\ (Q2.5,Q97.5)\end{tabular} & \begin{tabular}[c]{@{}l@{}}TrueOpt2\\ (Q2.5,Q97.5)\end{tabular} & \begin{tabular}[c]{@{}l@{}}SetCP95\\ $\%$\end{tabular} & \begin{tabular}[c]{@{}l@{}}SetPerc\\ $\%$\end{tabular} & \begin{tabular}[c]{@{}l@{}}BandsCP95\\ $\%$\end{tabular} \\ \hline
\multicolumn{10}{l}{Scenario 1 ($J^{(1)} = J^{(2)} = J$)} \\
 &  & 50 & 100 & 6 & (-5.246, -4.726) & (-5.176, -4.906) & 94.40 & 4.15 & 96.10 \\
 &  & & 100 & 10 & (-5.167, -4.827) & (-5.097, -4.927) & 94.65 & 3.33 & 96.00 \\
 &  & & 100 & 20 & (-5.246, -4.906) & (-5.076, -4.976) & 94.70 & 2.26 & 96.55 \\
 &  & 50 & 200 & 6 & (-5.246, -4.716) & (-5.176, -5.006) & 94.20 & 3.57 & 95.70 \\
 &  &  & 200 & 10 & (-5.167, -4.827) & (-5.097, -4.997) & 94.25 & 2.83 & 95.90 \\
 &  &  & 200 & 20 & (-5.246, -4.906) & (-5.076, -4.976) & 95.00 & 1.95 & 96.15 \\
 &  & 100 & 100 & 6 & (-5.246, -4.836) & (-5.176, -4.906) & 93.90 & 3.35 & 95.95 \\
 &  &  & 100 & 10 & (-5.167, -4.927) & (-5.097, -4.927) & 94.65 & 2.68 & 96.10 \\
 &  &  & 100 & 20 & (-5.146, -4.906) & (-5.076, -4.976) & 94.70 & 1.81 & 96.10 \\
 &  & 100 & 200 & 6 & (-5.246, -4.836) & (-5.176, -5.006) & 95.45 & 2.96 & 96.85 \\
 &  &  & 200 & 10 & (-5.167, -4.927) & (-5.097, -4.997) & 95.50 & 2.36 & 96.10 \\
 &  &  & 200 & 20 & (-5.146, -4.906) & (-5.076, -4.976) & 94.65 & 1.61 & 96.45 \\
\end{tabular}
\footnotesize{
    \fontsize{8pt}{10pt}\selectfont
    \raggedright{
    \medskip \\
    $n_j^{(1)}$: number of participants per center $j$ at stage 1, 
    $n_j^{(2)}$: number of participants per center $j$ at stage 2.\\
    $J$: number of centers for each stage. \\
    TrueOpt1: mean treatment effect under the stage 2 recommended intervention.\\
    TrueOpt2: mean treatment effect under the final estimated optimal intervention based on all data. \\
    $Q$2.5 and $Q$97.5: 2.5$\%$ and 97.5$\%$ quantiles.\\ 
    SetCP95$\%$: empirical coverage percentage of confidence set for the optimal intervention. \\
    SetPerc$\%$: mean percentage of the size of the confidence set as a percent of the total intervention space.\\
    BandsCP95$\%$: empirical coverage of 95$\%$ confidence band. \\
    }
}
\label{tab:confidence_metrics_fixed_lowconf}
\end{table}

\subsection{Simulation results for Scenario 1 with fixed \texorpdfstring{$\BoldZ_j$}{Z\_j} and \texorpdfstring{$\rho_1 = 0.1, \rho_2 = 0.2$}{rho1 = 0.1, rho2 = 0.2}}\label{simulationresults_linear_fixed_Zj_highconf}
\FloatBarrier
\begin{table}[t]
\caption{Simulation study results for individual package component effects with a linear cost function ($\rho_1 = 0.1, \rho_2 = 0.2$, fixed $Z_j$)}
\centering
\fontsize{8pt}{10pt}\selectfont
\begin{tabular}{llllllllll}
\makecell{$\boldsymbol{\beta}^*=(\beta_{1}^*, \beta_{2}^*)$ \\ (-1.70, -0.70)} & $n_j^{(1)}$ & $n_j^{(2)}$ & $J$ &  & $\widehat{\boldsymbol{\beta}}_{1}$ &  &  & $\widehat{\boldsymbol{\beta}}_{2}$ &  \\
  &  &  &  & $\%$RelBias & \begin{tabular}[c]{@{}l@{}}$\frac{SE}{EMP.SD}$ \\ ($\times 100$)\end{tabular} & CP95 & $\%$RelBias & \begin{tabular}[c]{@{}l@{}}$\frac{SE}{EMP.SD}$\\ ($\times 100$)\end{tabular} & CP95 \\ \hline
\multicolumn{4}{l}{Scenario 1 ($J^{(1)} = J^{(2)} = J$)} &  &  &  &  &  &  \\
 & 50 & 100 & 6 & -0.039 & 98.83 & 94.90 & 0.360 & 101.91 & 94.55 \\
 &  &  & 10 & -0.015 & 97.54 & 94.55 & 0.196 & 102.04 & 95.40 \\
 &  &  & 20 & 0.021 & 95.28 & 94.85 & 0.053 & 99.24 & 94.90 \\
 &  & 200 & 6 & 0.003 & 97.76 & 94.95 & 0.193 & 104.38 & 95.70 \\
 &  &  & 10 & -0.002 & 98.26 & 94.50 & 0.078 & 99.16 & 94.15 \\
 &  &  & 20 & -0.011 & 98.62 & 94.90 & 0.194 & 98.28 & 94.20 \\
 & 100 & 100 & 6 & -0.009 & 95.90 & 94.40 & -0.029 & 97.34 & 94.60 \\
 &  &  & 10 & -0.046 & 101.57 & 95.25 & 0.096 & 99.29 & 94.80 \\
 &  &  & 20 & 0.008 & 98.20 & 94.80 & -0.022 & 97.76 & 95.15 \\
 &  & 200 & 6 & 0.024 & 99.65 & 95.70 & 0.106 & 101.30 & 94.65 \\
 &  &  & 10 & 0.005 & 101.79 & 95.80 & 0.015 & 100.26 & 94.85 \\
 &  &  & 20 & 0.003 & 99.61 & 95.15 & -0.012 & 97.68 & 94.35 \\
\end{tabular}
\footnotesize{
    \fontsize{8pt}{10pt}\selectfont
    \raggedright{
    \medskip \\
    $n_j^{(1)}$: number of participants per center $j$ at stage 1, 
    $n_j^{(2)}$: number of participants per center $j$ at stage 2.\\
    $J$: number of centers for each stage. \\
    \%RelBias: percent relative bias $100(\hat{\beta}-\beta^\star)/\beta^\star$.\\
    SE: mean estimated standard error, 
    EMP.SD: empirical standard deviation.\\
    CP95: empirical coverage rate of 95\% confidence intervals.\\
    }
}
\label{tab:intervention_effects_fixed_highconf}
\end{table}
\FloatBarrier
\begin{table}[H]
\singlespacing
\caption{Simulation study results for estimated optimal intervention with a linear cost function ($\rho_1 = 0.1, \rho_2 = 0.2$, fixed $Z_j$)}
\centering
\fontsize{8pt}{10pt}\selectfont
\begin{tabular}{lllllllllllll}
\makecell{$\boldsymbol{\beta}^*=(\beta_1^*,\beta_2^*)$\\(-1.70, -0.70)} & \makecell{$\boldsymbol{x}^{\text{opt}}$\\(2.94, 0)} & $n_j^{(1)}$ & $n_j^{(2)}$ & $J$ & & 
\multicolumn{3}{c}{Stage 1} & & \multicolumn{3}{c}{Stage 2/LAGO optimized} \\
 &  &  &  &  & & \begin{tabular}[c]{@{}l@{}}Bias of\\ $\hat{x}_1^{\text{opt},(1)}$\\ \end{tabular} & \begin{tabular}[c]{@{}l@{}}Bias of\\ $\hat{x}_2^{\text{opt},(1)}$\\ \end{tabular} & \begin{tabular}[c]{@{}l@{}}rMSE\\ \end{tabular} & & \begin{tabular}[c]{@{}l@{}}Bias of\\ $\hat{x}_1^{\text{opt},(2,n^{(1)})}$\\ \end{tabular} & \begin{tabular}[c]{@{}l@{}}Bias of\\ $\hat{x}_2^{\text{opt},(2,n^{(1)})}$\\ \end{tabular} & \begin{tabular}[c]{@{}l@{}}rMSE\\ \end{tabular} \\ \hline
\multicolumn{13}{l}{Scenario 1 ($J^{(1)} = J^{(2)} = J$)} \\
 &  & 50 & 100 & 6 & & -0.399 & 0.161 & 0.473 & & -0.340 & 0.046 & 0.349 \\
 &  &  & 100 & 10 & & -0.744 & 0.145 & 0.787 & & -0.736 & 0.137 & 0.766 \\
 &  &  & 100 & 20 & & -0.081 & 0.051 & 0.112 & & -0.065 & 0.025 & 0.073 \\
 &  & 50 & 200 & 6 & & -0.396 & 0.147 & 0.462 & & -0.338 & 0.042 & 0.344 \\
 &  &  & 200 & 10 & & -0.741 & 0.143 & 0.781 & & -0.744 & 0.158 & 0.780 \\
 &  &  & 200 & 20 & & -0.079 & 0.046 & 0.106 & & -0.057 & 0.016 & 0.062 \\
 &  & 100 & 100 & 6 & & -0.352 & 0.074 & 0.376 & & -0.338 & 0.045 & 0.345 \\
 &  & & 100 & 10 & & -0.722 & 0.106 & 0.745 & & -0.737 & 0.140 & 0.767 \\
 &  & & 100 & 20 & & -0.073 & 0.037 & 0.091 & & -0.062 & 0.022 & 0.070 \\
 &  & 100 & 200 & 6 & & -0.347 & 0.067 & 0.368 & & -0.339 & 0.043 & 0.345 \\
 &  & & 200 & 10 & & -0.722 & 0.108 & 0.745 & & -0.743 & 0.156 & 0.778 \\
 &  & & 200 & 20 & & -0.074 & 0.039 & 0.091 & & -0.056 & 0.015 & 0.060 \\
\end{tabular}
\footnotesize{
    \fontsize{8pt}{10pt}\selectfont
    \raggedright{
    \medskip \\
    $n_j^{(1)}$: number of participants per center $j$ at stage 1, 
    $n_j^{(2)}$: number of participants per center $j$ at stage 2.\\
    $J$: number of centers for each stage. \\
    Bias of $\hat{x}_{1}^{\text{opt}}$, $\hat{x}_{2}^{\text{opt}}$: bias of the first and second components of the estimated optimal intervention.\\
    rMSE: root mean squared error, $\{\operatorname{mean}(\|\hat{\boldsymbol{x}}^{o p t}-\boldsymbol{x}^{o p t}\|^{2})\}^{1 / 2}$, estimated by average over simulation iterations.\\
    }
}
\label{tab:optimization_performance_fixed_highconf}
\end{table}

\begin{table}[H]
\singlespacing
\caption{Simulation study results for estimated optimal intervention, confidence set, and confidence band with a linear cost function ($\rho_1 = 0.1, \rho_2 = 0.2$, fixed $Z_j$)}
\centering
\fontsize{8pt}{10pt}\selectfont
\begin{tabular}{llllllllll}
\makecell{$\boldsymbol{\beta}^*=(\beta_1^*,\beta_2^*)$\\(-1.70, -0.70)} & \makecell{$\boldsymbol{x}^{\text{opt}}$\\(2.94, 0)} & $n_j^{(1)}$ & $n_j^{(2)}$ & $J$ & \begin{tabular}[c]{@{}l@{}}TrueOpt1\\ (Q2.5,Q97.5)\end{tabular} & \begin{tabular}[c]{@{}l@{}}TrueOpt2\\ (Q2.5,Q97.5)\end{tabular} & \begin{tabular}[c]{@{}l@{}}SetCP95\\ $\%$\end{tabular} & \begin{tabular}[c]{@{}l@{}}SetPerc\\ $\%$\end{tabular} & \begin{tabular}[c]{@{}l@{}}BandsCP95\\ $\%$\end{tabular} \\ \hline
\multicolumn{10}{l}{Scenario 1 ($J^{(1)} = J^{(2)} = J$)} \\
 &  & 50 & 100 & 6 & (-5.246, -4.725) & (-5.176, -4.906) & 94.55 & 4.20 & 96.00 \\
 &  & & 100 & 10 & (-5.167, -4.827) & (-5.097, -4.927) & 94.75 & 3.41 & 96.10 \\
 &  & & 100 & 20 & (-5.246, -4.906) & (-5.076, -4.976) & 94.65 & 2.26 & 96.60 \\
 &  & 50 & 200 & 6 & (-5.246, -4.686) & (-5.176, -5.006) & 94.30 & 3.61 & 95.85 \\
 &  &  & 200 & 10 & (-5.167, -4.827) & (-5.097, -4.997) & 94.10 & 2.90 & 96.00 \\
 &  &  & 200 & 20 & (-5.246, -4.906) & (-5.076, -4.976) & 94.95 & 1.96 & 96.15 \\
 &  & 100 & 100 & 6 & (-5.246, -4.836) & (-5.176, -4.906) & 93.75 & 3.38 & 96.10 \\
 &  &  & 100 & 10 & (-5.167, -4.927) & (-5.097, -4.927) & 95.15 & 2.74 & 96.20 \\
 &  &  & 100 & 20 & (-5.146, -4.906) & (-5.076, -4.976) & 94.70 & 1.82 & 96.05 \\
 &  & 100 & 200 & 6 & (-5.246, -4.836) & (-5.176, -5.006) & 95.35 & 2.99 & 97.00 \\
 &  &  & 200 & 10 & (-5.167, -4.927) & (-5.097, -4.997) & 95.40 & 2.42 & 96.10 \\
 &  &  & 200 & 20 & (-5.146, -4.906) & (-5.076, -4.976) & 94.85 & 1.61 & 96.40 \\
\end{tabular}
\footnotesize{
    \fontsize{8pt}{10pt}\selectfont
    \raggedright{
    \medskip \\
    $n_j^{(1)}$: number of participants per center $j$ at stage 1, 
    $n_j^{(2)}$: number of participants per center $j$ at stage 2.\\
    $J$: number of centers for each stage. \\
    TrueOpt1: mean treatment effect under the stage 2 recommended intervention.\\
    TrueOpt2: mean treatment effect under the final estimated optimal intervention based on all data. \\
    $Q$2.5 and $Q$97.5: 2.5$\%$ and 97.5$\%$ quantiles.\\ 
    SetCP95$\%$: empirical coverage percentage of confidence set for the optimal intervention. \\
    SetPerc$\%$: mean percentage of the size of the confidence set as a percent of the total intervention space.\\
    BandsCP95$\%$: empirical coverage of 95$\%$ confidence band. 
    }
}
\label{tab:confidence_metrics_fixed_highconf}
\end{table}

\subsection{Simulation results for Scenario 2 with varying \texorpdfstring{$\BoldZ_j$}{Z\_j} and \texorpdfstring{$\rho_1 = 0.05, \rho_2 = 0.07$}{rho1 = 0.05, rho2 = 0.07}}\label{simulationresults_cubic_varying_Zj_lowconf}
\FloatBarrier
\begin{table}[t]
\caption{Simulation study results for individual package component effects with a cubic cost function ($\rho_1 = 0.05, \rho_2 = 0.07$, varying $Z_j$)}
\centering
\fontsize{8pt}{10pt}\selectfont
\begin{tabular}{llllllllll}
\makecell{$\boldsymbol{\beta}^*=(\beta_{1}^*, \beta_{2}^*)$ \\ (-1.70, -0.70)} & $n_j^{(1)}$ & $n_j^{(2)}$ & $J$ &  & $\widehat{\boldsymbol{\beta}}_{1}$ &  &  & $\widehat{\boldsymbol{\beta}}_{2}$ &  \\
  &  &  &  & $\%$RelBias & \begin{tabular}[c]{@{}l@{}}$\frac{SE}{EMP.SD}$ \\ ($\times 100$)\end{tabular} & CP95 & $\%$RelBias & \begin{tabular}[c]{@{}l@{}}$\frac{SE}{EMP.SD}$\\ ($\times 100$)\end{tabular} & CP95 \\ \hline
\multicolumn{4}{l}{Scenario 1 ($J^{(1)} = J^{(2)} = J$)} &  &  &  &  &  &  \\
 & 50 & 100 & 6 & -0.025 & 96.01 & 94.95 & 0.304 & 100.78 & 95.35 \\
 &  &  & 10 & -0.038 & 98.40 & 95.10 & 0.285 & 101.27 & 95.55 \\
 &  &  & 20 & 0.009 & 99.28 & 95.40 & 0.016 & 98.83 & 94.45 \\
 &  & 200 & 6 & -0.021 & 96.74 & 94.60 & 0.276 & 105.31 & 95.60 \\
 &  &  & 10 & -0.000 & 95.38 & 94.45 & -0.084 & 99.66 & 95.20 \\
 &  &  & 20 & -0.011 & 98.77 & 94.80 & 0.093 & 98.79 & 94.90 \\
 & 100 & 100 & 6 & -0.017 & 96.56 & 94.40 & 0.029 & 98.41 & 94.50 \\
 &  &  & 10 & -0.018 & 101.90 & 95.40 & 0.026 & 97.76 & 94.40 \\
 &  &  & 20 & -0.026 & 96.91 & 94.05 & 0.089 & 100.71 & 94.90 \\
 &  & 200 & 6 & 0.009 & 99.67 & 95.70 & 0.150 & 101.75 & 95.50 \\
 &  &  & 10 & -0.003 & 99.05 & 94.75 & -0.051 & 100.28 & 95.15 \\
 &  &  & 20 & -0.002 & 95.52 & 93.75 & 0.027 & 97.92 & 94.65 \\
\end{tabular}
\footnotesize{
    \fontsize{8pt}{10pt}\selectfont
    \raggedright{
    \medskip \\
    $n_j^{(1)}$: number of participants per center $j$ at stage 1, 
    $n_j^{(2)}$: number of participants per center $j$ at stage 2.\\
    $J$: number of centers for each stage. \\
    \%RelBias: percent relative bias $100(\hat{\beta}-\beta^\star)/\beta^\star$.\\
    SE: mean estimated standard error, 
    EMP.SD: empirical standard deviation.\\
    CP95: empirical coverage rate of 95\% confidence intervals.\\
    }
}
\label{tab:intervention_effects_cubic_lowconf}
\end{table}
\FloatBarrier
\begin{table}[H]
\singlespacing
\caption{Simulation study results for estimated optimal intervention with a cubic cost function ($\rho_1 = 0.05, \rho_2 = 0.07$, varying $Z_j$)}
\centering
\fontsize{8pt}{10pt}\selectfont
\begin{tabular}{lllllllllllll}
\makecell{$\boldsymbol{\beta}^*=(\beta_1^*,\beta_2^*)$\\(-1.70, -0.70)} & \makecell{$\boldsymbol{x}^{\text{opt}}$\\(2.94, 0.01)} & $n_j^{(1)}$ & $n_j^{(2)}$ & $J$ & & 
\multicolumn{3}{c}{Stage 1} & & \multicolumn{3}{c}{Stage 2/LAGO optimized} \\
 &  &  &  &  & & \begin{tabular}[c]{@{}l@{}}Bias of\\ $\hat{x}_1^{\text{opt},(1)}$\\ \end{tabular} & \begin{tabular}[c]{@{}l@{}}Bias of\\ $\hat{x}_2^{\text{opt},(1)}$\\ \end{tabular} & \begin{tabular}[c]{@{}l@{}}rMSE\\ \end{tabular} & & \begin{tabular}[c]{@{}l@{}}Bias of\\ $\hat{x}_1^{\text{opt},(2,n^{(1)})}$\\ \end{tabular} & \begin{tabular}[c]{@{}l@{}}Bias of\\ $\hat{x}_2^{\text{opt},(2,n^{(1)})}$\\ \end{tabular} & \begin{tabular}[c]{@{}l@{}}rMSE\\ \end{tabular} \\ \hline
\multicolumn{13}{l}{Scenario 1 ($J^{(1)} = J^{(2)} = J$)} \\
 &  & 50 & 100 & 6 & & -0.065 & 0.193 & 0.616 & & -0.026 & 0.114 & 0.529 \\
 &  &  & 100 & 10 & & -0.021 & 0.093 & 0.433 & & 0.003 & 0.043 & 0.380 \\
 &  &  & 100 & 20 & & 0.008 & 0.044 & 0.286 & & 0.014 & 0.034 & 0.273 \\
 &  & 50 & 200 & 6 & & -0.063 & 0.182 & 0.599 & & -0.014 & 0.092 & 0.502 \\
 &  &  & 200 & 10 & & 0.005 & 0.098 & 0.445 & & 0.027 & 0.047 & 0.390 \\
 &  &  & 200 & 20 & & 0.002 & 0.040 & 0.280 & & 0.006 & 0.031 & 0.267 \\
 &  & 100 & 100 & 6 & & -0.040 & 0.146 & 0.564 & & -0.015 & 0.096 & 0.509 \\
 &  & & 100 & 10 & & -0.004 & 0.055 & 0.387 & & -0.001 & 0.047 & 0.375 \\
 &  & & 100 & 20 & & 0.005 & 0.035 & 0.273 & & 0.005 & 0.033 & 0.269 \\
 &  & 100 & 200 & 6 & & -0.029 & 0.158 & 0.582 & & 0.003 & 0.089 & 0.510 \\
 &  & & 200 & 10 & & -0.005 & 0.063 & 0.394 & & -0.001 & 0.051 & 0.380 \\
 &  & & 200 & 20 & & 0.013 & 0.033 & 0.270 & & 0.011 & 0.032 & 0.268 \\
\end{tabular}
\footnotesize{
    \fontsize{8pt}{10pt}\selectfont
    \raggedright{
    \medskip \\
    $n_j^{(1)}$: number of participants per center $j$ at stage 1, 
    $n_j^{(2)}$: number of participants per center $j$ at stage 2.\\
    $J$: number of centers for each stage. \\
    Bias of $\hat{x}_{1}^{\text{opt}}$, $\hat{x}_{2}^{\text{opt}}$: bias of the first and second components of the estimated optimal intervention.\\
    rMSE: root mean squared error, $\{\operatorname{mean}(\|\hat{\boldsymbol{x}}^{o p t}-\boldsymbol{x}^{o p t}\|^{2})\}^{1 / 2}$, estimated by average over simulation iterations.\\
    }
}
\label{tab:optimization_performance_cubic_lowconf}
\end{table}

\begin{table}[H]
\singlespacing 
\caption{Simulation study results for estimated optimal intervention, confidence set, and confidence band with a cubic cost function ($\rho_1 = 0.05, \rho_2 = 0.07$, varying $Z_j$)}
\centering
\fontsize{8pt}{10pt}\selectfont
\begin{tabular}{llllllllll}
\makecell{$\boldsymbol{\beta}^*=(\beta_1^*,\beta_2^*)$\\(-1.70, -0.70)} & \makecell{$\boldsymbol{x}^{\text{opt}}$\\(2.94, 0.01)} & $n_j^{(1)}$ & $n_j^{(2)}$ & $J$ & \begin{tabular}[c]{@{}l@{}}TrueOpt1\\ (Q2.5,Q97.5)\end{tabular} & \begin{tabular}[c]{@{}l@{}}TrueOpt2\\ (Q2.5,Q97.5)\end{tabular} & \begin{tabular}[c]{@{}l@{}}SetCP95\\ $\%$\end{tabular} & \begin{tabular}[c]{@{}l@{}}SetPerc\\ $\%$\end{tabular} & \begin{tabular}[c]{@{}l@{}}BandsCP95\\ $\%$\end{tabular} \\ \hline
\multicolumn{10}{l}{Scenario 1 ($J^{(1)} = J^{(2)} = J$)} \\
 &  & 50 & 100 & 6 & (-5.399, -4.671) & (-5.173, -4.906) & 95.35 & 4.08 & 96.15 \\
 &  & & 100 & 10 & (-5.329, -4.767) & (-5.144, -4.941) & 95.10 & 3.17 & 96.00 \\
 &  & & 100 & 20 & (-5.243, -4.861) & (-5.129, -4.966) & 95.30 & 2.25 & 96.65 \\
 &  & 50 & 200 & 6 & (-5.388, -4.680) & (-5.146, -4.931) & 94.90 & 3.51 & 96.80 \\
 &  &  & 200 & 10 & (-5.316, -4.776) & (-5.136, -4.959) & 94.40 & 2.73 & 96.40 \\
 &  &  & 200 & 20 & (-5.252, -4.857) & (-5.112, -4.977) & 94.70 & 1.95 & 96.80 \\
 &  & 100 & 100 & 6 & (-5.296, -4.786) & (-5.166, -4.920) & 94.85 & 3.28 & 96.40 \\
 &  &  & 100 & 10 & (-5.239, -4.861) & (-5.139, -4.942) & 95.50 & 2.54 & 96.25 \\
 &  &  & 100 & 20 & (-5.180, -4.907) & (-5.124, -4.967) & 94.45 & 1.81 & 96.45 \\
 &  & 100 & 200 & 6 & (-5.286, -4.798) & (-5.141, -4.946) & 95.80 & 2.89 & 96.50 \\
 &  &  & 200 & 10 & (-5.242, -4.859) & (-5.122, -4.964) & 95.05 & 2.25 & 95.95 \\
 &  &  & 200 & 20 & (-5.198, -4.904) & (-5.115, -4.976) & 93.50 & 1.60 & 96.25 \\
\end{tabular}
\footnotesize{
    \fontsize{8pt}{10pt}\selectfont
    \raggedright{
    \medskip \\
    $n_j^{(1)}$: number of participants per center $j$ at stage 1, 
    $n_j^{(2)}$: number of participants per center $j$ at stage 2.\\
    $J$: number of centers for each stage. \\
    TrueOpt1: mean treatment effect under the stage 2 recommended intervention.\\
    TrueOpt2: mean treatment effect under the final estimated optimal intervention based on all data. \\
    $Q$2.5 and $Q$97.5: 2.5$\%$ and 97.5$\%$ quantiles.\\ 
    SetCP95$\%$: empirical coverage percentage of confidence set for the optimal intervention. \\
    SetPerc$\%$: mean percentage of the size of the confidence set as a percent of the total intervention space.\\
    BandsCP95$\%$: empirical coverage of 95$\%$ confidence band. 
    }
}
\label{tab:confidence_metrics_cubic_lowconf}
\end{table}

\subsection{Simulation results for Scenario 2 with varying \texorpdfstring{$\BoldZ_j$}{Z\_j} and \texorpdfstring{$\rho_1 = 0.1, \rho_2 = 0.2$}{rho1 = 0.1, rho2 = 0.2}}\label{simulationresults_cubic_varying_Zj_highconf}
\FloatBarrier
\begin{table}[t]
\caption{Simulation study results for individual package component effects with a cubic cost function ($\rho_1 = 0.1, \rho_2 = 0.2$, varying $Z_j$)}
\centering
\fontsize{8pt}{10pt}\selectfont
\begin{tabular}{llllllllll}
\makecell{$\boldsymbol{\beta}^*=(\beta_{1}^*, \beta_{2}^*)$ \\ (-1.70, -0.70)} & $n_j^{(1)}$ & $n_j^{(2)}$ & $J$ &  & $\widehat{\boldsymbol{\beta}}_{1}$ &  &  & $\widehat{\boldsymbol{\beta}}_{2}$ &  \\
  &  &  &  & $\%$RelBias & \begin{tabular}[c]{@{}l@{}}$\frac{SE}{EMP.SD}$ \\ ($\times 100$)\end{tabular} & CP95 & $\%$RelBias & \begin{tabular}[c]{@{}l@{}}$\frac{SE}{EMP.SD}$\\ ($\times 100$)\end{tabular} & CP95 \\ \hline
\multicolumn{4}{l}{Scenario 1 ($J^{(1)} = J^{(2)} = J$)} &  &  &  &  &  &  \\
 & 50 & 100 & 6 & -0.025 & 96.32 & 94.95 & 0.310 & 100.99 & 95.10 \\
 &  &  & 10 & -0.038 & 98.65 & 94.95 & 0.289 & 101.28 & 95.45 \\
 &  &  & 20 & 0.008 & 99.43 & 95.15 & 0.007 & 99.03 & 94.65 \\
 &  & 200 & 6 & -0.023 & 96.42 & 94.40 & 0.294 & 105.18 & 95.75 \\
 &  &  & 10 & 0.001 & 95.56 & 94.30 & -0.084 & 99.68 & 95.20 \\
 &  &  & 20 & -0.012 & 98.81 & 94.95 & 0.087 & 98.74 & 95.15 \\
 & 100 & 100 & 6 & -0.018 & 96.60 & 94.35 & 0.049 & 98.60 & 94.95 \\
 &  &  & 10 & -0.018 & 101.85 & 95.40 & 0.030 & 97.73 & 94.55 \\
 &  &  & 20 & -0.027 & 97.12 & 94.05 & 0.095 & 101.00 & 95.20 \\
 &  & 200 & 6 & 0.011 & 99.59 & 95.80 & 0.146 & 101.70 & 95.45 \\
 &  &  & 10 & -0.003 & 99.01 & 94.65 & -0.050 & 100.37 & 94.70 \\
 &  &  & 20 & -0.002 & 95.50 & 93.55 & 0.030 & 97.97 & 94.50 \\
\end{tabular}
\footnotesize{
    \fontsize{8pt}{10pt}\selectfont
    \raggedright{
    \medskip \\
    $n_j^{(1)}$: number of participants per center $j$ at stage 1, 
    $n_j^{(2)}$: number of participants per center $j$ at stage 2.\\
    $J$: number of centers for each stage. \\
    \%RelBias: percent relative bias $100(\hat{\beta}-\beta^\star)/\beta^\star$.\\
    SE: mean estimated standard error, 
    EMP.SD: empirical standard deviation.\\
    CP95: empirical coverage rate of 95\% confidence intervals.\\
    }
}
\label{tab:intervention_effects_cubic_highconf}
\end{table}
\FloatBarrier
\begin{table}[H]
\singlespacing
\caption{Simulation study results for estimated optimal intervention with a cubic cost function ($\rho_1 = 0.1, \rho_2 = 0.2$, varying $Z_j$)}
\centering
\fontsize{8pt}{10pt}\selectfont
\begin{tabular}{lllllllllllll}
\makecell{$\boldsymbol{\beta}^*=(\beta_1^*,\beta_2^*)$\\(-1.70, -0.70)} & \makecell{$\boldsymbol{x}^{\text{opt}}$\\(2.94, 0.01)} & $n_j^{(1)}$ & $n_j^{(2)}$ & $J$ & & 
\multicolumn{3}{c}{Stage 1} & & \multicolumn{3}{c}{Stage 2/LAGO optimized} \\
 &  &  &  &  & & \begin{tabular}[c]{@{}l@{}}Bias of\\ $\hat{x}_1^{\text{opt},(1)}$\\ \end{tabular} & \begin{tabular}[c]{@{}l@{}}Bias of\\ $\hat{x}_2^{\text{opt},(1)}$\\ \end{tabular} & \begin{tabular}[c]{@{}l@{}}rMSE\\ \end{tabular} & & \begin{tabular}[c]{@{}l@{}}Bias of\\ $\hat{x}_1^{\text{opt},(2,n^{(1)})}$\\ \end{tabular} & \begin{tabular}[c]{@{}l@{}}Bias of\\ $\hat{x}_2^{\text{opt},(2,n^{(1)})}$\\ \end{tabular} & \begin{tabular}[c]{@{}l@{}}rMSE\\ \end{tabular} \\ \hline
\multicolumn{13}{l}{Scenario 1 ($J^{(1)} = J^{(2)} = J$)} \\
 &  & 50 & 100 & 6 & & -0.067 & 0.198 & 0.621 & & -0.027 & 0.114 & 0.530 \\
 &  &  & 100 & 10 & & -0.022 & 0.096 & 0.437 & & 0.003 & 0.043 & 0.380 \\
 &  &  & 100 & 20 & & 0.010 & 0.040 & 0.283 & & 0.014 & 0.034 & 0.273 \\
 &  & 50 & 200 & 6 & & -0.067 & 0.190 & 0.606 & & -0.013 & 0.090 & 0.500 \\
 &  &  & 200 & 10 & & 0.005 & 0.096 & 0.442 & & 0.027 & 0.047 & 0.390 \\
 &  &  & 200 & 20 & & 0.002 & 0.041 & 0.282 & & 0.006 & 0.031 & 0.267 \\
 &  & 100 & 100 & 6 & & -0.040 & 0.144 & 0.563 & & -0.015 & 0.095 & 0.509 \\
 &  & & 100 & 10 & & -0.005 & 0.056 & 0.388 & & -0.001 & 0.047 & 0.375 \\
 &  & & 100 & 20 & & 0.004 & 0.037 & 0.273 & & 0.005 & 0.033 & 0.269 \\
 &  & 100 & 200 & 6 & & -0.029 & 0.159 & 0.581 & & 0.002 & 0.090 & 0.512 \\
 &  & & 200 & 10 & & -0.007 & 0.067 & 0.397 & & -0.000 & 0.049 & 0.379 \\
 &  & & 200 & 20 & & 0.013 & 0.033 & 0.270 & & 0.011 & 0.032 & 0.268 \\
\end{tabular}
\footnotesize{
    \fontsize{8pt}{10pt}\selectfont
    \raggedright{
    \medskip \\
    $n_j^{(1)}$: number of participants per center $j$ at stage 1, 
    $n_j^{(2)}$: number of participants per center $j$ at stage 2.\\
    $J$: number of centers for each stage. \\
    Bias of $\hat{x}_{1}^{\text{opt}}$, $\hat{x}_{2}^{\text{opt}}$: bias of the first and second components of the estimated optimal intervention.\\
    rMSE: root mean squared error, $\{\operatorname{mean}(\|\hat{\boldsymbol{x}}^{o p t}-\boldsymbol{x}^{o p t}\|^{2})\}^{1 / 2}$, estimated by average over simulation iterations.\\
    }
}
\label{tab:optimization_performance_cubic_highconf}
\end{table}

\begin{table}[H]
\singlespacing
\caption{Simulation study results for estimated optimal intervention, confidence set, and confidence band with a cubic cost function ($\rho_1 = 0.1, \rho_2 = 0.2$, varying $Z_j$)}
\centering
\fontsize{8pt}{10pt}\selectfont
\begin{tabular}{llllllllll}
\makecell{$\boldsymbol{\beta}^*=(\beta_1^*,\beta_2^*)$\\(-1.70, -0.70)} & \makecell{$\boldsymbol{x}^{\text{opt}}$\\(2.94, 0.01)} & $n_j^{(1)}$ & $n_j^{(2)}$ & $J$ & \begin{tabular}[c]{@{}l@{}}TrueOpt1\\ (Q2.5,Q97.5)\end{tabular} & \begin{tabular}[c]{@{}l@{}}TrueOpt2\\ (Q2.5,Q97.5)\end{tabular} & \begin{tabular}[c]{@{}l@{}}SetCP95\\ $\%$\end{tabular} & \begin{tabular}[c]{@{}l@{}}SetPerc\\ $\%$\end{tabular} & \begin{tabular}[c]{@{}l@{}}BandsCP95\\ $\%$\end{tabular} \\ \hline
\multicolumn{10}{l}{Scenario 1 ($J^{(1)} = J^{(2)} = J$)} \\
 &  & 50 & 100 & 6 & (-5.403, -4.671) & (-5.173, -4.906) & 94.90 & 4.07 & 95.90 \\
 &  & & 100 & 10 & (-5.332, -4.763) & (-5.144, -4.941) & 95.25 & 3.17 & 96.10 \\
 &  & & 100 & 20 & (-5.244, -4.863) & (-5.129, -4.966) & 95.05 & 2.24 & 96.60 \\
 &  & 50 & 200 & 6 & (-5.392, -4.673) & (-5.144, -4.932) & 94.95 & 3.50 & 96.80 \\
 &  &  & 200 & 10 & (-5.315, -4.775) & (-5.136, -4.958) & 94.40 & 2.72 & 96.20 \\
 &  &  & 200 & 20 & (-5.252, -4.856) & (-5.112, -4.977) & 95.05 & 1.94 & 96.80 \\
 &  & 100 & 100 & 6 & (-5.292, -4.783) & (-5.166, -4.920) & 94.70 & 3.27 & 96.35 \\
 &  &  & 100 & 10 & (-5.240, -4.860) & (-5.139, -4.941) & 95.40 & 2.54 & 96.50 \\
 &  &  & 100 & 20 & (-5.182, -4.904) & (-5.123, -4.967) & 94.50 & 1.80 & 96.30 \\
 &  & 100 & 200 & 6 & (-5.288, -4.798) & (-5.144, -4.943) & 95.70 & 2.89 & 96.50 \\
 &  &  & 200 & 10 & (-5.242, -4.859) & (-5.122, -4.964) & 95.05 & 2.25 & 95.95 \\
 &  &  & 200 & 20 & (-5.198, -4.902) & (-5.115, -4.976) & 93.50 & 1.60 & 96.25 \\
\end{tabular}
\footnotesize{
    \fontsize{8pt}{10pt}\selectfont
    \raggedright{
    \medskip \\
    $n_j^{(1)}$: number of participants per center $j$ at stage 1, 
    $n_j^{(2)}$: number of participants per center $j$ at stage 2.\\
    $J$: number of centers for each stage. \\
    TrueOpt1: mean treatment effect under the stage 2 recommended intervention.\\
    TrueOpt2: mean treatment effect under the final estimated optimal intervention based on all data. \\
    $Q$2.5 and $Q$97.5: 2.5$\%$ and 97.5$\%$ quantiles.\\ 
    SetCP95$\%$: empirical coverage percentage of confidence set for the optimal intervention. \\
    SetPerc$\%$: mean percentage of the size of the confidence set as a percent of the total intervention space.\\
    BandsCP95$\%$: empirical coverage of 95$\%$ confidence band. 
    }
}
\label{tab:confidence_metrics_cubic_highconf}
\end{table}

\subsection{Simulation results for Scenario 2 with fixed \texorpdfstring{$\BoldZ_j$}{Z\_j} and \texorpdfstring{$\rho_1 = 0.05, \rho_2 = 0.07$}{rho1 = 0.05, rho2 = 0.07}}\label{simulationresults_cubic_fixed_Zj_lowconf}
\FloatBarrier
\begin{table}[t]
\caption{Simulation study results for individual package component effects with a cubic cost function ($\rho_1 = 0.05, \rho_2 = 0.07$, fixed $Z_j$)}
\centering
\fontsize{8pt}{10pt}\selectfont
\begin{tabular}{llllllllll}
\makecell{$\boldsymbol{\beta}^*=(\beta_{1}^*, \beta_{2}^*)$ \\ (-1.70, -0.70)} & $n_j^{(1)}$ & $n_j^{(2)}$ & $J$ &  & $\widehat{\boldsymbol{\beta}}_{1}$ &  &  & $\widehat{\boldsymbol{\beta}}_{2}$ &  \\
  &  &  &  & $\%$RelBias & \begin{tabular}[c]{@{}l@{}}$\frac{SE}{EMP.SD}$ \\ ($\times 100$)\end{tabular} & CP95 & $\%$RelBias & \begin{tabular}[c]{@{}l@{}}$\frac{SE}{EMP.SD}$\\ ($\times 100$)\end{tabular} & CP95 \\ \hline
\multicolumn{4}{l}{Scenario 1 ($J^{(1)} = J^{(2)} = J$)} &  &  &  &  &  &  \\
 & 50 & 100 & 6 & -0.042 & 98.11 & 94.70 & 0.405 & 102.17 & 94.45 \\
 &  &  & 10 & -0.007 & 96.28 & 94.70 & 0.438 & 104.47 & 95.10 \\
 &  &  & 20 & 0.017 & 97.17 & 94.85 & 0.044 & 99.40 & 94.95 \\
 &  & 200 & 6 & -0.006 & 97.02 & 94.70 & 0.277 & 105.43 & 96.00 \\
 &  &  & 10 & 0.010 & 95.03 & 94.75 & 0.262 & 100.21 & 95.00 \\
 &  &  & 20 & -0.012 & 99.04 & 94.80 & 0.185 & 98.08 & 94.25 \\
 & 100 & 100 & 6 & -0.011 & 95.81 & 94.55 & 0.000 & 97.72 & 94.45 \\
 &  &  & 10 & -0.029 & 98.58 & 94.90 & 0.175 & 101.17 & 94.60 \\
 &  &  & 20 & 0.008 & 98.26 & 94.70 & -0.024 & 97.69 & 95.10 \\
 &  & 200 & 6 & 0.023 & 99.75 & 95.80 & 0.116 & 101.47 & 94.80 \\
 &  &  & 10 & 0.013 & 100.48 & 96.05 & 0.082 & 102.00 & 94.95 \\
 &  &  & 20 & 0.003 & 99.54 & 95.20 & -0.016 & 97.51 & 94.20 \\
\end{tabular}
\footnotesize{
    \fontsize{8pt}{10pt}\selectfont
    \raggedright{
    \medskip \\
    $n_j^{(1)}$: number of participants per center $j$ at stage 1, 
    $n_j^{(2)}$: number of participants per center $j$ at stage 2.\\
    $J$: number of centers for each stage. \\
    \%RelBias: percent relative bias $100(\hat{\beta}-\beta^\star)/\beta^\star$.\\
    SE: mean estimated standard error, 
    EMP.SD: empirical standard deviation.\\
    CP95: empirical coverage rate of 95\% confidence intervals.\\
    }
}
\label{tab:intervention_effects_cubic_fixed_lowconf}
\end{table}
\FloatBarrier
\begin{table}[H]
\singlespacing
\caption{Simulation study results for estimated optimal intervention with a cubic cost function ($\rho_1 = 0.05, \rho_2 = 0.07$, fixed $Z_j$)}
\centering
\fontsize{8pt}{10pt}\selectfont
\begin{tabular}{lllllllllllll}
\makecell{$\boldsymbol{\beta}^*=(\beta_1^*,\beta_2^*)$\\(-1.70, -0.70)} & \makecell{$\boldsymbol{x}^{\text{opt}}$\\(2.94, 0.01)} & $n_j^{(1)}$ & $n_j^{(2)}$ & $J$ & & 
\multicolumn{3}{c}{Stage 1} & & \multicolumn{3}{c}{Stage 2/LAGO optimized} \\
 &  &  &  &  & & \begin{tabular}[c]{@{}l@{}}Bias of\\ $\hat{x}_1^{\text{opt},(1)}$\\ \end{tabular} & \begin{tabular}[c]{@{}l@{}}Bias of\\ $\hat{x}_2^{\text{opt},(1)}$\\ \end{tabular} & \begin{tabular}[c]{@{}l@{}}rMSE\\ \end{tabular} & & \begin{tabular}[c]{@{}l@{}}Bias of\\ $\hat{x}_1^{\text{opt},(2,n^{(1)})}$\\ \end{tabular} & \begin{tabular}[c]{@{}l@{}}Bias of\\ $\hat{x}_2^{\text{opt},(2,n^{(1)})}$\\ \end{tabular} & \begin{tabular}[c]{@{}l@{}}rMSE\\ \end{tabular} \\ \hline
\multicolumn{13}{l}{Scenario 1 ($J^{(1)} = J^{(2)} = J$)} \\
 &  & 50 & 100 & 6 & & -0.398 & 0.152 & 0.472 & & -0.340 & 0.038 & 0.349 \\
 &  &  & 100 & 10 & & -0.823 & 0.301 & 0.964 & & -0.705 & 0.069 & 0.724 \\
 &  &  & 100 & 20 & & -0.076 & 0.034 & 0.101 & & -0.063 & 0.014 & 0.070 \\
 &  & 50 & 200 & 6 & & -0.399 & 0.146 & 0.471 & & -0.337 & 0.031 & 0.342 \\
 &  &  & 200 & 10 & & -0.833 & 0.326 & 0.987 & & -0.694 & 0.049 & 0.701 \\
 &  &  & 200 & 20 & & -0.075 & 0.031 & 0.098 & & -0.055 & 0.005 & 0.060 \\
 &  & 100 & 100 & 6 & & -0.353 & 0.068 & 0.378 & & -0.337 & 0.035 & 0.343 \\
 &  & & 100 & 10 & & -0.747 & 0.150 & 0.809 & & -0.693 & 0.046 & 0.700 \\
 &  & & 100 & 20 & & -0.071 & 0.025 & 0.087 & & -0.061 & 0.012 & 0.068 \\
 &  & 100 & 200 & 6 & & -0.348 & 0.061 & 0.370 & & -0.337 & 0.033 & 0.342 \\
 &  & & 200 & 10 & & -0.741 & 0.142 & 0.800 & & -0.687 & 0.037 & 0.691 \\
 &  & & 200 & 20 & & -0.071 & 0.027 & 0.087 & & -0.054 & 0.004 & 0.059 \\
\end{tabular}
\footnotesize{
    \fontsize{8pt}{10pt}\selectfont
    \raggedright{
    \medskip \\
    $n_j^{(1)}$: number of participants per center $j$ at stage 1, 
    $n_j^{(2)}$: number of participants per center $j$ at stage 2.\\
    $J$: number of centers for each stage. \\
    Bias of $\hat{x}_{1}^{\text{opt}}$, $\hat{x}_{2}^{\text{opt}}$: bias of the first and second components of the estimated optimal intervention.\\
    rMSE: root mean squared error, $\{\operatorname{mean}(\|\hat{\boldsymbol{x}}^{o p t}-\boldsymbol{x}^{o p t}\|^{2})\}^{1 / 2}$, estimated by average over simulation iterations.\\
    }
}
\label{tab:optimization_performance_cubic_fixed_lowconf}
\end{table}

\begin{table}[H]
\singlespacing
\caption{Simulation study results for estimated optimal intervention, confidence set, and confidence band with a cubic cost function ($\rho_1 = 0.05, \rho_2 = 0.07$, fixed $Z_j$)}
\centering
\fontsize{8pt}{10pt}\selectfont
\begin{tabular}{llllllllll}
\makecell{$\boldsymbol{\beta}^*=(\beta_1^*,\beta_2^*)$\\(-1.70, -0.70)} & \makecell{$\boldsymbol{x}^{\text{opt}}$\\(2.94, 0.01)} & $n_j^{(1)}$ & $n_j^{(2)}$ & $J$ & \begin{tabular}[c]{@{}l@{}}TrueOpt1\\ (Q2.5,Q97.5)\end{tabular} & \begin{tabular}[c]{@{}l@{}}TrueOpt2\\ (Q2.5,Q97.5)\end{tabular} & \begin{tabular}[c]{@{}l@{}}SetCP95\\ $\%$\end{tabular} & \begin{tabular}[c]{@{}l@{}}SetPerc\\ $\%$\end{tabular} & \begin{tabular}[c]{@{}l@{}}BandsCP95\\ $\%$\end{tabular} \\ \hline
\multicolumn{10}{l}{Scenario 1 ($J^{(1)} = J^{(2)} = J$)} \\
 &  & 50 & 100 & 6 & (-5.246, -4.735) & (-5.176, -4.906) & 94.40 & 4.16 & 96.25 \\
 &  & & 100 & 10 & (-5.167, -4.807) & (-5.167, -4.927) & 94.85 & 3.33 & 95.95 \\
 &  & & 100 & 20 & (-5.246, -4.906) & (-5.076, -4.976) & 94.75 & 2.26 & 96.60 \\
 &  & 50 & 200 & 6 & (-5.246, -4.716) & (-5.176, -5.006) & 94.35 & 3.58 & 95.95 \\
 &  &  & 200 & 10 & (-5.167, -4.807) & (-5.097, -4.997) & 94.35 & 2.83 & 96.25 \\
 &  &  & 200 & 20 & (-5.246, -4.906) & (-5.076, -4.976) & 95.05 & 1.95 & 96.15 \\
 &  & 100 & 100 & 6 & (-5.246, -4.836) & (-5.176, -4.906) & 94.00 & 3.35 & 95.95 \\
 &  &  & 100 & 10 & (-5.167, -4.907) & (-5.097, -4.927) & 95.00 & 2.68 & 96.00 \\
 &  &  & 100 & 20 & (-5.146, -4.906) & (-5.076, -4.976) & 94.70 & 1.81 & 96.10 \\
 &  & 100 & 200 & 6 & (-5.246, -4.836) & (-5.176, -5.006) & 95.50 & 2.96 & 96.90 \\
 &  &  & 200 & 10 & (-5.167, -4.907) & (-5.097, -4.997) & 95.50 & 2.37 & 96.20 \\
 &  &  & 200 & 20 & (-5.146, -4.906) & (-5.076, -4.976) & 94.65 & 1.61 & 96.40 \\
\end{tabular}
\footnotesize{
    \fontsize{8pt}{10pt}\selectfont
    \raggedright{
    \medskip \\
    $n_j^{(1)}$: number of participants per center $j$ at stage 1, 
    $n_j^{(2)}$: number of participants per center $j$ at stage 2.\\
    $J$: number of centers for each stage. \\
    TrueOpt1: mean treatment effect under the stage 2 recommended intervention.\\
    TrueOpt2: mean treatment effect under the final estimated optimal intervention based on all data. \\
    $Q$2.5 and $Q$97.5: 2.5$\%$ and 97.5$\%$ quantiles.\\ 
    SetCP95$\%$: empirical coverage percentage of confidence set for the optimal intervention. \\
    SetPerc$\%$: mean percentage of the size of the confidence set as a percent of the total intervention space.\\
    BandsCP95$\%$: empirical coverage of 95$\%$ confidence band. 
    }
}
\label{tab:confidence_metrics_cubic_fixed_lowconf}
\end{table}

\subsection{Simulation results for Scenario 2 with fixed \texorpdfstring{$\BoldZ_j$}{Z\_j} and \texorpdfstring{$\rho_1 = 0.1, \rho_2 = 0.2$}{rho1 = 0.1, rho2 = 0.2}}\label{simulationresults_cubic_fixed_Zj_highconf}
\FloatBarrier
\begin{table}[t]
\caption{Simulation study results for individual package component effects with a cubic cost function ($\rho_1 = 0.1, \rho_2 = 0.2$, fixed $Z_j$)}
\centering
\fontsize{8pt}{10pt}\selectfont
\begin{tabular}{llllllllll}
\makecell{$\boldsymbol{\beta}^*=(\beta_{1}^*, \beta_{2}^*)$ \\ (-1.70, -0.70)} & $n_j^{(1)}$ & $n_j^{(2)}$ & $J$ &  & $\widehat{\boldsymbol{\beta}}_{1}$ &  &  & $\widehat{\boldsymbol{\beta}}_{2}$ &  \\
  &  &  &  & $\%$RelBias & \begin{tabular}[c]{@{}l@{}}$\frac{SE}{EMP.SD}$ \\ ($\times 100$)\end{tabular} & CP95 & $\%$RelBias & \begin{tabular}[c]{@{}l@{}}$\frac{SE}{EMP.SD}$\\ ($\times 100$)\end{tabular} & CP95 \\ \hline
\multicolumn{4}{l}{Scenario 1 ($J^{(1)} = J^{(2)} = J$)} &  &  &  &  &  &  \\
 & 50 & 100 & 6 & -0.046 & 98.20 & 94.80 & 0.439 & 102.30 & 94.60 \\
 &  &  & 10 & -0.012 & 96.67 & 94.75 & 0.512 & 105.01 & 95.25 \\
 &  &  & 20 & 0.016 & 97.05 & 94.90 & 0.051 & 99.26 & 94.90 \\
 &  & 200 & 6 & -0.006 & 97.26 & 95.05 & 0.289 & 105.34 & 95.95 \\
 &  &  & 10 & 0.007 & 95.02 & 94.55 & 0.282 & 100.13 & 94.75 \\
 &  &  & 20 & -0.013 & 99.05 & 94.95 & 0.190 & 98.20 & 94.20 \\
 & 100 & 100 & 6 & -0.012 & 95.91 & 94.65 & 0.004 & 97.80 & 94.70 \\
 &  &  & 10 & -0.030 & 98.52 & 95.15 & 0.186 & 101.34 & 94.75 \\
 &  &  & 20 & 0.008 & 98.16 & 94.80 & -0.023 & 97.74 & 95.15 \\
 &  & 200 & 6 & 0.021 & 99.79 & 95.65 & 0.130 & 101.76 & 94.80 \\
 &  &  & 10 & 0.013 & 100.38 & 96.05 & 0.094 & 101.97 & 94.95 \\
 &  &  & 20 & 0.003 & 99.59 & 95.15 & -0.015 & 97.50 & 94.30 \\
\end{tabular}
\footnotesize{
    \fontsize{8pt}{10pt}\selectfont
    \raggedright{
    \medskip \\
    $n_j^{(1)}$: number of participants per center $j$ at stage 1, 
    $n_j^{(2)}$: number of participants per center $j$ at stage 2.\\
    $J$: number of centers for each stage. \\
    \%RelBias: percent relative bias $100(\hat{\beta}-\beta^\star)/\beta^\star$.\\
    SE: mean estimated standard error, 
    EMP.SD: empirical standard deviation.\\
    CP95: empirical coverage rate of 95\% confidence intervals.\\
    }
}
\label{tab:intervention_effects_cubic_fixed_highconf}
\end{table}
\FloatBarrier
\begin{table}[H]
\singlespacing
\caption{Simulation study results for estimated optimal intervention with a cubic cost function ($\rho_1 = 0.1, \rho_2 = 0.2$, fixed $Z_j$)}
\centering
\fontsize{8pt}{10pt}\selectfont
\begin{tabular}{lllllllllllll}
\makecell{$\boldsymbol{\beta}^*=(\beta_1^*,\beta_2^*)$\\(-1.70, -0.70)} & \makecell{$\boldsymbol{x}^{\text{opt}}$\\(2.94, 0.01)} & $n_j^{(1)}$ & $n_j^{(2)}$ & $J$ & & 
\multicolumn{3}{c}{Stage 1} & & \multicolumn{3}{c}{Stage 2/LAGO optimized} \\
 &  &  &  &  & & \begin{tabular}[c]{@{}l@{}}Bias of\\ $\hat{x}_1^{\text{opt},(1)}$\\ \end{tabular} & \begin{tabular}[c]{@{}l@{}}Bias of\\ $\hat{x}_2^{\text{opt},(1)}$\\ \end{tabular} & \begin{tabular}[c]{@{}l@{}}rMSE\\ \end{tabular} & & \begin{tabular}[c]{@{}l@{}}Bias of\\ $\hat{x}_1^{\text{opt},(2,n^{(1)})}$\\ \end{tabular} & \begin{tabular}[c]{@{}l@{}}Bias of\\ $\hat{x}_2^{\text{opt},(2,n^{(1)})}$\\ \end{tabular} & \begin{tabular}[c]{@{}l@{}}rMSE\\ \end{tabular} \\ \hline
\multicolumn{13}{l}{Scenario 1 ($J^{(1)} = J^{(2)} = J$)} \\
 &  & 50 & 100 & 6 & & -0.400 & 0.156 & 0.476 & & -0.340 & 0.038 & 0.349 \\
 &  &  & 100 & 10 & & -0.829 & 0.314 & 0.977 & & -0.705 & 0.070 & 0.725 \\
 &  &  & 100 & 20 & & -0.076 & 0.034 & 0.101 & & -0.063 & 0.014 & 0.071 \\
 &  & 50 & 200 & 6 & & -0.401 & 0.149 & 0.474 & & -0.337 & 0.033 & 0.343 \\
 &  &  & 200 & 10 & & -0.831 & 0.322 & 0.983 & & -0.694 & 0.050 & 0.703 \\
 &  &  & 200 & 20 & & -0.075 & 0.031 & 0.097 & & -0.055 & 0.005 & 0.060 \\
 &  & 100 & 100 & 6 & & -0.353 & 0.068 & 0.379 & & -0.338 & 0.036 & 0.345 \\
 &  & & 100 & 10 & & -0.749 & 0.156 & 0.814 & & -0.693 & 0.046 & 0.700 \\
 &  & & 100 & 20 & & -0.071 & 0.026 & 0.087 & & -0.061 & 0.012 & 0.068 \\
 &  & 100 & 200 & 6 & & -0.348 & 0.061 & 0.370 & & -0.337 & 0.033 & 0.343 \\
 &  & & 200 & 10 & & -0.742 & 0.145 & 0.803 & & -0.687 & 0.037 & 0.690 \\
 &  & & 200 & 20 & & -0.072 & 0.027 & 0.087 & & -0.055 & 0.005 & 0.060 \\
\end{tabular}
\footnotesize{
    \fontsize{8pt}{10pt}\selectfont
    \raggedright{
    \medskip \\
    $n_j^{(1)}$: number of participants per center $j$ at stage 1, 
    $n_j^{(2)}$: number of participants per center $j$ at stage 2.\\
    $J$: number of centers for each stage. \\
    Bias of $\hat{x}_{1}^{\text{opt}}$, $\hat{x}_{2}^{\text{opt}}$: bias of the first and second components of the estimated optimal intervention.\\
    rMSE: root mean squared error, $\{\operatorname{mean}(\|\hat{\boldsymbol{x}}^{o p t}-\boldsymbol{x}^{o p t}\|^{2})\}^{1 / 2}$, estimated by average over simulation iterations.\\
    }
}
\label{tab:optimization_performance_cubic_fixed_highconf}
\end{table}

\begin{table}[H]
\singlespacing
\caption{Simulation study results for estimated optimal intervention, confidence set, and confidence band with a cubic cost function ($\rho_1 = 0.1, \rho_2 = 0.2$, fixed $Z_j$)}
\centering
\fontsize{8pt}{10pt}\selectfont
\begin{tabular}{llllllllll}
\makecell{$\boldsymbol{\beta}^*=(\beta_1^*,\beta_2^*)$\\(-1.70, -0.70)} & \makecell{$\boldsymbol{x}^{\text{opt}}$\\(2.94, 0.01)} & $n_j^{(1)}$ & $n_j^{(2)}$ & $J$ & \begin{tabular}[c]{@{}l@{}}TrueOpt1\\ (Q2.5,Q97.5)\end{tabular} & \begin{tabular}[c]{@{}l@{}}TrueOpt2\\ (Q2.5,Q97.5)\end{tabular} & \begin{tabular}[c]{@{}l@{}}SetCP95\\ $\%$\end{tabular} & \begin{tabular}[c]{@{}l@{}}SetPerc\\ $\%$\end{tabular} & \begin{tabular}[c]{@{}l@{}}BandsCP95\\ $\%$\end{tabular} \\ \hline
\multicolumn{10}{l}{Scenario 1 ($J^{(1)} = J^{(2)} = J$)} \\
 &  & 50 & 100 & 6 & (-5.246, -4.716) & (-5.176, -4.906) & 94.40 & 4.20 & 96.10 \\
 &  & & 100 & 10 & (-5.167, -4.807) & (-5.167, -4.927) & 94.85 & 3.41 & 95.95 \\
 &  & & 100 & 20 & (-5.246, -4.906) & (-5.076, -4.976) & 94.70 & 2.26 & 96.65 \\
 &  & 50 & 200 & 6 & (-5.246, -4.716) & (-5.176, -5.006) & 94.30 & 3.62 & 95.95 \\
 &  &  & 200 & 10 & (-5.167, -4.807) & (-5.097, -4.997) & 94.20 & 2.91 & 96.30 \\
 &  &  & 200 & 20 & (-5.246, -4.906) & (-5.076, -4.976) & 94.95 & 1.96 & 96.15 \\
 &  & 100 & 100 & 6 & (-5.246, -4.836) & (-5.176, -4.906) & 93.95 & 3.38 & 96.15 \\
 &  &  & 100 & 10 & (-5.167, -4.907) & (-5.097, -4.927) & 95.40 & 2.74 & 96.00 \\
 &  &  & 100 & 20 & (-5.146, -4.906) & (-5.076, -4.976) & 94.70 & 1.82 & 96.05 \\
 &  & 100 & 200 & 6 & (-5.246, -4.836) & (-5.176, -5.006) & 95.40 & 2.99 & 97.05 \\
 &  &  & 200 & 10 & (-5.167, -4.907) & (-5.097, -4.997) & 95.55 & 2.42 & 96.25 \\
 &  &  & 200 & 20 & (-5.146, -4.906) & (-5.076, -4.976) & 94.85 & 1.61 & 96.35 \\
\end{tabular}
\footnotesize{
    \raggedright{
    \fontsize{8pt}{10pt}\selectfont
    \medskip \\
    $n_j^{(1)}$: number of participants per center $j$ at stage 1, 
    $n_j^{(2)}$: number of participants per center $j$ at stage 2.\\
    $J$: number of centers for each stage. \\
    TrueOpt1: mean treatment effect under the stage 2 recommended intervention.\\
    TrueOpt2: mean treatment effect under the final estimated optimal intervention based on all data. \\
    $Q$2.5 and $Q$97.5: 2.5$\%$ and 97.5$\%$ quantiles.\\ 
    SetCP95$\%$: empirical coverage percentage of confidence set for the optimal intervention. \\
    SetPerc$\%$: mean percentage of the size of the confidence set as a percent of the total intervention space.\\
    BandsCP95$\%$: empirical coverage of 95$\%$ confidence band. \\
    }
}
\label{tab:confidence_metrics_cubic_fixed_highconf}
\end{table}

\subsection{Visualizations of the cubic cost function in Section 5}\label{sim_visualizations}
Figure \ref{fig:cubic_cost} shows the total and marginal cost plots for each intervention component. The intervention comprised two components, $\boldsymbol{x} = (x_1, x_2)$ with bounds $[L_1, U_1] = [0, 4]$ and $[L_2, U_2] = [0, 3]$. The dashed horizontal lines represent the marginal linear cost function based on the linear cost function used in Scenario 1: 1.0 for the nurse counseling component and 0.5 for the home BP measurement component. The cubic cost function was specified as $C_2(\boldsymbol{x}) = 1.25x_{1} - 0.04x_{1}^{3} + 0.0055x_{1}^{4} + 0.63x_{2} - 0.09x_{2}^{3} + 0.026x_{2}^{4}$. The cubic cost function was calibrated such that its average marginal cost over the feasible intervention range approximately equals the constant marginal cost of the linear function.  

\FloatBarrier
\begin{figure}[t]
    \centering
    \includegraphics[width=1\linewidth]{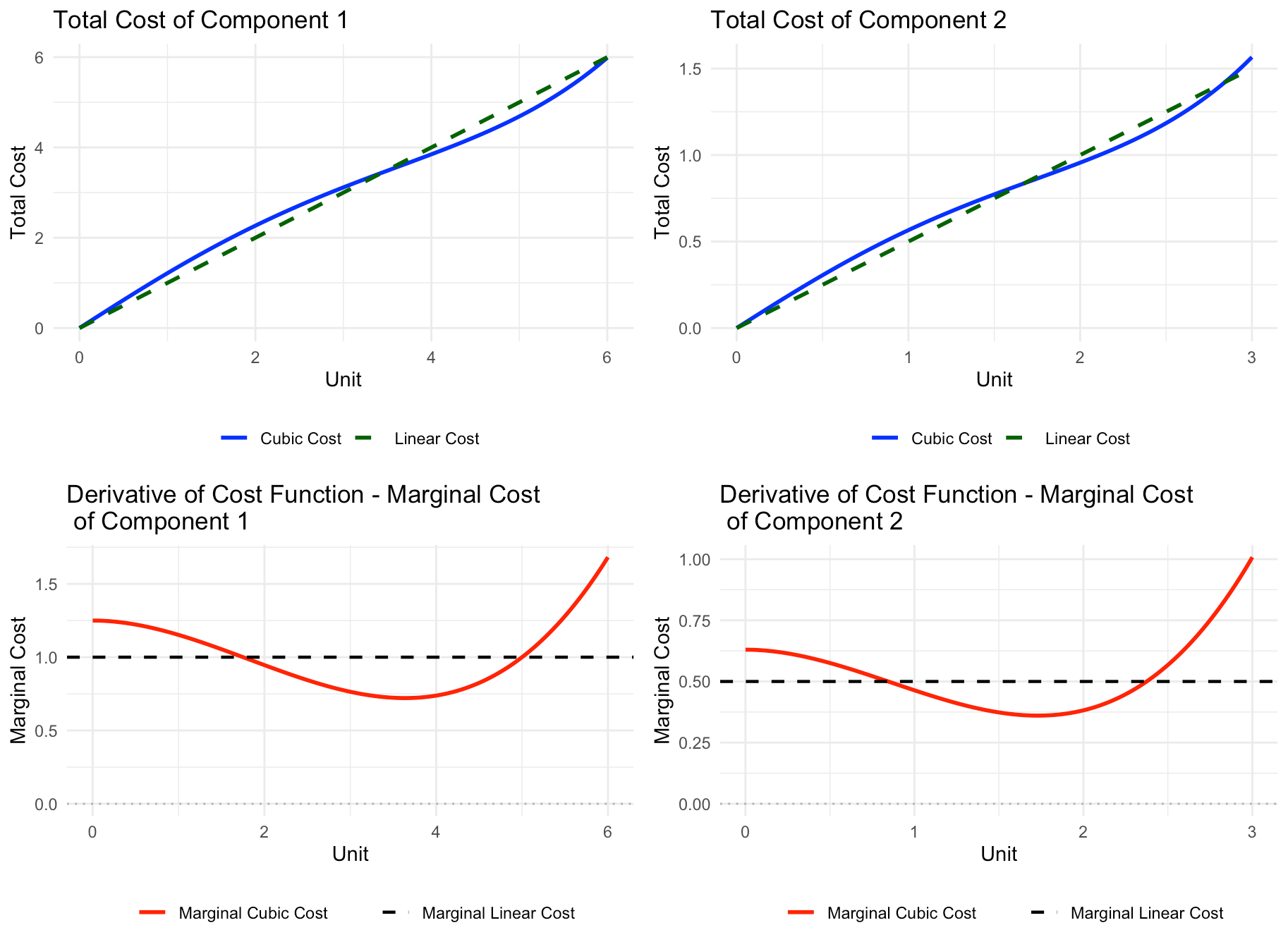}
    \caption{Total (top) and marginal (bottom) costs of the cubic cost function for each intervention component, component 1 (left) and component 2 (right).}
    \label{fig:cubic_cost}
\end{figure}
\FloatBarrier

\section{Additional results from the EXTRA-CVD study}
\subsection{Additional information on EXTRA-CVD: Variability in \texorpdfstring{$\Delta\text{SBP}$}{Delta SBP} from baseline to month 12}\label{variability_in_outcomes}
\begin{figure}[H]
    \centering
    \begin{minipage}{0.49\textwidth}
        \centering
        \includegraphics[width=\textwidth]{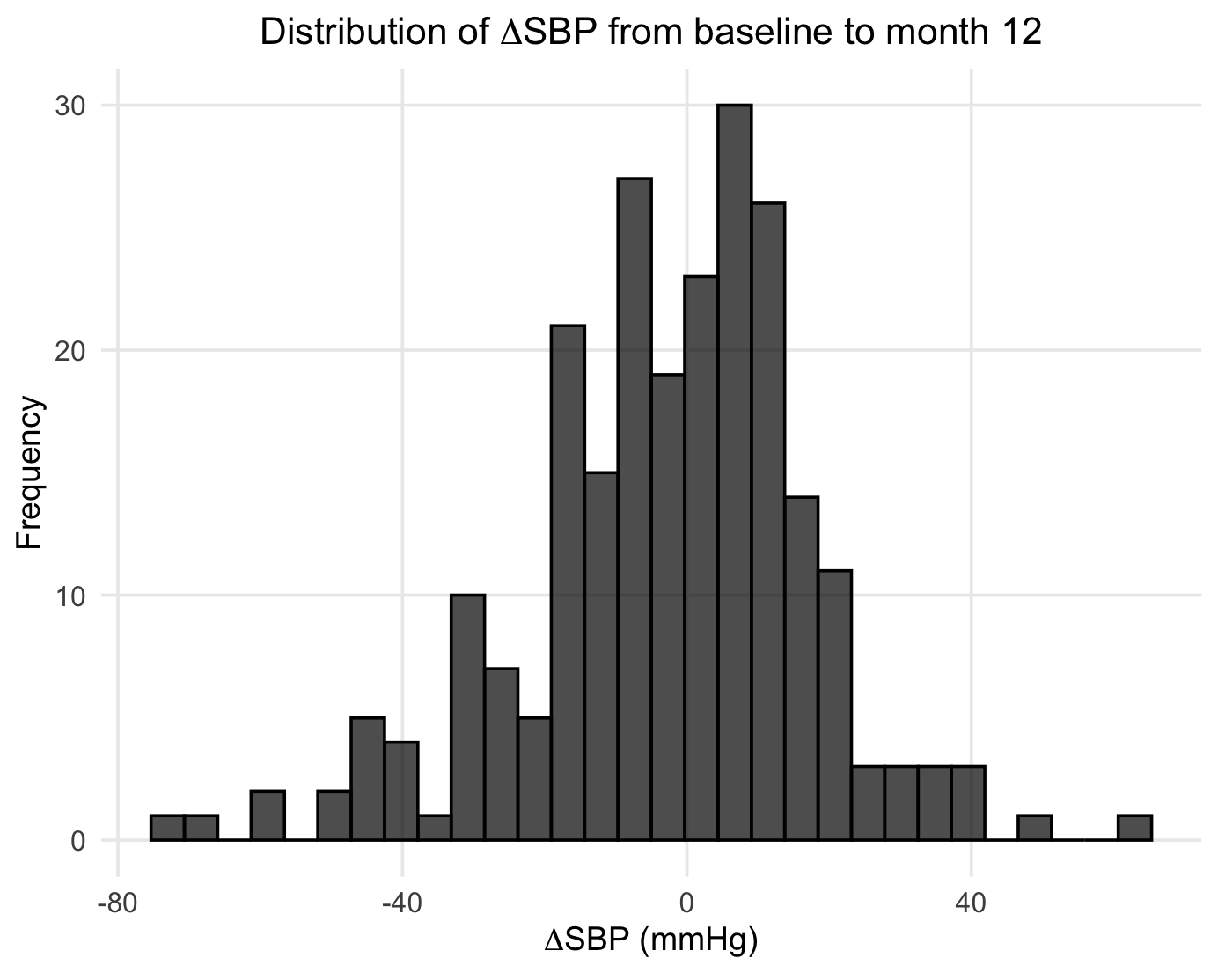}
    \end{minipage}
    \hfill
    \begin{minipage}{0.49\textwidth}
        \centering
        \includegraphics[width=\textwidth]{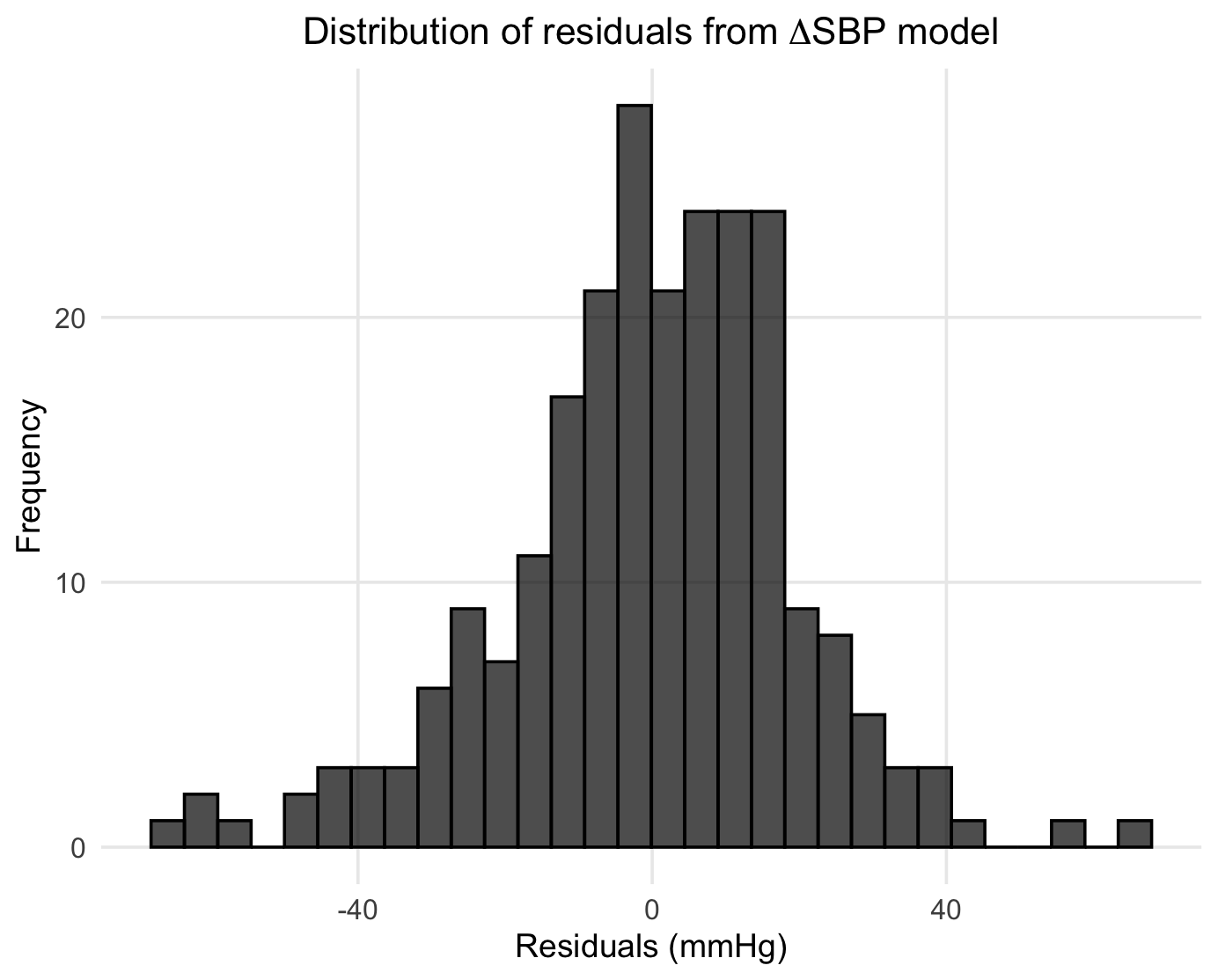}
    \end{minipage}
    \small
    \caption{\centering Distribution of $\Delta\text{SBP}$ (left) and residuals $\epsilon_{ijt}$ (right) in the month-12 EXTRA-CVD data across intervention and control groups.}
    \normalsize
    \label{fig:sbp_combined}
\end{figure}
\subsection{Additional information on EXTRA-CVD}\label{AdditionalInfoEXTRA-CVD}
\begin{table}[H]
\singlespacing
\caption{EXTRA-CVD: Average $\Delta$SBP by center and study arm by months 4, 8, and 12}
\fontsize{8pt}{10pt}\selectfont
\label{tab:bp-outcomes}
\centering
\begin{tabular}{llrrrr}
Center& Study Arm & SBP Baseline & \multicolumn{3}{c}{$\Delta$SBP } \\
& & & Month 4 & Month 8 & Month 12 \\
\hline
\\
\multirow{2}{*}{A} & Control & 134.72 & 1.46 (-5.95, 8.88) & -2.87 (-8.11, 2.37)& -3.15 (-9.66, 3.36) \\
                   & Intervention & 137.59 & -6.43 (-11.63, -1.23) & -5.91 (-11.04, -0.77)& -5.61 (-11.35, 0.14)\\
                   \hline
                   \\
\multirow{2}{*}{B} & Control & 133.29 & 0.64 (-5.10, 6.39) & 0.93 (-4.21, 6.07)& 0.39 (-4.80, 5.58) \\
                   & Intervention & 138.02 & -4.86 (-10.65, 0.92) & -8.35 (-16.10, -0.60) & -8.25 (-16.33, -0.17) \\
                   \hline
                   \\
\multirow{2}{*}{C} & Control & 133.66 & 6.11 (-0.25, 12.47)& -1.37 (-6.85, 4.12) & 3.67 (-2.54, 9.88)\\
                   & Intervention & 131.87 & -2.95 (-6.82, 0.92) & -4.61 (-9.58, 0.36)& -4.28 (-9.95, 1.38)\\
\end{tabular}
\vspace*{0.3cm} 
\\
\raggedright\footnotesize\fontsize{8pt}{10pt}\selectfont Negative $\Delta$SBP values indicate a reduction in blood pressure from baseline.
\end{table}

\subsection{Visualizations of the cubic cost function used to analyze in the EXTRA-CVD study}\label{visualizations}

Figure \ref{fig:cubic_cost} shows the total and marginal cost plots for each intervention component. The intervention consisted of two components, $\boldsymbol{x} = (x_1, x_2)$ with bounds $[L_1, U_1] = [0, 6.5]$ and $[L_2, U_2] = [0, 3]$, as shown in Table 2 in Section 6 for the EXTRA-CVD study. The dashed horizontal lines are the marginal linear cost function from the EXTRA-CVD analysis (Table 2 in Section 6): 1.0 for the nurse counseling component and 0.5 for the home BP measurement component. The cubic cost function was specified as $C_2(\boldsymbol{x}_t) = 1.25x_{1t} - 0.043x_{1t}^{3} + 0.0055x_{1t}^4 + 0.63x_{2t} - 0.09x_{2t}^3 + 0.026x_{2t}^4$, with $t = 4, 8, 12$ for each visit. The cubic cost function was calibrated such that its average marginal cost over the feasible intervention range approximately equals the constant marginal cost of the linear function.  
\begin{figure}[H]
    \centering
    \includegraphics[width=1\linewidth]{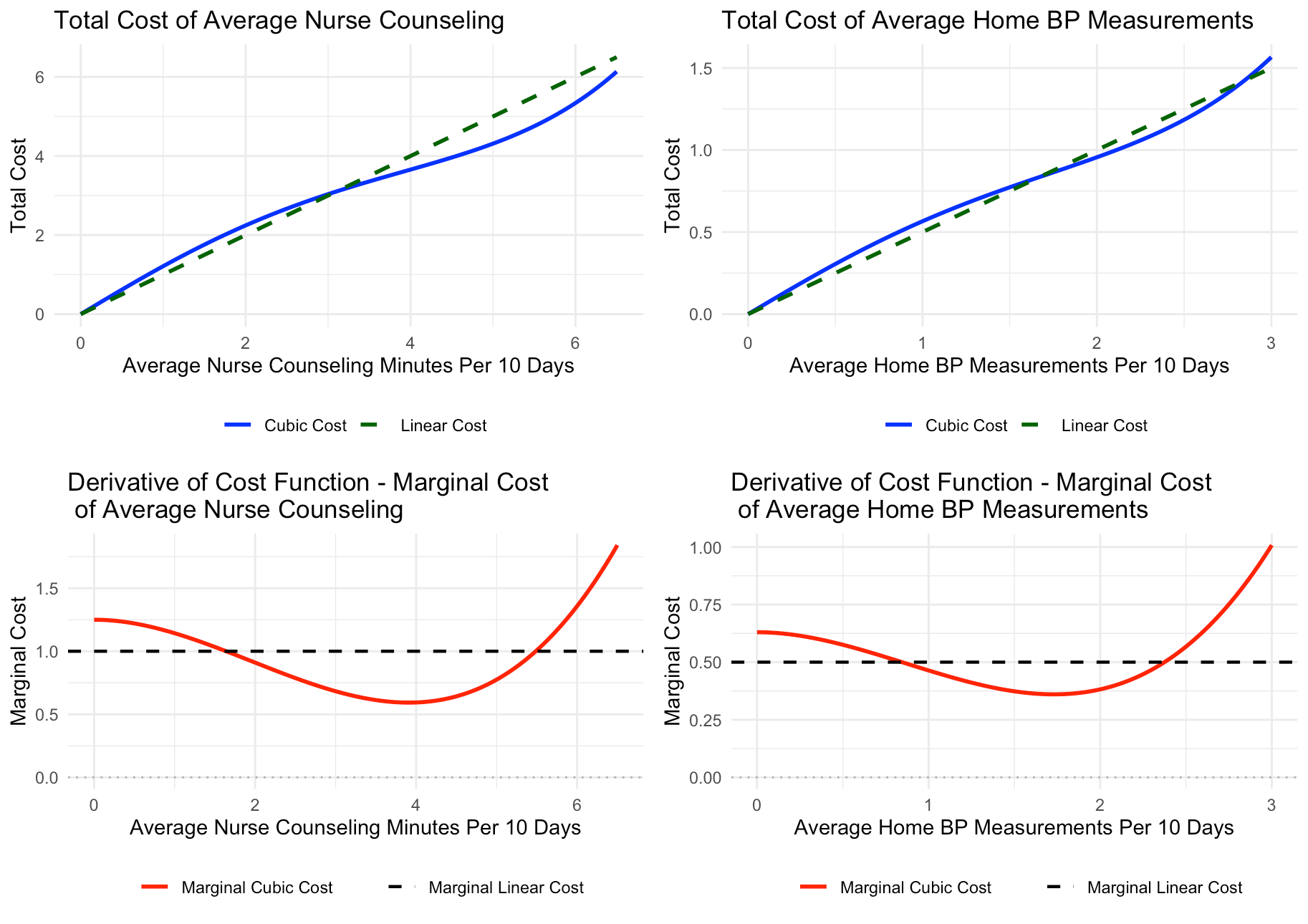}
    \caption{Total (top) and marginal (bottom) costs of the cubic cost function for each intervention component, nurse counseling (left) and home BP measurement (right).}
    \label{fig:cvd_cubic_cost}
\end{figure}

\subsection{Additional simulation results on EXTRA-CVD}\label{sim_outcome_goal_5}
Table~\ref{combined_simulations_goal5} shows the 
simulation results for scenarios 1 and 2 with a linear cost function and an outcome goal of \mbox{$-5$~mmHg}.
\begin{table}[htbp]
\caption{Simulation study results comparing scenarios with and without LAGO optimization for intervention components and optimal intervention package with a linear cost function and an outcome goal of -5 mmHg.}
\label{combined_simulations_goal5}
\fontsize{8pt}{10pt}\selectfont
\begin{tabular}{llllllllll}
\multicolumn{10}{l}{\textbf{(a) Individual intervention components}} \\
 & & & & \multicolumn{3}{c}{Scenario 1 (With LAGO)} & \multicolumn{3}{c}{Scenario 2 (Without LAGO)} \\
$\boldsymbol{\beta}^{*} = (\beta^{*}_1, \beta^{*}_2)$ & $n^{(1)}$ & $n^{(2)}$ & $\hat{\boldsymbol{\beta}}$ & \%RelBias & SE/EMP.SD & CP95 & \%RelBias & SE/EMP.SD & CP95 \\ \hline
\multirow{2}{*}{$(-1.59, -0.59)$} & \multirow{2}{*}{238} & \multirow{2}{*}{300} & $\hat{\beta}_1$ & -1.28 & 93.2 & 94.5 & 1.86 & 100.0 & 95.2 \\
 & &  & $\hat{\beta}_2$ & 19.40 & 91.1 & 94.4 & -8.04 & 101.0 & 95.3 \\
\end{tabular}
\fontsize{8pt}{10pt}\selectfont
\begin{tabular}{lllll}
\multicolumn{5}{l}{\textbf{(b) Estimated optimal intervention}} \\[0.5ex]
$\boldsymbol{x}^{\text{opt}}$ & Metrics &
\multicolumn{1}{c}{\shortstack[c]{Stage 1\\($n^{(1)} = 238$)}} &
\multicolumn{2}{c}{\shortstack[c]{Stage 2\\($n^{(1)} + n^{(2)} = 528$)}} \\
\cmidrule(lr){3-5}
 &  &  & \shortstack[c]{Scenario 1\\(With LAGO)} & \shortstack[c]{Scenario 2\\(Without LAGO)} \\
\midrule
\multirow{3}{*}{(3.15, 0.00)} 
 & Bias of $\hat{\boldsymbol{x}}^{\text{opt}}_1$ & -1.75 & -1.16 & -1.57 \\
 & Bias of $\hat{\boldsymbol{x}}^{\text{opt}}_2$ & 1.04 & 1.09 & 1.12 \\
 & rMSE & 2.71 & 2.44 & 2.61 \\
\end{tabular}
\\
\begin{tabular}{llll}
\multicolumn{4}{l}{\textbf{(c) Estimated optimal intervention, confidence set, and confidence band}} \\[1ex]
Metrics &
\multicolumn{1}{c}{\shortstack[c]{Stage 1}} &
\multicolumn{2}{c}{\shortstack[c]{Stage 2}} \\
\cline{2-4}\\[-0.8ex]
 &  & \shortstack[c]{Scenario 1 \\ (With LAGO)} & \shortstack[c]{Scenario 2\\(Without LAGO)} \\[0.5ex]
\hline\\[-0.8ex]
ExpectedOutActInt & -5.69 & -3.37 & -5.71 \\
ExpectedOutRecInt & -5.69 & -2.81 & -5.69 \\
ExpectedOutEstOptInt & --$^1$ & -3.78 & -3.15 \\
MeanOpt (Q2.5, Q97.5) & (-7.93,\,-0.45) & (-7.77,\,-0.86)$^*$ & (-6.34,\,-0.75)$^*$ \\
AvgObsOut & -5.83 & -3.35 & -5.68 \\
MeanCostAct & --$^2$ & 2.28 & 3.81 \\
MeanCostRec & 1.91 & 2.53 & 1.91 \\
SetCP95 (\%) & --$^3$ & 94.9 & 95.4 \\
SetPerc (\%) & --$^3$ & 70.7 & 70.4 \\
BandsCP95 (\%) & --$^3$ & 98.7 & 98.8 \\
\hline
\end{tabular}
\begin{tablenotes}[flushleft]
\footnotesize{
\fontsize{8pt}{10pt}\selectfont
\item See Table \ref{simulationresults_linear_varying_Zj_lowconf} and \ref{combined_simulations_goal9_adjust} footnotes for definitions.
}
\end{tablenotes}
\end{table}

Table \ref{combined_simulations_goal5}a presents the relative bias percentages, ratios of standard errors to empirical standard deviations, and coverage probabilities for the confidence intervals for $\beta_1$ and $\beta_2$ based on the combined data from stages 1 and 2. All parameter estimates demonstrated relative bias across all scenarios, but the 95\% coverage probabilities closely approximated the nominal 95\% level regardless of whether LAGO optimization was used to determine the stage 2 recommended~interventions.

Table \ref{combined_simulations_goal5}b summarizes the accuracy of the estimated optimal intervention components. The true optimal intervention was $(3.15, 0.00)$. Overall, LAGO achieved smaller bias and root mean squared error (rMSE) than non-LAGO. For stage 2, LAGO reduced the rMSE of the estimated optimal intervention from 2.71 to 2.44, while non-LAGO achieved 2.61. Both methods reduced bias between stages, with greater improvement for LAGO.

Table \ref{combined_simulations_goal5}c compares the expected and actual outcomes and summarizes confidence set results across scenarios. The expected outcome under the actual stage 2 interventions was -3.37 mmHg for LAGO compared to -5.71 mmHg for non-LAGO, with actual interventions under LAGO had lower total cost (mean cost: 2.28 minutes with LAGO and 3.81 minutes without LAGO). The expected outcome under the recommended interventions was -2.81 mmHg for LAGO versus -5.69 mmHg for non-LAGO. The expected outcome under the final estimated optimal interventions was -3.78 mmHg for LAGO compared to -3.15 mmHg for non-LAGO. Both methods maintained good confidence set coverage (LAGO: 94.9\%; non-LAGO: 95.4\%) and confidence band coverage (LAGO: 98.7\%; non-LAGO: 98.8\%), with similar confidence set sizes (approximately 70\% of the intervention space). These results reflect the challenges of adaptive optimization in settings with limited stage 1 sample size ($n^{(1)} = 238$) and high outcome variability (SD $\approx$ 20 mmHg). In this high-noise setting, stage 1 coefficient estimates exhibited substantial variability, leading to LAGO recommending suboptimal stage 2 intervention packages. This limitation is evident in the performance of traditional LAGO, which substaintially worsened outcomes to -3.37 mmHg compared to -5.69 mmHg for non-LAGO, highlighting the vulnerability of LAGO to stage 1 noise when outcome goals are already within reach.
\subsection{Additional simulation results on EXTRA-CVD}\label{sim_outcome_5_lowerbound}

\begin{table}[H]
\singlespacing
\caption{Simulation study results comparing scenarios with and without LAGO optimization for intervention components and optimal intervention package with a linear cost function and an outcome goal of -5 mmHg, with lower bound equal to stage 1 recommended interventions.}
\label{combined_simulations_goal5_adjust}
\fontsize{8pt}{10pt}\selectfont
\begin{tabular}{llllllllll}
\multicolumn{10}{l}{\textbf{(a) Individual intervention components}} \\
 & & & & \multicolumn{3}{c}{Scenario 1 (With LAGO)} & \multicolumn{3}{c}{Scenario 2 (Without LAGO)} \\
$\boldsymbol{\beta}^{*} = (\beta^{*}_1, \beta^{*}_2)$ & $n^{(1)}$ & $n^{(2)}$ & $\hat{\boldsymbol{\beta}}$ & \%RelBias & SE/EMP.SD & CP95 & \%RelBias & SE/EMP.SD & CP95 \\ \hline
\multirow{2}{*}{$(-1.59, -0.59)$} & \multirow{2}{*}{238} & \multirow{2}{*}{300} & $\hat{\beta}_1$ & 0.75 & 96.0 & 94.3 & 1.86 & 100.0 & 95.2 \\
 & &  & $\hat{\beta}_2$ & 1.30 & 96.0 & 94.4 & -8.04 & 101.0 & 95.3 \\
\end{tabular}
\fontsize{8pt}{10pt}\selectfont
\begin{tabular}{lllll}
\multicolumn{5}{l}{\textbf{(b) Estimated optimal intervention}} \\[0.5ex]
$\boldsymbol{x}^{\text{opt}}$ & Metrics &
\multicolumn{1}{c}{\shortstack[c]{Stage 1\\($n^{(1)} = 238$)}} &
\multicolumn{2}{c}{\shortstack[c]{Stage 2\\($n^{(1)} + n^{(2)} = 528$)}} \\
\cmidrule(lr){3-5}
 &  &  & \shortstack[c]{Scenario 1\\(With LAGO)} & \shortstack[c]{Scenario 2\\(Without LAGO)} \\
\midrule
\multirow{3}{*}{(3.15, 0.00)} 
 & Bias of $\hat{\boldsymbol{x}}^{\text{opt}}_1$ & 0.04 & -0.06 & -0.07 \\
 & Bias of $\hat{\boldsymbol{x}}^{\text{opt}}_2$ & 1.73 & 1.70 & 1.71 \\
 & rMSE & 1.82 & 1.73 & 1.74 \\
\end{tabular}
\fontsize{8pt}{10pt}\selectfont
\begin{tabular}{llll}
\multicolumn{4}{l}{\textbf{(c) Estimated optimal intervention, confidence set, and confidence band}} \\[1ex]
Metrics &
\multicolumn{1}{c}{\shortstack[c]{Stage 1}} &
\multicolumn{2}{c}{\shortstack[c]{Stage 2}} \\
\cline{2-4}\\
 &  & \shortstack[c]{Scenario 1 \\ (With LAGO)} & \shortstack[c]{Scenario 2\\(Without LAGO)} \\[0.5ex]
\hline\\[-0.8ex]
ExpectedOutActInt & -5.69 & -6.07 & -5.71 \\
ExpectedOutRecInt & -5.69 & -6.06 & -5.69 \\
ExpectedOutEstOptInt & --$^1$ & -5.89 & -5.87 \\
MeanOpt (Q2.5, Q97.5) & (-8.87,\,-5.69) & (-7.28,\,-5.69)$^*$ & (-7.25,\,-5.69)$^*$ \\
AvgObsOut & -5.83 & -6.05 & -5.68 \\
MeanCostAct & --$^2$ & 4.05 & 3.81 \\
MeanCostRec & 4.05 & 3.94 & 4.05 \\
SetCP95 (\%) & --$^3$ & 94.7 & 95.4 \\
SetPerc (\%) & --$^3$ & 68.9 & 70.4 \\
BandsCP95 (\%) & --$^3$ & 98.6 & 98.8 \\
\hline
\end{tabular}
\begin{tablenotes}[flushleft]
\fontsize{8pt}{10pt}\selectfont
\item See Table \ref{simulationresults_linear_varying_Zj_lowconf} and \ref{combined_simulations_goal9_adjust} footnotes for definitions.
\end{tablenotes}
\end{table}

\subsection{Additional simulation results on EXTRA-CVD}\label{sim_outcome_goal_9}
\begin{table}[htbp]
\singlespacing
\caption{Simulation study results comparing scenarios with and without LAGO optimization for intervention components and optimal intervention package with a linear cost function and an outcome goal of -9 mmHg.}
\label{combined_simulations_goal9}
\fontsize{8pt}{10pt}\selectfont
\begin{tabular}{llllllllll}
\multicolumn{10}{l}{\textbf{(a) Individual intervention components}} \\
 & & & & \multicolumn{3}{c}{Scenario 1 (With LAGO)} & \multicolumn{3}{c}{Scenario 2 (Without LAGO)} \\
$\boldsymbol{\beta}^{*} = (\beta^{*}_1, \beta^{*}_2)$ & $n^{(1)}$ & $n^{(2)}$ & $\hat{\boldsymbol{\beta}}$ & \%RelBias & SE/EMP.SD & CP95 & \%RelBias & SE/EMP.SD & CP95 \\ \hline
\multirow{2}{*}{$(-1.59, -0.59)$} & \multirow{2}{*}{238} & \multirow{2}{*}{300} & $\hat{\beta}_1$ & -0.77 & 96.6 & 94.9 & 1.86 & 100.0 & 95.2 \\
 & &  & $\hat{\beta}_2$ & 17.7 & 101.8 & 95.7 & -8.04 & 101.0 & 95.3 \\
\end{tabular}
\\[0.5em]
\fontsize{8pt}{10pt}\selectfont
\begin{tabular}{lllll}
\multicolumn{5}{l}{\textbf{(b) Estimated optimal intervention}} \\[0.5ex]
$\boldsymbol{x}^{\text{opt}}$ & Metrics &
\multicolumn{1}{c}{\shortstack[c]{Stage 1\\($n^{(1)} = 238$)}} &
\multicolumn{2}{c}{\shortstack[c]{Stage 2\\($n^{(1)} + n^{(2)} = 528$)}} \\
\cmidrule(lr){3-5}
 &  &  & \shortstack[c]{Scenario 1\\(With LAGO)} & \shortstack[c]{Scenario 2\\(Without LAGO)} \\
\midrule
\multirow{3}{*}{(5.67, 0.00)} 
 & Bias of $\hat{\boldsymbol{x}}^{\text{opt}}_1$ & -2.59 & -1.28 & -1.89 \\
 & Bias of $\hat{\boldsymbol{x}}^{\text{opt}}_2$ & 1.35 & 1.25 & 1.35 \\
 & rMSE & 3.69 & 2.62 & 3.18 \\
\end{tabular}
\\[0.5em]
\fontsize{8pt}{10pt}\selectfont
\begin{tabular}{llll}
\multicolumn{4}{l}{\textbf{(c) Estimated optimal intervention, confidence set, and confidence band}} \\[1ex]
Metrics &
\multicolumn{1}{c}{\shortstack[c]{Stage 1}} &
\multicolumn{2}{c}{\shortstack[c]{Stage 2}} \\[0.5ex]
\cline{2-4}\\[-0.8ex]
 &  & \shortstack[c]{Scenario 1 \\ (With LAGO)} & \shortstack[c]{Scenario 2\\(Without LAGO)} \\[0.5ex]
\hline\\[-0.8ex]
ExpectedOutActInt & -5.69 & -5.93 & -5.71 \\[0.8ex]
ExpectedOutRecInt & -5.69 & -5.66 & -5.69 \\[0.8ex]
ExpectedOutEstOptInt & --$^1$ & -7.69 & -6.77 \\[0.8ex]
MeanOpt (Q2.5, Q97.5) & (-11.1,\,-0.98) & (-11.5,\,-3.47)$^*$ & (-11.2,\,-1.39)$^*$ \\[0.8ex]
AvgObsOut & -5.83 & -5.91 & -5.68\\[0.8ex]
MeanCostAct & --$^2$ & 3.93 & 3.81\\[0.8ex]
MeanCostRec & 3.75 & 5.01 & 3.75 \\[0.8ex]
SetCP95 (\%) & --$^3$ & 95.2 & 95.4 \\[0.8ex]
SetPerc (\%) & --$^3$ & 53.3 & 58.1 \\[0.8ex]
BandsCP95 (\%) & --$^1$ & 98.9 & 98.8 \\[0.8ex]
\hline
\end{tabular}
\vspace{0.5em}
\begin{tablenotes}[flushleft]
\fontsize{8pt}{10pt}\selectfont
\item See Table \ref{simulationresults_linear_varying_Zj_lowconf} and \ref{combined_simulations_goal9_adjust} footnotes for definitions.
\end{tablenotes}
\end{table}
Table \ref{combined_simulations_goal9} examines the performance of LAGO with a more ambitious goal of -9 mmHg (true optimal intervention: $(5.67, 0.00)$). This is a scenario where LAGO can be expected to help avoid a failed trial. LAGO exhibited percent coefficient bias of -0.77\% for $\beta_1$ and 17.7\% for $\beta_2$, somewhat lower than with the -5 mmHg outcome goal. The average outcome under the actual stage 2 interventions was -5.93 mmHg for LAGO versus -5.71 mmHg for non-LAGO, showing minimal outcome differences between the two approaches. Stage 2 optimization accuracy (rMSE: 2.62 for LAGO versus 3.18 for non-LAGO) favored LAGO. 

\subsection{Additional information on EXTRA-CVD: Complete results using the cubic cost function}\label{extracvdsimulations_cubic}
\FloatBarrier
\begin{table}[htbp]
\caption{Simulation study results comparing with and without LAGO optimization for intervention components and optimal intervention package with a cubic cost function and an outcome goal of -5 mmHg.}
\label{combined_simulations_goal5_cubic}
\fontsize{8pt}{10pt}\selectfont
\begin{tabular}{llllllllll}
\multicolumn{10}{l}{\textbf{(a) Individual intervention component}} \\
 & & & & \multicolumn{3}{c}{Scenario 1 (With LAGO)} & \multicolumn{3}{c}{Scenario 2 (Without LAGO)} \\
$\boldsymbol{\beta}^{*} = (\beta^{*}_1, \beta^{*}_2)$ & $n^{(1)}$ & $n^{(2)}$ & $\hat{\boldsymbol{\beta}}$ & \%RelBias & SE/EMP.SD & CP95 & \%RelBias & SE/EMP.SD & CP95 \\ \hline
\multirow{2}{*}{$(-1.59, -0.59)$} & \multirow{2}{*}{238} & \multirow{2}{*}{300} & $\hat{\beta}_1$ & -1.53 & 92.8 & 94.4 & 1.86 & 100.0 & 95.2 \\
 & &  & $\hat{\beta}_2$ & 20.72 & 91.6 & 94.3 & -8.04 & 101.0 & 95.3 \\
\end{tabular}
\\[1em]
\fontsize{8pt}{10pt}\selectfont
\begin{tabular}{lllll}
\multicolumn{5}{l}{\textbf{(b) Estimated optimal intervention}} \\[0.5ex]
$\boldsymbol{x}^{\text{opt}}$ & Metrics &
\multicolumn{1}{c}{\shortstack[c]{Stage 1\\($n^{(1)} = 238$)}} &
\multicolumn{2}{c}{\shortstack[c]{Stage 2\\($n^{(1)} + n^{(2)} = 528$)}} \\
\cmidrule(lr){3-5}
 &  &  & \shortstack[c]{Scenario 1\\(With LAGO)} & \shortstack[c]{Scenario 2\\(Without LAGO)} \\
\midrule
\multirow{3}{*}{(3.20, 0.00)} 
 & Bias of $\hat{\boldsymbol{x}}^{\text{opt}}_1$ & -1.78 & -1.14 & -1.57 \\
 & Bias of $\hat{\boldsymbol{x}}^{\text{opt}}_2$ & 1.00 & 0.96 & 1.04 \\
 & rMSE & 2.71 & 2.37 & 2.59 \\
\end{tabular}
\\[0.7em]
\fontsize{8pt}{10pt}\selectfont
\begin{tabular}{llll}
\multicolumn{4}{l}{\textbf{(c) Estimated optimal intervention, confidence set, and confidence band}} \\[1ex]
Metrics &
\multicolumn{1}{c}{\shortstack[c]{Stage 1}} &
\multicolumn{2}{c}{\shortstack[c]{Stage 2}} \\[0.5ex]
\cline{2-4}\\[-0.8ex]
 &  & \shortstack[c]{Scenario 1 \\ (With LAGO)} & \shortstack[c]{Scenario 2\\(Without LAGO)} \\[0.5ex]
\hline\\[-0.8ex]
ExpectedOutActInt & -5.69 & -3.39 & -5.71 \\[0.8ex]
ExpectedOutRecInt & -5.69 & -2.83 & -5.69 \\[0.8ex]
ExpectedOutEstOptInt & --$^1$ & -3.83 & -3.19 \\[0.8ex]
MeanOpt (Q2.5, Q97.5) & (-7.93,\,-0.45) & (-7.89,\,-0.87)$^*$ & (-6.62,\,-0.75)$^*$ \\[1.0ex]
AvgObsOut & -5.83 & -3.36 & -5.68\\[0.8ex]
MeanCostAct & --$^2$ & 2.29 & 3.81 \\[0.8ex]
MeanCostRec & 2.00 & 2.61 & 2.00 \\[0.8ex]
SetCP95 (\%) & --$^3$ & 94.6 & 95.4 \\[0.8ex]
SetPerc (\%) & --$^3$ & 70.6 & 70.4 \\[0.8ex]
BandsCP95 (\%) & --$^3$ & 98.7 & 98.8 \\[0.8ex]
\hline
\end{tabular}
\vspace{0.3em}
\begin{tablenotes}[flushleft]
\footnotesize{
\fontsize{8pt}{10pt}\selectfont
\item See Table \ref{simulationresults_linear_varying_Zj_lowconf} and \ref{combined_simulations_goal9_adjust} footnotes for definitions. The cubic cost function used here was $C(\boldsymbol{x}) = 1.25x_1 - 0.043x_1^3 + 0.0055x_1^4 + 0.63x_2 - 0.09x_2^3 + 0.026x_2^4$.
}
\end{tablenotes}
\end{table}
\begin{table}[htbp]
\caption{Simulation study results comparing with and without LAGO optimization for intervention components and optimal intervention package with a cubic cost function and an outcome goal of -5 mmHg, with lower bound adjustment}
\label{combined_simulations_goal5_cubic_adjust}
\fontsize{8pt}{10pt}\selectfont
\begin{tabular}{llllllllll}
\multicolumn{10}{l}{\textbf{(a) Individual intervention component}} \\
 & & & & \multicolumn{3}{c}{Scenario 1 (With LAGO)} & \multicolumn{3}{c}{Scenario 2 (Without LAGO)} \\
$\boldsymbol{\beta}^{*} = (\beta^{*}_1, \beta^{*}_2)$ & $n^{(1)}$ & $n^{(2)}$ & $\hat{\boldsymbol{\beta}}$ & \%RelBias & SE/EMP.SD & CP95 & \%RelBias & SE/EMP.SD & CP95 \\ \hline
\multirow{2}{*}{$(-1.59, -0.59)$} & \multirow{2}{*}{238} & \multirow{2}{*}{300} & $\hat{\beta}_1$ & 0.67 & 95.9 & 94.3 & 1.86 & 100.0 & 95.2 \\
 & &  & $\hat{\beta}_2$ & 1.66 & 96.1 & 94.4 & -8.04 & 101.0 & 95.3 \\
\end{tabular}
\\[0.5em]
\fontsize{8pt}{10pt}\selectfont
\begin{tabular}{lllll}
\multicolumn{5}{l}{\textbf{(b) Estimated optimal intervention}} \\[0.5ex]
$\boldsymbol{x}^{\text{opt}}$ & Metrics &
\multicolumn{1}{c}{\shortstack[c]{Stage 1\\($n^{(1)} = 238$)}} &
\multicolumn{2}{c}{\shortstack[c]{Stage 2\\($n^{(1)} + n^{(2)} = 528$)}} \\
\cmidrule(lr){3-5}
 &  &  & \shortstack[c]{Scenario 1\\(With LAGO)} & \shortstack[c]{Scenario 2\\(Without LAGO)} \\
\midrule
\multirow{3}{*}{(3.20, 0.00)} 
 & Bias of $\hat{\boldsymbol{x}}^{\text{opt}}_1$ & -0.01 & -0.11 & -0.12 \\
 & Bias of $\hat{\boldsymbol{x}}^{\text{opt}}_2$ & 1.72 & 1.70 & 1.69 \\
 & rMSE & 1.82 & 1.73 & 1.73 \\
\end{tabular}
\\[0.5em]
\fontsize{8pt}{10pt}\selectfont
\begin{tabular}{llll}
\multicolumn{4}{l}{\textbf{(c) Estimated optimal intervention, confidence set, and confidence band}} \\[1ex]
Metrics &
\multicolumn{1}{c}{\shortstack[c]{Stage 1}} &
\multicolumn{2}{c}{\shortstack[c]{Stage 2}} \\[0.5ex]
\cline{2-4}\\[-0.8ex]
 &  & \shortstack[c]{Scenario 1 \\ (With LAGO)} & \shortstack[c]{Scenario 2\\(Without LAGO)} \\[0.5ex]
\hline\\[-0.8ex]
ExpectedOutActInt & -5.69 & -6.07 & -5.71 \\[0.8ex]
ExpectedOutRecInt & -5.69 & -6.06 & -5.69 \\[0.8ex]
ExpectedOutEstOptInt & --$^1$ & -5.89 & -5.88 \\[0.8ex]
MeanOpt (Q2.5, Q97.5) & (-8.93,\,-5.69) & (-7.44,\,-5.69)$^*$ & (-7.28,\,-5.69)$^*$ \\[0.8ex]
AvgObsOut & -5.83 & -6.05 & -5.68 \\[0.8ex]
MeanCostAct & --$^2$ & 4.06 & 3.81\\[0.8ex]
MeanCostRec & 4.02 & 3.94 & 4.02 \\[0.8ex]
SetCP95 (\%) & --$^3$ & 94.7 & 95.4 \\[0.8ex]
SetPerc (\%) & --$^3$ & 68.9 & 70.4 \\[0.8ex]
BandsCP95 (\%) & --$^3$ & 98.6 & 98.8 \\[0.8ex]
\hline
\end{tabular}
\vspace{0.3em}
\begin{tablenotes}[flushleft]
\fontsize{8pt}{10pt}\selectfont
\item See Table \ref{simulationresults_linear_varying_Zj_lowconf} and \ref{combined_simulations_goal9_adjust} footnotes for definitions. The cubic cost function used here was $C(\boldsymbol{x}) = 1.25x_1 - 0.043x_1^3 + 0.0055x_1^4 + 0.63x_2 - 0.09x_2^3 + 0.026x_2^4$.
\end{tablenotes}
\end{table}

\begin{table}[htbp]
\caption{Simulation study results comparing scenarios with and without LAGO optimization for intervention components and optimal intervention package with a cubic cost function and an outcome goal of -9 mmHg.}
\label{combined_simulations_goal9_cubic}
\fontsize{8pt}{10pt}\selectfont
\begin{tabular}{llllllllll}
\multicolumn{10}{l}{\textbf{(a) Individual intervention component}} \\
 & & & & \multicolumn{3}{c}{Scenario 1 (With LAGO)} & \multicolumn{3}{c}{Scenario 2 (Without LAGO)} \\
$\boldsymbol{\beta}^{*} = (\beta^{*}_1, \beta^{*}_2)$ & $n^{(1)}$ & $n^{(2)}$ & $\hat{\boldsymbol{\beta}}$ & \%RelBias & SE/EMP.SD & CP95 & \%RelBias & SE/EMP.SD & CP95 \\ \hline
\multirow{2}{*}{$(-1.59, -0.59)$} & \multirow{2}{*}{238} & \multirow{2}{*}{300} & $\hat{\beta}_1$ & -0.38 & 96.3 & 95.0 & 1.86 & 100.0 & 95.2 \\
 & &  & $\hat{\beta}_2$ & 13.47 & 101.8 & 95.7 & -8.04 & 101.0 & 95.3 \\
\end{tabular}
\\[0.5em]
\fontsize{8pt}{10pt}\selectfont
\begin{tabular}{lllll}
\multicolumn{5}{l}{\textbf{(b) Estimated optimal intervention}} \\[0.5ex]
$\boldsymbol{x}^{\text{opt}}$ & Metrics &
\multicolumn{1}{c}{\shortstack[c]{Stage 1\\($n^{(1)} = 238$)}} &
\multicolumn{2}{c}{\shortstack[c]{Stage 2\\($n^{(1)} + n^{(2)} = 528$)}} \\
\cmidrule(lr){3-5}
 &  &  & \shortstack[c]{Scenario 1\\(With LAGO)} & \shortstack[c]{Scenario 2\\(Without LAGO)} \\
\midrule
\multirow{3}{*}{(5.67, 0.00)} 
 & Bias of $\hat{\boldsymbol{x}}^{\text{opt}}_1$ & -2.58 & -1.28 & -1.88 \\
 & Bias of $\hat{\boldsymbol{x}}^{\text{opt}}_2$ & 1.28 & 1.16 & 1.28 \\
 & rMSE & 3.66 & 2.54 & 3.12 \\
\end{tabular}
\\[0.5em]
\fontsize{8pt}{10pt}\selectfont
\begin{tabular}{llll}
\multicolumn{4}{l}{\textbf{(c) Estimated optimal intervention, confidence set, and confidence band}} \\[1ex]
Metrics &
\multicolumn{1}{c}{\shortstack[c]{Stage 1}} &
\multicolumn{2}{c}{\shortstack[c]{Stage 2}} \\[0.5ex]
\cline{2-4}\\[-0.8ex]
 &  & \shortstack[c]{Scenario 1 \\ (With LAGO)} & \shortstack[c]{Scenario 2\\(Without LAGO)} \\[0.5ex]
\hline\\[-0.8ex]
ExpectedOutActInt & -5.69 & -5.97 & -5.71 \\[0.8ex]
ExpectedOutRecInt & -5.69 & -5.70 & -5.69 \\[0.8ex]
ExpectedOutEstOptInt & --$^1$ & -7.69 & -6.81 \\[0.8ex]
MeanOpt (Q2.5, Q97.5) & (-11.1,\,-0.98) & (-11.4,\,-3.18)$^*$ & (-11.2,\,-1.39)$^*$ \\[0.8ex]
AvgObsOut & -5.83 & -5.95 & -5.68\\[0.8ex]
MeanCostAct & --$^2$ & 3.95 & 3.81\\[0.8ex]
MeanCostRec & 3.57 & 4.56 & 3.57 \\[0.8ex]
SetCP95 (\%) & --$^3$ & 95.4 & 95.4 \\[0.8ex]
SetPerc (\%) & --$^3$ & 53.6 & 58.1 \\[0.8ex]
BandsCP95 (\%) & --$^3$ & 98.9 & 98.8 \\[0.8ex]
\hline
\end{tabular}
\vspace{0.3em}
\begin{tablenotes}[flushleft]
\fontsize{8pt}{10pt}\selectfont
\item See Table \ref{simulationresults_linear_varying_Zj_lowconf} and \ref{combined_simulations_goal9_adjust} footnotes for definitions. The cubic cost function used here was $C(\boldsymbol{x}) = 1.25x_1 - 0.043x_1^3 + 0.0055x_1^4 + 0.63x_2 - 0.09x_2^3 + 0.026x_2^4$.
\end{tablenotes}
\end{table}

\begin{table}[htbp]
\caption{Simulation study results comparing with and without LAGO optimization for intervention components and optimal intervention package with a cubic cost function and an outcome goal of -9 mmHg, with lower bound adjustment}
\label{combined_simulations_goal9_cubic_adjust}
\fontsize{8pt}{10pt}\selectfont
\begin{tabular}{llllllllll}
\multicolumn{10}{l}{\textbf{(a) Individual intervention component}} \\
 & & & & \multicolumn{3}{c}{Scenario 1 (With LAGO)} & \multicolumn{3}{c}{Scenario 2 (Without LAGO)} \\
$\boldsymbol{\beta}^{*} = (\beta^{*}_1, \beta^{*}_2)$ & $n^{(1)}$ & $n^{(2)}$ & $\hat{\boldsymbol{\beta}}$ & \%RelBias & SE/EMP.SD & CP95 & \%RelBias & SE/EMP.SD & CP95 \\ \hline
\multirow{2}{*}{$(-1.59, -0.59)$} & \multirow{2}{*}{238} & \multirow{2}{*}{300} & $\hat{\beta}_1$ & -0.77 & 96.5 & 94.7 & 1.86 & 100.0 & 95.2 \\
 & &  & $\hat{\beta}_2$ & 30.3 & 99.2 & 94.3 & -8.04 & 101.0 & 95.3 \\
\end{tabular}
\\[0.5em]
\fontsize{8pt}{10pt}\selectfont
\begin{tabular}{lllll}
\multicolumn{5}{l}{\textbf{(b) Estimated optimal intervention}} \\[0.5ex]
$\boldsymbol{x}^{\text{opt}}$ & Metrics &
\multicolumn{1}{c}{\shortstack[c]{Stage 1\\($n^{(1)} = 238$)}} &
\multicolumn{2}{c}{\shortstack[c]{Stage 2\\($n^{(1)} + n^{(2)} = 528$)}} \\
\cmidrule(lr){3-5}
 &  &  & \shortstack[c]{Scenario 1\\(With LAGO)} & \shortstack[c]{Scenario 2\\(Without LAGO)} \\
\midrule
\multirow{3}{*}{(5.67, 0.00)} 
 & Bias of $\hat{\boldsymbol{x}}^{\text{opt}}_1$ & -1.64 & -1.14 & -1.33 \\
 & Bias of $\hat{\boldsymbol{x}}^{\text{opt}}_2$ & 2.02 & 2.07 & 2.09 \\
 & rMSE & 2.84 & 2.61 & 2.72 \\
\end{tabular}
\\[0.5em]
\fontsize{8pt}{10pt}\selectfont
\begin{tabular}{llll}
\multicolumn{4}{l}{\textbf{(c) Estimated optimal intervention, confidence set, and confidence band}} \\[1ex]
Metrics &
\multicolumn{1}{c}{\shortstack[c]{Stage 1}} &
\multicolumn{2}{c}{\shortstack[c]{Stage 2}} \\[0.5ex]
\cline{2-4}\\[-0.8ex]
 &  & \shortstack[c]{Scenario 1 \\ (With LAGO)} & \shortstack[c]{Scenario 2\\(Without LAGO)} \\[0.5ex]
\hline\\[-0.8ex]
ExpectedOutActInt & -5.69 & -7.63 & -5.71 \\[0.8ex]
ExpectedOutRecInt & -5.69 & -7.63 & -5.69 \\[0.8ex]
ExpectedOutEstOptInt & -- & -8.45 & -8.17 \\[0.8ex]
MeanOpt (Q2.5, Q97.5) & (-11.3,\,-5.69) & (-11.4,\,-5.99)$^*$ & (-11.3,\,-5.93)$^*$ \\[1.0ex]
AvgObsOut & -5.83 & -7.61 & -5.68\\[0.8ex]
MeanCostAct & -- & 5.08 & 3.81 \\[0.8ex]
MeanCostRec & 4.76 & 5.13 & 4.76 \\[0.8ex]
SetCP95 (\%) & -- & 94.7 & 95.4 \\[0.8ex]
SetPerc (\%) & -- & 54.2 & 58.1 \\[0.8ex]
BandsCP95 (\%) & -- & 98.9 & 98.8 \\[0.8ex]
\hline
\end{tabular}
\vspace{0.3em}
\begin{tablenotes}[flushleft]
\fontsize{8pt}{10pt}\selectfont
\item See Table \ref{simulationresults_linear_varying_Zj_lowconf} and \ref{combined_simulations_goal9_adjust} footnotes for definitions. The cubic cost function used here was~$C(\boldsymbol{x}) = 1.25x_1 - 0.043x_1^3 + 0.0055x_1^4 + 0.63x_2 - 0.09x_2^3 + 0.026x_2^4$.
\end{tablenotes}
\end{table}
\FloatBarrier

\end{document}